\pdfoutput= 1 
\RequirePackage{fixltx2e}
\documentclass[11pt]{article}  
\usepackage[utf8]{inputenc}
\usepackage[T1]{fontenc}
\usepackage{amsmath, amsfonts, amsthm, amssymb, amscd, aliascnt}
\usepackage[normalem]{ulem}     
\usepackage{relsize}
\usepackage{cancel}
\usepackage{slashed}
\usepackage{stmaryrd}
\usepackage{esint} 
\usepackage{slashed}
\usepackage{float} 
\usepackage{setspace}
\usepackage{titlesec}
\titlespacing*{\paragraph}{0pt}{1.25ex plus 1ex minus .2ex}{.5em}

\usepackage{subcaption} 
\usepackage{tikz}

\usepackage{bbm} 
\usepackage{mathrsfs} 
\usepackage{accents, verbatim}
\makeatletter
\@ifundefined{c@localmathalphabets}{}{\setcounter{localmathalphabets}{0}}
\makeatother
\input{epsf}
\usepackage{pdfpages, enumerate, graphicx, xcolor, array, booktabs, microtype}
\usepackage{fancyhdr}
\usepackage[nottoc]{tocbibind} 
\usepackage[pdfusetitle]{hyperref}
\hypersetup{linktoc = all}
\hypersetup{hidelinks}
\hypersetup{bookmarksnumbered}

\numberwithin{equation}{section}
\numberwithin{table}{section}
\numberwithin{figure}{section}
\usepackage{footnotebackref}

\usepackage{needspace,multirow,accents,comment,xparse}
\renewcommand\dot{\accentset{\bullet}}

\usepackage{amsmath}
\def\over{\csname @@over\endcsname} 

\usepackage{tikz}
\usetikzlibrary{shadings}

\newcommand{\refwithname}[2]{\hyperref[#2]{#1~\ref*{#2}}}
\newcommand{\appref}{\refwithname{Appendix}}

\theoremstyle{plain}
\newtheorem{theorem}{Theorem}[section]
\newcommand{\mynewtheorem}[2]{
  \newaliascnt{#1}{theorem}
  \newtheorem{#1}[#1]{#2}
  \aliascntresetthe{#1}
  \expandafter\providecommand\csname #1autorefname\endcsname{#2}
}
\mynewtheorem{definition}{Definition}
\mynewtheorem{proposition}{Proposition}
\mynewtheorem{lemma}{Lemma}
\mynewtheorem{corollary}{Corollary}
\mynewtheorem{remark}{Remark}
\mynewtheorem{claim}{Claim}

\providecommand{\coloneqq}{\mathrel{\mathop:}\mathrel{\mkern-1.2mu}=}
\providecommand{\eqqcolon}{=\mathrel{\mkern-1.2mu}\mathrel{\mathop:}}
\newcommand{\smallbullet}{{\scriptscriptstyle\mspace{.5mu}\bullet\mspace{.5mu}}}

\newcommand{\Hess}{\operatorname{Hess}}
\newcommand \Riem {{\operatorfont Riem}}
\newcommand \Ric {{\operatorfont Ric}}
\newcommand \vol {{\mathbf{vol}}}
\newcommand \Hstar {H_\star}
\newcommand \Mstar {M_\star}
\newcommand \la 	{\langle}
\newcommand \ra 	{\rangle}
\newcommand \dbf 	{{\mathbf d}}
\newcommand \Dcal 	{\mathcal D}
\newcommand \Gcal 	{\mathcal G}

\newcommand \bei 	{\begin{itemize}}
\newcommand \eei 	{\end{itemize}}
\newcommand \del	{\partial}

\newcommand \Fcal 	{\mathcal F}
\newcommand \Hcal 	{\mathcal H}
\newcommand \Lcal 	{\mathcal L}
\newcommand \Mcal 	{\mathcal M}
\newcommand \Qcal 	{\mathcal Q}
\newcommand \Rcal 	{\mathcal R}
\newcommand \RR 	{\mathbb R}

\newcommand \eps 	{\epsilon}
\newcommand \be 	{\begin{equation}\mathopen{}} 
\newcommand \ee 	{\end{equation}}
\newcommand \bel [1]	{\be \label{#1}\mathopen{}} 
\newcommand \Bcal 	{\mathcal B}

\newcommand \lambdabf {\boldsymbol \lambda}

\newcommand \Acal {\mathcal A}

\newcommand \Mbf 	{\mathbf M}
\newcommand \omegabf {\boldsymbol{\omega}}

\newcommand \rbf     {\mathbf r}
\newcommand \Jcal   {\mathcal J}
\newcommand{\Lie}{\mathcal L}

\newcommand \diag {{\bf diag}}
\newcommand \bse {\begin{subequations}}
\newcommand \ese {\ifcase\value{equation}\relax\GenericError{}{Subequation (\theequation) is empty}{}{}\or\GenericError{}{Subequation (\theequation) with single equation}{}{}\fi\end{subequations}}

\newcommand \Jbb {\mathbb J}
\newcommand \Nbb {\mathbb N}
\newcommand \Mbb {\mathbb M}
\DeclareMathOperator \Tr {\bf Tr}
\DeclareFontFamily{U}{matha}{\hyphenchar\font45}
\DeclareFontShape{U}{matha}{m}{n}{ <5> <6> <7> <8> <9> <10> gen * matha <10.95> matha10 <12> <14.4> <17.28> <20.74> <24.88> matha12 }{}
\DeclareSymbolFont{matha}{U}{matha}{m}{n}
\DeclareFontSubstitution{U}{matha}{m}{n}
\DeclareFontFamily{U}{mathx}{\hyphenchar\font45}
\DeclareFontShape{U}{mathx}{m}{n}{ <5> <6> <7> <8> <9> <10> <10.95> <12> <14.4> <17.28> <20.74> <24.88> mathx10 }{}
\DeclareSymbolFont{mathx}{U}{mathx}{m}{n}
\DeclareFontSubstitution{U}{mathx}{m}{n}
\DeclareMathAccent{\widecheck}{0}{mathx}{"71}
\DeclareMathAccent{\wideparen}{0}{mathx}{"75}

\DeclareMathDelimiter{\VV}{0}{matha}{"7E}{mathx}{"17} 
\ExplSyntaxOn
\NewDocumentCommand\norm{som}{%
  \IfBooleanTF{#1}{\left\lVert#3\right\rVert}%
    {\IfValueTF{#2}{\use:c{\cs_to_str:N#2l}\lVert#3\use:c{\cs_to_str:N#2r}\rVert}{\lVert#3\rVert}}}
\NewDocumentCommand\Norm{som}{%
  \IfBooleanTF{#1}{\left\VV #3\right\VV}%
    {\IfValueTF{#2}{\use:c{\cs_to_str:N#2l}\VV #3\use:c{\cs_to_str:N#2r}\VV}{\mathopen{\VV} #3\mathclose{\VV}}}}
\ExplSyntaxOff

\makeatletter
\providecommand\clap[1]{\hbox to 0pt{\hss#1\hss}}
\providecommand\mathclap[1]{\mathpalette\mathclap@aux{#1}}
\newcommand\mathclap@aux[2]{\clap{$#1#2$}}
\providecommand\mathllap[1]{\mathpalette\mathllap@aux{#1}}
\newcommand\mathllap@aux[2]{\llap{$#1#2$}}
\providecommand\mathrlap[1]{\mathpalette\mathrlap@aux{#1}}
\newcommand\mathrlap@aux[2]{\rlap{$#1#2$}}
\makeatother

\newcommand \Obig {\mathcal O}
\usepackage{calligra}
\DeclareMathAlphabet{\mathcalligra}{T1}{calligra}{m}{n}
\newsavebox{\calligrabox}
\newcommand \osmall {\scalebox{1.2}{%
    \mbox
      {%
        \sbox{\calligrabox}{$\mathcalligra{o}$}%
        \hskip\wd\calligrabox
        \pdfsave
        \pdfsetmatrix{1 0 -.7 1}%
        \llap{\usebox{\calligrabox}}%
        \pdfrestore
      }}}

\makeatletter
\newcommand \lot {{\bf l.o.t.}\@ifnextchar.\@gobble{}} 
\makeatother

\renewcommand \ker {\operatorname{\bf ker}}
\renewcommand \dim {\operatorname{dim}}


\newcommand{\putabove}[2]{%
  \ifx\Tr#2\!\fi
  \mathrlap{\overset{\hbox{\lower1.5pt\hbox{\smash{$#1$}}}}{\phantom{#2}}}#2}

\newcommand{\zg}{\putabove{\scriptscriptstyle 0}g}
\newcommand{\zh}{\putabove{\scriptscriptstyle 0}h}
\newcommand{\zvol}{\putabove{\scriptscriptstyle 0}{\dVol}}
\newcommand{\zR}{\putabove{\scriptscriptstyle 0}R}

\newcommand{\zTr}{\putabove{\scriptscriptstyle 0}{\Tr}}
\newcommand{\znabla}{\putabove{\scriptscriptstyle 0}\nabla}

\newcommand{\zDelta}{\putabove{\scriptscriptstyle 0}\Delta}
\newcommand{\zdG}{\putabove{\scriptscriptstyle 0}{d\Gcal}}
\newcommand{\zdH}{\putabove{\scriptscriptstyle 0}{d\Hcal}}
\newcommand{\zdM}{\putabove{\scriptscriptstyle 0}{d\Mcal}}

\newcommand{\gdiff}{\gamma} 
\newcommand{\hdiff}{\eta} 
\newcommand{\yg}{\putabove{\scriptscriptstyle 1}g}
\newcommand{\yh}{\putabove{\scriptscriptstyle 1}h}
\newcommand{\yR}{\putabove{\scriptscriptstyle 1}R}
\newcommand{\yRic}{\putabove{\scriptscriptstyle 1}{\Ric}}
\newcommand{\yTr}{\putabove{\scriptscriptstyle 1}{\Tr}}
\newcommand{\ynabla}{\putabove{\scriptscriptstyle 1}\nabla}

\newcommand{\yDelta}{\putabove{\scriptscriptstyle 1}\Delta}

\newcommand{\ydH}{\putabove{\scriptscriptstyle 1}{d\Hcal}}
\newcommand{\ydM}{\putabove{\scriptscriptstyle 1}{d\Mcal}}

\newcommand{\yQH}{\putabove{\scriptscriptstyle 1}{\Qcal\Hcal}}
\newcommand{\yQM}{\putabove{\scriptscriptstyle 1}{\Qcal\Mcal}}

\newcommand \Poin {\textnormal{\textbf{P}}}  
\newcommand \Korn {\textnormal{\textbf{K}}}       
\newcommand \Hardy {\textnormal{\textbf{H}}}

\newcommand \Deltaslash {\slashed\Delta{}}

\newcommand \nablaslash {\slashed\nabla{}}

\newcommand \xh {\widehat{x}}

\newcommand \wtrr {\rbf}

\newcommand \Oneone {\mathbf{1}}

\newcommand \ut {\widetilde{u}}
\newcommand \vt {\widetilde{v}}

\newcommand \unquad {\hspace{-1em}\ifmmode\mathopen{}\fi}

\newcommand \astar {a_\star}
\newcommand \bstar {b_\star}
\newcommand \pstar {p_\star}

\newcommand \Sphe {\mathbf{S}}

\newcommand \Ball {\mathbf{B}}

\newcommand \Arr {\Acal}
\newcommand \Brr {\Bcal}

\newcommand \Ars {\slashed \Acal}
\newcommand \Brs {\slashed \Bcal}

\newcommand \nut{\widetilde\nu}

\newcommand \gslash {\slashed{g}{}}

\newcommand \Zpar {Z^{\parallel}}
\newcommand \Zperp {Z^{\perp}}

\newcommand \xipar {\xi^{\parallel}}
\newcommand \xiperp {\xi^{\perp}}

\newcommand{\ssAAUX}{\slashed{\Acal}}
\newcommand{\ssA}{\slashed{\ssAAUX}{}}
\declareslashed{}{/}{.1}{0}{\Acal}
\declareslashed{}{/}{.3}{0}{\ssAAUX}
\newcommand{\ssBAUX}{\slashed{\Bcal}}
\newcommand{\ssB}{\slashed{\ssBAUX}{}}
\declareslashed{}{/}{.1}{0}{\Bcal}
\declareslashed{}{/}{.3}{0}{\ssBAUX}
\newcommand{\ssFAUX}{\slashed{\Fcal}}

\declareslashed{}{/}{.1}{0}{F}
\declareslashed{}{/}{.3}{0}{\ssFAUX}
\newcommand{\ssrmAAUX}{\slashed{\mathrm{A}}}
\newcommand{\ssrmA}{\slashed{\ssrmAAUX}{}}
\declareslashed{}{/}{.1}{0}{\mathrm{A}}
\declareslashed{}{/}{.3}{0}{\ssrmAAUX}
\newcommand{\ssrmBAUX}{\slashed{\mathrm{B}}}
\newcommand{\ssrmB}{\slashed{\ssrmBAUX}{}}
\declareslashed{}{/}{.1}{0}{\mathrm{B}}
\declareslashed{}{/}{.3}{0}{\ssrmBAUX}

\newcommand \Seed {\textnormal{\textbf{Seed}}}

\newcommand \Ehat {\widehat{E}}

\newcommand \mbb {\mathbbm{m}}
\newcommand \notreH {\mathscr{H}}

\newcommand \notreM {\mathscr{M}}

\AtBeginDocument{%
  \mathchardef\ordinarycolon=\mathcode`\:
  \mathcode`\:="8000
  \mathchardef\ordinaryequal=\mathcode`\=
  \mathcode`\=="8000
}
\begingroup
\makeatletter
\lccode`~=`: \lowercase{\gdef~{\@ifnextchar={\coloneqq\@gobble}%
    {\begingroup\def~{\ordinarycolon}\colon\endgroup}}}
\lccode`~=`= \lowercase{\gdef~{\@ifnextchar:{\eqqcolon\@gobble}{\ordinaryequal}}}
\endgroup
\renewcommand{\coloneqq}{\mathrel{\mathop\ordinarycolon}\mathrel{\mkern-1.2mu}\ordinaryequal}
\renewcommand{\eqqcolon}{\ordinaryequal\mathrel{\mkern-1.2mu}\mathrel{\mathop\ordinarycolon}}

\newcommand \expoP       {P}
\newcommand \expoPp {\overline{P}}
\newcommand \expoPm {\underline{P}}

\newcommand \dVol {\textnormal{d\textbf{V}}}

\newcommand \seedg {g_{\textnormal{\textbf{s}}}}
\newcommand \seedh {h_{\textnormal{\textbf{s}}}}
\newcommand \seedmodg {g_{\textnormal{\textbf{s}}}^{\textnormal{\textbf{m}}}}
\newcommand \seedmodh {h_{\textnormal{\textbf{s}}}^{\textnormal{\textbf{m}}}}
\newcommand \seedmodgpstarp {g_{\textnormal{\textbf{s}}\,\pstar'}^{\textnormal{\textbf{m}}}}
\newcommand \seedmodhpstarp {h_{\textnormal{\textbf{s}}\,\pstar'}^{\textnormal{\textbf{m}}}}

\newcommand \modu {\infty}

\newcommand \bfDiv {\textnormal{\textbf{Div}}}

\newcommand \aire {\textbf{Area}}

\newcommand \normal {\textnormal{\textbf{si}}} 

\newcommand \Solu {\textnormal{\textbf{Sol}}}

\newcommand \Fpar {F^{\parallel}}
\newcommand \Fperp {F^{\perp}}

\newcommand \Mu {\mathrm{M}}
\newcommand \Nu {\mathrm{N}}
\newcommand \Nuh {\widehat{\Nu}}
\newcommand \Chi {\mathrm{X}}
\newcommand \Kappa {\mathrm{K}}

\newcommand \cutoff {\delta}

\newcommand \Span 	{\mathop{\bf Span}}
\newcommand \Sym 	{{\bf Sym}}

\newcommand \Ucal {\mathcal{U}}

\newcommand \rhopar {\rho^{\parallel}}
\newcommand \rhoperp {\rho^{\perp}}

\newcommand \bnotreH {b^\notreH}

\newcommand\gxedaux{\slashed{g}}

\declareslashed{}{\backslash}{.1}{0}{\gxedaux}
\newcommand\nablaxedaux{\slashed{\nabla}}

\declareslashed{}{\backslash}{.1}{0}{\nablaxedaux}
\newcommand\Deltaxedaux{\slashed{\Delta}}

\declareslashed{}{\backslash}{.1}{0}{\Deltaxedaux}
\newcommand\uxedaux{\slashed{u}}

\declareslashed{}{\backslash}{.1}{0}{\uxedaux}
\newcommand\vxedaux{\slashed{v}}

\declareslashed{}{\backslash}{.1}{0}{\vxedaux}
\newcommand\kxedaux{\slashed{k}}

\declareslashed{}{\backslash}{.1}{0}{\kxedaux}

\newcommand \matQ {Q}
\newcommand \projP {P}
\newcommand \polyP {P}

\newcommand \detbf {\operatorname{\bf det}}

\newcommand \seed {{\textnormal{\textbf{seed}}}}
\newcommand \lin {{\textnormal{\textbf{lin}}}}
\newcommand \qua {{\textnormal{\textbf{qua}}}}
\newcommand \adj {{\textnormal{\textbf{adj}}}}

\newcommand \diam {{\textnormal{\textbf{diam}}}}

\newcommand\mmax{m^{\max}}
\newcommand\Chu{C_{\textnormal{h}}^u}
\newcommand\ChE{C_{\textnormal{h}}^E}
\newcommand\ChNu{C_{\textnormal{h}}^{\Nu}}

\newcommand\barg{\overline{g}}

\newcommand\gammash{\gamma_{\textnormal{sh}}}
\newcommand\cellip{c_{\textnormal{ell}}}
\newcommand\pmax{p_{\max}}

\newcommand\Jmax{|J|^{\max}}
\newcommand\Dtot{\Dcal_{\textnormal{tot}}}
\newcommand\Dsym{\Dcal_{\textnormal{sym}}}
\newcommand\dtot{d_{\textnormal{tot}}}
\newcommand\dsym{d_{\textnormal{sym}}}
\newcommand{\Yperp}{Y^\perp}
\newcommand{\Ypar}{Y^\parallel}
\newcommand\Abb{\mathbb A}

\newcommand\Cstar{C_\star}
\newcommand\xistar{\xi^{\star}}
\newcommand\xistarpar{\xi^{\star\parallel}}
\newcommand\xistarperp{\xi^{\star\perp}}
\newcommand\mseed{m^{\textnormal{\textbf{s}}}}
\newcommand\Jseed{J^{\textnormal{\textbf{s}}}}
\newcommand\pushforward{\star}
\newcommand\pullback{\star}
\newcommand\mmodu{m^\modu}
\newcommand\Jmodu{J^\modu}
\newcommand\umodu{u^\modu}
\newcommand\Zmodu{Z^\modu}
\newcommand\gmodu{g^\modu}
\newcommand\hmodu{h^\modu}

\newcommand\cstun{\varpi_1}
\newcommand\cstdeux{\varpi_2}

\newcommand \uY {\underline{Y}}

\newcommand\ChZ{C_{\textnormal{h}}^Z}
\newcommand\Jrm{\mathrm{J}}

\newcommand\Est{\textnormal{\textbf{Est}}}

\newcommand\CPKHzero{C^{g_0,h_0}_{\textnormal{\textbf{PKH}}}}
\newcommand\unL{\underline{L}}
\newcommand\unH{\underline{H}}
\newcommand\charac{\operatorname{char}}
\newcommand \Err {\mathbf{Err}}

\newcommand{\compresseq}[1]{\medmuskip=#1\medmuskip}
\makeatletter
\newcommand{\compressmath}[2]{\mathpalette{\compressmath@{#1}}{#2}}
\newcommand{\compressmath@}[3]{\mbox{$\m@th#2\compresseq{#1}#3$}}
\makeatother

\newcommand{\cradialH}{c_{\textnormal{radial}}^\notreH}
\newcommand{\cradialHiota}{c_{\iota,\textnormal{radial}}^\notreH}
\newcommand{\CKornM}{C_{\textnormal{Korn}}^\notreM}
\newcommand{\CKornMiota}{C_{\iota,\textnormal{Korn}}^\notreM}

\newcommand{\gammashellHiota}{\gamma_{\iota,\textnormal{shell}}^\notreH}

\newcommand{\gammashellMiota}{\gamma_{\iota,\textnormal{shell}}^{\notreM}}
\newcommand{\gammaKornM}{\gamma_{\textnormal{Korn}}^{\notreM}}
\newcommand{\gammaKornMiota}{\gamma_{\iota,\textnormal{Korn}}^{\notreM}}
\newcommand{\shellKorn}{\mathfrak{K}}
\newcommand{\fluc}{\mathrm{fl}}

\newcommand{\etainv}{\eta_{\textnormal{inv}}}

\newcommand{\deltaH}{\delta^{\notreH}}
\newcommand{\deltaHiota}{\delta^{\notreH}_{\iota}}
\newcommand{\deltaM}{\delta^{\notreM}}
\newcommand{\deltaMiota}{\delta^{\notreM}_{\iota}}

\newcommand{\teps}{\widetilde\eps}
\newcommand{\DcalR}{\mathcal{D}_R}

\usepackage{footnotebackref}
\usepackage{perpage}
\MakePerPage{footnote}
 
\usepackage[textwidth=6in, textheight=9.5in]{geometry}

\newcounter{fig}[section] 

\usepackage{etoolbox}
\AtBeginEnvironment{table}{\refstepcounter{fig}}
\AtBeginEnvironment{figure}{\refstepcounter{fig}}

\begin{document} 

\phantomsection\label{firstpage}

\title{Optimal localization for the Einstein constraints}

\author{Bruno Le Floch\texorpdfstring{$^{1}$}{} and Philippe G. LeFloch\texorpdfstring{$^2$}{}}

\date{}

\maketitle
\footnotetext[1]{Laboratoire de Physique Théorique et Hautes Énergies, Sorbonne Universit\'e \& Centre National de la Recherche Scientifique, 4 Place Jussieu, 75252 Paris, France. Email: {\tt blefloch@lpthe.jussieu.fr}.
}
\footnotetext[2]{Laboratoire Jacques-Louis Lions, Sorbonne Universit\'e \& Centre National de la Recherche Scientifique, 4 Place Jussieu, 75252 Paris, France. Email: {\tt contact@philippelefloch.org}.
\newline
{\it Keywords and Phrases.} 
Shielding gravity;  optimal localization; 
harmonic, radial, and shell stability; curvature functional; 
Poincar\'e, Korn, and Hardy inequalities. 
\hfill  First version: Dec. 2023. This version: August 2026.\!%
}

\begin{abstract}
We establish the existence of asymptotically Euclidean solutions, or initial data sets, to Einstein's vacuum constraints exhibiting both \textit{gravitational shielding} and \textit{arbitrarily low decay}.  We resolve a conjecture of Carlotto and Schoen on gluing two solutions across an asymptotically conical domain: in the interpolation region we establish optimal estimates \textit{at and beyond harmonic decay}, together with ADM invariant estimates.  Starting from a seed data set that asymptotically solves the constraints, we project it to an exact solution, whose difference from the seed inherits the natural optimal radial decay rate.  The seed-to-solution projection operator composes the linearized constraint operator with its formal adjoint.  Its harmonic kernel at infinity generates energy-momentum modulators, which define \textit{relative invariants} equal to the differences of the usual ADM invariants whenever those exist.  We introduce harmonic, radial, and shell stability conditions on the angular gluing function; they ensure the positivity and radial monotonicity of the functionals governing optimal decay.  A follow-up paper proves these conditions from weighted Poincaré, Korn, and Hardy inequalities, leading to verifiable conditions on the localization function.  As particular cases, our theory applies to solutions \textit{without localization} and for solutions in the \textit{small aperture limit}.  Our analysis of the linearized scalar curvature operator relies crucially on a new fourth-order curvature functional, and also applies to gluing scalar-flat metrics.
\end{abstract}

{\small

\begin{spacing}{.95}

\setcounter{secnumdepth}{3}
\setcounter{tocdepth}{3}
 
\tableofcontents

\end{spacing}
}


\section{Introduction}
\label{section=1}

\subsection{Localization in Einstein gravity}
\label{section=1.1}

\paragraph{Main objective.}

We consider $n$-dimensional initial data sets for Einstein's vacuum field equations of general relativity, that is, spacelike hypersurfaces embedded in a Ricci flat, $(n+1)$-dimensional Lorentzian manifold. By definition, such data sets satisfy the Einstein constraints which are the Gauss-Codazzi equations satisfied by the induced geometry on such a hypersurface, namely its first and second fundamental forms. We are interested in asymptotically Euclidean initial data sets and in the anti-gravity phenomenon discovered by Carlotto and Schoen~\cite{CarlottoSchoen} and generalized by Chru\'sciel and Delay~\cite{ChruscielDelay-2018,ChruscielDelay-2021}: in short, the Einstein constraints admit classes of solutions that exhibit some angular localization at infinity, constructed by gluing solutions in different conical domains. Their method builds upon pioneering work on the gluing of solutions in compact domains~\cite{ChruscielDelay-memoir,Corvino-2000,CorvinoSchoen}.  

In the present paper, together with the companion paper~\cite{LL-PoincareKornHardy}, we investigate the decay properties of solutions
generated by gluing and develop an {\it optimal localization theory} for
constructing classes of solutions that

\bei

\item[(i)] satisfy estimates with \emph{(super-)harmonic decay};

\item[(ii)] allow for arbitrarily slow decay, so that the ADM energy
need not be finite; and

\item[(iii)] are obtained by gluing across conical domains with
arbitrarily small aperture.

\eei

\noindent
These results were first presented in~\cite{LL-first-version} in December~2023 and summarized in~\cite{LL-Letter}. Related recent developments include
\cite{AretakisCzimekRodnianski,LeFlochNguyen-preprint,LeFlochNguyen,
MaoOhTao,MaoTao}, which we discuss below. Our framework also includes the isotropic case, in which the localization function is identically one and no angular localization is imposed. In our companion paper
\cite{LL-PoincareKornHardy}, we establish the stability conditions introduced here.


\paragraph{Classes of initial data sets.}
 
The construction and analysis of physically relevant solutions to the Einstein constraints is a central topic in the physical, mathematical, and numerical literature. Recent reviews are provided in Carlotto~\cite{Carlotto-Review} and Galloway et al.~\cite{GallowayMiaoSchoen}. Historically, the subject started with a pioneering work by Lichnerowicz~\cite{Lichne} and the (now called) {\it conformal method,} which, later on, was expanded in many directions; cf.~\cite{ChrusCorvinoIsenberg,DruetHebey,DruetPremoselli,Isenberg-1995,IsenbergMaxwellPollack,IsenbergMoncrief,Maxwell-2005} as well as \cite{Gicquaud} (and the references cited therein). In particular, the pioneering paper by Isenberg~\cite{Isenberg-1995} provides one with a parametrization of all {\it closed} manifolds representing vacuum initial data sets with constant mean curvature. The problem of parametrizing classes of solutions, recently addressed by Maxwell~\cite{Maxwell-2021}, is an important issue that we also tackle in the present paper for {\it non-compact} manifolds.  

A different strategy, referred to as the {\it variational method,} was introduced by Corvino~\cite{Corvino-2000} and Corvino and Schoen~\cite{CorvinoSchoen}, who built on Fischer and Marsden's study of the deformations of the scalar curvature operator~\cite{FischerMarsden-1973,FischerMarsden-1975}. We will follow this line of study in the present paper, while further related results will be quoted throughout our presentation.  

Importantly, both lines of research led to major achievements in general relativity, but also made contact with central developments in Riemannian geometry, including gluing techniques that allow to combine two different manifolds and build ``new'' ones. The research in this direction encompasses classes of compact or non-compact manifolds, whose ends may be asymptotically Euclidean, asymptotically hyperbolic, etc. The techniques of geometric analysis in these papers are varied and rely on the linear and nonlinear differential structure of the Einstein equations: for basic material on the Einstein equations and on elliptic equations we refer to~\cite{Bartnik,Choquet-book} and to~\cite{ChoquetC,DouglisNirenberg,HanLin-book}, respectively. 


\paragraph{Shielding (or anti-gravity) phenomenon.}
  
A complete vacuum initial data set that agrees with the Euclidean data $(\delta,0)$ in a neighborhood of infinity has vanishing ADM energy-momentum.
By the rigidity statement in the positive mass theorem, the data arises from a spacelike hypersurface in Minkowski spacetime; in the time-symmetric case, the data are globally Euclidean.
In contrast, in~\cite{CarlottoSchoen}, Carlotto and Schoen made a remarkable discovery for manifolds with asymptotically Euclidean ends, namely, the existence of localized solutions that, in a neighborhood of infinity, coincide with the Euclidean geometry in all angular directions except within a conical domain with arbitrarily small aperture. Alternatively, one can interpolate between exactly Euclidean and Schwarzschild solutions through a gluing region that is conical at infinity. Subsequently, Chru\'sciel and Delay~\cite{ChruscielDelay-2021} presented a method that applied to more general gluing domains, and Beig and Chru\'sciel~\cite{BeigChrusciel-2017} contributed to this problem in the context of linearized gravity. We refer the interested reader to the reviews~\cite{Carlotto-Review,Chrusciel-bourbaki,GallowayMiaoSchoen}.  

We observe that the method in~\cite{CarlottoSchoen,ChruscielDelay-2021} allows the authors to establish {\it sub-harmonic} estimates within the gluing region, namely $r^{-n+2+ \eta}$ estimates on the metric with respect to a radial coordinate~$r$ at infinity, for any exponent $\eta > 0$.  Carlotto and Schoen in~\cite{CarlottoSchoen,Carlotto-Review} also conjectured that the localization at infinity should be achievable with harmonic  control, corresponding to taking $\eta$ to vanish. 

Next, in a preprint~\cite{LeFlochNguyen-preprint} posted in 2019 (and later published in~\cite{LeFlochNguyen}), {P.~LeFloch} and Nguyen proposed a different approach to the gluing problem and formulated what they called the {\it asymptotic} localization problem, as opposed to the {\it exact} localization problem originally proposed by Carlotto and Schoen~\cite{CarlottoSchoen}. In~\cite{LeFlochNguyen-preprint} a notion of {\it seed-to-solution map} was introduced and estimates at the \mbox{(super-)harmonic} level of decay were indeed proven, so that the gluing at harmonic rate was achieved in the sense that solutions match with the seed data in all angular directions at harmonic level at least, {\it up to} contributions with faster decay. 

After the publication of the optimal gluing results in~\cite{LL-first-version,LL-Letter}, many further advances on the gluing problem took place via independent methods. Namely, Aretakis, Czimek, and Rodnianski~\cite{AretakisCzimekRodnianski} addressed the {\it characteristic gluing problem,} for a class of metrics with prescribed Kerr behavior. Their work implies (as a corollary) Carlotto-Schoen's conjecture on spacelike gluing. Also recently, Mao and Tao~\cite{MaoTao} achieved localization with optimal decay in regions that are asymptotically conical, or even smaller, based on solution operators with support on half-lines, and reproduced the obstruction-free compact gluing in a purely spacelike context together with Oh~\cite{MaoOhTao}. 


\subsection{Overview of the optimal localization theory}  
\label{section=1.2}
 
\paragraph{Optimal localization theory.} 

We build upon the works (cited above) by Carlotto, Corvino, Chru\'sciel, Delay, and Schoen, as well as the recent contribution by P.~LeFloch and Nguyen, and we establish an {\it optimal localization theory,} in which we control the behavior of harmonic terms arising from the gluing. Namely, we describe precisely the \emph{harmonic terms} associated with seed data sets and we provide a {\it parametrization of solutions.} Importantly, this theory, 
combined with the stability criteria proven in our companion paper~\cite{LL-PoincareKornHardy},
yields a complete and direct proof of Carlotto and Schoen's conjecture in the variational approach. Our results are summarized in the next paragraphs, while precise definitions and statements are given throughout \refwithname{Sections}{section=2}, \ref{section=3}, \refwithname{and}{section=4}, below. For a schematic illustration, we refer to \autoref{figure---111}.
  
 
\paragraph{Einstein's constraint equations.}
 
For $n \geq 3$, we are interested in $n$-dimensional Riemannian manifolds $(\Mbf, g, k)$ with (possibly) several asymptotic ends, endowed with a Riemannian metric $g$ and a symmetric $(0,2)$-tensor field $k$, which represents the extrinsic curvature in the spacetime picture. The {\it Hamiltonian and momentum constraints} read as follows\footnote{For instance, see~the textbook~\cite[Chap.~VII]{Choquet-book}.}: 
\bel{eq:ee11}
\aligned
R_g + (\Tr_g k)^2 -  | k |_g^2 & = 0,
\qquad
\bfDiv_g \bigl( k - (\Tr_g k) g \bigr) = 0.
\endaligned
\ee  
Here, $R_g$ denotes\footnote{In any local coordinate chart $(x^j)$ we write $g = g_{ij} dx^i dx^j$, so that $\Tr_g k = g_{ij} k^{ij}$ and $| k |_g^2 = k^{ij} k^{lm} g_{il} g_{jm}$, while the divergence operator reads $(\bfDiv_g k)_j = \nabla_i k^i_j$. The  Levi-Civita connection of $g$ is denoted by $\nabla$ and the range of Latin indices is taken to be $i,j, \ldots = 1,2,\dots,n$.} the scalar curvature of $g$, while $\Tr_g k$  and $| k |_g$ denote the trace and norm of $k$, respectively, and $\bfDiv_g$ stands for the divergence operator.   It is convenient to introduce the $(2,0)$-tensor $h$ by 
\be 
h \coloneqq \big( k - \Tr_g(k) g\big)^{\sharp\sharp},
\ee
where the sharp symbol refers to the identification between covariant and contravariant tensors, induced by the metric $g$. From now on, we work exclusively with the unknowns $(g,h)$, and we express the Hamiltonian and momentum operators as 
\be
\aligned
\Hcal(g,h) 
& \coloneqq R_g + {1  \over n-1} \bigl( \Tr_g h \bigr)^2 - | h |^2_g,
\qquad\quad
\Mcal(g,h) \coloneqq \bfDiv_g h, 
\endaligned
\ee
which are scalar-valued and vector-valued, respectively. We then formulate Einstein's  vacuum constraints as 
\bel{eq:Einstein00}
\Gcal(g,h) \coloneqq \big( \Hcal(g,h), \Mcal(g,h) \big) = 0
\ee
and observe that $\Gcal(g,h)$ can be interpreted as an $(n+1)$-dimensional vector field in the spacetime picture.  


\paragraph{Localized seed-to-solution projection.}

Before we tackle the main problem of interest in this paper (namely, the control of harmonic terms), we are going to introduce a  framework that encompasses a broad class of initial data sets and is formulated so as to provide us with basic continuity and decay   estimates in suitably weighted norms. Specifically, we parametrize solutions in the ``vicinity'' of a given {\it localization data set,} as we call it. The localization data set serves to specify the underlying geometry of interest, especially the gluing domain and the asymptotic structure. On a given manifold $\Mbf$ we introduce a localization domain $\Omega \subset \Mbf$ together with a weight of the form 
\bel{omegabfp}
\omegabf_p = \lambdabf^{\expoP} \wtrr^{n/2-p}, 
\ee
which provides a localization in angular directions: here, $\lambdabf \geq 0$ vanishes linearly on the boundary of $\Omega$, the exponent $\expoP$ is a large integer, and the factor $\wtrr^{n/2-p}$ controls the radial behavior in each asymptotic end. Our parametrization is presented in terms of a {\it localized seed-to-solution projection operator} denoted by 
\bel{equa-sol-map-repeat} 
\Solu^{\lambdabf}_{n,p}:  (\seedg,\seedh) \in \Seed(\Omega, g_0,h_0,p_G, p_A,\eps_G,\eps_A)
\mapsto (g,h) \in [(\seedg,\seedh)], 
\ee
which maps an approximate solution (the {\it seed\/}) to an exact solution of the Einstein constraints that lies in an affine space based on the seed. This standpoint extends the (non-localized) formulation in~\cite{LeFlochNguyen} and, for further details, we refer to \autoref{def-mapping-2}, below. 

The parametrization of solutions to the Einstein constraints by \emph{seed data sets} prescribed at infinity was a key ingredient in~\cite{LeFlochNguyen-preprint} in order to solve the asymptotic localization problem, as well as in our 2023 preprint~\cite{LL-first-version} to solve the Carlotto-Schoen localization problem beyond the harmonic rate.  The technique of prescribing an Ansatz at infinity to generate a class of solutions to the constraints can be traced back to the work of Beig and \'O~Murchadha~\cite{BeigMurchadha}, which was based on the conformal method, further extended in the work by Bieri, Garfinkle, Isenberg, Maxwell, and Wheeler in a 2025 preprint~\cite{Bieri-constraints} for the construction of broad classes of solutions. 


\paragraph{Conical gluing framework.}

Within this general setup, we specialize our investigation to asymptotically Euclidean manifolds and conical gluing domains. In a preliminary stage, we prove the existence of solutions enjoying {\it sub-harmonic estimates;} cf.~\autoref{thm:sts-existence}. While we follow closely Carlotto and Schoen~\cite{CarlottoSchoen}, our presentation (based on a different iteration scheme) has the advantage of separating clearly between the roles (and ranges) of the relevant decay exponents. 

\bei 
 
\item The \textbf{projection exponent} for the solution map $\Solu^{\lambdabf}_{n,p}$ is denoted by $p \in (0,n-2)$ and arises in the variational formulation of the linearized Einstein operator.
 
\item The \textbf{geometry exponent} for the seed data set $(\seedg,\seedh)$ is denoted by $p_G >0$ and specifies the (possibly very low) decay of the metric and extrinsic curvature. Namely, the solutions at infinity are solely required to decay as $r^{-p_G}$ for some $p_G>0$ (possibly close to~$0$). The ADM energy-momentum vector of such a solution need not be well-defined. 
 
\item The \textbf{accuracy exponent} $p_A \geq \max(p,p_G)$ describes the accuracy of the seed data set, regarded as an approximate solution of the Einstein equations in the vicinity of the asymptotic ends. 
 
\eei 

\noindent We work in suitably weighted H\"older--Lebesgue norms and derive first basic estimates which will be useful throughout this paper. At this stage of our analysis, we content ourselves with the sub-harmonic estimates available for both the linearized constraints and the nonlinearities of Einstein constraints.


\paragraph{Harmonic stability.}

Next, building upon the proposed projection framework we turn our attention to the core of the present paper, devoted to a new method for deriving sharp integral and pointwise estimates for Einstein's initial data sets. Our harmonic results involve 

\bei 

\item a \textbf{sharp decay exponent}, which is denoted by $\pstar \in [p, p_A]$ and describes the control of the difference ${(\seedg,\seedh) - (g,h)}$ between the seed data set and the actual solution to the Einstein constraints, up to harmonic terms (cf.~\eqref{equa-xjs3}). 

\eei 

\noindent In the course of our analysis, we are led to introduce a lower bound for the projection exponent,
\bel{pflatn-def}
p^\flat_n \coloneqq \frac{(n-1)(n-3)}{2(n-2)} = \frac{n-2}{2} - \frac{1}{2(n-2)} ,
\ee
and an upper bound $p^{\lambdabf}_{n,p}> n-2$ for the sharp decay exponent, that depends on the geometry and on~$p$,
so that the range of interest for the estimates is
\bel{exponent-range}
\aligned
& 0 < p_G \leq p_A, \qquad
&& p^\flat_n < p < n-2, \qquad
\\
& p \leq \pstar \leq p_A , \qquad  
&& \pstar < \min(p^{\lambdabf}_{n,p}, n-2+p_G). 
\endaligned
\ee
Importantly, to accommodate localized solutions with possibly very low decay (small $p_G>0$) our analysis encompasses a wide range of decay exponents. After applying the projection operator $\Solu^{\lambdabf}_{n,p}$, we prove that ${(\seedg,\seedh) - (g,h)}$ contains a harmonic contribution which arises in the asymptotic structure of the solution at each asymptotic end, even if the seed has very high accuracy exponent~$p_A$. By examining an asymptotic version of the constraints at infinity, we unveil certain (asymptotic kernel) contributions associated with the sphere at infinity, which we refer to as \emph{energy-momentum modulators.}  


\paragraph{Overview of the main results.}

Our main statements will be presented in \autoref{theo--beyond-harmonic}, \autoref{thm:informal-sufficient-stability}, \autoref{thm:informal-sufficient-stability-M}, and \autoref{thm:stable-s2}, below, after introducing required concepts. At this stage we state an informal version of  \autoref{theo--beyond-harmonic}. (For an overview, we also refer to~\cite{LL-Letter}.) 

\begin{theorem}[Optimal localization theory]
\label{theoinformel}  
Consider a \emph{conical localization data set} denoted by $(\Mbf, \Omega, g_0,h_0, \wtrr, \lambdabf)$ (\autoref{def-conical}) together with (projection, geometry, accuracy, sharp decay) exponents $(p, p_G,p_A, \pstar)$ satisfying~\eqref{exponent-range} and, within the gluing domain $\Omega$ and in preferred charts at infinity (equipped with a Euclidean metric~$\delta$), the pointwise decay conditions 
\be
\aligned
g_0 - \delta & = \Obig(r^{-p_G}),
&
h_0 & = \Obig(r^{-p_G-1}) \quad 
& \text{ in each end.}
\endaligned
\ee
Assume that $\lambdabf$ is a \emph{stable localization function}\footnote{Sufficient criteria for this stability are stated next} in the sense of \autoref{def-41-stable}, below. 
\bei 

\item For any localized seed data set $(\seedg, \seedh)$ (\autoref{def-aset}) that is sufficiently close to $(g_0, h_0)$ and provides an ``asymptotic solution'' in the sense that in each end\footnote{In the case $\pstar=n-2$, it is in fact sufficient that the \emph{signed} componentwise integrals in~\eqref{equa-mstarJstar} converge.}
\bel{Hgh-Mgh-cond}
\Gcal(\seedg, \seedh) = \Obig(r^{-p_A-2}),
\quad \text{and} \quad
|\Gcal(\seedg, \seedh)|_{g_0} \text{ integrable in the case $\pstar=n-2$,}
\ee
there exists a solution $(g,h)$ to the Einstein constraints~\eqref{eq:Einstein00} defined by variational projection of $(\seedg, \seedh)$ (\autoref{thm:sts-existence}) and enjoying the pointwise decay  
\bel{equa-xjs3}
\aligned
g &= \seedg + \osmall(r^{-\pstar}),
\qquad 
&&h = \seedh + \osmall(r^{-\pstar-1})  && \text{ when } \pstar = [p, n-2), 
\\
g &= \seedmodg + \osmall(r^{-\pstar}),
\qquad 
&&h = \seedmodh + \osmall(r^{-\pstar-1})  && \text{ when } \pstar \geq n-2.
\endaligned
\ee 

\item In~\eqref{equa-xjs3}, the so-called \emph{modulated seed data set} $(\seedmodg, \seedmodh)$ (\autoref{def-317}) is the sum of the seed data set $(\seedg, \seedh)$ and of \emph{energy-momentum modulators} in the (one-dimensional and $n$-dimensional) kernels of the \emph{harmonic operators} associated with the Hamiltonian and momentum operators (\autoref{def-operators-decomposer}).  

\eei
\end{theorem}

Sufficient conditions for Hamiltonian and momentum stability will be explained in \refwithname{Theorems}{thm:informal-sufficient-stability} \refwithname{and}{thm:informal-sufficient-stability-M}, below, informally summarized as follows.  They apply to gluing in thin cones.
The limit where the angular domain is close to the whole sphere (that is, a thick gluing) is also of interest. However, we will not pursue this question in the present paper. 

\begin{theorem}[Optimal localization theory with small aperture]
\label{theo-smallaper}
The stability conditions in \autoref{theoinformel} are satisfied in any gluing domain $(\Omega, \lambdabf)$ whose asymptotic ends have sufficiently small\/\footnote{Deriving the optimal constants for our theory is left as an open problem.} weighted Poincaré constants and angular aperture,
and for a projection exponent $p<n-2$ sufficiently close to the harmonic limit $p\to n-2$.
For a fixed angular profile of $\lambdabf$, these geometric conditions on the sphere at infinity hold whenever $\Omega$ has sufficiently small aperture and $n-2-p$ is sufficiently small.
\end{theorem}


\paragraph{A new result even without localization.}

Interestingly, our framework also applies to the construction of solutions \emph{without} localization at infinity.
We summarize \autoref{thm:stable-s2}, presented below, informally as follows, and mention a special case of \autoref{theoinformel} of particular interest to the parametrization of initial data sets with arbitrarily low decay rate $p_G>0$. Our results are new even in this case, as the analogous results in~\cite{LeFlochNguyen} require $p_G\geq 1/2$ in dimension $n=3$.

\begin{theorem}[Optimal radial decay without localization]
\label{thm:no-loc}
The localization domain $\Omega=\Mbf$ with a constant localization function $\lambdabf\equiv 1$ satisfies the stability conditions in \autoref{theoinformel} for $p<n-2$ sufficiently close to $n-2$ in dimension $n\leq 17$.
The conclusions of \autoref{theoinformel} thus hold.
In particular, for any data set $(\seedg, \seedh)$ on $\RR^n$ with sufficiently small
\be
\seedg-\delta=\Obig(r^{-p_G}) , \qquad
\seedh=\Obig(r^{-p_G-1}) , \qquad
\Gcal(\seedg, \seedh) = \Obig(r^{-p_A-2}),
\ee
for some exponents $0<p_G\leq n-2 < p_A$, there exists an exact solution $(g,h)$ of the constraints that differs from $(\seedg,\seedh)$ by energy-momentum modulators with harmonic decay and terms with some strictly larger radial decay exponent $\pstar>n-2$.
\end{theorem}
 
\paragraph{Arbitrarily slow decay.}

Our earlier papers~\cite{LL-first-version} posted in 2023 and~\cite{LL-Letter} in 2024, together with all the theorems in the present paper, encompass solutions with \emph{arbitrarily slow decay}, namely $g=\delta+\Obig(r^{-p_G})$ with $p_G>0$. The problem of constructing solutions at such an arbitrarily slow rate was recently systematically studied by several authors with distinct aims and methods: Fang--Szeftel--Touati in~2024, and Chen--Klainerman in~2025, and Shen--Wan in~2026.
We refer the reader to both the posted\footnote{The relevant preprints listed in chronological order, are:
Le~Floch--LeFloch,
\href{https://arxiv.org/abs/2312.17706}{\ttfamily arXiv:2312.17706};
Fang--Szeftel--Touati,
\href{https://arxiv.org/abs/2401.14353}{\ttfamily arXiv:2401.14353};
Le~Floch--LeFloch,
\href{https://arxiv.org/abs/2402.17598}{\ttfamily arXiv:2402.17598};
Fang--Szeftel--Touati,
\href{https://arxiv.org/abs/2405.02071}{\ttfamily arXiv:2405.02071};
and Chen--Klainerman,
\href{https://arxiv.org/abs/2512.22704}{\ttfamily arXiv:2512.22704}, and 
Shen--Wan \href{https://arxiv.org/abs/2602.01557}{\ttfamily arXiv:2602.01557}.} and published 
versions~\cite{LL-first-version,Fang-1, LL-Letter, Fang-2,ChenK,ShenWan}.

\begin{table}
\centering
\setlength{\tabcolsep}{5pt}
\caption{Notation associated with the localized Hamiltonian and momentum}
\centerline{%
\begin{tabular}{cccccc}
\toprule
& \textbf{decomposition} & \textbf{functionals} & \textbf{averages}& \textbf{normalization} & \hspace{-2em}\textbf{silhouette}
\\
\midrule
$\notreH^\lambda$  
& $\Arr, \, \Ars^\lambda, \, \ssA^\lambda$ 
&
$\Phi^\notreH\!, \Psi^\notreH_\beta\!, \Kappa^\notreH$
& $\la \nu \ra$ 
& $\la - \Deltaslash\nu + d_{n,p} \nu \ra$ 
& $\nu^\normal$
\\
\midrule
$\notreM^\lambda$  
& $\Brr, \, \Brs^\lambda, \, \ssB^\lambda$ 
&
$\Phi^\notreM\!, \Psi^\notreM_\beta\!, \Kappa^\notreM$
& $\compresseq{0}\bigl\la 2 \xh_l\, \xi^{\perp} + \xi_{l}^{\parallel} \ra$
& $\compresseq{0}\bigl\la - \nablaslash_l \xi^{\perp} + 2a_{n,p}  \xh_l\, \xi^{\perp} + (a_{n,p} +1) \xi_{l}^{\parallel} \ra$
& $\xi^{\normal(j)}$
\\
\bottomrule
\end{tabular}%
}
\label{table-structure}
\end{table}

\subsection{Originality and perspectives}
\label{section=1.3}

\paragraph{Other recent developments.}

A vast literature is available on the construction of initial data sets with specific local or asymptotic properties. In addition to the contributions already mentioned, let us point out further recent developments. 
By leveraging nonlinearities of the Einstein constraints, Czimek and Rodnianski~\cite{CzimekRodnianski} lifted the codimension~$10$ restriction on the Kerr data at infinity.  At the $C^3$-regularity level, Sansom~\cite{Sansom-2024} established null gluing modulo a $20$-dimensional obstruction space.
More recently, Fang, Szeftel, and Touati~\cite{Fang-1,Fang-2} constructed initial data sets that are suitable for the evolution of arbitrarily decaying perturbations of the Kerr solution. Other recent advances on the gluing problem include contributions by Corvino and Huang~\cite{CorvinoHuang} (solutions with matter) and Anderson et al.~\cite{AndersonCorvinoPasqualotto} (multi-localized solutions). In addition, we refer to Lee et al.~\cite{LeeLesourdUnger} for the study of initial data sets with boundaries.
We also refer to Henneaux~\cite{Henneaux} who extends the Corvino-Schoen theorem 
with supertranslations, as well as to the series of papers by Hintz~\cite{Hintz1,Hintz2} on gluing small black holes along timelike geodesics. We also mention the recent contributions by Bieri et al.~\cite{Bieri-constraints}, concerning prescribed asymptotic structures, and Chen et al.~\cite{ChenK}, concerning general free data.


\paragraph{A decomposition of the Einstein constraints.}

Let us outline here some key steps of our method, contrasting it with existing techniques in the literature.
Our construction of the solution $(g,h)$ starting from $(\seedg,\seedh)$ uses a fixed-point method based on repeatedly solving the linearized problem, as in~\cite{CarlottoSchoen}.  Replacing the Picard iteration scheme by a Newton iteration (see \autoref{appendix=E}) allows $(g,h)$ to have the same regularity as the seed, which is necessary to interpret the seed-to-solution map as a projection onto the space of constraints.
At this stage, $g-\seedg$ and $h-\seedh$ are known to decay as $r^{-p}$ and~$r^{-p-1}$, respectively.

This decay must be improved (assuming stability of~$\lambdabf$) from an exponent $p$ to~$\pstar$ to establish \autoref{theoinformel}.

After a suitable reduction at infinity and by neglecting various perturbation terms, we will be led to consider a single conical domain $\Omega_R$ in the Euclidean space $(\RR^n, \delta)$, consisting of the intersection of the exterior of a ball with radius $R>1$ and an open cone (that is, a union of half-lines $\RR_+x$) in $\RR^n$. Within this domain, we study the asymptotic properties of (the square of) the linearization of the localized Hamiltonian and momentum operators, defined as\footnote{Here, indices are freely raised and lowered using the metric~$\delta$, and implicit summation over repeated indices $i,j,\dots$ is used, {\it even when} they are both lower or upper indices.}
\bel{equa:acalew0} 
\aligned
\notreH^\lambda[u] 
& \coloneqq \omega_p^{-2} \, \Bigl(
(n-1) \, \Delta(\omega_p^2 \, \Delta u)
+ \del_i\del_j(\omega_p^2) \del_i\del_j u - \Delta(\omega_p^2) \Delta u
\Bigr),
\\
\notreM^\lambda[Z]^i 
& \coloneqq - \frac{1}{2} \, (\Delta Z_i + \del_j \del_i Z_j) - (\del_j \log \omega_{p+1}) \, (\del_j Z_i + \del_i Z_j), 
\endaligned
\ee
in which $u: \Omega_R \to \RR$ is a scalar-valued unknown and $Z: \Omega_R \to \RR^n$ is a one-form-valued unknown. As in~\eqref{omegabfp}, $\omega_p:\Omega_R \to [0,+\infty)$ is a weight of the form $\omega_p = \lambda^{\expoP} r^{n/2-p}$ with a power-law dependence in the radial distance $r\geq R$ and which vanishes as powers of the distance to the angular boundary of $\Omega_R$, but not the $r=R$ boundary, in contrast to~$\omegabf_p$ which vanishes at all boundaries of~$\Omega$.

We introduce here the \textbf{harmonic-spherical decomposition} of these operators (as we call it), which is adapted to the study of the harmonic decay of solutions and takes the form 
\bel{equa-key-decompose-H-repeat} 
r^4 \notreH^\lambda[u] = \Arr[u] + \Ars^\lambda[u] + \ssA^\lambda[u], 
\qquad   
r^2 \notreM^\lambda[Z] = \Brr[Z] + \Brs^\lambda[Z] + \ssB^\lambda[Z]. 
\ee
The (``double slashed'') harmonic operators are defined by plugging harmonic-decaying functions of the form $\nu \, r^{-2( n-2-p)}$ and $\xi \, r^{-2( n-2-p)}$ and, specifically, we define 
\bel{equa-asymptoOper-reapeat} 
\aligned 
\ssA{}^\lambda[\nu] & \coloneqq r^{4+2(n-2-p)} \notreH^\lambda[\nu \, r^{-2( n-2-p)}],
\qquad   
\ssB{}^\lambda[\xi] \coloneqq r^{2+2(n-2-p)} \notreM^\lambda[\xi \, r^{-2( n-2-p)}]. 
\endaligned
\ee
The Hamiltonian and momentum harmonic operators $\ssA{}^\lambda,\ssB{}^\lambda$ are non-self-adjoint elliptic operators.  If the localization function~$\lambda$ satisfies the \textbf{harmonic stability conditions} (coercivity of quadratic functionals associated to $\ssA{}^\lambda,\ssB{}^\lambda$), we prove that their kernels are of dimension $1$ and $n$, respectively.  These kernel elements give solutions of $\notreH^\lambda[u]=0$ and $\notreM^\lambda[Z]=0$ with $r^{-2(n-2-p)}$ decay that are responsible for the energy-momentum modulators with harmonic decay in \autoref{theoinformel}.  These, in turn, adjust the ADM energy and momentum at each asymptotic end.


\paragraph{A novel strategy.}

The core of our linearized analysis relies on new quadratic functionals for the Einstein constraints, combined with radial differential equations on spherical averages.  These enable us to prove the (integral, pointwise) radial decay of solutions of $\notreH^\lambda[u]=E$ and $\notreM^\lambda[Z]=F$ for suitably decaying sources.
We describe here the strategy for the Hamiltonian operator (\refwithname{Sections}{section=5}--\ref{section=7}); the treatment of the momentum operator (\refwithname{Sections}{section=8}--\ref{section=9}) is analogous.
To simplify the presentation we omit the source~$E$, and simply show that solutions of $\notreH^\lambda[u]=0$ decaying faster than~$r^{-(n-2-p)}$ (provided by the existence theorem) must decay at least as~$r^{-2(n-2-p)}$ (corresponding to a metric with harmonic decay).
This setting is rich enough to uncover the stability conditions that must be imposed on the localization function~$\lambda$.

We introduce the notation $a=a_{n,p}=2(n-2-p)$, which is positive and is small for $p$~close to harmonic.  The functionals of interest are integrals over spherical shells $\Lambda_r\coloneqq\Omega\cap\Sphe_r$ of radius $r\in[R,+\infty)$, whose integrand involves $u$ and its radial derivatives restricted to~$\Lambda_r$.
The short-hand notation $\norm{u}^\notreH$ for the sum of normalized weighted angular $L^2$ norms of $u$ and its (angular and radial) derivatives of order up to~$2$ will prove useful, that is, 
\bel{equa-normsH-0}
\bigl( \norm{u}^\notreH \bigr)^2
\coloneqq \| \vartheta^2 u\|^2_{\unL^2_{-\expoP}(\Lambda_{r})} + \| \vartheta u\|^2_{\unH^1_{-\expoP}(\Lambda_{r})} + \| u\|^2_{\unH^2_{-\expoP}(\Lambda_{r})}, 
\ee
defined in \eqref{equa-normsH} below in suitable weighted Sobolev spaces. 

We now outline our method.

\bei 

\item 
The starting point of our journey is the \textbf{shell functional} $\Phi^\notreH[u](r)$, a non-negative functional~\eqref{PhinotreH-expr} obtained by integrating on~$\Lambda_r$ a quadratic combination of $u$ and its first and second-order derivatives.  In other words, $0\leq\Phi^\notreH[u]\lesssim(\norm{u}^\notreH)^2$.
This functional obeys the shell identity (cf.~\eqref{equa-condition-monotone})
\bel{intro-shell}
a (r\del_r + a) (r\del_r + 2a) \Phi^\notreH[u]
= (r\del_r + 2a) \Psi^\notreH_{a}[u]
- (r\del_r + a) \Psi^\notreH_{2a}[u]
\ee
where the dissipation functionals $\Psi^\notreH_{\beta}[u]\lesssim(\norm{r\del_r u}^\notreH + \norm{u}^\notreH)^2$ for $\beta\in\{a,2a\}$ are quadratic in~$u$ and up to three derivatives (except for third angular derivatives).
Solving this ordinary differential equation (with $u=\osmall(r^{-a})$) gives
\bel{intro-Phi-solve}
\Phi^\notreH[u](r) = \Cstar r^{-2a}
- \sum_{\beta=a,2a} \int_R^{+\infty} \Psi^\notreH_\beta[u](s) f_\beta(r,s) ds ,
\ee
with weights $f_\beta(r,s)\geq 0$.
If $\Psi^\notreH_\beta$ were coercive functionals, \eqref{intro-Phi-solve}~would force $\Phi^\notreH[u]$ and the integrals of $\Psi^\notreH_\beta$ to decay as~$r^{-2a}$, hence $u$~to decay as~$r^{-a}$.
Unfortunately, coercivity is too much to ask for: instead, we find that $\Psi^\notreH_\beta[u] + C \la u\ra^2$ can be coercive for a large constant $C>0$, where $\la u\ra$ is the weighted average of~$u$ on~$\Lambda_r$.

\item The second ingredient is a radial differential equation obeyed by the average~$\la u\ra$, with a source expressed in terms of the fluctuations $\ut = u-\la u\ra$.
Solving this equation yields (for $u=\osmall(r^{-a/2})$)
\bel{intro-uave}
\la u(r)\ra = \Cstar r^{-a} + C_+ r^{-\beta_+} + \int_R^{+\infty} f_\Kappa(r,s) \Kappa^\notreH[\ut](s) ds
\ee
where $\beta_+>a/2$,
for a suitable weight $f_\Kappa$, and where $\Kappa^\notreH$ is a linear functional bounded by $\norm{r\del_r u}^\notreH + \norm{u}^\notreH$.
If $\beta_+<a$, then the homogeneous term $r^{-\beta_+}$ makes~\eqref{intro-uave} unusable to reach the harmonic decay, while for $\beta_+>a$ this term is negligible.  This motivates the \textbf{radial stability condition} $\beta_+>a$.

\item
Combining \eqref{intro-Phi-solve} and~\eqref{intro-uave} leads schematically to
\bel{intro-Phi-plus-Psi}
\Phi^\notreH[u](r) + \bigl(\text{integral of } \Psi^\notreH_\beta[u] + C \la u\ra^2 \bigr)
- \bigl(\text{integral of } C \, \Kappa^\notreH[\ut]^2\bigr) = \Obig(r^{-2a})
\ee
The need to control $\Kappa^\notreH$ motivates the \textbf{shell stability condition}, namely
\be
\Psi^\notreH_\beta[u](s) + C \bigl( \la u\ra^2 - \cradialH \Kappa^\notreH[\ut]^2 \bigr) \gtrsim (\norm{r\del_r u}^\notreH + \norm{u}^\notreH)^2 , \qquad \beta \in \{a, 2a\}
\ee
for an explicitly known $\cradialH>0$ that arises from radial Hardy inequalities.
Under this condition, the left-hand side of~\eqref{intro-Phi-plus-Psi} is positive and $\norm{r\del_r u}^\notreH, \norm{u}^\notreH = \Obig(r^{-a})$, which implies $u=\Obig(r^{-a})$.

\item To go beyond this harmonic decay rate, the shell identity~\eqref{intro-shell} is perturbed by replacing $2a$ by a slightly larger exponent.  All the steps go through except that the $\Cstar r^{-a}$ term in the expression~\eqref{intro-uave} of~$\la u\ra$ remains.  The harmonic stability condition ensures that $\ker\ssA{}^\lambda$ is one-dimensional and the radial stability condition that kernel elements have non-vanishing average.  Thus, one can shift $u$ by $r^{-a}$ times an element of the kernel to cancel $\Cstar r^{-a}$ in $\la u\ra$.  The resulting function then decays with an exponent strictly better than harmonic, as desired.

\eei

\paragraph{Improving the radial decay of the nonlinear solution.}

In \refwithname{Sections}{section=10} to~\ref{section=11}, we leverage the linear estimates to control the radial decay of ${(g-\seedmodg,h-\seedmodh)}$, or equivalently of some underlying scalar and vector fields~$(u,Z)$.
To avoid clutter we shall focus on~$u$ only.
Initially, from the existence of the seed-to-solution projection, $u$ is known to decay as~$r^{-a}$.
To improve this decay, the key is that $u$~solves a non-linear equation $\notreH^\lambda[u]=E+\Qcal[u]$ where $E$ decays strongly (as determined by~$p_A$) and the non-linearities $\Qcal[u]$ decay with a strictly better exponent than the known decay of~$u$.
The estimates on $\notreH^\lambda$ then ensure that $u = (\notreH^\lambda)^{-1}\bigl[\Ehat\bigr]$ decays like $\Ehat=E+\Qcal[u]$.
This improvement stops when reaching either the decay rate of~$E$, or the maximum exponent available in linear estimates, or an exponent of resonances between energy-momentum modulators and the seed metric.  These three cases explain the upper bound on $\pstar$ in~\eqref{exponent-range}.

Technical difficulties plague this simple story.

\bei

\item 
Since several iterations are needed to reach the desired decay exponent, it is important that all Sobolev estimates on~$u$ feature the same derivative order to avoid a large loss of regularity radially.  Elliptic regularity does not make up for such a loss, as it degrades the exponent of~$\lambda$ and would require dealing with a discrepancy between the norms of interest and the powers of~$\lambda$ appearing in the operators.

\item The radial control of solutions of the linear problem $\notreH^\lambda[u]=\Ehat$ relies on shell functionals~$\Psi^\notreH_\beta$ that involve up to three radial derivatives of~$u$.  Accounting for the invertibility of~$\notreH^\lambda$, the functionals are well-defined provided $\Ehat$ and $r\del_r\Ehat$ are \emph{both} in a negative Sobolev space~$H^{-2}$.  This additional radial regularity has to be tracked throughout each improvement of the radial decay estimates.

\item 
If $u$ is known to have a given radial decay rate in an $H^2$~sense, the coercivity of \emph{curved} linearized constraints, suitably commuted with~$r\del_r$, is used to deduce that $r\del_r u$ has that decay rate in an $H^2$~sense too.  This then implies the desired $H^{-2}$ control of $\Qcal[u]$ and $r\del_r\Qcal[u]$, which are needed in order for linear estimates on~$\notreH^\lambda$ to improve the rate.  The cyclic process is illustrated in \autoref{fig:estimates} below.
\eei
\noindent While the treatment of nonlinearities uses many of the same estimates as the iteration scheme used to construct $(g,h)$, we emphasize that at this stage $(g,h)$ (or $(u,Z)$) is \emph{already determined}.  The number of iterations is finite, but is large if $p_G>0$ is small.


\paragraph{Perspectives.}

The stability mechanism developed in this paper is not tied to the
optimal localization problem. Built around the Hamiltonian and momentum
shell functionals introduced and applied here ---and fully investigated in the
companion paper~\cite{LL-PoincareKornHardy}--- it encodes the radial
evolution of (semi-)coercive quantities while retaining the contribution
of angular fluctuations. We expect this structure to provide a
natural framework for other curvature problems in which radial
propagation, monotonicity, and semi-coercivity interact. Such functionals
play a central role in geometric analysis, notably in the study of
minimal surfaces, mean-curvature flow, and Ricci flow, where they govern
global evolution and the formation of singularities; see, for instance,
\cite{BrendleHuisken:2017,CDHS,ColdingMinicozzi}.

\begin{figure}
    \centering
    \begin{subfigure}[b]{\textwidth}
        \centering
        \includegraphics[height=3.5cm,width= 0.3\textwidth]{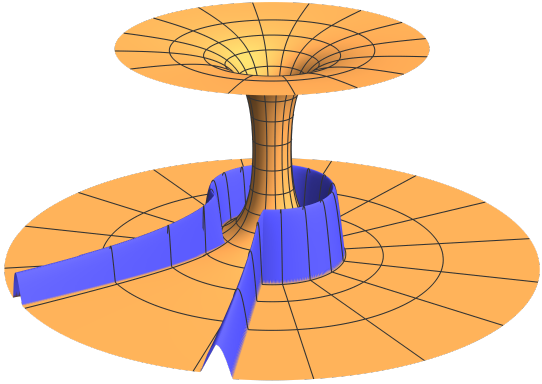}
        \includegraphics[height=3.5cm, width= 0.3\textwidth]{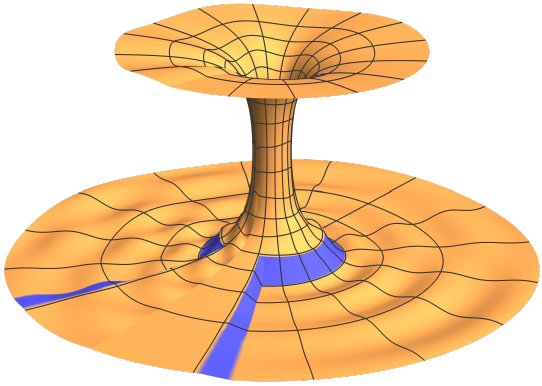}
        \includegraphics[height=3.5cm, width= 0.3\textwidth]{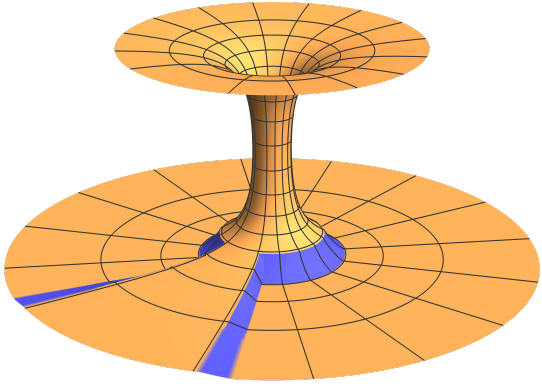}
        \caption{Gluing of the Euclidean metric (outside a conical domain) and the Schwarzschild metric (inside a conical domain). 
        {Left:~\it exact localization with sub-harmonic control.} 
        {Middle:~\it asymptotic localization with harmonic control.} 
        {Right:~\it exact localization with harmonic control.}}
        \label{figure---2} 
    \end{subfigure}
    
    \bigskip
    
    \begin{subfigure}[b]{\textwidth}
        \centering
        \includegraphics[height=4.cm,width= 0.4\textwidth]{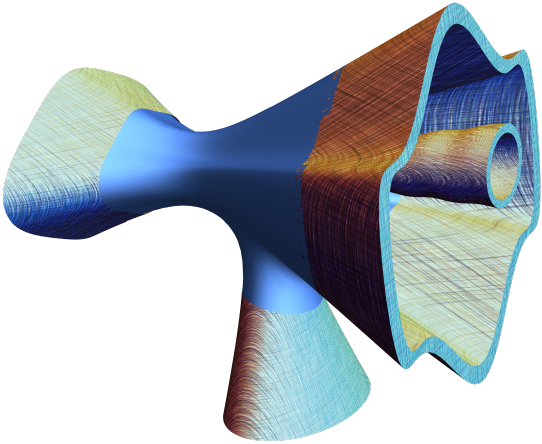}
        \caption{An example of gluing domain with multiple asymptotic ends}
        \label{figure---25}
    \end{subfigure}
    \caption{Schematic illustration for the optimal localization method  in three dimensions}
 \label{figure---111}
\end{figure}

\subsection{Outline and notation} 
\label{section=1.4}
 
\paragraph{Outline of this paper.}

In \autoref{section=2}, we begin with some notation and the notion of seed-to-solution projection (cf.~\refwithname{Definitions}{def-locali-weight} \refwithname{to}{def-mapping-2}). In \autoref{thm:sts-existence} we state standard sub-harmonic estimates which hold for a broad class of decay exponents. In \autoref{section=3}, we introduce the stability conditions that are required for the analysis of the harmonic decay of solutions, including the notions associated with the harmonic Hamiltonian and momentum operators and the modulators. Next, in \autoref{section=4} we state our main result concerning the \mbox{(super-)harmonic} estimates enjoyed by solutions to the localized constraints; cf.~\autoref{theo--beyond-harmonic}. 

While \autoref{thm:sts-existence} relies on basic properties of the linearized Einstein constraints (and is proven in \autoref{appendix=E}), most of the present paper is devoted to the proof of \autoref{theo--beyond-harmonic}.  Relying on the proposed stability structure, we derive harmonic decay estimates for the linearization of the Einstein constraints and their formal adjoints around a Euclidean data set in \refwithname{Sections}{section=5} \refwithname{to}{section=7} (Hamiltonian operator) and \refwithname{Sections}{section=8} \refwithname{and}{section=9} (momentum operator).  In particular, \refwithname{Sections}{section=6} \refwithname{and}{section=9} construct the energy-momentum modulators.
Finally, in \refwithname{Sections}{section=10} \refwithname{and}{section=11} 
we build upon all of the previous sections and complete the proof of the harmonic decay leading to \autoref{theo--beyond-harmonic}.
This is done by iteratively improving the decay exponent for the solution and its (nonlinear) source term.
Operator coefficients and structure constants that arise in our study of the Einstein operators are collected in \autoref{appendix=A}.  See also \autoref{table-structure} for the most important objects associated to the Hamiltonian and momentum.  Technical material is postponed to \autoref{appendix=B} (expressions of dissipation functionals), \autoref{appendix=C} (expansion of the constraints), \autoref{appendix=D} (geometric inequalities on cones), \autoref{appendix=E} (sub-harmonic analysis proving \autoref{thm:sts-existence}), and \autoref{appendix=F} (differential equations). 


\paragraph{Notation for the averages.}

The weighted average of a function $f\colon \Lambda \to \RR$ defined on an open set $\Lambda \subset \Sphe^{n-1}$ of the $(n-1)$-dimensional, unit sphere $\Sphe^{n-1} \subset \RR^n$ is denoted by 
\bel{equa-average} 
\la f \ra \coloneqq \fint_{\Lambda} f \, d\chi 
\coloneqq {1 \over \aire[\Lambda,\lambda]} \int_{\Lambda} f \, d\chi,   
\qquad  
\aire[\Lambda,\lambda] \coloneqq \int_{\Lambda} \, d\chi.
\ee
The notation $d\xh$ stands for the standard homogeneous measure on the unit sphere, and the weighted measure $d\chi \coloneqq \lambda^{2 \expoP} d\xh$ is determined by a function $\lambda \geq 0$ with positivity set~$\Lambda$ (a large enough integer~$\expoP$ being fixed). This notation is used at each asymptotic end labelled with the symbol $\iota=1,\ldots,I$, where $I$ denotes the finite number of asymptotic ends, and we write $\lambda_\iota$ as well as $d\chi_\iota$, $\la u \ra_\iota$, etc. Whenever a spherical shell $\Lambda_{\iota,r}$ is identified with $\Lambda_\iota$ by radial projection, $d\chi_\iota$ denotes the pullback of this angular measure; in particular, it carries no factor~$r^{n-1}$.

Sobolev norms on~$\Lambda$ are defined in~\eqref{equa-norm-poids} using this integral normalized by $\aire[\Lambda,\lambda]$, in contrast to norms on $n$-dimensional domains, defined in \autoref{section=2.2}.
 

\section{The localized seed-to-solution projection}
\label{section=2}

\subsection{A construction scheme}
\label{section=2.1}

\paragraph{Localization and weight functions.}

Our preliminary aim is to parametrize the class of Einstein's initial data sets that are close to a given ``localization data set'' ---a notion we are going to define first. We do so by introducing the notion of a seed-to-solution map, which is defined via a projection along the image of the formal adjoint of the linearized constraints. We include suitable weights in this projection in order to ensure the desired localization. While we are mainly interested in gluing at asymptotically Euclidean ends, our abstract framework encompasses, simultaneously, compact or non-compact manifolds with or without localization. In a first stage, we introduce general definitions and, next, specialize our notions to asymptotically Euclidean manifolds and finally arrive at the sub-harmonic estimates; cf.~\autoref{thm:sts-existence}, below. In the following, $\dbf_{g_0}(x,y)$ denotes the geodesic distance between any two points $x,y $ in a manifold $(\Mbf, g_0)$. (Cf.~\autoref{tab:example} for a summary of notation.)

All the reference data in \autoref{def-locali-weight} and, especially, the weights\footnote{For any functions $A,B \geq 0$, we use the notation $A \lesssim B$ whenever there exists a constant $C>0$ such that for $0 \leq A \leq C \, B$. We write $A \simeq B$ when both  $A \lesssim B$ and $B \lesssim A$ hold. When necessary, the dependence of the implied constants will be specified. The constants will never depend on the point $x$ in the manifold.} are sufficiently regular (with bounds specified next) ---except for $\lambdabf$ which is regular up to the boundary from within $\Omega$, but only Lipschitz continuous across the boundary of the gluing domain. In contrast, the regularity of other tensors will be of crucial importance throughout our construction.

\begin{table}
\centering
\begin{tabular}{lccc}
\toprule
& \textbf{gluing domain} & \textbf{localization} & \textbf{weight} 
\\
\midrule
\textbf{Manifold} & $\Omega \subset \Mbf$ & $\lambdabf: \Omega \to [0, \lambda_0]$ & $\omegabf_p = \lambdabf^{\expoP} \wtrr^{n/2-p}$
\\
\midrule
\textbf{Asymptotic end} & $\Omega_\iota \cong \Omega_R \subset \RR^n$ & $\lambda_\iota: \Lambda_\iota \to [0, \lambda_0]$  & $\omega_{\iota,p} = \lambda_\iota^{\expoP} r^{n/2-p}$
\\
\bottomrule
\end{tabular}
\caption{Main notation for the localization}
\label{tab:example}
\end{table}

\begin{definition}
\label{def-locali-weight}
A \textbf{localization manifold} $(\Mbf,\Omega, g_0,\wtrr, \lambdabf)$ consists of a Riemannian manifold $(\Mbf, g_0)$  endowed with an open set $\Omega \subset \Mbf$ (with smooth boundary) referred to as the \textbf{gluing domain}, and a pair of functions $(\wtrr, \lambdabf)$.  
\bei 
 
 \item The \textbf{radius function} $\wtrr: \Mbf \to [R_0, + \infty)$ is smooth and satisfies (for some $x_0 \in \Mbf$ and $R_0>0$)
\bel{equa-with-base-point}
\wtrr(x) \simeq \big( (R_0)^2+ (\dbf_{g_0}(x, x_0))^2 \big)^{1/2}, 
\qquad
x\in \Mbf.
\ee
For every integer $j\geq1$, it also satisfies the bound $|\nabla^j\wtrr|_{g_0}\leq C_j\wtrr^{1-j}$.

\item If $\Omega=\Mbf$ (the non-localized, or isotropic, case), one allows the choice $\lambdabf\equiv1$.  Otherwise, the \textbf{localization function} $\lambdabf : \Mbf \to [0,\lambda_0]$ is positive and smooth in~$\Omega$, extends smoothly up to $\partial\Omega$ from within~$\Omega$, is extended by zero to ${}^{\complement}\Omega$, and satisfies, for some $\lambda_0>0$, 
\bel{equa-localiz-weight} 
\lambdabf(x)
\simeq\begin{cases} 
{1 \over \wtrr(x)} \dbf_{g_0}\bigl(x, {}^{\complement} \Omega\bigr),  
& x \in \Omega, 
\\
 0, & x \in {}^{\complement} \Omega \coloneqq \Mbf \setminus \Omega. 
\end{cases}
\quad
\ee   
Moreover, for every integer $j\geq1$, $|\nabla^j\lambdabf|_{g_0}\leq C_j\wtrr^{-j}$ in~$\Omega$, and
$|\nabla\lambdabf|_{g_0}\simeq\wtrr^{-1}$ in a fixed collar of~$\partial\Omega$.
\eei
\end{definition}

In the above definition, the specific choice of the base point $x_0$ and the parameters $R_0,\lambda_0$ is unimportant. On the one hand, for {\it sub-harmonic} estimates the specific choice of the functions $\wtrr$ and $\lambdabf$ also is unimportant. On the other hand, deriving {\it harmonic} estimates will require us to be more specific about the functions $\wtrr$ (and normalize it to the standard radius in the chosen coordinate chart at each Euclidean end) and $\lambdabf$ (in order to ensure certain stability conditions). Further specifications will be indicated in the course of this paper. 


\paragraph{Proposed scheme.}

We proceed first with a {\it formal} construction scheme for data in the ``vicinity'' of a localization manifold in the sense of \autoref{def-locali-weight}. Further notions will be defined only afterwards. We are interested in deriving a parametrization of a large class of solutions to the Einstein constraints and we find it convenient to introduce a {\it projection operator}: a \textbf{localization data set} $(\Mbf, \Omega, g_0, h_0, \wtrr, \lambdabf)$ being fixed throughout (with $h_0$ a symmetric two-tensor on~$\Mbf$), we view the solution mapping of interest as a projection along ${(g_0,h_0)}$. Our construction depends upon the choice of a \textbf{variational weight} (as we call it), that is,  
\bel{equa-omega-choice}
\omegabf_p \coloneqq \lambdabf^{\expoP} \wtrr^{n/2-p}, 
\ee
in which the power of the radial variable $\wtrr$ is determined by a \textbf{projection exponent} denoted by $p \in (0,n-2)$. The \textbf{localization exponent} denoted\footnote{This is a basic requirement to prevent certain norms from blowing-up near the boundary (if $\expoP$ approaches $3/2$), but in fact $\expoP$ needs to be larger for existence results as specified in our statements later on.}  by ${\expoP \geq 2}$ determines the regularity at the boundary of the gluing domain. More precisely, we use $\omegabf_p$ for the Hamiltonian component and $\omegabf_{p+1}$ for the momentum component. 

To any localized seed data set $(\seedg, \seedh)$ (cf.~\autoref{def-aset}, below) we want to associate an actual solution $(g,h)$ to the Einstein constraints
$\Gcal(g,h) = 0$, 
which we seek as a deformation of $(\seedg,\seedh)$ in the form 
\bel{equa-deform} 
g = \seedg + \gdiff,
\qquad 
h = \seedh + \hdiff. 
\ee
Specifically, we require that the \emph{deformation} $(\gdiff,\hdiff)$ belongs to the image of the adjoint $d\Gcal_{(g_0,h_0)}^*$ of the linearized Hamiltonian and momentum operators, in the sense that there exist a scalar field $u$ and a one-form field $Z$ so that\footnote{The duality is understood with respect to the metric structure $g_0$.  Detailed expressions are given in \autoref{appendix=C}.}
\bel{equa--221}
\aligned
\gdiff & = \omegabf_p^2 \, d\Hcal_{(g_0,h_0)}^{*\flat\flat}[u,Z], 
\qquad
\hdiff = \omegabf_{p+1}^2  \, d\Mcal_{(g_0,h_0)}^{*\sharp\sharp}[u,Z].
\endaligned
\ee
Observe that, in our formulation, the linearization is taking place {\it at the (fixed) base point} $(g_0, h_0)$ rather than at the seed data, which was the original choice in~\cite{CarlottoSchoen,CorvinoSchoen}). 


Later in this section, we will specify a choice of localization (i.e.~asymptotically Euclidean solutions in conical gluing domains) for which a solution $(u,Z)$ exists (cf.~\autoref{thm:sts-existence}, below).
The proof (in \autoref{appendix=E}) relies on a Newton iteration scheme for the nonlinear problem
\be
\Gcal\bigl(\seedg + \omegabf_p^2 \, d\Hcal_{(g_0,h_0)}^{*\flat\flat}[u,Z], \seedh + \omegabf_{p+1}^2  \, d\Mcal_{(g_0,h_0)}^{*\sharp\sharp}[u,Z] \bigr) = 0 .
\ee
Starting from $(g_1,h_1)=(\seedg,\seedh)$, we iteratively construct (for $k\geq 1$)
\bse\label{Newton-it}
\be
g_{k+1} = g_k + \omegabf_p^2 \, d\Hcal_{(g_0,h_0)}^{*\flat\flat}[u_k,Z_k] ,
\qquad
h_{k+1} = h_k + \omegabf_{p+1}^2  \, d\Mcal_{(g_0,h_0)}^{*\sharp\sharp}[u_k,Z_k] ,
\ee
with fields $(u_k,Z_k)$ taken as the solution of a linear problem that involves the linearization of $\Gcal$ \emph{at the point} $(g_k,h_k)$ in addition to the dual operators $d\Hcal^*$ and~$d\Mcal^*$ at $(g_0,h_0)$,
\bel{equa-Jcal-def}
\aligned
\Jcal_{(g_k,h_k;g_0,h_0)}[u_k,Z_k] & = - \Gcal(g_k,h_k) ,
\\
\Jcal_{(g_k,h_k;g_0,h_0)}[u,Z]
\, & \! \coloneqq d\Gcal_{(g_k,h_k)}\Bigl[ \omegabf_p^2 \, d\Hcal_{(g_0,h_0)}^{*\flat\flat}[u,Z], \omegabf_{p+1}^2  \, d\Mcal_{(g_0,h_0)}^{*\sharp\sharp}[u,Z] \Bigr] .
\endaligned
\ee
\ese

Convergence of the iteration relies on two observations.
Firstly, the linear operator $\Jcal_{(g_k,h_k;g_0,h_0)}$ is invertible thanks to the Lax--Milgram theorem, whose coercivity condition is ensured by weighted Poincaré, Korn, and Hardy inequalities (and non-existence of Killing initial data sets).
The operator is further checked to be \emph{elliptic in the sense of Douglis--Nirenberg}~\cite{DouglisNirenberg} so that standard interior elliptic regularity results apply and suitable pointwise norms of the solution $(u_k,Z_k)$ are controlled by the source term~$\Gcal(g_k,h_k)$.
Secondly, the next source term $\Gcal(g_{k+1},h_{k+1})$ is small since it is the (quadratic) nonlinear remainder of the Einstein operators
\bel{Gcalgk1hk1}
\Gcal(g_{k+1},h_{k+1})
= \Gcal(g_k + \gdiff_k, h_k + \hdiff_k) - \Gcal(g_k,h_k) - d\Gcal_{(g_k,h_k)}[\gdiff_k, \hdiff_k] \eqqcolon \Qcal\Gcal_{(g_k,h_k)}[\gdiff_k, \hdiff_k] ,
\ee
evaluated at $\gdiff_k=\omegabf_p^2 \, d\Hcal_{(g_0,h_0)}^{*\flat\flat}[u_k,Z_k]$ and $\hdiff_k=\omegabf_{p+1}^2  \, d\Mcal_{(g_0,h_0)}^{*\sharp\sharp}[u_k,Z_k]$.
An important subtlety here is that the nonlinearities involve higher derivatives of $(g_0,h_0)$ than of $(g_k,h_k)$.  As a result, while the reference data $(g_0,h_0)$ must feature higher regularity than the solutions $(g,h)$ we eventually construct, there is no loss of regularity from one step to the next in the iteration.
Our choice of defining the seed-to-solution projection in terms of a fixed operator $d\Gcal_{(g_0,h_0)}^*$ is thus crucial in allowing us to use Newton iteration.
In turn, Newton iteration only requires us to prove smallness of the quadratic terms~$\Qcal\Gcal$.
This is simpler than showing their Lipschitz continuity with respect to $(\gdiff,\hdiff)$, needed in the Picard iteration scheme used originally in~\cite{CarlottoSchoen,CorvinoSchoen}.
 

\subsection{Proposed parametrization and norms}
\label{section=2.2}

\paragraph{Localized weighted norms.}

We now introduce the norms that are suitable for transforming our formal scheme above into an actual proof leading to quantitative statements. We are interested in tensor fields defined on the localization manifold $(\Mbf,\Omega, g_0,\wtrr, \lambdabf)$ and, for simplicity, we omit from the notation the metric and the functions $\wtrr, \lambdabf$, and sometimes do not specify the domains of integration $\Mbf, \Omega$. For a comprehensive study of the standard weighted H\"older and Sobolev spaces, we refer to~\cite{Bartnik,ChoquetC}. 
\bei

\item {\it Weighted H\"older norms.} Given $\alpha \in (0,1)$, a non-negative integer $l$, and reals $p, a \in \RR$, we define the space $C_{p,a}^{l,\alpha}(\Omega)$ as the set of tensor fields $f$ on $\Omega$ with local H\"older regularity of order $l + \alpha$ and finite weighted pointwise norm\footnote{Observe that the norm depends on the metric $g_0$ even for scalar functions.} 
\begin{subequations} 
\bel{equa-norm-Hol}
\aligned
& \| f \|_{\Omega, p,a}^{l, \alpha}
= 
\| f \|_{C_{p,a}^{l, \alpha}(\Omega)}
\\
& \coloneqq 
\sum_{|L| \leq l}   \sup_\Omega \Big( \lambdabf^{-a+|L|} \, \wtrr^{p+|L|} \, |\nabla^L f|_{g_0} \Big) 
 + \sum_{|L| = l}   \sup_\Omega \Big( \lambdabf^{-a + |L|+ \alpha} \, \wtrr^{p+|L|+ \alpha} \, \llbracket \nabla^L f \rrbracket_{\Omega,\alpha} \Big). 
\endaligned
\ee
Here, $\nabla^L f$ denotes covariant derivatives for the metric $g_0$ and $L$ represents a multi-index, while
\bel{equa-doublebarre}
\llbracket f \rrbracket_{\Omega,\alpha} (x) 
\coloneqq   
\sup_{y \in A_{\Omega,g_0}(x)}\biggl(
\frac{|\mathsf P_{y\to x}f(y)-f(x)|_{g_0(x)}}
{\dbf_{g_0}(x,y)^\alpha}
\biggr),
\qquad x \in \Omega
\ee  
with $A_{\Omega,g_0}(x)$ the set of points $y\in\Omega$ such that
\be
\dbf_{g_0}(x,y) < \frac{1}{2} \min\bigl(\operatorname{inj}_{g_0}(x),\dbf_{g_0}(x,{}^\complement\Omega)\bigr),
\ee
where $\mathsf P_{y\to x}$ denotes parallel transport along the unique minimizing $g_0$-geodesic from $y$ to~$x$, and the distance to the empty set is understood to be~$+\infty$. 
We occasionally use the $C_{p,a}^l(\Omega)$ norm, in which the term $\llbracket \nabla^L f \rrbracket_{\Omega,\alpha}$ is omitted.
\end{subequations}


\item {\it Weighted Lebesgue norms.} Given any $p, a \in \RR$, we define the space $L^m_{p,a}(\Omega)$ (for $m=1,2$) by completion of the set of smooth tensor fields $f$ with finite weighted integral norm
\bel{equa-def-weightL}
\|f\|_{L^m_{p,a}(\Omega)}
\coloneqq  \Big(
\int_\Omega  | f |_{g_0}^m \lambdabf^{-ma} \, \wtrr^{m p-n} \, \dVol_{g_0}
\Big)^{1/m}, 
\ee
where we recall that $\dVol_{g_0}$ is the volume form associated with the metric $g_0$. When the second index vanishes, we simply write $L^m_p(\Omega)$.  Similarly, we also define weighted Sobolev spaces $H^l_{p,a}(\Omega)$
\bel{equa-def-Hl-space}
\|f\|_{H^l_{p,a}(\Omega)}
\coloneqq  \Big(
\sum_{|L|\leq l} \int_\Omega  | \nabla^L f |_{g_0}^2 \lambdabf^{-2a} \, \wtrr^{2|L|+2p-n} \, \dVol_{g_0}
\Big)^{1/2} .
\ee
In contrast to the H\"older norm~\eqref{equa-norm-Hol}, the power of $\lambdabf$ is the same for all derivatives.

\item {\it Weighted H\"older--Lebesgue norms.} By combining the previous two definitions, it is also useful to consider the space $C_{p,a}^{l,\alpha}(\Omega) \cap L^2_{p,b}(\Omega)$ with the norm
\bel{equa-def-weightLC}
\Norm{f}_{\Omega,p,a,b}^{l,\alpha}
\coloneqq \| f \|_{\Omega,p,a}^{l,\alpha} + \| f \|_{L^2_{p,b}(\Omega)}
= \| f \|_{C_{p,a}^{l,\alpha}(\Omega)} + \| f \|_{L^2_{p,b}(\Omega)} ,
\ee  
in which a possibly different exponent $b\in\RR$ is introduced for the $L^2$ factor. This norm requires some integrability at infinity for the (undifferentiated) function, which is only slightly stronger radially than the pointwise bound implied by the H\"older factor. (We typically take $b\geq a$, so that the integrability is stronger in the angular direction as well.)

\eei

\noindent When emphasis is required, we specify our choice of metric and write $\| f \|_{\Omega, g_0, p,a}^{l, \alpha}$, $\|f\|_{L^m_{p,a}(\Omega, g_0)}$, etc. Furthermore, when localization is not required, a similar notation as above is also used by replacing $\lambdabf$ by $1$ (the irrelevant exponent $a$ being then omitted) and possibly $\Omega$ by~$\Mbf$.  This is \emph{not} equivalent to setting $a$ to a particular value.  Explicitly,
\be
\aligned
\| f \|_{\Omega, p}^{l, \alpha}
= 
\| f \|_{C_p^{l, \alpha}(\Omega)}
\coloneqq 
\sum_{|L| \leq l}   \sup_\Omega \Big( \wtrr^{p+|L|} \, |\nabla^L f|_{g_0} \Big) 
 + \sum_{|L| = l}   \sup_\Omega \Big( \wtrr^{p+|L|+ \alpha} \, \llbracket \nabla^L f \rrbracket_{\Omega,\alpha} \Big). 
\endaligned
\ee


\paragraph{Choice of regularity.}

We wish to modify a given data set {\it within} the gluing domain~$\Omega$, while the (possibly non-vacuum) constraints are assumed to be `already' satisfied by the data {\it outside} the gluing domain. The regularity at the boundary of the gluing domain is determined by a \textbf{localization exponent} $\expoP$ which arises in the variational formulation and, later on, by shifted exponents $\expoPm < \expoP$ and $\expoPp = 2\expoP - \expoPm > \expoP$ (after applying elliptic regularity).
We fix throughout our work these exponents as well as parameters $N,\alpha$ for the H\"older regularity of data sets on~$\Mbf$. We summarize the conditions on these exponents as follows.

\begin{definition}
\label{def-expoP}
Given a \textbf{regularity exponent} $N\geq 3$ and a \textbf{H\"older exponent} $\alpha \in (0,1)$, a triple $(\expoPm, \expoP, \expoPp)$ is called an \textbf{admissible set of localization exponents} provided\footnote{The condition $\expoPm>N+\alpha$ ensures that the data sets we obtain are $C^{N,\alpha}$ regular across the boundary of~$\Omega$, which is required to control linear operators in the iterative construction of solutions.  The condition $\expoPm>1+\expoP/2$ is only used to avoid technicalities explained above \autoref{lem:appE-control-nonlin}, and could be relaxed to $\expoPm>2$.}
\bel{equa-expoPPP}
\max(N + \alpha, 1+\expoP/2) < \expoPm < \expoP < \expoPp , \qquad \expoPp - \expoP = \expoP - \expoPm \geq \frac{n+4}{2} .
\ee
\end{definition}


\paragraph{The notion of seed data.}

Our construction relies on projecting a `seed' data set, which is an approximate solution of the constraints, onto the space of exact solutions.
The localization data set, defined next, will play the role of a `reference' in our projection scheme and it is natural to assume it to have slightly better differentiability in comparison to the seed data sets (defined below) or the actual solutions.  Observe that the smallness conditions stated below concern the differences $\seedg - g_0$ and $\seedh - h_0$ only; no support condition is imposed on these differences, whereas the correction from the seed to the solution is localized to the gluing domain. At this stage, we state a definition for arbitrary exponents but, later on, restrictions will be required in order to reach existence and decay results. 

\begin{definition}
\label{def-locali-weight-3}
Consider a localization manifold $(\Mbf,\Omega, g_0, \wtrr, \lambdabf)$.
Given an exponent $p_G>0$ referred to as the \textbf{geometry exponent}, one calls $(\Mbf, \Omega, g_0, h_0, \wtrr, \lambdabf)$ a (reference) \textbf{localization data set} if, in addition, 
\bei 

\item $h_0$ is a symmetric $(2,0)$-tensor defined on~$\Mbf$, 

\item $g_0$ and $h_0$ are $C^{N+2,\alpha}$ and $C^{N+1, \alpha}$ regular, respectively, and 

\item and the integral decay condition
\bel{equa-nearEins-00-new}
\Norm{\Riem_{g_0}}^{N,\alpha}_{\Mbf, g_0, p_G+2}
+ \Norm{h_0}^{N+1, \alpha}_{\Mbf,g_0,p_G+1}
< +\infty. 
\ee
\eei 
\end{definition}

Observe that, under the condition \eqref{equa-nearEins-00-new}, the reference data satisfies 
\bel{equa-nearEins-00}
\Norm{\Hcal(g_0,h_0)}^{N,\alpha}_{\Mbf,g_0, p_G+2}
+ \Norm{\Mcal(g_0, h_0)}^{N, \alpha}_{\Mbf, g_0,p_G+2}
< +\infty ,
\ee

\begin{definition} 
\label{def-aset}
Let $(\Mbf, \Omega, g_0,h_0, \wtrr, \lambdabf)$ be a localization data set with geometry exponent $p_G>0$. Given $\eps_G,\eps_A>0$ and any exponent $p_A \geq p_G$ referred to as the \textbf{accuracy exponent}, a \textbf{localized seed data set} $(\seedg, \seedh)$ consists of fields defined on the whole manifold~$\Mbf$ and satisfying the following conditions.  
\bei 

\item \emph{Near-localization data:} $\seedg$ is a Riemannian metric and $\seedh$ is a symmetric $(2,0)$-tensor, both with $C^{N,\alpha}$~H\"older regularity, satisfying
\bel{equa-near-refe} 
\| \seedg - g_0 \|^{N, \alpha}_{\Mbf, p_G}
+ \| \seedh - h_0 \|^{N,\alpha}_{\Mbf, p_G+1}\leq \eps_G .
\ee

\item \emph{Near-Einsteinian data:} the Einstein operators satisfy, in H\"older--Lebesgue norms in the gluing domain, 
\bel{equa-nearEins}
\Err_{p_A}[\seedg,\seedh] = \Norm{\Hcal(\seedg,\seedh)}^{N-2,\alpha}_{\Omega,  p_A+2, \expoPm-2,\expoP}
+ \Norm{\Mcal(\seedg, \seedh)}^{N-1, \alpha}_{\Omega, p_A+2, \expoPm-1,\expoP} \leq \eps_A ,
\ee
and $\Gcal(\seedg,\seedh)=0$ in the complement domain~${}^\complement\Omega$.
\eei  
\end{definition}

In view of the above definitions, we introduce the following collection of all localized seed data sets associated with a given localization data set $(\Mbf, \Omega, g_0, h_0)$,
\be
(\seedg,\seedh) \in \Seed(\Omega, g_0,h_0,p_G, p_A,\eps_G,\eps_A) .
\ee 
The conditions~\eqref{equa-near-refe} and~\eqref{equa-nearEins} on the localized seed data are requirements\footnote{The subscripts $G$ and $A$ in $(p_G,\eps_G)$ and $(p_A,\eps_A)$ refer the words ``geometry'' and ``accuracy'', respectively.} of very different nature:~\eqref{equa-near-refe} determines the decay of the geometry, while~\eqref{equa-nearEins} controls the `remainder' associated with the Einstein constraints. The following observations are in order. 

\bei 

\item The extrinsic curvature $\seedh$ is expected to decay faster than the metric itself (as in the Kerr solution), so the pair of exponents $(p_G,p_G+1)$ in~\eqref{equa-near-refe}  and $(p_A+2,p_A+2)$ in~\eqref{equa-nearEins} are natural, in view of the fact that $\Hcal$ and $\Mcal$ are second and first-order operators, respectively. 
 
\item At the boundary of the domain~$\Omega$, the data set $(\seedg, \seedh)$ satisfies Einstein's vacuum constraints since $\Hcal(\seedg,\seedh)$ and $\Mcal(\seedg,\seedh)$ decay pointwise as $\lambdabf^{\expoPm-2}$ and~$\lambdabf^{\expoPm-1}$, respectively.  This is compatible with imposing that $(\seedg, \seedh)$ satisfies Einstein's vacuum constraints in the complement domain ${}^{\complement} \Omega$.

\item The inequalities~\eqref{equa-nearEins} on the Einstein operators, in principle, could be deduced for $p_A=p_G$ from the inequalities~\eqref{equa-near-refe} on the geometric data, with a suitable~$\eps_A$. However, we emphasize that we are interested in seed data that do represent ``accurate'' approximate solutions, namely those for which $p_A$ is {\it strictly greater} than $p_G$. 

\eei


\paragraph{Proposed parametrization.}

We recall that an order of regularity $N\geq 3$ and a H\"older exponent $\alpha \in (0,1)$ are fixed. 
No smallness assumption is imposed on the localization data set in the following definition; the existence theorem will require the reference constraint error~$\Err_{p_G}[g_0,h_0]$ and the seed perturbation to be sufficiently small. 
We recall from \autoref{def-expoP} that $\expoP \geq 2$ is fixed for the variational functional, while $\expoPm$ and $\expoPp = 2\expoP-\expoPm$, subject to $\max(N+\alpha,1+\expoP/2) < \expoPm \leq \expoP - (n+4)/2$, arise when applying interior elliptic regularity.  On the other hand, the range of the projection exponent $p$ is essential and will become clear in \autoref{section=2.3}. 

\begin{definition}
\label{def-mapping}
Let $(\Mbf, \Omega, g_0,h_0, \wtrr, \lambdabf)$  be a localization data set with geometry exponent $p_G>0$. Fix $\eps_G,\eps_A>0$, an accuracy exponent $p_A \geq p_G$, and a projection exponent $p\in(0,n-2) \cap(0,p_A]$. By definition, the \textbf{localized seed-to-solution projection} in the vicinity of this localization data set, denoted by~$\Solu^{\lambdabf}_{n,p}$, maps any data set $(\seedg, \seedh) \in \Seed(\Omega, g_0,h_0,p_G, p_A, \eps_G,\eps_A)$ (cf.~\autoref{def-aset}) to a solution $(g,h) = \Solu^{\lambdabf}_{n,p}(\seedg,\seedh)$ of the Einstein constraints $\Gcal(g,h)= 0$, enjoying  the following properties.

\bei 

\item The pairs $(g,h)$ and $(\seedg,\seedh)$ coincide in the complement domain, that is, 
\bel{equa-gh-uZ-00} 
g = \seedg, 
\quad  h = \seedh 
\qquad \text { in } {}^{\complement} \Omega.
\ee

\item The solution $(g,h)$ is close to the data $(\seedg,\seedh)$ in the gluing domain, in the sense that (in weighted H\"older norm with the positive exponent~$\expoPm$)
\bel{equa-near-seed}
\aligned
\bigl\| g - \seedg \bigr\|^{N, \alpha}_{\Omega, g_0,p,\expoPm}
+ \bigl\| h - \seedh \bigr\|^{N,\alpha}_{\Omega, g_0,p+1,\expoPm}
& \lesssim \Err_p[\seedg,\seedh],  
\endaligned
\ee
in terms of the H\"older--Lebesgue norm\footnote{This norm is generalized in~\eqref{equa-Ecal-pstar} to a norm $\Err^+_{\pstar}[\seedg,\seedh]$ with an additional term for $\pstar=n-2$, absent here since $p<n-2$.} 
\bel{equa-Ecal-pq} 
\aligned
\Err_p[\seedg,\seedh]
\coloneqq
\Norm[\big]{\Hcal(\seedg,\seedh)}^{N-2,\alpha}_{\Omega, g_0, p+2, \expoPm-2, \expoP}
+ \Norm[\big]{\Mcal(\seedg,\seedh)}^{N-1,\alpha}_{\Omega, g_0, p+2, \expoPm-1, \expoP}
\endaligned
\ee
and the implied constant depends upon the localization data set and the exponents. 

\item  The pair $(g,h)$ selected by the construction is a vacuum solution for which there exist a scalar field $u$ and a vector field $Z$ defined in the gluing domain $\Omega$ such that 
\bel{equa-gh-uZ} 
\aligned 
& \Hcal(g,h) = 0,   
&& \Mcal(g,h) = 0,
\\
& g = \seedg + \omegabf_p^2 \, d\Hcal^{*\flat\flat}_{(g_0,h_0)}(u,Z), 
\quad
&& h = \seedh + \omegabf_{p+1}^2 \, d\Mcal^{*\sharp\sharp}_{(g_0,h_0)}(u,Z), 
\endaligned
\ee
where $(u,Z)$ satisfies (in a weighted H\"older norm\footnote{We will also control suitable weighted $H^k$ Sobolev norms of $(u,Z)$, hence weighted $L^2$ norms of $(g,h)$.} with the negative exponent $- \expoPp$) 
\bel{equa-213-2-u-Z}
\aligned
& \| u \|^{N+2, \alpha}_{\Omega, n-2-p, -\expoPp+2}
+ \| Z \|^{N+1, \alpha}_{\Omega, n-2-p, -\expoPp+1}
 \lesssim \Err_p[\seedg,\seedh]. 
\endaligned 
\ee 
\eei
\end{definition}

Finally, the proposed parametrization is made evident by defining an equivalence relation, i.e.
the \emph{equivalentce modulo adjoint-constraints.} We point out the analogy with the Yamabe problem and the conformal method, for which conformal classes are introduced in order to make a certain classification of metrics (cf., for instance,~\cite{Holst,LeeParker}). When dealing with initial data sets, there is no canonical notion of equivalence class, and our parametrization offers a standpoint, which depends on the prescribed localization data set $(\Mbf,\Omega,g_0,h_0,\wtrr,\lambdabf)$ and on the projection exponent~$p$.
Describing the projection in terms of equivalence classes relies on the feature that the data sets of interest $(\seedg,\seedh)$ and $(g,h)$ have the same regularity.

\begin{definition} 
\label{def-mapping-2}
Assume the conditions in \autoref{def-mapping} on the localization data set $(\Mbf, \Omega, g_0,\allowbreak h_0, \wtrr, \lambdabf)$ and exponents $(p,p_G, \expoPm,\expoP,\expoPp)$. Two data sets $(g,h)$ and $(g',h')$, for which $g-g_0,g'-g_0\in C^{N,\alpha}_{p_G}(\Mbf)$ and $h-h_0,h'-h_0\in C^{N,\alpha}_{p_G+1}(\Mbf)$, are deemed \textbf{equivalent modulo adjoint-constraints} if they agree on ${}^{\complement}\Omega$ and, in~$\Omega$,
\be
g - g' = \omegabf_p^2 \, d\Hcal^{*\flat\flat}_{(g_0,h_0)}(u,Z), 
\qquad
h - h' = \omegabf_{p+1}^2 \, d\Mcal^{*\sharp\sharp}_{(g_0,h_0)}(u,Z)
\ee 
for a scalar field $u\in C^{N+2,\alpha}_{n-2-p,-\expoPp+2}(\Omega)$ and a one-form $Z \in C^{N+1,\alpha}_{n-2-p,-\expoPp+1}(\Omega;T^*\Omega)$. 
In particular, $g-g'$ and $h-h'$ have finite $C^{N,\alpha}_{p,\expoPm}$ and $C^{N,\alpha}_{p+1,\expoPm}$ norms, respectively.
In this language, the solution map $\Solu^{\lambdabf}_{n,p}$ sends an element $(\seedg,\seedh)$ to a representative $(g,h)$ of its equivalence class $[(\seedg,\seedh)]$ that satisfies Einstein's constraint equations, namely,
\bel{equa-sol-map}
\Solu^{\lambdabf}_{n,p} \colon
\Seed(\Omega, g_0,h_0,p_G, p_A,\eps_G,\eps_A) \ni (\seedg,\seedh)
\longmapsto (g,h) \in [(\seedg,\seedh)] .
\ee
\end{definition}

\begin{remark}
1. Instead of vacuum solutions we could consider solutions to the Einstein-matter system $\Gcal(g,h) = (\Hstar, \Mstar)$, when the matter fields are $(\Hstar,\Mstar)={(\phi'^2 + |d\phi|_g^2, \phi' \, d\phi)}$. Here, the data $\phi, \phi'$ should be  scalar fields prescribed over the manifold $\Mbf$. The techniques providing the existence of the seed-to-solution map are expected to apply without significant change. Importantly, in the fixed-point argument used in the construction of solutions, the additional terms involve {\it no} derivatives with respect to the principal variables. 

2. A slightly more general class of solutions can be constructed by replacing the weight $\omega_{p+1}$ in front of the momentum operator by a weight $\omega_q$, where the exponent $q$ is assumed to be smaller than, or equal to, $p+1$. The relevant range for the exponent $q$ was investigated in~\cite{LeFlochNguyen} for non-localized solutions; rather direct modifications would allow one to rewrite our theory of localized solutions for such a range of exponents $q$. 
\end{remark} 


\subsection{Existence theory for conical localization}
\label{section=2.3}

\paragraph{Data sets of interest.}

Within the proposed framework, our next aim is to establish the existence of solutions to the projection problem and investigate their decay properties. We focus on the class of asymptotically Euclidean solutions and, in the light of \autoref{def-mapping-2}, we study the projection map $\Solu^{\lambdabf}_{n,p}$ around an asymptotically Euclidean data set $(g_0,h_0)$ with an asymptotically conical localization domain~$\Omega$, in a sense we now define. Specific conditions on the exponents $(p,p_G, p_A)$ arise at this stage and, in order to proceed, we need some terminology.

\begin{definition}
\label{def-conical}
Given a geometry exponent $p_G>0$, a localization data set $(\Mbf, \Omega, g_0, h_0, \wtrr, \lambdabf)$ is called a \textbf{conical localization data set} if the gluing domain has a finite nonzero number $I$ of connected asymptotic ends~$\Omega_\iota$ (labeled with $\iota = 1, \ldots, I$) on which the following conditions hold, and whose union covers all but a (large) bounded domain, namely
\bel{equa-decomp-Omega-Omega}
\Omega = \Omega_0 \cup \hskip-.15cm \bigcup_{\iota =1,\ldots,I} \Omega_\iota, 
\qquad
\Omega_\iota \cap {\Omega_\iota'} = \emptyset \text{ (for $\iota \neq \iota'$),} 
\qquad 
\Omega_0  \, \text{ bounded.}
\ee

\bei 
 
\item \textbf{Asymptotically Euclidean.} Each $\Omega_\iota$ is connected and endowed with a global chart of coordinates $x=(x^j)$ to which one associates the Euclidean metric $\delta$ in these coordinates, namely $
\delta \coloneqq \sum_j (dx^j)^2$ in $\Omega_\iota$, 
in which $r^2 = |x|^2 \coloneqq \sum_j (x^j)^2 = \wtrr(x)^2$ is identified (by convention) with the decay weight of \autoref{def-locali-weight}. It is required that the reference data set enjoys the following decay\footnote{This decay can equivalently be written with $(\Omega_\iota,g_0)$ replaced by $(\Omega_R,\delta)$ in the norms, namely as the finiteness of the norms $\| g_0 - \delta  \|^{N+2, \alpha}_{\Omega_R, \delta, p_G}$ and $\| h_0 \|^{N+1,\alpha}_{\Omega_R, \delta, p_G+1}$.  The equivalence is based on noting that first-order derivatives $\nabla_0(g_0-\delta)$ and $\del(g_0-\delta)$ are two descriptions of the Christoffel symbols of~$g_0$.} in each end~$\Omega_\iota$ (as defined above, this norm is not localized, namely does not involve any power of~$\lambda$)
\bel{equa-slightlybetter}
\max_\iota 
\Bigl( \| g_0 - \delta  \|^{N+2, \alpha}_{\Omega_\iota, g_0, p_G}
+ \| h_0 \|^{N+1,\alpha}_{\Omega_\iota, g_0, p_G+1}
\Bigr) < +\infty .
\ee
Every connected component of~$\Omega$ is required to contain at least one of the ends~$\Omega_\iota$; in particular, bounded connected components are excluded.
 
\item \textbf{Conical localization.} The coordinates under consideration define a diffeomorphism $\phi_\iota\colon\Omega_\iota \overset{\sim}\to K_\iota \cap {}^{\complement}\Ball_R$ to the intersection of the exterior of a ball $\Ball_R \subset \RR^n$ with radius $R>1$ (fixed henceforth) and a cone\footnote{Explicitly, $K_\iota$ is the open cone $\{x\in\RR^n\setminus\{0\}:x/|x|\in\Lambda_\iota\}$ over an open subset $\Lambda_\iota$ of the sphere at infinity.}  $K_\iota \subset\RR^n$.

\bei 

\item The restriction of the localization function $\lambdabf: \Omega_\iota \to (0, \lambda_0]$ is \emph{scale-invariant} in the sense that, in the coordinates provided at each asymptotic end, $\lambdabf(\mu x) = \lambdabf(x)$ for all $x \in \Omega_i$ and all $\mu\geq 1$. 

\item Consequently, one can identify it with a function $\lambda_\iota: \Sphe^{n-1} \to [0, \lambda_0]$ (without boldface, as it is not defined on the whole gluing domain), the interior of its support being denoted by $\Lambda_\iota = K_\iota \cap \Sphe^{n-1}$. Connectedness of $\Omega_\iota$ ensures that $\Lambda_\iota$ is connected.

\eei
\eei  
\end{definition} 

We emphasize that our definition requires at least one asymptotic end for each connected component of $\Omega$. 
This is used in the global Poincaré--Korn--Hardy inequalities of \autoref{lem:PoincareKornHardyD} below. It also allows for different ends~$\Omega_\iota$ of the gluing domain to lie in the {\it same} asymptotic end of the manifold under consideration. We are now interested in the solutions determined by the localized projection mapping. Interestingly, such solutions exist for a broad range of  decay exponents and general localization functions, as we now present. 
We introduce a collection of smooth radial cutoff functions $\kappa_\iota\colon\Omega_\iota\to[0,1]$ ($\iota=1,\ldots,I$) defined on~$\Omega_\iota$ and such that, for some radii $R < R_2 < R_3 < R_4$ (in the coordinates on~$\Omega_\iota$ defined by \autoref{def-conical}), $\kappa_\iota(x)$~vanishes identically for $|x|\leq R_2$ and is identically~$1$ for $|x|\geq R_3$, and $\Omega_0$ contains the pre-image under~$\phi_\iota$ of the truncated conical region $K_\iota\cap\Ball_{R_4}\cap{}^\complement\Ball_R$, in other words
\bel{equa-partition}
(\phi_{\iota *}\kappa_\iota)|_{K_\iota\cap \Ball_{R_2}} = 0 , \qquad
(\phi_{\iota *}\kappa_\iota)|_{K_\iota\cap {}^\complement \Ball_{R_3}} = 1 , \qquad
\phi_\iota^{-1}(K_\iota\cap\Ball_{R_4}\cap {}^\complement \Ball_R) \subset \Omega_0 .
\ee
The first condition ensures that $\kappa_\iota$ vanishes in a neighborhood of $\del\Omega_\iota$, which avoids difficulties related to boundary conditions at $r=R$ when analyzing the Hamiltonian and momentum operators on asymptotic ends. The second and third state that $\Omega_0$ is large enough so that the complement $\Omega\setminus\Omega_0$ consists of asymptotic ends on which exactly one function~$\kappa_\iota$ is identically equal to~$1$ and all others vanish.


\paragraph{Existence theory.}

We arrive at the statement of our existence result. Our basic assumptions on the regularity and localization exponents are given in \autoref{def-expoP}, and those on the decay exponents are as follows. The lower bound $p>0$ below is required for the gluing construction to provide asymptotically Euclidean solutions, while the upper bound $p<n-2$ is required for suitably weighted Poincaré, Korn, and Hardy inequalities to hold. 

\begin{definition} 
\label{def-triple}
A triple $(p,p_G,p_A)$ is called an \textbf{admissible set of decay exponents} provided the following conditions hold: 
\bel{condi--1-repeat} 
\aligned 
& \text{Admissible projection exponent } \, & p & \in (0,n-2). 
\\
& \text{Admissible geometry exponent } & p_G & > 0.
\\
& \text{Admissible accuracy exponent } & p_A & \geq \max(p_G, p).
\endaligned
\ee
\end{definition}

We also introduce a constant that quantifies the coercivity of the squared localized constraint operator $\Jcal_{(g_0,h_0;g_0,h_0)}[u,Z]$ defined in~\eqref{equa-Jcal-def}, considered now with a single data set~$(g_0,h_0)$.  The existence of $\CPKHzero$ is proven in \autoref{lem:squaredquadraticform}.  The weighted Sobolev norms on the left-hand side involve different powers of $\rbf$~for each order of derivatives of $u,Z$.

\begin{definition}\label{def:PoincareKornHardyConst}
For the fixed exponents $p\in(0,n-2)$ and $\expoP\geq2$, the \textbf{Poincaré--Korn--Hardy constant} of a conical localization data set $(\Mbf, \Omega,\allowbreak g_0, h_0, \wtrr, \lambdabf)$ is the optimal (smallest admissible) constant $\CPKHzero>0$ such that for every $u\in H^2_{n-2-p,-\expoP}(\Omega,g_0)$ and every one-form $Z\in H^1_{n-2-p,-\expoP}(\Omega,g_0)$, one has
\be
\aligned
& \|u\|_{H^2_{n-2-p,-\expoP}(\Omega,g_0)}^2 + \|Z\|_{H^1_{n-2-p,-\expoP}(\Omega,g_0)}^2
\\
& \leq (\CPKHzero)^2 \int_{\Mbf} \Bigl( \omegabf_p^2 \bigl| d\Hcal_{(g_0,h_0)}^*[u,Z]\bigr|_{g_0}^2 + \omegabf_{p+1}^2 \bigl| d\Mcal_{(g_0,h_0)}^*[u,Z]\bigr|_{g_0}^2 \Bigr) \dVol_{g_0} .
\endaligned
\ee
\end{definition}

We have at the main result of this section.

\begin{theorem}[The localized seed-to-solution projection]
\label{thm:sts-existence}
Given admissible sets of decay and localization exponents $(p, p_G,p_A)$ and $(\expoPm, \expoP, \expoPp)$, and a conical localization data set $(\Mbf,\Omega,g_0,h_0,\wtrr,\lambdabf)$, there exists a threshold $\overline{\eps}_A \in(0,1]$ with the following property. 

\bei 

\item The threshold depends only on the exponents and on fixed a priori bounds for the localization data set, namely
\be
\aligned
\CPKHzero, \quad
\|\Riem_{g_0}\|_{C^{N,\alpha}_{2+p_G}(\Omega,g_0)},
\quad
\|h_0\|_{C^{N+1,\alpha}_{1+p_G}(\Omega,g_0)},
\quad \operatorname{inj}(\Mbf, g_0),
\endaligned
\ee
($\operatorname{inj}$ denoting the injectivity radius). It is independent of $\eps_G,\eps_A$ and of the seed data.

\item Then, for 
$
0<\eps_A\leq\overline{\eps}_A
$
and for small enough $\eps_G>0$ (as a function of the localization data set and exponents), the localized seed-to-solution map $\Solu^{\lambdabf}_{n,p}$ in~\eqref{equa-sol-map} is well-defined over the collection of seed data sets $\Seed(\Omega, g_0,h_0,p_G, p_A,\eps_G,\eps_A)$ and generates an exact solution of Einstein's vacuum constraints satisfying \eqref{equa-gh-uZ-00}--\eqref{equa-213-2-u-Z} starting from any localized seed data set satisfying, by definition, the near-localization condition~\eqref{equa-near-refe} and the near-Einsteinian condition~\eqref{equa-nearEins}. 
\eei
\end{theorem}

A convenient class of examples for \autoref{thm:sts-existence} is to take $(g_0,h_0)$ to be an exact solution such as the Euclidean solution $(\delta,0)$, and $(\seedg,\seedh)$ to differ from it only in some angular regions.  This includes the gluing of Euclidean and Schwarzschild solutions by Carlotto and Schoen.

The proof of \autoref{thm:sts-existence} and of analogous Sobolev estimates stated in \autoref{thm:sts-Sobolev} is given in \autoref{appendix=E}.  It relies on standard techniques, extended to accommodate a wider range of geometric exponents~$p_G$, and with additional error terms due to constraints being 
linearized at the successive Newton iterates, while the formal adjoint is fixed at the reference data $(g_0,h_0)$, 
as outlined in \autoref{appendix=E}.  The distinction between these data sets allows the seed and solution data sets to share the same regularity $C^{N,\alpha}$ (while the reference data is more regular), and allows $\Solu^{\lambdabf}_{n,p}$~to be a projection along a particular \emph{linear} subspace (determined by the reference data set) and onto the space of solutions of the constraints.

\bei 

\item The linearization $d\Gcal_{(g_0,h_0)}$ of the Einstein operator $\Gcal$ around $(g_0,h_0)$ is not elliptic unless a gauge choice is made~\cite{Corvino-2000,CorvinoSchoen} and, specifically, the solution deformations $(\gdiff,\hdiff)$ are restricted to lie in the image of the adjoint  $d\Gcal_{(g_0,h_0)}^*$, multiplied by weights that suitably localize the deformations to the gluing domain. 

\item The solution is obtained by a Newton iteration scheme~\eqref{Newton-it}, whose convergence requires a good control of the underlying linear operator~$\Jcal_{(g_k,h_k;g_0,h_0)}$.  The proof involves the following ingredients: a weighted Poincaré--Korn--Hardy inequality applied to the linearized Hamiltonian and momentum operators; the Lax--Milgram theorem providing integral bounds on the solution; the interior elliptic regularity for Douglis--Nirenberg systems (away from the boundary of the gluing domain); and the higher radial and angular decay enjoyed by nonlinearities of the Einstein constraints.

\eei

\noindent \autoref{thm:sts-existence} allows us to validate the proposed projection framework for a large class of asymptotically Euclidean solutions.  We encompass solutions with \emph{arbitrarily slow decay} (since ${p_G >0}$) while we establish their \emph{continuous dependence} with respect to the seed data in appropriate weighted spaces.  

 
\section{Harmonic, radial, and shell stability}
\label{section=3}

\subsection{The harmonic-spherical decomposition of the Einstein constraints}
\label{section=3.1}

For our derivation of harmonic estimates, we need to exploit some structure of the Einstein constraints, and specifically of the relevant linearized operators $\Jcal_{(g_0,h_0;g_0,h_0)}$.
Within a conical localization data set $(\Mbf,\Omega, g_0,h_0, \wtrr, \lambdabf)$, we focus on an asymptotic end $\Omega_\iota \simeq K_\iota\cap{}^\complement \Ball_R \subset \RR^n$ as introduced in \autoref{def-conical}.
After a suitable reduction and neglecting comparatively smaller terms (cf.~\autoref{section=10}), we are led to consider these operators with the Euclidean metric $\delta_{ij}$ and vanishing~$h_0$, and we study now their asymptotic properties in the truncated cone $K_\iota\cap{}^\complement \Ball_R$.
Here, a projection exponent $p>0$ is given, together with a localization function $\lambda_\iota\colon\Lambda_\iota = K_\iota \cap \Sphe^{n-1} \to (0, \lambda_0]$, and a variational weight (cf.~\eqref{equa-omega-choice}) 
\be
\omega_{\iota,p} = \lambda_\iota^{\expoP} r^{n/2-p}. 
\ee
This weight has radial homogeneity $n/2-p$ in the distance $\rbf=|x| \geq R$, and the localization function $\lambda_\iota$ vanishes linearly with respect to the distance to the boundary of~$\Lambda_\iota$. With this notation we introduce the following terminology, in which the Laplace operator $\Delta$ is defined in $(\RR^n, \delta)$, and we identify lower and upper indices using the standard metric~$\delta$.

\begin{definition}
With the notation in \autoref{def-conical}, at each asymptotic end $\Omega_\iota$ ($\iota=1,\ldots,I$) the  operators 
\bel{equa:acalew0-deux} 
\aligned
\notreH^{\lambda_\iota}[u] 
& \coloneqq \omega_{\iota,p}^{-2} d\Hcal_{(\delta,0)}\bigl[\omega_{\iota,p}^2 \, d\Hcal_{(\delta,0)}^{*\flat\flat}[u]\bigr]
= \omega_{\iota,p}^{-2} \, \Bigl(
(n-2) \, \Delta(\omega_{\iota,p}^2 \, \Delta u)
+ \del_i\del_j(\omega_{\iota,p}^2 \del_i\del_j u)
\Bigr),
\\
\notreM^{\lambda_\iota}[Z]^i 
& \coloneqq \omega_{\iota,p+1}^{-2} d\Mcal_{(\delta,0)}\bigl[\omega_{\iota,p+1}^2  \, d\Mcal_{(\delta,0)}^{*\sharp\sharp}[Z]\bigr]
= - \frac{1}{2} \omega_{\iota,p+1}^{-2} \del_j\Bigl(\omega_{\iota,p+1}^2(\del_j Z^i + \del_i Z^j) \Bigr),
\endaligned
\ee
are referred to as the (squared) {\bf localized Hamiltonian operator} and (squared) {\bf localized momentum operators,} respectively. Here, $u: \Omega_\iota \subset \RR^n \to \RR$ is a scalar-valued field, while $Z: \Omega_\iota \subset \RR^n \to \RR^n$ is a vector-valued field. 
\end{definition}

We introduce a decomposition of these operators which is adapted to the study of the harmonic decay and takes the form 
\bel{equa-key-decompose-H-repeat-deux} 
\aligned 
r^4 \notreH^{\lambda_\iota}[u] 
& = \Arr[u] + \Ars^{\lambda_\iota}[u] + \ssA^{\lambda_\iota}[u], 
\\
r^2 \notreM^{\lambda_\iota} [Z] 
& = \Brr[Z] + \Brs^{\lambda_\iota}[Z] + \ssB^{\lambda_\iota}[Z]. 
\endaligned
\ee 
Our decomposition is mostly derived by considering how the operators act on functions $\nu(x/r)\, r^{-a_{n,p}}$ and $\xi(x/r) \, r^{-a_{n,p}}$ with radial decay given by the {\bf harmonic exponent} $a_{n,p} = 2(n-2-p)$ (cf.~\eqref{equa-our-parame-00}).
At this stage, we content ourselves with definitions, while the actual decompositions will be made explicit in \refwithname{Lemmas}{lem:sph-Ham} \refwithname{and}{lem:sph-mom} later on.
Importantly, the operators $\ssA^{\lambda_\iota}$ and~$\ssB^{\lambda_\iota}$ are not self-adjoint, but their principal parts are self-adjoint and positive.

\begin{definition} 
\label{def-operators-decomposer}  
The decomposition~\eqref{equa-key-decompose-H-repeat-deux} of the localized Hamiltonian and momentum operators $r^4 \notreH^{\lambda_\iota}[u]$ and $r^2 \notreM^{\lambda_\iota}[u]$, respectively, is characterized by the following two properties. 

\bei 

\item The operators $\Arr[u]$ and $ \Brr[Z]$ involve all the terms without any angular derivatives of the fields $u,Z$ or the localization function $\lambda^{2\expoP}$.

\item The operators $\ssA^{\lambda_\iota}[u]$ and $\ssB^{\lambda_\iota}[Z]$ are defined by restricting attention to fields with harmonic decay, namely for a scalar function~$\nu$ of $x/r$ (respectively a vector-valued function~$\xi$ of~$x/r$), one has
\bel{equa-asymptoOper-reapeat-deux} 
\ssA{}^{\lambda_\iota}[\nu] \coloneqq r^{4 +a_{n,p}} \notreH^{\lambda_\iota}[\nu \, r^{-a_{n,p}}],
\qquad
\ssB{}^{\lambda_\iota}[\xi] \coloneqq r^{2 +a_{n,p}} \notreM^{\lambda_\iota}[\xi \, r^{-a_{n,p}}]. 
\ee 
\eei 
This decomposition is referred to as the \textbf{harmonic-spherical decomposition}, while $\ssA{}^{\lambda_\iota}$ and $\ssB{}^{\lambda_\iota}$ are referred to as the \textbf{harmonic operators.}
\end{definition}

In the following, by analogy with the weighted Sobolev norms defined earlier for functions defined on subsets of $\RR^n$, we use for the weighted measure $d\chi_\iota\coloneqq\lambda_\iota^{2\expoP}d\xh$
and, for any $q\in\RR$, the normalized angular Sobolev norms
\bel{equa-norm-poids}
\|v\|_{\unH^k_q(\Lambda_\iota)}^2
\coloneqq
\sum_{0\leq j\leq k}\fint_{\Lambda_\iota}
|\nablaslash^jv|_{\gslash}^2\,\lambda_\iota^{-2\expoP-2q}\,d\chi_\iota ,
\ee
and $\unL^2_q\coloneqq\unH^0_q$, with the underline denoting the normalization of the integral by the same factor
$\aire[\Lambda_\iota,\lambda_\iota]\coloneqq\int_{\Lambda_\iota}d\chi_\iota$ regardless of~$q$.
In particular, $q=-\expoP$ gives the normalized Sobolev norm associated directly with~$d\chi_\iota$.
If $v$ is defined on the spherical shell $\Lambda_{\iota,r}$, its norm is defined as that of $v(r,\cdot)$ on~$\Lambda_\iota$, namely all angular derivatives are taken with respect to the unit-sphere metric~$\gslash$, and no additional power of~$r$ is implicit.
The space $\unH^k_q(\Lambda_\iota)$ consists of the distributions on $\Lambda_\iota$ with finite norm~\eqref{equa-norm-poids}; the same definition is used componentwise for vector- and tensor-valued fields. This notation does not include any trace condition.

Within proofs, we also consider $\unH^{k*}_{-\expoP}(\Lambda_\iota)$, the continuous dual of $\unH^k_{-\expoP}(\Lambda_\iota)$ for the normalized pairing. Explicitly,
\bel{dual-Sobolev-Lambda}
\|\varphi\|_{\unH^{k*}_{-\expoP}(\Lambda_\iota)}
\coloneqq
\sup_{0\neq v\in\unH^k_{-\expoP}(\Lambda_\iota)}
| \langle\varphi,v\rangle_{\Lambda_\iota}| \bigm/ \|v\|_{\unH^k_{-\expoP}(\Lambda_\iota)} ,
\ee
where $\langle\varphi,v\rangle_{\Lambda_\iota}
=\fint_{\Lambda_\iota}\varphi\cdot v\,d\chi_\iota$ whenever $\varphi$ is represented by a locally integrable field.
For a distribution $\varphi\in\unH^{k*}_{-\expoP}(\Lambda_\iota)$ and test function~$v$ (and likewise for tensors) we use the convention
\be
\fint_{\Lambda_\iota}\varphi\cdot v\,d\chi_\iota
\coloneqq
\langle\varphi,v\rangle_{\Lambda_\iota},
\qquad
\la\varphi\ra_\iota
\coloneqq \langle\varphi,1\rangle_{\Lambda_\iota} .
\ee
Every average or normalized integral involving a dual Sobolev element below is to be understood as a duality pairing, not a claim that the integrand is integrable.

We also introduce operator coefficients (also collected in \autoref{appendix=A})
\bel{equa-our-parame-00} 
\aligned
a_{n,p} & \coloneqq 2( n-2-p) ,
\\
b_{n,p} & \coloneqq 2 + (n-3)(n-2-a_{n,p}), 
\\ 
c_{n,p} & \coloneqq a_{n,p} \bigl( 1 + (n-2) ( n-2- a_{n,p}) \bigr), 
\\
d_{n,p} & \coloneqq (n-1) a_{n,p} b_{n,p} / (n^2-4n+5) ,
\endaligned
\ee
which are all positive for a projection exponent $p\in(p^\flat_n, n-2)$, where $p^\flat_n< (n-2)/2$ was given in~\eqref{pflatn-def}.
 
   
\subsection{Harmonic and radial stability for the Hamiltonian}
\label{section=3.2}

\paragraph{Harmonic stability.}

We present the relevant notion of stability, which addresses three issues: What is the dimension of the harmonic kernel at infinity? Does the spherical average of a solution decay at infinity? Do the anisotropic fluctuations decay at infinity?
We answer the first two questions for the Hamiltonian operator (\autoref{section=3.2}), then for the momentum operator (\autoref{section=3.3}), and then provide quadratic functionals for both operators to control fluctuations (\refwithname{Sections}{section=3.4} \refwithname{and}{section=3.5}).

First of all, we want to ensure that the harmonic operator~$\ssA^{\lambda_\iota}$ has a one-dimensional kernel. We emphasize that it is not self-adjoint.
To this operator we associate the quadratic functional $\nu \mapsto \ssrmA^{\lambda_\iota}[\nu] \coloneqq \fint_{\Lambda_\iota} \nu \, \ssA{}^{\lambda_\iota}[\nu] \, d\chi_\iota$ (after formal integration by parts). We thus consider the quadratic functional 
\bel{ssAalpha-quaform-0}
\aligned
\ssrmA^{\lambda_\iota}[\nu]
& = \fint_{\Lambda_\iota} \Bigl( (n-2) (\Deltaslash\nu)^2
+ |\nablaslash^2\nu|^2
+ 2 (1+a_{n,p}) |\nablaslash\nu|^2 
- c_{n,p} \, \nu \Deltaslash\nu
\Bigr) \, d\chi_\iota. 
\endaligned
\ee
The coefficients $a_{n,p}$ and $c_{n,p}$ were defined in~\eqref{equa-our-parame-00}. We will also use the notation $\la f \ra_\iota$ for the weighted average of a function $f$ on spheres at the asymptotic end $\Omega_\iota$ (cf.~\eqref{equa-average}). 
 
\begin{definition}
\label{def-harmonic-Hstab} 
The localization function $\lambda_\iota\colon \Lambda_\iota \to (0, \lambda_0]$ is said to satisfy the Hamiltonian \textbf{harmonic stability condition} provided\footnote{Here and throughout we use weighted Sobolev norms~\eqref{equa-norm-poids} normalized by $\aire[\Lambda_\iota,\lambda_\iota]$.}
\bel{equa-stable-H-414}
\aligned
  \ssrmA^{\lambda_\iota}[\nu] \gtrsim \| \nu \|_{\unH^2_{-\expoP}(\Lambda_\iota)}^2,
  \qquad \nu \in \unH^2_{-\expoP}(\Lambda_\iota) 
  \text{ with } 
  \la\nu\ra_\iota = 0. 
\endaligned
\ee
\end{definition}
 
Our stability condition above (also sometimes referred to as the harmonic coercivity) is nothing but a weighted Poincaré-type inequality for functions with a vanishing average.  It can be equivalently stated as the coercivity of $\ssrmA^{\lambda_\iota}[\nu]$ plus a large enough multiple of~$\la\nu\ra_\iota^2$.
The harmonic stability condition implies that $\ssA^{\lambda_\iota}$ has a one-dimensional kernel and cokernel, as we prove in \autoref{section=6.2}. Choose any nonzero element $\nu_\iota\in\ker(\ssA^{\lambda_\iota})$, set
$\widetilde\nu_\iota\coloneqq\nu_\iota-\la\nu_\iota\ra_\iota$, and define the structure constants 
\bel{equa-Ptwo-0-struct}
\aligned
\bnotreH_{\iota 1}
&\coloneqq
(n-1)c_{n,p}\la\nu_\iota\ra_\iota
-(n-2)^2\la\Deltaslash\widetilde\nu_\iota\ra_\iota,
\\
\bnotreH_{\iota 0}
&\coloneqq
(n-1)b_{n,p}c_{n,p}\la\nu_\iota\ra_\iota
-(n^2-4n+5)\frac{c_{n,p}}{a_{n,p}}
\la\Deltaslash\widetilde\nu_\iota\ra_\iota.
\endaligned
\ee
Their common rescaling under $\nu_\iota\mapsto C\nu_\iota$ leaves the radial stability condition below unchanged. Once the normalized silhouette is introduced in \autoref{def:normalized-kernel-basis}, below, we take $\nu_\iota=\nu_\iota^{\normal}$ in these formulas.


\paragraph{Spherical averages.}

We are interested in the radial decay of general solutions $u\colon\Omega_\iota \to \RR$ to 
\bel{equa-solutionH} 
\notreH^{\lambda_\iota}[u] = E, 
\ee
where $E\colon\Omega_\iota \to \RR$ is a given scalar field. The operators $\Arr$, $\Ars^{\lambda_\iota}$, and $\ssA^{\lambda_\iota}$ in the decomposition~\eqref{equa-key-decompose-H-repeat-deux}
are quite involved and, especially, are fourth-order in the radial variable as well as in tangential spherical derivatives. 

By contracting the Hamiltonian equation~\eqref{equa-solutionH} with an element of the co-kernel (namely $1$) and an element of the kernel of the harmonic Hamiltonian operator, we are able to derive in \autoref{section=6.3} a fourth-order system of two {\it coupled} differential equations satisfied by the averages $\la u\ra_\iota$ and $\la\Deltaslash u\ra_\iota$. Using the notation 
\be
\vartheta = r \, \del_r
\ee
after radial integration (twice), we find the second-order equation\footnote{The stability condition~\eqref{equa-b2-positive}, below, imposes that $\bnotreH_{\iota 1}$ does not equal zero.} 
\bel{equa-defb-b-b} 
\bigl( - \bnotreH_{\iota 1} \vartheta(\vartheta+a_{n,p}) + \bnotreH_{\iota 0} \bigr) \, \la u\ra_\iota 
= \bigl( -(n-2) \vartheta + {c_{n,p} / a_{n,p}} \bigr) \Kappa^\notreH_\iota[\ut]
+ \Nu^\notreH_\iota[E] + C_\iota^u r^{-a_{n,p}} ,
\ee
where, throughout this paragraph, the harmonic constant $C_\iota^u\in\RR$ denotes a solution-dependent integration constant.
We distinguish here the following objects.

\bei 

\item  The {\bf structure constants} $\bnotreH_{\iota 1}$ and $\bnotreH_{\iota 0}$ (defined in~\eqref{equa-Ptwo-0-struct} above) depend upon the localization function $\lambda_\iota$ as well as $n,p$ and, more precisely, are made explicit in terms of a suitably normalized element  of the kernel of the harmonic operator $\ssA^{\lambda_\iota}$, referred to as the ``silhouette function'' and denoted by $\nu^\normal$; cf.~\autoref{def:normalized-kernel-basis}, below. 

\item The {\bf fluctuation operator} $\Kappa^\notreH_\iota$ is a linear functional (given explicitly in \autoref{def:KappaH}, below): it is an integral over $\Lambda_{\iota,r}$ of a linear form in the derivatives $\vartheta^j \nablaslash^k u$ with $j+k\leq 3$ and $k\leq 2$. 
Interestingly, this operator is bilinear in the {\it fluctuations} of $u$, defined as the function 
\bel{equa-tildenotation}
\ut \coloneqq u - \la u \ra_\iota ,
\ee
and in the fluctuations $\nut^\normal$ of~$\nu^\normal$.
As a result, $\Kappa^\notreH_\iota$ can be thought of as a {\it lower-order term.}
For instance, it vanishes identically when $\Lambda_\iota$ is the whole sphere and $\lambda_\iota$ is taken to be constant.

\item The {\bf source term} $\Nu^\notreH_\iota[E]$ is an integral operator acting on the function $E$ (and given explicitly in~\eqref{equa-hereOmegaHcompact}, below). 

\eei   
  
We now state our condition for the stability of spherical averages.  It guarantees that the characteristic polynomial of the differential operator $- \bnotreH_{\iota 1} \vartheta(\vartheta+a_{n,p}) + \bnotreH_{\iota 0} $ in~\eqref{equa-defb-b-b} admits two roots in $(-\infty,-a_{n,p})\cup(0,+\infty)$, which correspond to a {\it super-harmonic} mode and a {\it growing} mode (excluded by variational bounds).

\begin{definition}
\label{def-radial-Hstab}
A localization function satisfying the Hamiltonian harmonic stability condition~\eqref{equa-stable-H-414} is said to also enjoy the Hamiltonian \textbf{radial stability condition} provided the following product\footnote{This positivity condition is independent of the normalization of the kernel element~$\nu^\normal$.  Strictly speaking, the normalization condition on~$\nu^\normal$ that we use relies on the non-vanishing of~$\bnotreH_{\iota 0}$.} is positive:
\bel{equa-b2-positive}
\bnotreH_{\iota 1} \, \bnotreH_{\iota 0}>0.  
\ee
\end{definition}


\subsection{Harmonic and radial stability for the momentum}
\label{section=3.3}

\paragraph{Harmonic stability.}

It is convenient to introduce the decomposition of a vector $Z$ into radial (perpendicular) and angular (parallel to the spheres) components,
\bel{Zi-split}
Z_i = \xh_i\Zperp+ \Zpar_i , \qquad
\Zperp \coloneqq \xh_j Z_j , \qquad
\Zpar_i \coloneqq Z_i - \xh_i\xh_j Z_j ,
\ee
in which $\xh= x/r$ in the coordinates at infinity. We proceed analogously as before but now deal with the momentum operator and consider first the (non-self-adjoint) harmonic operator $\xi \mapsto \ssB{}^{\lambda_\iota}[\xi]$ introduced in \autoref{def-operators-decomposer}. We thus introduce the quadratic functional 
\bel{ssBalpha-quaform-0}
\aligned
\ssrmB^{\lambda_\iota}[\xi]
= \fint_{\Lambda_\iota} \hskip-.1cm \Bigl(
(n-1) (\xiperp)^2 + \frac{1}{2} |\nablaslash\xiperp|^2
& - {a_{n,p}+2 \over 2} \xipar\cdot\nablaslash\xiperp + 2 \xiperp \nablaslash\cdot\xipar \\[-1ex]
& + {a_{n,p}+1\over 2}|\xipar|^2
+ \bigl|\Sym(\nablaslash\xipar) \bigr|^2 \Bigr) \, d\chi_\iota, 
\endaligned
\ee
where $\Sym(\nablaslash \xipar)_{ab} \coloneqq \frac{1}{2} \bigl( \nablaslash_a \xipar_b + \nablaslash_b \xipar_a\bigr)$. The following condition is nothing but a Korn-type inequality for the localization measure $d\chi_\iota$.
It implies that the kernel and cokernel of $\ssB^{\lambda_\iota}$ have dimension~$n$ (cf.~\autoref{section=9.2}).

\begin{definition}
\label{def-harmonic-Mstab}
A localization function $\lambda_\iota\colon \Lambda_\iota \subset \Sphe^{n-1} \to (0, \lambda_0]$ is said to satisfy the momentum \textbf{harmonic stability condition} provided  
\bel{equa-stable-M-414}
\aligned
& \ssrmB^{\lambda_\iota}[\xi] \gtrsim 
\| \xiperp \|_{\unH^1_{-\expoP}(\Lambda_\iota)}^2
+ \| \xipar \|_{\unL^2_{-\expoP}(\Lambda_\iota)}^2
+ \| \Sym(\nablaslash \xipar) \|_{\unL^2_{-\expoP}(\Lambda_\iota)}^2,
\\
& \xi \in \unH^1_{-\expoP}(\Lambda_\iota)
\quad \text{ with } \, 
 \bigl\la 2 \, \xh_l\, \xi^{\perp} + \xi_{l}^{\parallel} \bigr \ra_\iota  = 0.  
\endaligned
\ee  
\end{definition}
 

\paragraph{Spherical averages.}
 
We can integrate the equation $\notreM^{\lambda_\iota}[Z] = F$ on each sphere of radius $r \geq R$ after contraction with an element of the kernel, or the co-kernel, of the harmonic operator $\ssB{}^{\lambda_\iota}$. After some calculations in \autoref{section=9.3}, it turns out that, for the $n$ spherical averages $\la 2 \, \xh_l \Zperp + \Zpar_l\ra_\iota$ associated with the vector field $Z$ on each sphere, 
we find a second-order differential system, which integrates to (non-differential) linear equations ($j=1,2,\ldots,n$)
\bel{equa-defb-b-b-MM} 
(\Xi^\notreM_\iota T_\iota^{-1})^{(j)l} \bigl\la 2 \, \xh_l \Zperp + \Zpar_l \bigr\ra_\iota 
=  \Kappa_\iota^{\notreM (j)}[Z] + \Nu_\iota^{\notreM (j)}[F] + C_\iota^{Z(j)} \, r^{-a_{n,p}},
\ee
where the harmonic constants $C_\iota^{Z(j)}$ are solution-dependent constants. We distinguish the following contributions. 
\bei  

\item The (constant) {\bf structure matrices}, by definition, are
\bel{equa-the-matrix}
\aligned
(T_{\iota kl}) & \coloneqq (\delta_{kl} + \la \xh_k\xh_l\ra_\iota),
\\
(\Xi^\notreM_\iota)^{(j)}{}_k
& \coloneqq
 \bigl\la - \nablaslash_k \xi_\iota^{\normal (j) \perp} + 2 a_{n,p} \xh_k\, \xi_\iota^{\normal (j) \perp} \bigr \ra_\iota
 +(1+a_{n,p}) \la \xi_\iota^{\normal (j) \parallel}{}_k \ra_\iota. 
 \endaligned
\ee
The matrix $T_\iota$ is positive-definite.  The matrix $\Xi^\notreM_\iota$ is given explicitly in terms of averages of any choice of basis $\xi_\iota^{\normal(j)} \in \ker(\ssB^{\lambda_\iota})$.
The matrix product 
$(\Xi^\notreM_\iota T_\iota^{-1})^{(j)l} \coloneqq (\Xi^\notreM_\iota)^{(j)}{}_{k} (T_\iota^{-1})^{kl}$ in~\eqref{equa-defb-b-b-MM} is defined in a standard way. 

\item The  {\bf fluctuation operator} $\Kappa_\iota^\notreM = ( \Kappa_\iota^{\notreM (j)} )_{1 \leq j \leq n}$ is a vector-valued linear functional given explicitly in \autoref{def:KappaM}, below: it is an integral over $\Lambda_\iota$ of a linear form in the variables $Z, \vartheta Z, \nablaslash Z$.
Importantly, $\Kappa_\iota^\notreM$ can be thought of as a {\it lower-order term.}
For instance, it vanishes identically when $\Lambda_\iota$ is taken to be the sphere and $\lambda_\iota$ to be constant, and it depends upon the {\it fluctuations} $Z^\fluc$ and their derivatives, only. (See  \autoref{DEF-fluct-vectors} for the detailled derivation.)

\item The {\bf source term} $\Nu_\iota^{\notreM (j)}$ is a vector-valued integral operator acting on the source~$F$, whose exact formula and mapping properties are stated in \autoref{prop-87-moment}.

\eei 

\noindent
The stability condition~\eqref{equa-Xi-invertible}, below, implies that~\eqref{equa-defb-b-b-MM} controls the vector of spherical averages $\la 2 \, \xh_l \Zperp + \Zpar_l\ra_\iota$.  In addition, it allows us to normalize the basis $\xi_\iota^{\normal(j)} \in \ker(\ssB^{\lambda_\iota})$, referred to as the `silhouette vector' fields, in such a way that $\Xi^\notreM_\iota$ is a multiple of the identity matrix, as stated in \autoref{def:normalized-kernel-basis}.

\begin{definition}
\label{def-radial-Mstab}
A localization function satisfying the momentum harmonic stability condition~\eqref{equa-stable-M-414} is said to satisfy the momentum {\bf radial stability condition} provided  
the matrix $\Xi^\notreM_\iota$ is invertible, namely 
\bel{equa-Xi-invertible}  
\detbf (\Xi^\notreM_\iota) \neq 0. 
\ee 
\end{definition}


\subsection{Shell stability condition for the Hamiltonian}
\label{section=3.4}

\paragraph{Hamiltonian on spherical shells.}

To complement the average $\la u\ra_\iota$ on each shell $\Lambda_{\iota,r}$, we rely on a quadratic functional,
referred to as the {\bf Hamiltonian shell functional}. Fix auxiliary constants $\cstun>0$ and $\cstdeux\geq 0$, to be chosen as part of the shell stability condition so that the functional below is non-negative and its dissipation is semi-coercive. The functional is given by
\bel{PhinotreH-expr}
\aligned
\Phi^\notreH_\iota[u] 
& = \frac{1}{2(n-1)} \fint_{\Lambda_{\iota,r}} \Bigl(
(n-1) \bigl(\vartheta^2 u + a_{n,p} \vartheta u - \cstun u\bigr)^2
+ \bigl|\nablaslash^2 u\bigr|^2
+ (n-2) (\Deltaslash u)^2
\\[-1ex]
& \qquad\qquad\qquad\quad\ \
+ 2 (1+a_{n,p}+\cstun) \,  |\nablaslash u|^2
+ \cstdeux \, u^2
- 2\bigl(c_{n,p}+(n-2)\cstun\bigr) u \Deltaslash u
\Bigr) \, d\chi_\iota.
\endaligned
\ee
The admissible choice may depend on the localization geometry; for instance, in the isotropic case we can take $\cstdeux=0$. It obeys a crucial dissipation property\footnote{The word ``dissipation'' is used here in a loose sense.} obtained by multiplying~\eqref{equa-solutionH} by
\bel{test-function-theta2u}
{1 \over n-1} r^4 \bigl( \vartheta^2 u + a_{n,p} \vartheta u - \cstun u \bigr) .
\ee
Specifically, for any solution $u\colon\Omega_\iota \to \RR$ to~\eqref{equa-solutionH} we derive the {\bf Hamiltonian shell identity} 
\bse
\label{equa-condition-monotone}
\bel{main-func-identity}
- (\vartheta + a_{n,p}) (\vartheta + 2a_{n,p}) \Phi^\notreH_\iota[u]
+ \Chi^\notreH_\iota[u] = \Mu^\notreH_\iota[u,E], 
\ee 
which involves the following features.
Detailed expressions of the various functionals play no role in the main text and are postponed to \autoref{appendix=B}.
In~\cite{LL-PoincareKornHardy}, we consider functionals~$\Phi^\notreH_\iota$ given as integrals of general quadratic expressions in the variables $\vartheta^j \nablaslash^k u$ with $j+k \leq 2$, of which \eqref{PhinotreH-expr} is only an example.
\bei
\item The differential operator $(\vartheta + a_{n,p}) (\vartheta + 2a_{n,p})$ arises naturally since we are seeking to ``bridge'' the variational decay rate $r^{-a_{n,p}/2}$ (available from \autoref{thm:sts-existence}) and the harmonic decay rate $r^{-a_{n,p}}$ (which is our main aim), for which quadratic expressions of~$u$ behave as $r^{-a_{n,p}}$ and~$r^{-2a_{n,p}}$, respectively.

\item The {\bf bare dissipation functional} $\Chi^\notreH_\iota$ is an integral functional in the variables $\vartheta^j \nablaslash^k u$ for $j+k\leq 4$ with $j \leq 3$ and $k \leq 2$. It features products of second-order and fourth-order derivatives of~$u$ (and no quadratic term in these latter derivatives) hence cannot enjoy positivity properties.

\item To deal with these linear terms, we introduce the two (shifted) {\bf dissipation functionals}
\bel{equa-310c} 
\Psi^\notreH_{\beta\iota}
\coloneqq \Chi^\notreH_\iota - (\vartheta +\beta) \Upsilon^\notreH_\iota, 
\qquad \beta \in \{ a_{n,p}, 2 a_{n,p}\}, 
\ee
in which we have subtracted a {\bf radial integration functional} $\Upsilon^\notreH_\iota[u]$ in the variables $\vartheta^j \nablaslash^k u$ for $j+k\leq 3$ and $j,k\leq 2$, chosen such that $\Psi^\notreH_{\beta\iota}$ is a quadratic functional in $\vartheta^j \nablaslash^k u$ for $j+k\leq 3$ and $k\leq 2$, namely the same derivatives as the functional~$\Kappa^\notreH_\iota$. They are given explicitly in~\eqref{PsinotreH-expr}.

\item The {\bf source functional} $\Mu^\notreH_\iota$ is the negative of the product in~\eqref{test-function-theta2u} and the source~$E$,
\bel{Psi-gendef} 
\Mu^\notreH_\iota[u,E] \coloneqq 
\frac{1}{n-1} \fint_{\Lambda_{\iota,r}}  \bigl(- \vartheta ( \vartheta + a_{n,p}) u  + \cstun u \bigr) \, r^4 E \, d\chi_\iota.  
\ee 

\eei

\noindent The shifts from $\Chi^\notreH_\iota$ to $\Psi^\notreH_{\beta\iota}$ in~\eqref{equa-310c} amount to total derivatives in an explicit integration of the differential equation~\eqref{main-func-identity}. Rather than imposing a condition on $\Upsilon^\notreH_\iota$ itself we impose a (partial) coercivity condition (stated in~\eqref{equa-conditionH2}, below) on both functionals~$\Psi^\notreH_{\beta\iota}$. Namely, it can be checked that the positivity of the functionals $\Psi^\notreH_{\beta\iota}$ is possible only after adding a large multiple of~$\la u\ra_\iota^2$. Such a feature is shared with the harmonic stability condition and is typical for Poincaré inequalities.

\ese


\paragraph{Hamiltonian shell stability.}

The third stability condition is governed by a \emph{radial Hardy constant}, denoted by~$\cradialHiota$. This
constant arises by combining the Hamiltonian shell identity \eqref{main-func-identity} with the radial ODE
\eqref{equa-defb-b-b} satisfied by the spherical average
$\la u\ra_\iota$. Equivalently, $\cradialHiota$ is the optimal
coefficient in a weighted one-dimensional Hardy inequality (cf.~\autoref{section=7}). Its value depends
only on the characteristic exponents of the radial equations
governing $\la u\ra_\iota$ and~$\Phi^\notreH_\iota$. Since its explicit form is not needed here, we defer its expression to
\eqref{equa--517-repeat}.
With some abuse of notation it is convenient to introduce the norm 
\bel{equa-normsH}
\bigl( \norm{u}^\notreH_\iota \bigr)^2
\coloneqq \| \vartheta^2 u\|^2_{\unL^2_{-\expoP}(\Lambda_{\iota,r})} + \| \vartheta u\|^2_{\unH^1_{-\expoP}(\Lambda_{\iota,r})} + \| u\|^2_{\unH^2_{-\expoP}(\Lambda_{\iota,r})} 
\ee
of $u$ and its radial derivatives on each spherical shell $\Lambda_{\iota,r}$ of each asymptotic end.
From now on, we simplify the notation and no longer specify the variable $r$ explicitly. 

\begin{definition}
\label{def-shell-Hstab} 
A localization function satisfying the Hamiltonian harmonic and radial stability conditions~\eqref{equa-stable-H-414} and~\eqref{equa-b2-positive} is said to satisfy the Hamiltonian \textbf{shell stability condition} if there exists a Hamiltonian shell identity~\eqref{equa-condition-monotone} enjoying the following properties.
\bse
\label{equa-last-twoH} 
\bei 
\item {\bf Continuity of the shell functionals.} For every sufficient regular scalar field $u\colon\Omega_\iota \to \RR$ and each spherical shell~$\Lambda_{\iota, r}$, and for $\beta\in\{a_{n,p},2a_{n,p}\}$, one has
\bel{equa-conditionH}
0 \leq \Phi^\notreH_\iota[u] \lesssim \bigl( \norm{u}^\notreH_\iota\bigr)^2 ,
\qquad
\bigl|\Psi^\notreH_{\beta\iota}[u]\bigr|
\lesssim \bigl( \norm{\vartheta u}^\notreH_\iota \bigr)^2 + \bigl( \norm{u}^\notreH_\iota \bigr)^2 .
\ee

\item {\bf Semi-coercivity of the shell dissipation.}
There exists a constant $\gammashellHiota>0$ such that for $\beta\in\{a_{n,p},2a_{n,p}\}$, for every sufficiently regular scalar field $u\colon\Omega_\iota \to \RR$ and each spherical shell~$\Lambda_{\iota,r}$, one has
\bel{equa-conditionH2}
\Psi^\notreH_{\beta\iota}[u]
+ \gammashellHiota \Bigl( \la u\ra_\iota^2
- \cradialHiota \bigl( \Kappa^\notreH_\iota[\ut] \bigr)^2 \Bigr)
\gtrsim
\bigl( \norm{\vartheta u}^\notreH_\iota \bigr)^2 + \bigl( \norm{u}^\notreH_\iota \bigr)^2 .
\ee
\eei 
\ese
The localization function is called {\bf Hamiltonian-stable} when the conditions~\eqref{equa-stable-H-414},~\eqref{equa-b2-positive}, and~\eqref{equa-last-twoH}  hold. 
\end{definition}

We have thus defined the class of localization functions that will naturally arise from our investigation of the structure of the (squared) localized Hamiltonian operator $\notreH^{\lambda_\iota}$ defined in~\eqref{equa:acalew0-deux}. 
Note that \eqref{equa-conditionH}~is clearly satisfied, in view of~\eqref{PhinotreH-expr} and of the construction of~$\Psi^\notreH_{\beta\iota}$, and we state it here to abstract away the details of the functional.
In contrast, the semi-coercivity~\eqref{equa-conditionH2} is a non-trivial condition on the localization domain $(\Lambda,d\chi_\iota)$.  By analyzing the dissipation functionals, we establish sufficient conditions for stability in~\cite{LL-PoincareKornHardy} (cf.~\autoref{thm:informal-sufficient-stability}, below).
Another interesting fact is that shell stability implies harmonic stability: indeed,  the shell identity~\eqref{main-func-identity} applied to a harmonic function implies $\Psi^\notreH_{2a_{n,p}\,\iota}[\nu r^{-a_{n,p}}]=\frac{\cstun}{n-1}\ssrmA^{\lambda_\iota}[\nu]r^{-2a_{n,p}}$, so that semi-coercivity of $\Psi^\notreH_{2a_{n,p}\,\iota}$ implies that of~$\ssrmA^{\lambda_\iota}$.


\subsection{Shell stability condition for the momentum}
\label{section=3.5}

\paragraph{Momentum on spherical shells.}
 
\bse
\label{equ-shell-momen}
We then introduce the {\bf momentum shell functional} (as we call it) 
\bel{equa-quada-M}
\Phi^\notreM_\iota[Z] \coloneqq \frac{1}{2} \fint_{\Lambda_\iota} \bigl( 2 \, \Zperp{}^2 + |\Zpar|^2 \bigr) d\chi_\iota, 
\ee 
which obeys the {\bf shell identity} for the momentum
\bel{main-func-identity-MM}
- (\vartheta + a_{n,p}) (\vartheta + 2a_{n,p}) \Phi^\notreM_\iota[Z]
+ \Chi^\notreM_\iota[Z] = \Mu^\notreM_\iota[Z,F].
\ee 
It involves the following terms (given explicitly in \autoref{appendix=B}, below). 

\bei 

\item  The {\bf bare dissipation functional} $\Chi^\notreM_\iota[Z]$  is an integral functional in the variables 
$Z$, $\vartheta Z$, $\nablaslash Z$, $\vartheta\nablaslash\Zperp$. Since it features $\vartheta\nablaslash\Zperp$ {\it linearly} (and not quadratically) it cannot enjoy positivity properties.
 
\item To eliminate these linear terms that would prevent the functionals from being coercive we introduce the (shifted) {\bf dissipation functionals}  (for $\beta\in\{a_{n,p}, 2 a_{n,p}\}$)
\bel{equa-310c-MM} 
\Psi^\notreM_{\beta\iota} \coloneqq \Chi^\notreM_\iota - (\vartheta +\beta) \Upsilon^\notreM_\iota, 
\ee
in which we have subtracted a {\bf radial integration functional} $\Upsilon^\notreM_\iota[Z]$ in the variables $Z, \nablaslash \Zperp$. 

\item The {\bf source functional} $\Mu^\notreM_\iota$ is
\be
\Mu^\notreM_\iota[Z,F] 
\coloneqq 2 r^2 \fint_{\Lambda_{\iota, r}} Z\cdot F \, d\chi_\iota = \fint_{\Lambda_{\iota,r}} 2 \, \bigl( \Zperp \Fperp + \Zpar\cdot\Fpar \bigr) r^2 d\chi_\iota. 
\ee 
\eei 
\ese


\paragraph{Momentum shell stability.}

From now on, we simplify the notation and no longer specify the variable $r$ explicitly.
Our last stability condition on the localization function is as follows.
It involves a quadratic functional $\shellKorn_\iota[Z]$ on each shell, defined momentarily.  The weighted radial integrals of $\shellKorn_\iota[Z]$ are non-negative, so that its presence in~\eqref{equa-conditionM2} does not spoil radially-integrated coercivity properties of~$\Psi^\notreM_{\beta\iota}$.  As we explain below, this term is only crucial in dimension $n=3$.

\begin{definition}
\label{def-shell-Mstab} 
A localization function satisfying the momentum harmonic and radial stability conditions \eqref{equa-stable-M-414} and~\eqref{equa-Xi-invertible} is said to satisfy the momentum \textbf{shell stability condition} if there exists a momentum shell identity~\eqref{equ-shell-momen} enjoying the following properties.
\bse\label{equa-last-twoM}
\bei 
 
\item {\bf Continuity of the shell functionals.}
For every sufficiently regular vector field $Z\colon\Omega_\iota\to\RR^n$ and any spherical shell $\Lambda_{\iota,r}$, and for $\beta\in\{a_{n,p},2a_{n,p}\}$, one has
\bel{equa-shell-M-coer}
0 \leq \Phi^\notreM_\iota[Z] \lesssim \|Z\|_{\unL^2_{-\expoP}(\Lambda_{\iota,r})}^2 ,
\qquad
\bigl| \Psi^\notreM_{\beta\iota}[Z] \bigr| \lesssim \|\vartheta Z\|_{\unL^2_{-\expoP}(\Lambda_{\iota,r})}^2 + \|Z\|_{\unH^1_{-\expoP}(\Lambda_{\iota,r})}^2.
\ee

\item {\bf Semi-coercivity of the shell dissipation.}
There exist constants $\gammashellMiota,\gammaKornMiota\geq 0$ such that for $\beta\in\{a_{n,p},2a_{n,p}\}$, for every sufficiently regular vector field $Z\colon\Omega_\iota\to\RR^n$ and each spherical shell $\Lambda_{\iota,r}$, one has
\bel{equa-conditionM2}
\aligned
& \Psi^\notreM_{\beta\iota}[Z] 
+ \gammashellMiota \, \Bigl| \la 2 \,\xh \Zperp + \Zpar \ra_\iota - T_\iota(\Xi^\notreM_\iota)^{-1} \Kappa_\iota^\notreM[Z] \Bigr|^2
- \gammaKornMiota \shellKorn_\iota[Z]
\\
& \quad \gtrsim \bigl\|\vartheta\Zperp,\, (\vartheta\Zpar+\nablaslash\Zperp),\, \Zperp\bigr\|_{\unL^2_{-\expoP}(\Lambda_{\iota,r})}^2 + \bigl\|\Zpar\bigr\|_{\unH^1_{-\expoP}(\Lambda_{\iota,r})}^2
\endaligned
\ee
\eei 
\ese
The localization function is called {\bf momentum-stable} when the conditions~\eqref{equa-stable-M-414},~\eqref{equa-Xi-invertible}, and~\eqref{equa-last-twoM}  hold. 
\end{definition}

As for Hamiltonian shell stability, the continuity property is immediate in view of our choice~\eqref{equa-quada-M} and of the construction of~$\Psi^\notreM_{\beta\iota}$, but is stated here for clarity in the presentation.
The averages $\la 2\xh_l\Zperp+\Zpar_l\ra_\iota$ and fluctuation operators $\Kappa_\iota^{\notreM(j)}[Z]$ appear differently from the Hamiltonian case because the equation~\eqref{equa-defb-b-b-MM} relating them is algebraic instead of being a differential equation like~\eqref{equa-defb-b-b}.


\paragraph{Korn inequality.}

Observe that the norm controlled in~\eqref{equa-conditionM2} is strictly weaker than~\eqref{equa-shell-M-coer} on a given shell.  In \autoref{section=8}, when considering radial integrals of these norms, we will apply the Korn inequality on~$\Omega$ to bridge this gap.  A basic version of this inequality is the following standard weighted estimate on a dyadic annulus: there exists an optimal constant $\CKornMiota>0$ such that for all vector fields $Z=(\Zperp,\Zpar)$ on $\Omega_R$, and all $R'\geq R$, one has\footnote{The norms $L^2_{0,-\expoP}$ and $H^1_{0,-\expoP}$ involve the measure $\lambda^{2\expoP}r^{-n}d^nx$, whose radial power is convenient later on.}
\bel{ambient-Korn}
\bigl\|\Sym(\del Z)\bigr\|_{L^2_{0,-\expoP}(\Omega_{R'}\setminus\Omega_{2R'})}^2
\geq \frac{1}{(\CKornMiota)^2} \min_{\textnormal{ambient Killing } \zeta} \|Z-\zeta\|_{H^1_{0,-\expoP}(\Omega_{R'}\setminus\Omega_{2R'})}^2 ,
\ee
where the minimum is taken over all Killing vector fields~$\zeta$ of~$\RR^n$ (rotations and translations).  The constant is independent of~$R'$ by a scaling argument.
In the non-compact setting of the full domain~$\Omega_R$ with norms involving a radial weight, Killing vector fields are eliminated by radial decay conditions, so that the weighted norm of $\Sym(\del Z)$ on the whole domain~$\Omega$ is coercive.

When using the shell identity to control fluctuations of~$Z$, we will rely on positivity properties of \emph{radial integrals} of the shell dissipations~$\Psi^\notreM_{\beta\iota}$ (cf.\ \autoref{section=8.2}).  The Korn inequality~\eqref{ambient-Korn} thus motivates us to allow in \autoref{def-shell-Mstab} an additional term which we dub the {\bf Korn remainder functional} (which is essential only in dimension $n=3$)
\bel{def-Korn-shell}
\shellKorn_\iota[Z] \coloneqq \bigl\|\Sym(\del Z)\bigr\|_{\unL^2_{-\expoP}(\Lambda_\iota)}^2
- \frac{1}{2^{2a_{n,p}}(\CKornMiota)^2} \min_{\textnormal{ambient Killing }\zeta} \|\Zpar - \zeta^\parallel\|_{\unH^1_{-\expoP}(\Lambda_\iota)}^2 ,
\ee
where $\CKornMiota$ is the constant in~\eqref{ambient-Korn}.
We emphasize that the negative term only involves parallel components and their angular derivatives; a control of $\Zperp-\zeta^\perp$ will be not be necessary for our purposes.
The minimum is taken over Killing vectors~$\zeta$ of Euclidean space~$\RR^n$ (corresponding to translations and rotations) restricted to $\Lambda_\iota\subset\Sphe^{n-1}\subset\RR^n$ and projected to their component~$\zeta^\parallel$.
The ambient Korn inequality~\eqref{ambient-Korn} yields the positivity of weighted radial integrals\footnote{The rescaled constant in~\eqref{def-Korn-shell} accounts for the variation of radial factors across a dyadic annulus.  In~\eqref{def-Korn-shell} the minimum over~$\zeta$ is taken independently on every shell, which makes $\shellKorn_\iota[Z]$ less negative than the radial integrand in~\eqref{ambient-Korn}.} of~$\shellKorn_\iota[Z]$, which ensures that our applications of shell stability in \autoref{section=8} go through.

On the other hand, the term $\|\Zpar-\zeta^\parallel\|^2$ helps complete the dissipation~$\Psi^\notreM_{\beta\iota}$ into a coercive functional.
The companion paper~\cite{LL-PoincareKornHardy} proves that, for suitable constants $\gammashellMiota,C_1,C_2>0$,
\bse\label{PsiM-coer-companion}
\be
\aligned
\Psi^\notreM_{\beta\iota}[Z] 
+ \gammashellMiota \, \Bigl| \la 2 \,\xh \Zperp + \Zpar \ra_\iota - T_\iota(\Xi^\notreM_\iota)^{-1} \Kappa_\iota^\notreM[Z] \Bigr|^2 &
\\
{} + C_1 \min_{\textnormal{ambient Killing }\zeta} \|\Zpar - \zeta^\parallel\|_{\unH^1_{-\expoP}(\Lambda_\iota)}^2
& \geq C_2 \|\Sym(\del Z)\|_{\unL^2_{-\expoP}(\Lambda_\iota)}^2 ,
\endaligned
\ee
with the crucial property that
\be
C_1/C_2\to 0 \quad \text{in the } a_{n,p}\to 0 \text{ limit}.
\ee
\ese
Thus, this coercivity property implies the desired shell stability condition~\eqref{equa-conditionM2} provided $C_1/C_2 < 2^{-2a_{n,p}}(\CKornMiota)^{-2}$, which holds for small enough $a_{n,p}>0$.
In higher dimensions the $\shellKorn_\iota[Z]$ term is not harmful, but we could omit it from the definition of shell stability because at best it helps improve some optimal constants, as we now explain.

Intrinsically, the set of $\zeta^\parallel$ to minimize over in~\eqref{def-Korn-shell} are exactly the \emph{conformal} Killing vector fields on the sphere~$\Sphe^{n-1}$ (corresponding to special conformal transformations and rotations), restricted to~$\Lambda_\iota$.
In high dimensions $n\geq 4$, or for $n=3$ and $\Lambda_\iota=\Sphe^2$, these constitute the whole set of conformal Killing vector fields of~$\Lambda_\iota$.
In contrast, in dimension $n=3$ with $\Lambda_\iota\neq\Sphe^2$, the domain~$\Lambda_\iota$ admits an infinite-dimensional space of conformal Killing vector fields, parametrized by holomorphic functions on~$\Lambda_\iota$.  These fields are dubbed \emph{local} conformal Killing vector fields, as opposed to \emph{global} ones that extend to the whole sphere.

It is instructive to separate the radial and angular components of $\Sym(\del Z)$, which are $\vartheta\Zperp$, $(\vartheta-1)\Zpar+\nablaslash\Zperp$, and $\Sym(\nablaslash\Zpar)+\Zperp\gslash$.  Their weighted $L^2$ norm on a fixed shell $\Lambda_\iota$ gives a control of $\Zpar$ thanks to the \emph{conformal} Korn inequality, which states that
\bel{conf-Korn}
\bigl\|\Sym(\nablaslash\Zpar)^\circ\bigr\|_{\unL^2_{-\expoP}(\Lambda_\iota)}^2
\gtrsim \min_{\textnormal{conformal Killing }\zeta^\parallel} \|\Zpar-\zeta^\parallel\|_{\unH^1_{-\expoP}(\Lambda_\iota)}^2 ,
\ee
where the minimum is taken over all conformal Killing vector fields $\zeta^\parallel$ of the shell~$\Lambda_\iota$.
In dimension $n\geq 4$ (or $n=3$ and $\Lambda_\iota=\Sphe^2$), if $C_1/C_2$ appearing in~\eqref{PsiM-coer-companion} is less than the implicit constant in the conformal Korn inequality~\eqref{conf-Korn}, then~\eqref{PsiM-coer-companion} implies that $\Psi^\notreM_{\beta\iota}[Z] + \gammashellMiota \, \bigl| \la 2 \,\xh \Zperp + \Zpar \ra_\iota - T_\iota(\Xi^\notreM_\iota)^{-1} \Kappa_\iota^\notreM[Z] \bigr|^2$ is coercive, which is the desired shell stability condition~\eqref{equa-conditionM2} with no need for~$\gammaKornMiota$.


\section{The optimal localization theory}
\label{section=4}

\subsection{Main statement}
\label{section=4.1}
 
\paragraph{The notion of stable localization.}

We now collect our stability conditions into a definition. The regime of interest for the projection exponent should now be such that\footnote{Dealing with the regime $p \in (0,p^\flat_n)$ would require some change of signs throughout our analysis.}
\bel{equa-p-n-flat}
p^\flat_n  < p < n-2 \qquad \text{ with } p^\flat_n \coloneqq (n-1)(n-3)/(2(n-2)),
\ee 
in which the notation $p^\flat_n$ was already introduced in \autoref{section=1.2}. The lower bound $p^\flat_n$ will arise in our analysis in order for an operator coefficient denoted by $c_{n,p}$ to be non-negative (cf.~our notation in \autoref{appendix=A}). On the one hand, we pick an interval of exponents~$p$ that encompasses a (one-sided) neighborhood of the harmonic exponent $n-2$. On the other hand, we observe that $p_3^\flat=0$ so that the full range $p \in (0,n-2=1)$ in dimension $3$ is covered by our presentation, while in general dimensions this interval of~$p$ is larger than the ``ADM range'' $[(n-2)/2,n-2)$.

\begin{definition} 
\label{def-41-stable}
Fix a projection exponent $p \in (p^\flat_n,n-2)$ and a localization exponent $\expoP \geq 2$. A localization function $\lambdabf$ is called a \textbf{stable localization function} if, at each asymptotic end, it is 
Hamiltonian-stable and momentum-stable in the sense of \refwithname{Definitions}{def-harmonic-Hstab} \refwithname{to}{def-shell-Mstab}.
\end{definition}

Our stability conditions provide sufficient (and essentially necessary) conditions to guarantee, in the asymptotics of solutions to the localized Einstein constraints, 
\bei

\item[(i)] the uniqueness of the harmonic contributions, and  

\item[(ii)] the convergence to the seed data set at $r \to +\infty$, modulo a harmonic contribution.
\eei
\noindent Checking our conditions on examples relies on specific Poincaré, Korn and Hardy inequalities for the weighted measure space $(\Lambda_\iota, \lambda_\iota^{2\expoP} d\xh)$ associated with the localization manifold, together with Hardy inequalities in the radial variable. Consequently, broad classes of localization phenomena are covered by our theory.  


\paragraph{Harmonic kernels and modulated seed data.}

At the harmonic level of decay, certain corrector terms arise. As we will prove, solutions to the (squared) linearized Hamiltonian and momentum operators involve harmonic terms, which read ${\nu_\iota(\xh) r^{-a_{n,p}}}$ and ${\xi_\iota(\xh) r^{-a_{n,p}}}$, respectively, in which $\xh \coloneqq x/|x| \in \Lambda_\iota$ in the chart at the asymptotic end $\Omega_\iota$. Here, $\nu_\iota$ and $\xi_\iota$ are scalar-valued and vector-valued fields defined on the sphere at infinity and belong to the kernels of the harmonic operators, analyzed in \refwithname{Propositions}{prop:Ham-kernel} \refwithname{and}{prop:Mom-kernel}.
Under the harmonic stability conditions, the kernels $\ker(\ssA^{\lambda_\iota})$ and $\ker(\ssB^{\lambda_\iota})$ are {\it of dimension $1$ and $n$,} respectively.
Furthermore, under the radial stability conditions, we are able to select normalized bases of these kernels.
Our expressions below use the parameter $d_{n,p}$ defined in~\eqref{equa-our-parame-00}, as well as $\theta^{\lambda_\iota}$ and $\eta^{\lambda_\iota}$, given explicitly in~\eqref{equa-thetalambda} and~\eqref{equa-thetalambda-M}. 

\begin{definition}
\label{def:normalized-kernel-basis}
At each asymptotic end $\Omega_\iota$, under the harmonic and radial stability conditions, the \textbf{silhouette function} $\nu^{\normal}_\iota$ and the \textbf{silhouette vector fields} $\xi^{\normal (j)}_\iota$ ($j=1,2,\ldots, n$)
are characterized as follows. The function is a basis of $\ker(\ssA^{\lambda_\iota})$, while the vector fields form a basis of $\ker(\ssB^{\lambda_\iota})$ (in the chart in $\Omega_\iota$):
\bel{equa-norm-normal} 
\aligned
& \nu^{\normal}_\iota \in \ker(\ssA^{\lambda_\iota}),  
&
&& \la - \Deltaslash\nu^{\normal}_\iota + d_{n,p} \, \nu^{\normal}_\iota \ra_\iota 
&= \theta^{\lambda_\iota}, 
\\
& \xi^{\normal (j)}_\iota \in \ker(\ssB^{\lambda_\iota}), 
&   
&& \bigl\la - \nablaslash_l \xi^{\normal (j) \perp}_\iota  + 2 a_{n,p} \xh_l\, \xi_\iota^{\normal (j) \perp} 
 +(a_{n,p}+1) \, \xi_\iota^{\normal (j) \parallel}{}_l \bigr \ra_\iota 
 &= \eta^{\lambda_\iota} \, \delta^{(j)}{}_{l} .
\endaligned
\ee
\end{definition}

As we will prove, the normalization in~\eqref{equa-norm-normal} ensures that the energy and momentum modulators (defined next) can be interpreted in the context of the ADM formalism; cf.~\autoref{corollary-bonne-normalisations}, below. 

\begin{definition}  
\label{def-317}
Let $(\Mbf,\Omega, g_0,h_0, \wtrr, \lambdabf)$ be a conical localization data set and $(p, p_G,p_A)$ be an admissible set of exponents satisfying~\eqref{equa-p-n-flat}, and assume $\lambdabf$ is a stable localization function (\autoref{def-41-stable}). Consider the cutoff functions $\kappa_\iota$ in~\eqref{equa-partition}. For any localized seed data set $(\seedg, \seedh)$, a pair 
\bel{equa-2p8}
\aligned
\seedmodg  
& = \seedg + \sum_\iota \kappa_\iota \, \gmodu_\iota, 
\qquad
\seedmodh = \seedh + \sum_\iota \kappa_\iota \, \hmodu_\iota, 
\endaligned
\ee
is called a \textbf{modulated seed data set} (with respect to $(\seedg,\seedh)$) if, for each asymptotic end~$\Omega_\iota$, there exists a constant scalar $\mmodu_\iota$ and a constant vector $\Jmodu_\iota = (\Jmodu_{\iota j})$ so that the correctors (or modulators) at each end read
\bel{equa-2p8-correct} 
\aligned 
\gmodu_\iota
& \coloneqq  
\lambda_\iota^{2\expoP} r^{n-2p} \bigl(\del_i \del_j  u^\modu_\iota - \delta_{ij} \Delta  \umodu_\iota \bigr)_{1\leq i,j\leq n}, 
&& \umodu_\iota \coloneqq \mmodu_\iota \, {\nu^\normal_\iota(x/r) \over r^{a_{n,p}}}, 
\\  
\hmodu_\iota
& \coloneqq - {1 \over 2} \, \lambda_\iota^{2\expoP} r^{n-2p-2} \big(\del_i \Zmodu_{\iota k} + \del_k \Zmodu_{\iota i} \big)_{1 \leq i,k \leq n},   
&& \Zmodu_\iota \coloneqq \Jmodu_{\iota j} \, {\xi^{\normal (j)}_\iota (x/r) \over r^{a_{n,p}}}. 
\endaligned
\ee 
Here, $\nu^\normal_\iota \in \ker\big( \ssA{}^{\lambda_\iota}\big)$ and $\xi^{\normal(j)}_\iota \in \ker\big( \ssB{}^{\lambda_\iota} \big)$ denote the silhouette functions and vector fields at each end $\iota=1,\ldots,I.$ Each pair $(\mmodu_\iota, \Jmodu_\iota)$ forms a spacetime vector defined at each asymptotic end and (in the dynamical picture) is referred to as a \textbf{modulated energy-momentum vector.}\footnote{We use the unconventional letters $m$ and~$J$ to avoid confusion with other quantities denoted using the letters~$E$ and~$P$.}
\end{definition}


\paragraph{Main result.}

We answer positively, and extend, a question raised by Carlotto and Schoen~\cite{CarlottoSchoen} for the gluing problem: for a broad range of projection, geometry, and accuracy exponents we establish that the behavior prescribed by the seed data set is achieved by the solution up to (and beyond) the $1/r^{n-2}$ Schwarzschild rate. At this juncture, we need an additional exponent denoted by $\pstar \geq p$, which we refer to as the \textbf{sharp decay exponent.} In comparison with~\eqref{condi--1-repeat} we now require that $p_A$ is larger than or equal to $\pstar$, so that the given seed data set provides us with a ``sufficiently accurate'' approximate solution in the vicinity of each asymptotic end. We {\it do not restrict} the behavior of the metric itself, which may have slow (or fast) decay.

At the harmonic threshold $\pstar=n-2$, we assume the following signed integrability condition at each end $\iota = 1,2,\ldots$:  there exist a constant $\mseed_\iota = \mseed_\iota (\seedg, \seedh)$ and a vector $\Jseed_\iota = \Jseed_\iota(\seedg, \seedh) \in \RR^n$ defined using the diffeomorphism $\phi_\iota\colon\Omega_\iota\to K_\iota\cap{}^\complement \Ball_R$ of \autoref{def-conical} by
\bel{equa-mstarJstar}
\aligned
2(n-1) |\Sphe^{n-1}| \mseed_\iota & \coloneqq - \int_{\Omega_\iota} \Hcal( \seedg, \seedh) \, \dVol_{\seedg}
\coloneqq - \lim_{r \to +\infty} \int_{(K_\iota\cap{}^\complement \Ball_R)\cap \Ball_r} \phi_{\iota\pushforward} \bigl( \Hcal( \seedg, \seedh)  \, \dVol_{\seedg} \bigr) ,
\\
(n-1) |\Sphe^{n-1}| \Jseed_{\iota j} & \coloneqq - \int_{\Omega_\iota} \la\Mcal( \seedg, \seedh), \phi_\iota^{\pullback}\del_j\ra \, \dVol_{\seedg}
\\
&
\coloneqq - \lim_{r \to +\infty}  \int_{(K_\iota\cap{}^\complement \Ball_R)\cap \Ball_r} \bigl(\phi_{\iota\pushforward} \Mcal( \seedg, \seedh)\bigr)_j \, \phi_{\iota\pushforward}\dVol_{\seedg},
\endaligned
\ee
where $\Ball_r\subset\RR^n$ denotes the ball of radius $r\geq R$ and $\la\cdot,\cdot\ra$ the pairing of the one-form~$\Mcal$ with the pull-back $\phi_\iota^{\pullback}\del_j$ of the standard unit vector.
Before introducing the relevant norm, set $\cutoff_{\pstar}=0$ for $\pstar<n-2$ and $\cutoff_{\pstar}=1$ for $\pstar\geq n-2$, that is, 
\bel{cutoff-pstar-def}
\cutoff_{\pstar} = \Oneone_{\pstar\geq n-2} = \bigl( 0 \text{ when } \pstar < (n-2) \text{ and $1$ otherwise} \bigr).
\ee 
This leads us to introduce the relevant norm for the Einstein operator of the seed data set,
\bel{equa-Ecal-pstar}
\aligned
& \Err^+_{\pstar}[\seedg,\seedh]
\coloneqq
\Norm[\big]{\Hcal(\seedg,\seedh)}^{N-2,\alpha}_{\Omega, g_0, \pstar+2, \expoPm-2, \expoP}
+ \Norm[\big]{\Mcal(\seedg,\seedh)}^{N-1,\alpha}_{\Omega, g_0, \pstar+2, \expoPm-1, \expoP}
\\
& \quad + \cutoff_{\pstar} \sum_\iota \sup_{r\geq R} \biggl( \biggl| \int_{\Kappa_\iota\cap {}^\complement \Ball_r} \phi_{\iota\pushforward} \bigl(\Hcal( \seedg, \seedh) \, \dVol_{\seedg}\bigr) \biggr|
+ \sum_{1\leq j\leq n} \biggl| \int_{\Kappa_\iota\cap {}^\complement \Ball_r} \bigl(\phi_{\iota\pushforward}\Mcal( \seedg, \seedh)\bigr)_j \phi_{\iota\pushforward} \dVol_{\seedg} \biggr| \biggr). 
\endaligned
\ee 
Throughout the paper, whenever an asymptotic modulator~$X$ is defined only in the 
\mbox{(super-)}\allowbreak harmonic regime, the compound expression $\cutoff_{\pstar}X$ is understood \emph{by cases}: it is defined to be zero when $\pstar<n-2$, without requiring or evaluating~$X$, and it equals~$X$ when $\pstar\geq n-2$. The same convention applies with any iteration exponent~$\pstar'$ in place of~$\pstar$. Thus expressions such as $\cutoff_{\pstar}\mmax(E)$, $\cutoff_{\pstar}\umodu$, $\cutoff_{\pstar}\Jmax(F)$, and $\cutoff_{\pstar}\Zmodu$ never involve a product of zero with an undefined object.
The supremum is well-defined under the assumption that the limits~\eqref{equa-mstarJstar} are well-defined.
This norm coincides with~\eqref{equa-Ecal-pq} when evaluated for $\pstar<n-2$, and is equivalent to it for $\pstar>n-2$ because the additional integral terms are controlled by the H\"older--Lebesgue norms of $\Hcal(\seedg,\seedh)$ and $\Mcal(\seedg,\seedh)$.

We arrive at our main result, which we state for clarity in the harmonic and super-harmonic range $\pstar\geq n-2$.
The full range of exponents $(p,p_G,p_A,\pstar)$ summarized in~\eqref{exponent-range} also includes the simpler case of sub-harmonic exponents $\pstar<n-2$, which we treat in \autoref{theo--beyond-harmonic-II}.
We refer to \autoref{section=4.4} for an outline of the proof for both ranges.
We also provide in \autoref{section=4.2} some sufficient conditions for the stability assumption made presently, namely weighted Poincaré inequalities and smallness of $n-2-p$, as stated in \refwithname{Theorems}{thm:informal-sufficient-stability} \refwithname{and}{thm:informal-sufficient-stability-M}.
Recall that $p, p_G, p_A$ satisfy 
\bel{equapGpAp}
p \in (p^\flat_n , n-2) , \qquad
p_G > 0, \qquad
p_A \geq \max(p_G, p). 
\ee
Let us specify how our sharp decay exponent is bounded, in view of the linearized Hamiltonian and momentum estimates below. 
For each end~$\Omega_\iota$, let $\deltaHiota>0$ and $\deltaMiota>0$ denote margins for which the Hamiltonian and momentum estimates in \refwithname{Theorems}{thm-sharp-h-localized} \refwithname{and}{thm-sharp-m-localized}, respectively, hold. Since the number of ends is finite, we may define the positive quantities
\bse
\label{p-lambda-def}
\be
\delta^{\lambdabf}_{n,p}
\coloneqq
\min\Big(
n-2,
\min_{1\leq\iota\leq I}\deltaHiota,
\min_{1\leq\iota\leq I}\deltaMiota\Big),
\qquad
p^{\lambdabf}_{n,p}\coloneqq n-2+\delta^{\lambdabf}_{n,p},
\ee
and then 
\be
\pmax\coloneqq\min(p^{\lambdabf}_{n,p}, n-2+p_G).
\ee
\ese
In particular, 
every exponent below~$p^{\lambdabf}_{n,p}$ lies simultaneously in the ranges of both linear estimates, and below the exponent $2(n-2)$ where possible nonlinear resonances between modulator terms would occur.

\begin{theorem}[The localized seed-to-solution projection ---harmonic and super-harmonic decay]
\label{theo--beyond-harmonic} 
Consider a conical localization data set $(\Mbf, \Omega, g_0,h_0, \wtrr, \lambdabf)$  together with admissible exponents $(p, p_G,p_A)$ satisfying~\eqref{equapGpAp} with the super-harmonic accuracy exponent $p_A \geq (n-2)$, and admissible localization exponents $(\expoPm, \expoP, \expoPp)$. Assume that $\lambdabf$ is a stable localization function (cf.~\autoref{def-41-stable}), 
and let $\pmax>n-2$ be the upper bound defined in~\eqref{p-lambda-def}.
 Then the following property holds for any \mbox{(super-)harmonic} exponent \be
\pstar \in [n-2, \pmax) \, \text{ and } \, \pstar \leq p_A 
 \qquad \text{ (harmonic and super-harmonic cases).}
\ee
Provided $\eps_G,\eps_A>0$ are sufficiently small (as functions of the localization data set and exponents), for any localized seed data set $(\seedg, \seedh) \in \Seed(\Omega, g_0,h_0,p_G, p_A, \eps_G,\eps_A)$ satisfying additionally $\Err^+_{\pstar}[\seedg,\seedh]<\eps_A$, and in particular the integrability\footnote{The condition is required in the harmonic case but, for super-harmonic exponents, is automatically satisfied.} 
\be
\Hcal(\seedg,\seedh), \, \Mcal(\seedg,\seedh)  \text{ are integrable in the sense } \eqref{equa-mstarJstar} \text{ when } \pstar =n-2, 
\ee
the solution $(g,h)$ to the Einstein constraints given by the {seed-to-solution map} $\Solu^{\lambdabf}_{n,p}$ (cf.~\autoref{thm:sts-existence}) enjoys the following pointwise decay estimates for some modulated seed data set $(\seedmodg, \seedmodh)$ associated with a collection of modulated energy-momentum vectors $(\mmodu_\iota, \Jmodu_\iota)$. 

\medskip

$\bullet$ \mbox{\textbf{(Super-)harmonic estimate.}}
The solution enjoys the estimate 
\bse
\bel{pstar-less2}
\aligned 
\| g - \seedmodg\|_{\Omega, g_0, \pstar, \expoPm}^{N,\alpha} +
\| h - \seedmodh \|_{\Omega, g_0, \pstar+1, \expoPm}^{N,\alpha}
& \lesssim \Err^+_{\pstar}[\seedg,\seedh], 
\endaligned 
\ee
where $\Err^+_{\pstar}[\seedg,\seedh]$ is defined in~\eqref{equa-Ecal-pstar} above.
Moreover, one has additional pointwise decay in which pointwise norms restricted to the the exterior region
$\Omega_R^\infty\coloneqq\{x\in\Omega:\wtrr(x)>R\}$
tend to zero in the limit $R\to+\infty$,
\bel{large-pstar-pointwise-decay}
\lim_{R\to+\infty} \Bigl( \| g - \seedmodg\|_{C^1_{\pstar,\expoPm}({\Omega_R^\infty})} +
\| h - \seedmodh \|_{C^0_{\pstar+1,\expoPm}({\Omega_R^\infty})} \Bigr) .
\ee
\ese
 
\medskip
 
$\bullet$ \textbf{ADM energy-momentum estimate.}
The modulators in the expressions~\eqref{equa-2p8-correct} enjoy the estimates
\bse\label{energyg2b}
\bel{energyg2b-1}
\sup_{1 \leq \iota \leq I} \big| \mmodu_\iota \big| + \sup_{1 \leq \iota \leq I} \big|  \Jmodu_\iota \big| 
\lesssim \Err^+_{\pstar}[\seedg,\seedh] .
\ee
If $h_0=0$ and the gluing domain $(\Omega,g_0)$ is (isometric to) a subset\footnote{Note that $\Omega$ can have multiple asymptotic ends despite being embedded into~$\RR^n$ since the spherical cap $\Omega\cap \Sphe_r$ for large radius~$r$ may be disconnected.} $\Omega\subset\RR^n$ of Euclidean space with the Euclidean metric $g_0=\delta$, then the total energy modulator $\mmodu\coloneqq\sum_{\iota=1}^{I} \mmodu_\iota$, and the total momentum modulator $\Jmodu = \sum_{\iota=1}^{I}\Jmodu_\iota$ (with each vector~$\Jmodu_\iota$ being expressed in the same coordinate system on~$\Omega$), are close to the total energy and momentum $(\mseed,\Jseed)$ of the source term~$\Gcal(\seedg,\seedh)$,
\bel{energyg2b-2}
\aligned
|\mmodu - \mseed| & \lesssim (\eps_A+\eps_G) \Err^+_{\pstar}[\seedg,\seedh],
& \quad
2(n-1) |\Sphe^{n-1}| \mseed & \coloneqq - \int_{\Omega} \Hcal( \seedg, \seedh) \, d^nx ,
\\
|\Jmodu - \Jseed| & \lesssim (\eps_A+\eps_G) \Err^+_{\pstar}[\seedg,\seedh],
& \quad
(n-1) |\Sphe^{n-1}| \Jseed_j & \coloneqq - \int_\Omega \Mcal(\seedg, \seedh)_j\, d^nx .
\endaligned
\ee
\ese
\end{theorem} 
 

\paragraph{A notion of relative ADM invariants.}

The energy-momentum modulators are interpreted as deviations from the ADM energy-momentum vectors, as follows. 
 
\begin{definition}\label{def:relative-ADM}
Let $(\Mbf, \Omega, g_0,h_0, \wtrr, \lambdabf)$ be a conical localization data set and consider one of its asymptotic ends $\Omega_\iota$. Given two Riemannian metrics $g,g'$ and two symmetric $(2,0)$-tensors $h,h'$, the \textbf{relative energy} and the \textbf{relative momentum vector} are defined (whenever the limits exist) as  a scalar in $\RR$ and a vector in $\RR^n$, respectively; in the asymptotic chart, the indices of $h-h'$ are lowered with~$\delta$:
\bel{equa-def-energy-momentum}
\aligned  
\mbb(\Omega_\iota,g-g') 
& \coloneqq {1 \over 2(n-1) \, |\Sphe^{n-1}|} \lim_{r \to + \infty} 
r^{n-1} \int_{\Lambda_\iota} 
\sum_{i,j = 1}^n {x_j \over r} \, \bigl( (g-g')_{ij,i} - (g-g')_{ii,j} \bigr) \Big|_{|x|=r} \, d\xh,
\\
\Jbb(\Omega_\iota, h-h')_{j}
&
\coloneqq  {1 \over (n-1) \, |\Sphe^{n-1}|} \lim_{r \to + \infty} 
r^{n-1} \int_{\Lambda_\iota} \sum_{1 \leq k \leq n}{x_k \over r} \,  (h-h')_{jk} \Big|_{|x|=r} \, d\xh.
\endaligned
\ee 
\end{definition}

Our definition makes sense even when the initial data set does not admit a standard notion of energy and momentum, and only requires that the {\it difference} has sufficient decay for the integrals in~\eqref{equa-def-energy-momentum} to make sense. Namely, it is only the {\it difference} which is relevant in our setup.  
While we have $\mbb(\Omega_\iota,g-g') = \mbb(\Omega_\iota,g) - \mbb(\Omega_\iota,g')$ {\it provided} both energies are finite, it is possible for the relative energy to be finite for metrics $g$ and $g'$ having infinite energy. 
 
\begin{corollary}
\label{corollary-bonne-normalisations} 
 The energy-momentum vector $(\mmodu_\iota, \Jmodu_\iota)$ in \autoref{theo--beyond-harmonic} can be interpreted as the relative energy-momentum associated with the prescribed data set and the actual solution, that is, 
\be
\mmodu_\iota = \mbb(\Omega_\iota, g - \seedg),
\qquad
\Jmodu_\iota = \Jbb(\Omega_\iota, h - \seedh).
\ee 
\end{corollary}

\begin{proof}
\autoref{theo--beyond-harmonic} concerns the harmonic or super-harmonic case $\pstar\geq n-2$ and establishes pointwise decay of $(g-\seedmodg, h-\seedmodh)$.
Specifically, by \eqref{large-pstar-pointwise-decay}, $g-\seedmodg$ decays faster than $r^{-\pstar}\leq r^{-(n-2)}$, while its first derivatives and $h-\seedmodh$ decay faster than~$r^{-(n-1)}$.  Thus, from the explicit expressions~\eqref{equa-def-energy-momentum}, $\mbb(\Omega_\iota,g-\seedg) = \mbb(\Omega_\iota,\seedmodg-\seedg)$ and $\Jbb(\Omega_\iota, h-\seedh) = \Jbb(\Omega_\iota, \seedmodh-\seedh)$.
On the other hand, at sufficiently large radii, \eqref{equa-2p8} reduces to $\seedmodg-\seedg = \gmodu_\iota$ and $\seedmodh-\seedh=\hmodu_\iota$.
\autoref{lem-ADMenergymod} and \autoref{lem-ADMmomentmod}, below, determine $\mbb(\gmodu_\iota) = \mmodu_\iota$ and $\Jbb(\hmodu_\iota) = \Jmodu_\iota$ by a direct evaluation.
\end{proof}

Furthermore, we recall that standard notions of energy and momentum enjoy positivity properties~\cite{BeigChrusciel-1996,SchoenYau-79, Witten-81}, and it would be interesting to also investigate the notion of center of mass (cf.~for instance~\cite{CederbaumSakovich}) for localized initial data sets. Additional results will be presented in \cite{LL-next}, including an extension of the technique in~\cite{LeFlochNguyen-preprint,LeFlochNguyen} for ``engineering'' seed data sets in order to reach a specific energy-momentum at infinity. 


\paragraph{Revisiting Carlotto-Schoen's sub-harmonic estimates.}

In fact, we improve the (sub-harmonic) estimates established earlier in \autoref{thm:sts-existence} even in the sub-harmonic regime.
Assuming Hamiltonian and momentum stability conditions on the localization function~$\lambdabf$, the bounds are improved as follows from the projection exponent~$p$ to a larger sharp decay exponent~$\pstar$.

\begin{proposition}[The localized seed-to-solution projection ---refined sub-harmonic control]
\label{theo--beyond-harmonic-II}
Under the same assumptions as in \autoref{theo--beyond-harmonic} but for any \emph{sub-harmonic} exponent~$\pstar$ such that
\be
p \leq \pstar \leq p_A \, \text{ and } \,  \pstar < n-2 \qquad \text{ (sub-harmonic case),}
\ee
the following estimate holds:
\bel{pstar-less2-SUB}
\aligned
\| g  - \seedg \|_{\Omega, g_0, \pstar, \expoPm}^{N,\alpha} +
\| h  - \seedh \|_{\Omega, g_0, \pstar+1, \expoPm}^{N,\alpha}
& \lesssim \Err^+_{\pstar}[\seedg,\seedh].
\endaligned
\ee
\end{proposition} 


\subsection{Sufficient conditions for stability}
\label{section=4.2}

Our main results have uncovered the general functional inequalities that are required in order to control the decay of solutions and cope with harmonic contributions. It is natural to also seek {\sl sufficient conditions} that imply such stability conditions while being comparatively easier to check for a choice of localization domain and localization function. 
We quote here two results proved in our companion paper~\cite{LL-PoincareKornHardy}; they provide directly verifiable sufficient conditions for the abstract stability hypotheses used in the present paper. 
The principal localization theorems of the present paper are formulated under those abstract stability hypotheses; only their verification for the concrete classes below, and hence \autoref{theo-smallaper}, depends on the companion results~\cite{LL-PoincareKornHardy}.
Namely, our stability conditions are naturally implied by concrete inequalities to be imposed on the diameter and Poincaré constant of the localization domain under consideration and on the exponent~$a_{n,p}$.
For the Hamiltonian operator, harmonic and radial stability follow when $(n-2-p)C_\Poin^{\lambda_\iota}$ is sufficiently small, while shell stability follows from the geometric smallness assumptions stated below.  For the momentum operator, the domain may be arbitrary, but $n-2-p$ must be sufficiently small, with a threshold depending on the domain.
Allowing for a broad choice of gluing domains is of particular interest since this opens the way to design rather complex gluing structures at infinity, while still keeping a sharp control on the decay of the solutions.
 
To a localization domain $\Lambda_\iota \subset \Sphe^{n-1}$ together with a localization function $\lambda_\iota$ supported in $\Lambda_\iota$, we associate its diameter measured using the geodesic distance on the sphere (with the round metric~$\gslash$),
\bel{diamLambdadef}
\diam(\Lambda_\iota) = \sup_{x,y\in\Lambda_\iota} \dbf_{\gslash}(x,y) ,
\ee
and the weighted Poincaré constant on $\Lambda_\iota$ endowed with the weighted measure $d\chi_\iota= \lambda_\iota^{2 \expoP} d\xh$, defined as the optimal constant $C_\Poin^{\lambda_\iota}>0$ in   
\bel{equa-Poin1}
\| \nu - \la\nu\ra_\iota \|_{\unL^2_{-\expoP}(\Lambda_\iota)}
\leq  C_\Poin^{\lambda_\iota} \, \| \nablaslash\nu \|_{\unL^2_{-\expoP}(\Lambda_\iota)} 
\ee 
for every $\nu\in\unH^1_{-\expoP}(\Lambda_\iota)$. This constant can be evaluated by standard techniques and, for instance, without localization we would find $\diam(\Sphe^{n-1})=\pi$ and $C_\Poin^{\lambda\equiv 1} = 1/\sqrt{n-1}$.
Before stating the second-order inequality, define the symmetric positive-definite matrix
\be
G_{\iota kl}
\coloneqq
\delta_{kl}+\la\xh_k\xh_l\ra_\iota
-2\la\xh_k\ra_\iota\la\xh_l\ra_\iota.
\ee
The matrix $G$ is positive-definite as the sum of $\la\gslash\ra$ and twice the covariance matrix of~$\xh$.  We find for instance $C_{\Poin 2}^{\lambda\equiv 1} = 1/\sqrt{n+2}$. We also associate a second-order Poincaré constant defined as the optimal constant $C_{\Poin 2}^{\lambda_\iota}>0$ in
\bel{equa-Poin2}
\| \nablaslash\nu \|_{\unL^2_{-\expoP}(\Lambda_\iota)}^2
\leq  (C_{\Poin 2}^{\lambda_\iota})^2 \, \| \nablaslash^2\nu \|_{\unL^2_{-\expoP}(\Lambda_\iota)}^2 + \sum_{k,l} G_{\iota kl} \la\nablaslash_k\nu\ra_\iota \la\nablaslash_l\nu\ra_\iota. 
\ee

\begin{theorem}[Sufficient conditions for Hamiltonian stability~\cite{LL-PoincareKornHardy}]\label{thm:informal-sufficient-stability} 
There exist universal constants $c_1^\notreH(n)$ and~$c_2^\notreH(n)$ depending on the dimension~$n$, only, such that the following holds.

1. For all localization domains $(\Lambda_\iota, \lambda_\iota)$ whose Poincaré constant $C_\Poin^{\lambda_\iota}$ (cf.~\eqref{equa-Poin1}, above) is bounded above as
\be
(n-2-p) C_\Poin^{\lambda_\iota} \leq c_1^\notreH(n) ,
\ee
the harmonic stability condition~\eqref{equa-stable-H-414} in \autoref{def-harmonic-Hstab} and the radial stability condition~\eqref{equa-b2-positive} in \autoref{def-radial-Hstab} for the Hamiltonian operator do hold. \textup{(}Hence, these stability conditions always hold if the projection exponent $p$ is picked to be sufficiently close\footnote{or, equivalently, the coefficient $a_{n,p}$ to be sufficiently small} to the harmonic exponent $(n-2)$.\textup{)}

2. If the diameter and Poincaré constants $C_\Poin^{\lambda_\iota}$ and $C_{\Poin 2}^{\lambda_\iota}$ are bounded as
\be
\diam(\Lambda_\iota) + C_\Poin^{\lambda_\iota} + C_{\Poin 2}^{\lambda_\iota} < c_2^\notreH(n) ,
\ee
then the Hamiltonian shell stability condition in \autoref{def-shell-Hstab} also holds.
\end{theorem}

\begin{theorem}[Sufficient conditions for momentum stability~\cite{LL-PoincareKornHardy}]\label{thm:informal-sufficient-stability-M}
For any localization domain $(\Lambda_\iota, \lambda_\iota)$ there exists a geometric constant $c_{\Poin\Korn\Hardy}(\Lambda_\iota,\lambda_\iota)>0$ (related to Poincaré, Korn, and Hardy inequalities) such that, for any exponent~$p$ with
\[
0 < n-2-p < c_{\Poin\Korn\Hardy}(\Lambda_\iota,\lambda_\iota),
\]
the (harmonic, radial, shell) momentum stability conditions hold.
\end{theorem}

It is additionally shown in \cite{LL-PoincareKornHardy} that in families of domains obtained by scaling a given reference domain and projecting it onto the sphere the Poincaré constants $C_\Poin^{\lambda_\iota}$ and $C_{\Poin 2}^{\lambda_\iota}$ scale as the diameter, so that the premises of \autoref{thm:informal-sufficient-stability} are obeyed for sufficiently small $\diam(\Lambda_\iota)$.  Once the domain is chosen, the premises of \autoref{thm:informal-sufficient-stability-M} are obeyed for sufficiently small $a_{n,p} = 2(n-2-p)$.  Thus, all six stability conditions are satisfied.  In short, sufficiently small $(\Lambda_\iota,\lambda_\iota)$ with $p$ sufficiently close to $n-2$ give a stable localization setup. Consequently, the two results quoted from~\cite{LL-PoincareKornHardy} imply \autoref{theo-smallaper}. 


\subsection{Isotropic construction}
\label{section=4.3}

In the special case $\lambda_\iota \equiv 1$ when no localization is required in an asymptotic end~$\Omega_\iota$, the relevant geometric objects can be worked out explicitly, and the stability conditions hold, as we now state in \autoref{thm:stable-s2}.  Some details rely on general explicit expressions that are only worked out later on, but the proof can be read without referring to these details.

The harmonic contributions are given {\it explicitly}, namely the silhouette functions and vector fields are
\bel{equa-explicit-case}
\left.
\begin{aligned}
\nu^\normal_\iota(\xh) & = {2 \over a_{n,p}b_{n,p}}
\\
\xi^{\normal (j)}_\iota (\xh)_k & =
\frac{3n(n-1)}{a_{n,p}(n+a_{n,p}+2)} \Bigl(\delta_{jk} + \frac{a_{n,p}}{3} \, \xh_j \, \xh_k\Bigr)
\end{aligned}
\right\}
\quad \text{ (case $\lambda_\iota \equiv 1$).}  
\ee 
The first formula is obvious in view of the normalization~\eqref{equa-norm-normal}.
For the second formula we begin by seeking solutions $\xi^{(j)}_k = \alpha\delta_{jk} + \beta \xh_j \xh_k$ to $\notreM^{\lambda_\iota}[Z^{(j)}]=0$, with $Z^{(j)}=r^{-a_{n,p}}\xi$ and the operator~$\notreM^{\lambda_\iota}$ given explicitly in~\eqref{equa:acalew0-deux}.  We evaluate successively, using $\del_k r=x_k/r\eqqcolon\xh_k$ and $r\del_k(\xh_j)=\delta_{jk}-\xh_j\xh_k$,
\be
\aligned
\del_k Z^{(j)}_l + \del_l Z^{(j)}_k
& = r^{-a_{n,p}-1} \bigl(
2 \beta \xh_j \delta_{kl} + (\beta - a_{n,p}\alpha) (\xh_k \delta_{jl} + \xh_l \delta_{jk})
- 2 (a_{n,p} + 2) \beta \xh_j\xh_k\xh_l \bigr) ,
\\
r^{a_{n,p}+2} \notreM^{\lambda_\iota}[Z^{(j)}]_l
& = - {1 \over 2} r^n \del_k(r^{a_{n,p}-n+2}(\del_kZ_l+\del_lZ_k))
= - {1 \over 2} (3\beta - a_{n,p}\alpha) (\delta_{jl} - n\xh_j\xh_l) ,
\endaligned
\ee
which vanishes provided $\beta=\alpha a_{n,p}/3$.  This gives~\eqref{equa-explicit-case} after taking into account the normalization~\eqref{equa-norm-normal}. 

\begin{theorem}[Stability without localization]\label{thm:stable-s2}
In dimension $n\leq 17$ and with $p<n-2$ sufficiently close to $n-2$ (as a function of dimension, only), the localization function $\lambda_\iota\equiv 1$ is (Hamiltonian and momentum) stable in the sense of \autoref{def-41-stable}.
\end{theorem}

\begin{proof}
The restriction on dimension will only play a role for Hamiltonian shell stability.
Momentum stability is immediate from \autoref{thm:informal-sufficient-stability-M}, as the geometric constant $c_{\Poin\Korn\Hardy}(\Lambda_\iota,\lambda_\iota)$ there can only depend on the localization domain and measure $d\chi_\iota=\lambda_\iota^{2\expoP}d\xh$, which here only depends on the dimension.

For the Hamiltonian harmonic stability condition, we observe that in the absence of weight the quadratic functional $\ssrmA^{\lambda_\iota\equiv 1}[\nu]$ given in~\eqref{ssAalpha-quaform-0} is non-negative since the $-\nu\Deltaslash\nu$ term can be integrated by parts to the non-negative term $|\nablaslash\nu|^2$.
The functional is then coercive on functions with vanishing average due to the standard Poincaré inequality.

For the Hamiltonian radial stability condition, we observe that the kernel of $\ssA^{\lambda_\iota\equiv 1}$ manifestly contains the constant functions, and it is one-dimensional due to harmonic stability.  We then evaluate $\bnotreH_{\iota 1} = (n-1) c_{n,p}\la\nu^\normal\ra_\iota$ and $\bnotreH_{\iota 0} = (n^2 -4n+5) d_{n,p} (c_{n,p} / a_{n,p}) \la\nu^\normal\ra_\iota$, whose product is positive for any non-zero element~$\nu^\normal$ of the kernel.
The normalization of~$\nu^\normal$ is given in~\eqref{equa-explicit-case}, above.

As a warm-up for the shell dissipation functional considered next, we show that (in terms of the fluctuations $\ut = u - \la u\ra_\iota$),
\bel{rmA-iso-coer}
\ssrmA^{\lambda_\iota\equiv 1}[u]
\geq (n-1)(n^2-3n+5+2a_{n,p}+c_{n,p}) \fint_{\Sphe^{n-1}} \ut^2 d\xh .
\ee
An integration by parts rewrites the functional as
\bel{rmA-iso-Delta}
\ssrmA^{\lambda_\iota\equiv 1}[u]
= \fint_{\Sphe^{n-1}} \Bigl( (n-1) (\Deltaslash u)^2
+ (n-4-2a_{n,p}-c_{n,p}) u \Deltaslash u
\Bigr) \, d\xh .
\ee
Upon decomposing~$u$ into an orthonormal basis of modes of the spherical Laplacian~$\Deltaslash$, all cross terms integrate to zero, so that~\eqref{rmA-iso-Delta} splits into a sum over (non-zero) modes.
It is then enough to check the scalar inequality
\be
(n-1) s^2  - (n-4-2a_{n,p}-c_{n,p}) s
\geq (n-1)(n^2-3n+5+2a_{n,p}+c_{n,p})
\ee
for all non-zero eigenvalues $s$ of $-\Deltaslash$, namely $s = \ell(n-2+\ell)$ with integer $\ell\geq 1$.  This inequality holds for all $s\geq n-1$.

For the Hamiltonian shell stability condition, the explicit expression of $\Kappa^\notreH[\ut]$ given in \autoref{def:KappaH}, later on, involves only fluctuations of~$\nu^\normal$ and not its average, hence $\Kappa^\notreH[\ut]=0$.  This simplifies shell stability drastically since it now boils down to choosing the parameters $\cstun,\cstdeux$ of the shell functional in such a way that $\Phi^\notreH_\iota$ is non-negative and $\Psi^\notreH_{\beta\iota}[u]+C\la u\ra_\iota^2$ is coercive for a large enough $C>0$.
We shall select for convenience the specific values
\be
\cstun = b_{n,p} \geq 2 , \qquad \cstdeux = 0 .
\ee
The shell functional $\Phi^\notreH_\iota$ given in~\eqref{PhinotreH-expr} is manifestly positive by an integration by parts of the $-u\Deltaslash u$ term to $|\nablaslash u|^2$.

The dissipation functional, given explicitly in~\eqref{PsinotreH-expr}, features the cancellation of a $u\vartheta w$ term that motivates our choice of~$\cstun$.
We are interested in showing its coercivity (modulo $\la u\ra^2$) for small $a_{n,p}=2(n-2-p)$, so it is enough to show coercivity in the limit $a_{n,p}\to 0$, as the dissipation is continuous in~$a_{n,p}$.
We recall that $\beta=a_{n,p},2a_{n,p}$ so $\beta\to 0$ as well.
Note that in this limit,
\be
\cstun\to n^2-5n+8 \eqqcolon \cstun(n) .
\ee

Similarly to $\ssrmA^{\lambda_\iota\equiv 1}$, all angular derivatives in the dissipation functional can be rewritten in terms of~$\Deltaslash$.  Upon decomposing $u,\vartheta u,\vartheta^2 u,\vartheta^3 u$ on the sphere of radius~$r$ into modes of the (spherical) Laplacian, none of the cross-terms contribute.  It thus suffices to prove coercivity for $a_{n,p}=0$ and functions $\vartheta^k u = q_k \nu$ ($k=0,1,2,3$) that are multiples of the same Laplacian eigenmode~$\nu$, with $\Deltaslash\nu=-s\nu$ where $s=\ell(n-2+\ell)$ for some integer $\ell\geq 0$.  In that setting the functional reduces to
\bel{Psi-one-mode}
\aligned
\frac{\Psi^\notreH_{\beta\iota}[u]}{\fint_{\Sphe^{n-1}} \nu^2 \, d\xh}
& = \Bigl(q_3 + \frac{n-2}{n-1} s q_1\Bigr)^2
+ \frac{\cstun(n)}{n-1} \bigl((n-1) s^2  - (n-4) s\bigr) q_0^2
\\
& \quad + \Bigl( \frac{2 s}{n-1} + 2 \cstun(n) \Bigr) q_2^2
- \frac{2}{n-1} \cstun(n) s q_1 q_2
\\
& \quad + \Bigl( \frac{2n-3}{(n-1)^2} s^2 - \frac{n-4}{n-1} s + 2 \cstun(n) s + \cstun(n)^2 \Bigr) q_1^2 .
\endaligned
\ee
For the zero-mode $s=0$ (constant~$\nu$) this controls $q_3^2+q_2^2+q_1^2$, so that $\Psi^\notreH_{\beta\iota}[u]$ controls $\la\vartheta^k u\ra_\iota$ for $k=1,2,3$.
For $s>0$ modes, we must check that this controls $q_3^2 + s q_2^2 + s^2 q_1^2 + s^2 q_0^2$ uniformly in~$s$.  Note that the last term does not feature $s^3$ since shell stability only requires a control of two angular derivatives of~$u$.
The leading term of~\eqref{Psi-one-mode} in the $s\to+\infty$ limit at fixed $q_3,s^{1/2}q_2,sq_1,sq_0$ is
\bel{Psi-one-mode-larges}
\Bigl(q_3 + \frac{n-2}{n-1} s q_1\Bigr)^2
+ \frac{2}{n-1} s q_2^2
+ \frac{2n-3}{(n-1)^2} s^2 q_1^2
+ \cstun(n) s^2 q_0^2 ,
\ee
as needed.
Thus, there only remains to show that \eqref{Psi-one-mode} is positive-definite for all $s>0$ in the spectrum.
The first two terms manifestly are.  The next three constitute a quadratic form in $q_2,q_1$ whose $q_2^2$ term is positive, so altogether only the determinant
\be
P(s) \coloneqq \Bigl( \frac{2 s}{n{-}1} + 2 \cstun(n) \Bigr) \Bigl( \frac{2n{-}3}{(n{-}1)^2} s^2 - \frac{n{-}4}{n{-}1} s + 2 \cstun(n) s + \cstun(n)^2 \Bigr)
- \frac{(n^2{-}4n{+}5)^2 s^2}{(n{-}1)^2}
\ee
needs to be checked to be positive.
For $3\leq n\leq 9$ it has positive coefficients hence is positive for $s>0$.
For $10\leq n\leq 17$ its $s^2$ term has a negative coefficient but is bounded by its $s$ and $s^3$ terms.
For larger dimensions, $P$~is a cubic polynomial with three real roots and there are some modes with $P(s)<0$, so that the dissipation functional is not coercive and Hamiltonian shell stability fails for this choice of $\cstun,\cstdeux$.
As an illustration, for $n=18$ the modes $55\leq\ell\leq 63$ have $P(\ell(n-2+\ell))<0$.
\end{proof}


\subsection{Outline of the proof}
\label{section=4.4}

We summarize here the main steps leading us to the proof of \autoref{theo--beyond-harmonic}. The remainder of this paper is organized as follows.
\refwithname{Sections}{section=5}--\ref{section=7} and \refwithname{Sections}{section=8}--\ref{section=9} concern the localized Hamiltonian and momentum operators, respectively, and \refwithname{Sections}{section=10}--\ref{section=11} control nonlinearities.

\paragraph{\autoref{section=5}: Linear estimates for the squared localized Hamiltonian operator.}

We establish first three results for the Hamiltonian operator linearized around Euclidean data, which together provide us with the desired linearized decay estimates: \autoref{thm-sharp-h-localized-vari} (variational formulation) with sub-harmonic decay, \autoref{thm:decayL2} (integral estimates, proven in \refwithname{Sections}{section=6}--\ref{section=7}) which leverages our new stability conditions to improve the decay rate, and \autoref{thm-sharp-h-localized} (pointwise estimates) based on standard ellipticity arguments.

\paragraph{\autoref{section=6}: Consequences of Hamiltonian harmonic and radial stability.}

We derive the harmonic-spherical decomposition of the Hamiltonian operator in \autoref{lem:sph-Ham} and extract consequences of stability conditions.
Harmonic stability allows us to construct the silhouette function $\nu^{\normal}_\iota \in \ker(\ssA^{\lambda_\iota})$ in \autoref{prop:Ham-kernel}, whose normalization~\eqref{equa-norm-normal} is motivated by \autoref{lem-ADMenergymod} where we compute the ADM energy of the modulator, used in the proof of \autoref{corollary-bonne-normalisations}.
Radial stability leads to \autoref{prop-66} which expresses spherical averages in terms of the functional~$\Kappa^\notreH_\iota$.

\paragraph{\autoref{section=7}: Integral estimates for the Hamiltonian (\autoref{thm:decayL2}).}

We define the Hamiltonian shell identity in a weak sense in \autoref{section=7.1}.
By integrating this identity, we obtain in \autoref{section=7.2} that source terms control a radial integral of the dissipation functionals~$\Psi^\notreH_{\beta\iota}[u]$.  Thanks to Hamiltonian shell stability, these functionals control Sobolev norms of the solution, modulo a contribution of the average~$\la u\ra$ that is controlled in \autoref{section=7.3} based on \autoref{prop-66}.

\paragraph{\autoref{section=8}: Linear estimates for the squared localized momentum operator.}
 
Next, we establish three results for the momentum operator, which all together lead us to the desired linearized decay estimates, namely \autoref{thm-sharp-m-localized-vari} (variational formulation), \autoref{thm:decayL2-M} (integral estimates), and \autoref{thm-sharp-m-localized} (pointwise estimates). As in \autoref{section=5}, these results arise as consequences of our new stability conditions in combination with standard ellipticity arguments.  The proof of \autoref{thm:decayL2-M} relies on calculations in the next section.

\paragraph{\autoref{section=9}: Consequences of momentum harmonic and radial stability.}

We investigate the structure enjoyed by the momentum operator in \autoref{prop:Mom-kernel} (kernel properties of the asymptotic momentum) and \autoref{prop-87-moment} (spherical averages). The normalization of the silhouette vector fields $ \xi^{\normal (j)}_\iota \in \ker(\ssB^{\lambda_\iota})$ in~\eqref{equa-norm-normal} is motivated by \autoref{lem-ADMmomentmod}, which yields the ADM momentum of the modulator, thus completing the proof of \autoref{corollary-bonne-normalisations}.

\paragraph{\autoref{section=10}: Nonlinear analysis of the squared localized Einstein constraints.}

Finally, we are in a position to analyze curvature terms and nonlinearities to determine decay properties of the localized seed-to-solution projection $\Solu^{\lambdabf}_{n,p}$ in~\eqref{equa-sol-map}.  Starting from the low radial decay provided by the existence result \autoref{thm:sts-existence}, we iteratively improve the radial exponent based on \autoref{prop:improve-radial}, to reach the desired sub-harmonic, harmonic, or super-harmonic decay for the localized Einstein equations, thus proving \autoref{theo--beyond-harmonic} and \autoref{theo--beyond-harmonic-II}.

\paragraph{\autoref{section=11}: Proof of the radial decay improvement (\autoref{prop:improve-radial}).}

To close the proof of the main results, there remains to determine how the radial decay exponent of $(u,Z)$ in an asymptotic end can be improved.
We expand the nonlinear Einstein constraints and their adjoint, first around a curved-space linear operator to deduce bounds on $(\vartheta u,\vartheta Z)$ with the same radial decay as $(u,Z)$ in \autoref{section=11.2}, and then around Euclidean space in \autoref{section=11.3} to improve the decay rate of $(u,Z)$ using the linearized estimates in~\refwithname{Sections}{section=5} \refwithname{and}{section=8}.


\section{Linear estimates for the squared localized Hamiltonian operator}
\label{section=5}

\subsection{Conical localization framework}
\label{section=5.1}

In this section and the next, we focus on the linearization of the (squared) localized Hamiltonian operator~$\notreH^\lambda$ within a cone of the $n$-dimensional Euclidean space (with $n \geq 3$), and we aim at establishing sharp decay.
A basic variational formulation provides us first with a control of the decay of solutions in a mild integral sense (\autoref{thm-sharp-h-localized-vari}), and our challenge in the present section is to derive {\it sharp integral} estimates (\autoref{thm:decayL2}) and {\sl sharp pointwise} estimates (\autoref{thm-sharp-h-localized}) in weighted norms.
Results that rely on a decomposition of~$\notreH^\lambda$ into radial and angular derivatives are all collected in \autoref{section=6} and quoted as needed in the present section.
For clarity in the presentation, we summarize our notation as follows. (We point out that our analysis encompasses the special case where $\lambda$ is taken to be identically $1$ and $K$ is chosen to be $\RR^n$.) 
 
\bei 

\item {\bf Gluing domain.} We work in an open cone $K \subset \RR^n$ with smooth boundary (except at its tip $r= 0)$, truncated by restricting attention to the exterior ${}^\complement \Ball_R \subset \RR^n$ of the (closed) ball $\Ball_R$ with fixed radius $R>1$. That is, we consider the open set $\Omega_R \coloneqq K \cap {}^{\complement} \Ball_R$. The radius $R>1$ is arbitrarily fixed  (but will be chosen to be sufficiently large in \autoref{section=10}).  

\item {\bf Radius function.} In our coordinates in $\Omega_R$ the radius $r^2 = \sum (x^j)^2 > R^2$ is bounded away from zero. 

\item {\bf Localization function.} The function $\lambda \geq 0$ is positive in the interior of $K$, depends upon the angular variable $\xh=x/|x|$ only, and vanishes linearly as the boundary $\del K$ is approached. We regard $\lambda: \Lambda \subset \Sphe^{n-1}\to (0,\lambda_0]$ as a function defined over an open subset~$\Lambda$ of the sphere at infinity~$\Sphe^{n-1}$, and we assume that $\Lambda$ is \emph{connected}. For clarity, this connectedness property will be emphasized in our statements below (in order for the asymptotic kernel to be one-dimensional). 

\item {\bf Variational weight.} In this section, the weight~\eqref{equa-omega-choice} reads ${\omega_p = r^{n/2-p} \, \lambda^\expoP \geq 0}$ (without boldface) for some reals $p \in (p_n^\flat, n-2)$ and ${\expoP \geq 2}$. 

\item {\bf Boundary.} The boundary $\del\Omega_R$ consists of a subset of the boundary of $K$ together with a subset $\Sphe_R \cap K=\Lambda_R$ of the sphere. While no boundary condition is explicitly required along $\del \Omega_R \setminus \Sphe_R$ (due to the presence of the weight $\lambda^\expoP$), it will be necessary to specify a Neumann-type boundary condition along the boundary $\Lambda_R$. The specific choice is important in order to easily compute the harmonic terms without spurious contribution from this boundary.\footnote{However, this is unimportant for the application in \autoref{section=10} since the solutions there vanish in the vicinity of~$\Lambda_R$.}

\eei

In agreement with \autoref{section=2.2} we use  the notation $L^2_{q, -\expoP}(\Omega_R)$ (cf.~\eqref{equa-def-weightL}) for the weighted Lebesgue spaces and $H^m_{q,-\expoP}(\Omega_R)$ (for $m= 1,2, \ldots$) for the corresponding weighted Sobolev spaces.
Sharp pointwise and integral estimates for the metric in \autoref{theo--beyond-harmonic} involve a sharp decay exponent~$\pstar$.
Since our description here focuses on the solution~$u$ instead, it proves convenient to work in terms of the exponent
\be
\astar = n - 2 - 2p + \pstar ,
\ee
which accounts for the two derivatives and $\omega_p^2$ factor in~\eqref{equa--221}.
The variational exponent $\pstar=p$ corresponds to $\astar=a_{n,p}/2$ while the harmonic exponent $\pstar=n-2$ corresponds to $\astar=a_{n,p}$.
We recall also from~\eqref{equa:acalew0-deux} that the (squared) localized Hamiltonian operator
\be
\notreH^\lambda[u] 
= \omega_p^{-2} \, \Bigl(
(n-2) \, \Delta(\omega_p^2 \, \Delta u)
+ \del_i\del_j(\omega_p^2 \del_i\del_j u)
\Bigr)
\ee
involves a scaling factor~$\omega_p^{-2}$.
The norm~\eqref{equa-Ecal-pstar} of the Hamiltonian operator of the seed data set, restricted to one asymptotic end, and expressed in terms of a source term $E=(\omega_p)^{-2}\Hcal(\seedg,\seedh)$, is equivalent to the sum of a H\"older and Lebesgue norms of~$E$, and an energy term.
However, results of the present section will be applied in \autoref{section=10} with a source term constructed from fourth derivatives of the variational solution $u$ in a weighted $H^2$ space, so that the Lebesgue norm has to be replaced here by the (weaker) operator norm
\bel{H2star-def}
\|E\|_{H^{2*}_{\astar+4,-\expoP}(\Omega_R)}
= \sup_{v\in H^2_{-\astar-4,-\expoP}(\Omega_R)} \frac{\int_{\Omega_R} E v \lambda^{2\expoP} r^{-n} dx}{\|v\|_{H^2_{-\astar-4,-\expoP}(\Omega_R)}} ,
\ee
which defines $H^{2*}_{\astar+4,-\expoP}(\Omega_R)$ as the dual of $H^2_{-\astar-4,-\expoP}(\Omega_R)$ under the measure $\lambda^{2\expoP} r^{-n}dx=d\chi\,dr/r$, for any decay exponent $\astar+4$.

In \autoref{section=5.2} we provide a variational formulation of the equation $\notreH^\lambda[u]=E$ with Neumann boundary conditions and prove that a weighted $H^2$ norm of~$u$ is controlled by a weighted $H^{2*}$ norm of~$E$.
It will be important for our purposes to control one additional radial derivative of the source term~$E$ as well as the solution~$u$.
We introduce the distributional derivative $\vartheta_* E\in H^{3*}_{\astar+4,-\expoP}(\Omega_R)$ defined by\footnote{For smooth $E$, this derivative differs from $\vartheta E$ by boundary terms (cf.~\eqref{vartheta-star-def}).}
\be
\int_{\Omega_R} \vartheta_* E \, v \lambda^{2\expoP} r^{-n} dx
\coloneqq - \int_{\Omega_R} E \, \vartheta v \lambda^{2\expoP} r^{-n} dx ,
\qquad v\in H^3_{-\astar-4,-\expoP}(\Omega_R) ,
\ee
and we assume the improved regularity $\vartheta_* E\in H^{2*}_{\astar+4,-\expoP}(\Omega_R)$.
Since the radial derivative does not commute with Neumann boundary conditions, the $H^2$ norm of $\vartheta u$ with variational exponent is not controlled by this assumption, and we include it separately.

Before introducing the norm, we define the truncated energy, and its total value and tail bound, as follows: 
\bel{mmodu-def}
\aligned
\mmodu(E;r_1,r_2) & \coloneqq
\frac{1}{2(n-1)|\Sphe^{n-1}|} \int_{\Omega_{r_1}\setminus\Omega_{r_2}} E \, r^{n-2p} \lambda^{2\expoP} dx
\\
& = \frac{\aire[\Lambda,\lambda]}{2(n-1)|\Sphe^{n-1}|} \int_{r_1}^{r_2} \la E(s)\ra \, s^{3+a_{n,p}} ds ,
\\[1ex]
\mmodu(E) & \coloneqq \lim_{r_2 \to +\infty} \mmodu(E;R,r_2) ,
\qquad
\mmax(E) \coloneqq \sup_{r_1\geq R} \Bigl| \lim_{r_2\to+\infty} \mmodu(E;r_1,r_2) \Bigr| .
\endaligned
\ee
In the \mbox{(super-)harmonic} regime, these quantities are part of the hypotheses whenever they are used; in the sub-harmonic regime the convention following~\eqref{cutoff-pstar-def} applies.  In full, the norm of interest here is thus
\bel{NbbastarE-def}
\Nbb_{\astar}^{\notreH}
\coloneqq \| E \|_{\Omega_R,\astar+4,-\expoPp-2}^{N-2,\alpha} + \| E,\vartheta_* E \|_{H^{2*}_{\astar+4,-\expoP}(\Omega_R)} + \cutoff_{\pstar} \mmax(E) + \|\vartheta u\|_{H^2_{n-2-p,-\expoP}(\Omega_R)} ,
\ee
where $\cutoff_{\pstar}$ is equal to $0$ in the sub-harmonic case $\pstar<n-2$ (namely $\astar<a_{n,p}$) and otherwise~$1$ (cf.~\eqref{cutoff-pstar-def}).  The energy term $\mmax(E)$ controls truncated energies $\mmodu(E;r_1,r_2)$ and the total energy~$\mmodu(E)$ ---which is an analogue of~\eqref{equa-mstarJstar} for the linear problem in flat space.
Their existence is predicated on the (signed) integrability of $\la E\ra r^{a_{n,p}+3}$ at infinity\footnote{Since $\expoPp+2<2\expoP$ (owing to $\expoPm>2$), the pointwise bound $|E|\lesssim r^{-\astar-4} \lambda^{-\expoPp-2}$ ensures that $\la E(r)\ra$ is well-defined for all $r\in[R,+\infty)$.}.
We find it convenient to include this truncated energy also in the super-harmonic case $\astar>a_{n,p}$, even though it is controlled by the H\"older norm in~\eqref{NbbastarE-def} in that case.
For the integral estimates in \autoref{thm:decayL2} below, only the (dual) Sobolev norms and truncated energy appear, whereas the pointwise estimates in \autoref{thm-sharp-h-localized} involve all of~$\Nbb_{\astar}^{\notreH}$.
We emphasize that \eqref{NbbastarE-def} features three radial exponents: $r^4 E$ with exponent $\astar$, the energy $\mmax(E)$ related to the harmonic exponent~$a_{n,p}$ (if $\astar\geq a_{n,p}$), and the solution $\vartheta u$ with variational exponent $a_{n,p}/2=n-2-p$.


\subsection{Variational formulation}
\label{section=5.2}

We are now interested in the boundary value problem for the localized Hamiltonian operator in suitably weighted norms associated with the function $\lambda$. Our variational formulation for the operator $\notreH^\lambda$ relies on a choice of boundary conditions of Neumann-type along a sphere of radius $R$. To the Hamiltonian operator $\notreH^\lambda$ we associate the \textbf{Hamiltonian boundary operators} 
\bel{equa-bound-ope-2-lambda}
\aligned
\Abb_{3}^\lambda[u] 
& \coloneqq (n-1) \vartheta ( \vartheta^2 + a_{n,p} \vartheta - b_{n,p} ) u
+ 
2 \, \lambda^{-2\expoP} \nablaslash \cdot \bigl(\lambda^{2\expoP} \nablaslash (\vartheta u -u) \bigr)
\\
& \quad + \Bigl( (n-2) \vartheta - (n-2)(n-2-a_{n,p}) - 1 \Bigr) \Deltaslash u, 
\\
\Abb_{2}^\lambda[u] 
& \coloneqq (n-1) \vartheta (\vartheta + n-3) u + (n-2) \Deltaslash u,
\endaligned
\ee 
which are third-order and second-order, respectively, and are defined on any sphere~$\Sphe_r$ (for $r \geq R$). These boundary operators provide us with an analogue of the Neumann boundary operator associated with the Laplace operator.  We point out that the variational decay we achieve in this preliminary stage is \textsl{much weaker} than the \mbox{(super-)harmonic} decay that we will establish later on in this section.

\begin{theorem}[Variational formulation for the localized Hamiltonian operator]
\label{thm-sharp-h-localized-vari}
Consider a conical domain $\Omega_R = K \cap {}^{\complement} \Ball_R \subset \RR^n$ together with a localization function $\lambda\colon \Lambda \to (0, \lambda_0]$ with connected support  $\Lambda \subset \Sphe^{n-1}$. Fix some arbitrary localization exponent~$\expoP \geq 2$ and consider a projection exponent\footnote{The restriction based on the lower bound $p_n^\flat$ is unnecessary here.} $p \in (0, n-2)$. Given any $E \in H^{2*}_{n+2-p,-\expoP}(\Omega_R)$, there exists a unique solution $u \in H^2_{n-2-p, -\expoP}(\Omega_R)$ characterized by a variational identity (cf.~\eqref{equa-variH} below).  In particular, it solves the following localized fourth-order Hamiltonian equation in the weak sense
\bel{equa-opera5} 
\aligned
\notreH^\lambda[u]
& = E \, 
&& \text{ in the exterior domain } \Omega_R = {}^\complement \Ball_R\cap K , 
\qquad 
\\
\Abb_{3}^\lambda [u]  
& = \Abb_{2}^\lambda[u]  = 0 
&& \text{ on the spherical shell } \Lambda_R = \Sphe_R \cap K.
\endaligned
\ee   
Moreover, one has 
\be
\|u\|_{H^2_{n-2-p,-\expoP}(\Omega_R)}
\lesssim
\|E\|_{H^{2*}_{n+2-p,-\expoP}(\Omega_R)} .
\ee 
\end{theorem}

 The first line in~\eqref{equa-opera5} holds in $\Dcal'(\Omega_R)$.  The two conditions on~$\Lambda_R$ are the natural boundary conditions encoded by the variational identity; they are asserted as classical trace identities only when $u$ has enough regularity up to~$\Lambda_R$ for the second- and third-order boundary operators to possess traces.  

\begin{proof} 
We determine the variational formulation by suitable integrations by parts.
Since $\expoP\geq2$, Proposition~\ref{prop-density-H2} provides a subspace $\DcalR(\Omega_R)$ that is dense in the space $H^2_{n-2-p,-\expoP}(\Omega_R)$ of interest.
This space consists of smooth test functions~$w$ that vanish near the lateral boundary but have arbitrary traces $w|_{\Lambda_R}$ and $\partial_rw|_{\Lambda_R}$.
Hence no lateral boundary term occurs, while the two coefficients of these inner traces are the two natural boundary operators. For a variational solution of the regularity stated in the theorem, the equalities $\Abb_2^\lambda[u]=\Abb_3^\lambda[u]=0$ on $\Lambda_R$ are understood in this natural weak sense; when $u$ is sufficiently regular up to $\Lambda_R$, they agree with the stated differential expressions.  

Consider thus a sufficiently regular solution $u$ to~\eqref{equa-opera5} and integrate by parts.
We use the short-hand $\phi_{kl} = \lambda^{2\expoP} r^{n-2p} (\del_k\del_l u-\delta_{kl}\Delta u)$ in the calculation.  For arbitrary test functions $w\in\DcalR(\Omega_R)$, and taking into account that the outer unit normal to $\Omega_R$ is $-\xh$, we find 
\[
\aligned
\int_{\Omega_R} w \, E\, r^{n-2p} \lambda^{2\expoP} dx  
& = \int_{\Omega_R} w \, (\del_k \del_l - \delta_{kl} \Delta) \phi_{kl} dx 
= \int_{\Omega_R} (\del_k  \del_l w - \delta_{kl} \Delta w) \, \phi_{kl} dx + T_B[w]
\\
T_B[w] \, & \! \coloneqq \int_{\Lambda_R} \Bigl( w \, \bigl( \xh_k \del_k \phi_{ll} - \xh_k \del_l \phi_{kl} \bigr) + \del_k w \, \bigl(\xh_l \phi_{kl} - \xh_k \phi_{ll} \bigr) \Bigr) \, R^{n-1}d\xh, 
\endaligned
\]
therefore
\bel{equa-variH} 
\aligned
Q^\notreH[w,u] 
& \coloneqq
 \int_{\Omega_R} \big( \del_k  \del_l  w \, - \delta_{kl} \Delta w \big) \, (\del_k \del_l u - \delta_{kl} \Delta u) \, \lambda^{2\expoP} r^{n-2p} dx  
\\
& = \int_{\Omega_R} w \, E\, r^{n-2p} \lambda^{2\expoP} dx, 
\endaligned
\ee
provided the boundary term~$T_B[w]$ vanishes for all test functions~$w$.
We note that $\del_k w = \frac{1}{r} \bigl( \xh_k \, \vartheta w + \nablaslash_k w \bigr)$, where $\nablaslash$ is the Levi--Civita connection of the sphere metric~$\gslash$, and we integrate by parts the $\nablaslash w$ terms along the angular shell $\Lambda_R\subset\del\Omega_R$ to get
\[
\aligned
T_B[w] & = \frac{R^{n-1}}{R} \int_{\Lambda_R} \vartheta w \bigl(\xh_k \xh_l \phi_{kl} - \phi_{ll} \bigr) d\xh \\
& \quad - \frac{R^{n-1}}{R} \int_{\Lambda_R} w \Bigl( \nablaslash_k \bigl(\xh_l \phi_{kl} - \xh_k \phi_{ll} \bigr) + R(\xh_k \del_l \phi_{kl} - \xh_k \del_k \phi_{ll}) \Bigr) d\xh
\endaligned
\]
and the required boundary conditions are that the coefficients of $\vartheta w$ and of~$w$ vanish on the shell~$\Lambda_R$.
For the first condition we write (using $\xh_k\del_k=r^{-1}\vartheta$ and $\vartheta\xh_l=0$)
\[
\aligned
\lambda^{-2\expoP} r^{2p-n} \bigl( \xh_k\xh_l \phi_{kl} - \phi_{ll} \bigr)
& = (n-2) \Delta u + \xh_l \xh_k \del_k \del_l u
= (n-2) \Delta u + r^{-1}\vartheta (r^{-1}\vartheta u)
\\
& = r^{-2} \Bigl( (n-1)\vartheta(\vartheta+n-3) u + (n-2)\Deltaslash u \Bigr) , 
\endaligned
\]
which leads us to the boundary operator $\Abb_{2}^\lambda[u]$ defined in~\eqref{equa-bound-ope-2-lambda}.
For the second condition we first evaluate
\[
\aligned
\xh_l \phi_{kl} & = \lambda^{2\expoP} r^{n-2p-2} \bigl( - (n-1) \xh_k \vartheta u + (\vartheta-1)\nablaslash_k u - \xh_k\Deltaslash u \bigr) ,
\\
r \xh_k \del_l \phi_{kl} - r \xh_k \del_k \phi_{ll} & = \del_l(r \xh_k \phi_{kl}) - \phi_{ll} - \vartheta \phi_{ll} ,
\endaligned
\]
and use them, together with $\nablaslash_k(\xh_k f)=(n-1)f$, to obtain
\[
\aligned
& \lambda^{-2\expoP} r^{2p-n+2} \Bigl( \nablaslash_k \bigl(\xh_l \phi_{kl} - \xh_k \phi_{ll} \bigr) + r \xh_k \del_l \phi_{kl} - r \xh_k \del_k \phi_{ll} \Bigr)
\\
& = \lambda^{-2\expoP} \nablaslash_k \Bigl( \lambda^{2\expoP} \Bigl(- (n-1) \xh_k \vartheta u + (\vartheta-1)\nablaslash_k u - \xh_k\Deltaslash u  \\[-1ex]
& \qquad\qquad\qquad\quad\ \
+ (n-1) \xh_k \vartheta(\vartheta+n-2) u + (n-1) \xh_k \Deltaslash u \Bigr) \Bigr) \\
& \quad
 + \lambda^{-2\expoP} r^{2p-n+2} \Bigl(\del_l\Bigl( \lambda^{2\expoP} r^{n-2p-1} \bigl( - (n-1) \xh_l \vartheta u + (\vartheta-1)\nablaslash_l u - \xh_l\Deltaslash u \bigr) \Bigr)  \\[-1ex]
& \qquad\qquad\qquad\qquad\
+ (n-1) (\vartheta+1)(\lambda^{2\expoP} r^{n-2p} \Delta u) \Bigr)
\\
& = (n-1)^2 \vartheta(\vartheta+n-3) u
+ 2 \lambda^{-2\expoP} \nablaslash\cdot \bigl( \lambda^{2\expoP} \nablaslash(\vartheta u - u) \bigr)
+ (n-2) (n-1) \Deltaslash u \\
& \quad - (\vartheta+2n-2p-2) \bigl( (n-1) \vartheta u + \Deltaslash u \bigr)
+ (n-1) (\vartheta+n-2p-1)\bigl(\vartheta(\vartheta+n-2)u+\Deltaslash u\bigr)
\\
& = (n-1) \vartheta\bigl( \vartheta^2 + (a_{n,p}+n-1)\vartheta + (n-3)a_{n,p} + n - 5\bigr) u
\\
& \quad + 2 \lambda^{-2\expoP} \nablaslash\cdot \bigl( \lambda^{2\expoP} \nablaslash(\vartheta u - u) \bigr)
+ \bigl((n-2)\vartheta+(n-2)a_{n,p}+n-3\bigr) \Deltaslash u .
\endaligned
\]
This leads us to a linear combination $\Abb_{3}^\lambda[u] + (n-1) \Abb_{2}^\lambda[u]$ of the boundary operators defined in~\eqref{equa-bound-ope-2-lambda}.
This completes the derivation of the variational formulation with boundary.
In view of the decomposition in \autoref{lem:sph-Ham}, below, the variational formulation is also equivalent to saying 
\bel{eq-521}
\aligned 
& \int_R^{+ \infty} \int_{\Lambda_r} \Bigl(
(n-1) \vartheta^2 w \vartheta^2 u
+ (n-1) b_{n,p} \vartheta w \vartheta u
+ 2 \nablaslash\vartheta w \cdot \nablaslash\vartheta u
+ (n-2) (\vartheta^2+a_{n,p}\vartheta) w \, \Deltaslash u
\\
& \qquad
+ (n-2) \Deltaslash w \, (\vartheta^2+a_{n,p}\vartheta) u
+ \bigl( (n-2)(n-2-a_{n,p}) + 1\bigr) \bigl( \Deltaslash w \, \vartheta u + \vartheta w \Deltaslash u
\bigr)
\\
& \qquad + (n-2) \Deltaslash w \Deltaslash u + \nablaslash^a\nablaslash^b w \nablaslash_a\nablaslash_b u 
+ 2 (a_{n,p}+1) \nablaslash w \cdot \nablaslash u
\Bigr) \, r^{a_{n,p}} \, d\chi \frac{dr}{r}
\\
& = \int_{\Omega_R} w \, E\, r^{n-2p} \lambda^{2\expoP} dx. 
\endaligned
\ee
Consequently, for $E\in H^{2*}_{n+2-p,-\expoP}(\Omega_R)$, the variational solution $u \in H^2_{n-2-p, -\expoP}(\Omega_R)$ is defined by requiring that~\eqref{equa-variH} (equivalently~\eqref{eq-521}) holds for all $w \in H^2_{n-2-p, -\expoP}(\Omega_R)$, the right-hand side of~\eqref{equa-variH} being understood as the duality bracket between $E\in H^{2*}_{n+2-p,-\expoP}(\Omega_R)$ and $r^{2n-2p}w\in H^2_{-n-2+p,-\expoP}(\Omega_R)$. Continuity and coercivity properties for the linearized Hamiltonian follow from the weighted Poincaré inequality in~$\Omega_R$ (cf.~\autoref{lem:PoincareKornHardyD}).  The Lax--Milgram theorem therefore gives existence and uniqueness of the variational solution to~\eqref{equa-opera5}, and the estimate stated above.
Testing against compactly supported functions proves $\notreH^\lambda[u]=E$ in $\Dcal'(\Omega_R)$; the preceding Green formula identifies the natural boundary operators whenever their traces exist.
\end{proof}


\subsection{Statement of the localized integral estimates}
\label{section=5.3}

We now consider the solutions to the localized Hamiltonian and, by integrating differential equations \eqref{equa-defb-b-b}~for the average $\la u\ra$ and \eqref{main-func-identity}~for the shell functional~$\Phi^\notreH[u]$, and applying our Hamiltonian stability conditions, we establish sharp decay for Sobolev norms on spherical caps $\Lambda_r=\{x\in\Omega_R,\,|x|=r\}$.
The proof spans \autoref{section=7.1}, which defines in particular the parameter $\deltaH>0$ appearing in the theorem statement, \autoref{section=7.2} analysing the shell identity, and \autoref{section=7.3} proving a technical bound based on the equation for~$\la u\ra$, and relies on some explicit expressions (especially~\autoref{prop-66}) derived in \autoref{section=6} based on a harmonic-spherical decomposition of~$\notreH^\lambda$.

\begin{theorem}[Integral estimates for the localized Hamiltonian operator]
\label{thm:decayL2}
Consider a conical domain $\Omega_R = K \cap {}^{\complement} \Ball_R \subset \RR^n$ together with a localization function $\lambda\colon \Lambda \to (0, \lambda_0]$ with connected support $\Lambda \subset \Sphe^{n-1}$ and some ${\expoP \geq 2}$. Fix  a projection exponent $p\in (p_n^\flat,n-2)$ and assume that the localization function is Hamiltonian-stable (cf.~\autoref{def-shell-Hstab}).  
There exists $\deltaH>0$, depending only on the exponents, the Hamiltonian-stability constants, and the fixed localization geometry, such that the following holds.
For every sharp decay exponent $\astar \in [a_{n,p}/2, a_{n,p}+\deltaH)$, consider the variational solution $u \in H^2_{n-2-p, -\expoP}(\Omega_R)$ of the localized Hamiltonian equation~\eqref{equa-opera5} associated with a source term obeying the radial decay (cf.~\autoref{rem:decay-reg})
\bse
\bel{equa-hyo-EEE}
T_E \coloneqq \|E,\vartheta_* E\|_{H^{2*}_{\astar+4,-\expoP}(\Omega_R)} + \cutoff_{\pstar} \mmax(E) < +\infty ,
\ee
where $\cutoff_{\pstar}=\Oneone_{\astar\geq a_{n,p}}$ and $\mmax(E)$ is defined in~\eqref{mmodu-def} if $\astar\geq a_{n,p}$.
In that regime, set $\umodu \coloneqq \mmodu(E) \nu^\normal r^{-a_{n,p}}$, with the energy $\mmodu(E)$ given in~\eqref{mmodu-def}.
Assume the following additional radial regularity on the solution \textup{(}with radial exponent $a_{n,p}/2$\textup{)}
\bel{equa-hyo-EEu}
T_u \coloneqq \|\vartheta u\|_{H^2_{n-2-p,-\expoP}(\Omega_R)} < +\infty .
\ee
\ese
Fix an exponent $\bstar\in(\astar,a_{n,p}+\cutoff_{\pstar}\deltaH)$ in the same regime \textup{(}sub-harmonic or not\textup{)} as~$\astar$, and let $o_E(1)$ denote a bounded function of~$r$ that tends to zero as $r\to+\infty$; its rate may depend on the source~$E$. 
Then in terms of the norm 
\be
(\norm{\,{\cdot}\,}^\notreH)^2 = \| \vartheta^2\,{\cdot}\, \|^2_{\unL^2_{-\expoP}(\Lambda)} + \| \vartheta\,{\cdot}\, \|^2_{\unH^1_{-\expoP}(\Lambda)} + \|\,{\cdot}\,\|^2_{\unH^2_{-\expoP}(\Lambda)}
\ee 
on each spherical shell $\Lambda_r$ identified with $\Lambda$ by radial projection, the solution decays radially in a pointwise and integral sense
\bel{equa-bound-int}
\aligned
& \bigl( \norm{u-\cutoff_{\pstar}\umodu}^\notreH(r) \bigr)^2 + \sum_{k=0,1} \int_r^{2r} \bigl(\norm{\vartheta^k (u - \cutoff_{\pstar} \umodu)}^\notreH(s)\bigr)^2 \frac{ds}{s}
\\
& \lesssim (T_E)^2 r^{-2\astar} o_E(1) + (T_u)^2 r^{-2\bstar} , \quad r \geq R. 
\endaligned
\ee
The implicit constant depends on the geometry and exponents, including~$\bstar$, and on the constants implicit in Hamiltonian stability conditions.
\end{theorem} 

\begin{remark}\label{rem:decay-reg}
1. The condition $E,\vartheta_*E \in H^{2*}_{\astar+4,-\expoP}(\Omega_R)$ in~\eqref{equa-hyo-EEE} implies that the function $r\mapsto\la E(r)\ra$ belongs to $H^{2*}_{\astar+4}([R,+\infty))$.  If $\astar>a_{n,p}$, it can be integrated against the test function $r^{a_{n,p}+4}\in H^2_{-\astar-4}([R,+\infty))$ to obtain the energy $\mmodu(E)$ defined in~\eqref{mmodu-def}.  Integrating on a finite interval $[R,r]$ defines a distribution $\mmodu(E;R,{\,\cdot\,})$ in $H^{1*}_\gamma([R,+\infty))$ for any $\gamma<\astar-a_{n,p}$, and we assume here that this distribution is bounded and has a limit.  Boundedness is automatic when assuming H\"older regularity later on.

2. The regularity of $\vartheta u$ is not implied by that of $\vartheta_*E$, because Neumann boundary conditions are not preserved by~$\vartheta$.  In \autoref{section=10} we apply these results to $E$ and~$u$ vanishing near $r=R$, in which case $T_U = \|\vartheta u\|_{H^2_{n-2-p,-\expoP}(\Omega_R)}\lesssim\|\vartheta_*E\|_{H^{2*}_{n+2-p,-\expoP}(\Omega_R)} \leq T_E$.
\end{remark}

 
\subsection{Localized pointwise estimates}
\label{section=5.4}

We now reach our main conclusion for the localized Hamiltonian operator, namely sharp decay properties in pointwise norms when $r \to + \infty$. Observe that the solutions may be unbounded near the part of the boundary of $\Omega_R$ where $\lambda$ vanishes. Moreover, since interior ellipticity is used it is natural also that our estimates below be stated slightly away from the boundary $\Lambda_R=\Sphe_R\cap K$, say in $\Omega_{R'}$ with $R' > R$. For convenience in applications, we also state Sobolev bounds on this reduced domain that are slightly weaker than \autoref{thm:decayL2}.
We also point out that the connectedness of $\Lambda$ is necessary for the kernel at infinity to have dimension one.  We recall that the weighted localized H\"older norms were introduced in \autoref{section=2}. To describe the regularity, we fix an integer $N\geq 3$ and a H\"older exponent $\alpha \in (0,1)$, together with a localization exponent~$\expoP \geq 2$. (Later on, this exponent will be taken to be sufficiently large to deal with nonlinear equations.)
We recall that to a decay exponent $\pstar\geq p$, one associates
$\astar = \pstar + n - 2 - 2p \geq a_{n,p} / 2$,
and $\cutoff_{\pstar}=0$ in the sub-harmonic regime $\astar<a_{n,p}$ (namely, $p<n-2$) and $\cutoff_{\pstar}=1$ otherwise.

\begin{theorem}[Pointwise estimates for the localized Hamiltonian operator]
\label{thm-sharp-h-localized}
In the setup of \autoref{thm:decayL2}, with in particular a Hamiltonian-stable localization function $\lambda$ with connected support~$\Lambda$,
and a sharp decay exponent $\astar\in[a_{n,p}/2,a_{n,p}+\deltaH)$
consider the variational solution $u \in H^2_{n-2-p,-\expoP}(\Omega_R)$ to~\eqref{equa-opera5} and assume additional regularity for~$\vartheta u$ and radial decay of the source~$E$ in dual Sobolev norm and H\"older norm such that the norm~\eqref{NbbastarE-def} is finite, namely
\bel{equa-hyo-EEE-Hold}
\Nbb_{\astar}^{\notreH}
\coloneqq \| E \|_{\Omega_R,\astar+4,-\expoPp-2}^{N-2,\alpha} + \| E,\vartheta_*E \|_{H^{2*}_{\astar+4,-\expoP}(\Omega_R)} + \cutoff_{\pstar} \mmax(E) + \|\vartheta u\|_{H^2_{n-2-p,-\expoP}(\Omega_R)}
< +\infty ,
\ee
for some exponent $\expoPp \geq \expoP + (n+3)/2$.
In the \mbox{(super-)harmonic} case $\astar\geq a_{n,p}$, recall the energy term $\mmax(E)$ defined in~\eqref{mmodu-def}, and the energy~$\mmodu(E)$ and modulator $\umodu=\mmodu(E)r^{-a_{n,p}}\nu^\normal$ defined\/\footnote{The H\"older bounds ensure existence of $\mmax(E)$ and $\mmodu(E)$ for $\astar>a_{n,p}$ but in the harmonic case their existence is a signed integrability condition on~$E$.  One has $|\mmodu(E)|\leq\mmax(E)$.} in~\eqref{mmodu-def}.
Fix some radius $R'> R$.  Then the following estimate holds,\footnote{We emphasize that it is necessary to take the H\"older norms over $\Omega_R$ or $\Omega_{R'}$, as specified.  Integral norms are controlled over all of~$\Omega_R$ by \autoref{thm:decayL2} but using the same domain allows for a convenient statement about the $R'\to+\infty$ limit.}
\bel{u-HolderSobolev}
\bigl\| u - \cutoff_{\pstar} \umodu \bigr\|^{N+2,\alpha}_{\Omega_{R'}, \astar, -\expoPp+2}
+ \sup_{\beta\in[a_{n,p}/2,\astar)}\biggl( (\astar-\beta)^{1/2} \sum_{k=0,1} \bigl\| \vartheta^k(u - \cutoff_{\pstar} \umodu) \bigr\|_{H^2_{\beta,-\expoP}(\Omega_{R'})} \biggr)
\lesssim \Nbb_{\astar}^{\notreH} .
\ee
When $\astar=a_{n,p}/2$, the interval in the supremum is empty and that term is, by convention, omitted.
Furthermore, the left-hand side, with the $C^{N+2,\alpha}$ H\"older norm replaced by a $C^3$ norm, tends to zero in the limit $R'\to+\infty$.
\end{theorem}  

\begin{remark}
\label{rem:EE}
To be precise, stating that the source~$E$ is in a dual Sobolev space and a H\"older space means that the distribution $E\in H^{2*}_{\astar+4,-\expoP}(\Omega_R)$, restricted to test functions $w\in H^2_{-\astar-4,-Q}$ for some (unimportant) $0<Q<\expoPm-2$, is given by integration against a H\"older function~$E_{\textnormal{reg}}$, namely $\la E,w\ra = \int_{\Omega_R} E_{\textnormal{reg}} w \lambda^{2\expoP} r^{-n}dx$.
\end{remark}

\begin{proof}
\noindent{\it 1. Weighted Sobolev decay.} Thanks to the shell stability condition and in view of \autoref{thm:decayL2}, we already have the desired sharp decay properties~\eqref{equa-bound-int} on spherical shells~$\Lambda_r$ and on their unions over intervals~$[r,2r]$.
Summing the latter control over dyadic intervals covering $[R',+\infty)$ with a weight $r^{2\beta}$, $\beta<\astar$, yields a Sobolev bound on $\vartheta^k(u-\cutoff_{\pstar}\umodu)$, $k=0,1$, with exponent $\beta<\astar$.
The upper bound involves a sum $\sum_{m\geq 0}2^{-2m\gamma}$ with $\gamma=\astar-\beta$, which behaves as $1/\gamma$ in the limit $\beta\to\astar$.
The factor of $(\astar-\beta)^{1/2}$ ensures that the Sobolev bound in~\eqref{u-HolderSobolev} holds uniformly in that limit.
As $R'\to+\infty$, fewer and fewer dyadic intervals contribute, and the Sobolev norm decays to zero.


To proceed towards the H\"older part of that bound, we begin with the control~\eqref{equa-bound-int} in $\unH^2_{-\expoP}(\Lambda_r)$, within which we only keep the $\unL^2$~part (and use $\bstar>\astar$). 
\be
\gathered
\bigl\| (u- \cutoff_{\pstar} \umodu) (r) \bigr\|_{\unL^2_{-\expoP}(\Lambda)}
\lesssim \Mbb_0 r^{-\astar} o(1) ,
\\
\Mbb_0 \coloneqq \| E,\vartheta_*E \|_{H^{2*}_{\astar+4,-\expoP}(\Omega_R)} + \cutoff_{\pstar} \mmax(E) + \|\vartheta u\|_{H^2_{n-2-p,-\expoP}(\Omega_R)} ,
\endgathered
\ee
where $o(1)$ is bounded by~$1$ and decays to zero as $r\to+\infty$.
We measure this decay by the monotonically decreasing
\bel{eq522aa}
\Mbb(r) \coloneqq \sup_{s\in[r,+\infty)} \Bigl( s^{\astar} \bigl\| (u- \cutoff_{\pstar} \umodu) (s) \bigr\|_{\unL^2_{-\expoP}(\Lambda)} \Bigr) \lesssim \Mbb_0 ,
\qquad
\lim_{r\to+\infty} \Mbb(r) = 0 .
\ee

\medskip


\bse 
\noindent{\it 2. Interior estimates.}
For a given $R'>R$,
we set
\be
\cellip(R') = \min\biggl( 1 , \ \frac{1}{\lambda_0} , \ \inf_{x\in\Omega_{R'}} \frac{\dbf(x,{}^\complement\Omega_R)}{2\lambda(x)r(x)} \biggr) > 0 .
\ee
Denote $d(x) \coloneqq c_1 \lambda(x) r(x)$ for some $c_1\in(0,\cellip(R'))$.
By construction, $c_1<1$, one has $d(x)<r(x)$ everywhere on~$\Omega_{R'}$, and for all $x\in\Omega_{R'}$ the ball $\Ball_{2d(x)}(x)$ with radius $2d(x)$ is included in $\Omega_R$.
In particular, one has $\lambda\simeq\lambda(x)$, $r\simeq r(x)$, and $d\simeq d(x)$ on the ball $\Ball_{d(x)}(x)$.
Note that $\cellip(R')\lesssim R'-R$ in the (uninteresting) limit $R'\to R$, whereas $\cellip(R')$ tends to a positive limit as $R'\to+\infty$.

From the above weighted Sobolev bounds, we are now in a position to establish the desired weighted H\"older estimates, as follows, by applying the interior elliptic regularity enjoyed by the Hamiltonian operator. In any compact subset of the open set $\Omega_R$ and, in particular, after excluding the small domain $R < r <R'$, standard interior elliptic estimates in H\"older norm are available.\footnote{As per \autoref{rem:EE}, the source term is actually a pair of a dual Sobolev distribution and a H\"older function.  These coincide away from the boundary, so that the elliptic problem of interest is not ambiguous.} Indeed, we rewrite the equation as a singular perturbation of a bi-Laplacian problem and, from the expression of the operator $\notreH^\lambda[u]$ in~\eqref{equa-2933} and by observing that (for instance) 
$
|\del_i(\log\lambda^\expoP)| = \expoP \, |\del_i(\log\lambda)| \lesssim \expoP \, ( \lambda(x) \, r(x))^{-1}, 
$  
we can write 
\be
\aligned
& (n-1) \Delta^2 (u- \cutoff_{\pstar} \umodu) + \sum_{2 \leq |\beta| \leq 3}a_{\beta}^{(1)} \del^\beta (u- \cutoff_{\pstar} \umodu)
 = E, 
\qquad 
 |a_\beta^{(1)} | \lesssim (\lambda \, r)^{|\beta|-4}.  
\endaligned
\ee
This equation fulfills the ellipticity conditions in Douglis--Nirenberg~\cite{DouglisNirenberg} and, therefore, at any $x \in \Omega_{R'}$ we have an interior elliptic estimate in the ball $\Ball_\rho(x)$ for a radius $\rho$ smaller than the distance from $x$ to the boundary $\del \Omega_R$. We find (using specifically~\cite[Theorem 6.2.6]{Morrey-1966}, see \autoref{rem:DN} below for further comments)
\bel{equa-421-H}
\aligned
& \sum_{i=0}^{N+2} d(x)^i  \, \bigl| \del^i  (u- \cutoff_{\pstar} \umodu) (x) \bigr| 
+ d(x)^{N+2+ \alpha}\big[ \del^{N+2}  (u- \cutoff_{\pstar} \umodu)  \big]_{\alpha, \Ball_{d(x)/2}(x)} 
\\
& \lesssim {1 \over d(x)^n} \|  u- \cutoff_{\pstar} \umodu \|_{L^1(\Ball_{d(x)/2}(x))}
\\
& \quad +
 \sum_{i= 0}^{N-2} d(x)^{4+i} \sup_{\Ball_{3d(x)/4}(x)} |\del^i E| + d(x)^{N+2+ \alpha} \big[ \del^{N-2} E \big]_{\alpha, \Ball_{3d(x)/4}(x)}, 
\endaligned
\ee
in which the implied constants are {\it independent} of $x$. The terms involving $E$ are bounded thanks to our assumptions on the source. On the other hand, for the first term in the right-hand side of~\eqref{equa-421-H} we write 
\bel{equa-L1versusL2} 
\aligned
d(x)^{-n} \|  u- \cutoff_{\pstar} \umodu \|_{L^1(\Ball_{d(x)/2}(x))} 
\lesssim d(x)^{-n/2} \, \|  u- \cutoff_{\pstar} \umodu \|_{L^2(\Ball_{d(x)/2}(x))},
\endaligned
\ee 
which requires a control of the $L^2$ norm, derived next. 
\ese


\medskip

\bse
\noindent{\it 3. Pointwise decay estimates.} There remains to use our estimate on $\fint_{\Lambda_r} (u- \cutoff_{\pstar} \umodu)^2  \, d\chi$ by $r^{-2\astar}$ in~\eqref{eq522aa} in order to control the right-hand side of~\eqref{equa-L1versusL2}. For an arbitrary point $x$ and with the notation  $d(x) = c_1 \lambda(x) r(x)$ (introduced earlier such that $\Ball_{d(x)}(x)\subset\Omega_R$), we 
set $R_*\coloneqq\max\{R,R'/2\}$ and find
\bel{eq443}
\aligned
& d(x)^{-n} \, \| u- \cutoff_{\pstar} \umodu \|_{L^2(\Ball_{d(x)/2}(x))}^2
\\
& = d(x)^{-n} \int_{r(x) - d(x)/2}^{r(x)+d(x)/2} \int_{\Sphe_r \cap \Ball_{d(x)/2}(x)} \! \bigl((u- \cutoff_{\pstar} \umodu)(r, \xh)\bigr)^2 \, d\xh \, r^{n-1} dr
\\
& \lesssim 
{d(x)^{-n} \over \min_{\Ball_{d(x)/2}(x)} \lambda^{2\expoP}}
\int_{r(x) - d(x)/2}^{r(x)+d(x)/2} \Bigg(\int_{\Sphe_r \cap \Ball_{d(x)/2}(x)} \bigl((u- \cutoff_{\pstar} \umodu)(r, \xh)\bigr)^2 \, \lambda^{2\expoP} d\xh  \Bigg) \, r^{n-1} dr
\\
& \lesssim 
{d(x)^{-n} \over \min_{\Ball_{d(x)/2}(x)} \lambda^{2\expoP}} 
\Mbb(R_*)^2 \int_{r(x) - d(x)/2}^{r(x)+d(x)/2} r^{-2\astar} \, r^{n-1} dr
\endaligned
\ee
where we used $r(x)-d(x)/2>r(x)/2\geq R'/2$ and, when $R'/2<R$, the fact that every point of the ball still belongs to~$\Omega_R$.
Given that $\lambda\simeq\lambda(x)$ on the ball, the minimum can be replaced by~$\lambda(x)^{2\expoP}$.
The integrand is bounded up to a constant by $r(x)^{n-1-2\astar}$ hence
\bel{equa---444} 
\aligned
d(x)^{-n} \, \| u- \cutoff_{\pstar} \umodu \|_{L^2(\Ball_{d(x)/2}(x))}^2
& \lesssim {d(x)^{1-n} r(x)^{n-1-2\astar} \over \lambda^{2\expoP}} \Mbb(R_*)^2
\\
& = \lambda^{-2\expoP+ 1-n} \, r(x)^{-2\astar} c_1^{-n+1} \Mbb(R_*)^2.
\endaligned
\ee 
Consequently, in combination with~\eqref{equa-421-H}--\eqref{equa-L1versusL2}, we see that $u- \cutoff_{\pstar} \umodu$ is controlled in the relevant H\"older norm with angular weight $\lambda^{-\expoPp+2}$, namely 
\bel{equa-421-more}
\aligned
& \sum_{i=0}^{N+2}   \lambda^{\expoPp-2} \, d(x)^i  \, |\del^i (u- \cutoff_{\pstar} \umodu)(x)| 
+ \lambda^{\expoPp-2} \,  d(x)^{N +2 + \alpha}\big[ \del^{N+2} (u- \cutoff_{\pstar} \umodu) \big]_{\alpha, \Ball_{d(x)/2}(x)} 
\\
& \lesssim
\lambda(x)^{\expoPp-2 - \expoP - (n-1)/2} \, r(x)^{-\astar} \, c_1^{-(n-1)/2} \Mbb(R_*)
\\
& \quad + \sum_{i=0}^{N-2} \Bigl( \lambda(x)^{\expoPp-2} \, d(x)^{4+i} \sup_{\Ball_{3d(x)/4}(x)} |\del^i E| \Bigr)
+ \lambda(x)^{\expoPp-2} d(x)^{N+2+\alpha} \big[ \del^{N-2} E \big]_{\alpha, \Ball_{3d(x)/4}(x)}
\\
& \lesssim r(x)^{-\astar} \biggl(
c_1^{-(n-1)/2} \Mbb(R_*) \lambda(x)^{\expoPp-2 - \expoP - (n-1)/2} + c_1^4 \|E\|_{\Omega_R,\astar+4,-\expoPp-2}^{N-2,\alpha} \biggr)
\endaligned
\ee
where in the last step we bounded $\lambda(x)^{\expoPp-2}d(x)^{4+i}\lesssim c_1^4r(x)^{-\astar} r^{\astar+4+i}\lambda^{\expoPp+2+i}$ by using that $\lambda\simeq\lambda(x)$ and $r\simeq r(x)$ on the relevant ball.
This leads us to choose an exponent $\expoPp \geq \expoP + (n+3)/2$ in our statement, so that $\lambda^{\expoPp - 2 - \expoP - (n-1)/2}$ is bounded.
For any given value of $c_1\in(0,\cellip(R'))$ (which can be taken to be uniform as $R'\to+\infty$) one gets the H\"older bound~\eqref{u-HolderSobolev} stated in \autoref{thm-sharp-h-localized},
\be
c_1^{N+2+\alpha} \bigl\| u - \cutoff_{\pstar} \umodu \bigr\|^{N+2,\alpha}_{\Omega_{R'}, \astar, -\expoPp+2}
\lesssim c_1^{-(n-1)/2} \Mbb(R_*) + c_1^4 \|E\|_{\Omega_R,\astar+4,-\expoPp-2}^{N-2,\alpha} .
\ee
The factor $c_1^{N+2+\alpha}$ on the left-hand side arises from $d(x)^{N+2+\alpha}$ in~\eqref{equa-421-more}.  If we focus on bounding only up to three derivatives\footnote{We expect that the argument can be refined to control $u-\cutoff_{\pstar}\umodu$ with $C^{N+2,\beta}$ regularity for any $\beta<\alpha$.} of~$u$, this factor changes to~$c_1^3$, and 
and provided $\Mbb(R_*)>0$, we instead choose
\be
c_1=\frac{1}{2}\cellip(R')\min\Bigl(1,\bigl(\Mbb(R_*) /\Nbb_{\astar}^{\notreH}\bigr)^{\frac{2}{n+7}}\Bigr). 
\ee
This gives
\bel{u-cutoff-astar-prf}
\bigl\| u - \cutoff_{\pstar} \umodu \bigr\|_{C^3_{\astar, -\expoPp+2}(\Omega_{R'})}
\lesssim \Mbb({R_*})^{\frac{2}{n+7}} (\Nbb_{\astar}^{\notreH})^{\frac{n+5}{n+7}} .
\ee
If $\Mbb(R_*)=0$, then $u-\cutoff_{\pstar}\umodu=0$ almost everywhere in $\Omega_{R_*}$ and the same conclusion is immediate.
The implicit constant remains bounded in the limit $R'\to+\infty$, since $1/\cellip(R')$ remains bounded.
Finally, $\Mbb(R_*)\to 0$ as $R'\to+\infty$, thus ending the proof.
\ese
\end{proof}

\begin{remark}\label{rem:DN}
  Following Carlotto--Schoen~\cite{CarlottoSchoen} and Chrusciel--Delay~\cite{ChruscielDelay-memoir}, we use interior elliptic estimates to control H\"older norms by the weighted $L^2$~norms that are controlled at that stage.  It should be noted that these references only quote the results of Douglis--Nirenberg~\cite{DouglisNirenberg}, whose interior elliptic estimate only controls H\"older norms by the \emph{sup norm}, which is a weaker statement.
  We are grateful to Romain Gicquaud for pointing us to Morrey's textbook~\cite{Morrey-1966}, which has a detailed treatment of Douglis--Nirenberg elliptic systems and their elliptic estimates.
\end{remark}


\section{Consequences of Hamiltonian harmonic and radial stability}
\label{section=6}

\subsection{The harmonic-spherical decomposition}
\label{section=6.1}
 
Our next task is to investigate the kernel of the harmonic operator (\autoref{section=6.2}) and then 
the evolution of the radial averages  (\autoref{section=6.3}). Let us first present the relevant decomposition of the Hamiltonian operator.  In view of~\eqref{equa:acalew0}, in the Euclidean space $\RR^n$ the operator of interest reads 
\bel{equa-2933} 
\aligned
\notreH^\lambda[u]  
& = \omega_p^{-2} \, d\Hcal\bigl( \omega_p^2 d\Hcal^{*\flat\flat}[u] \bigr)
  = \lambda^{-2\expoP} r^{-n+2p} \,  (\del_i \del_j - \delta_{ij} \Delta) \Bigl( \lambda^{2\expoP} r^{n-2p} (\del_i \del_j u - \delta_{ij} \Delta u) \Bigr)
\\
& = \lambda^{-2\expoP} r^{-n+2p} \, \del_i \del_j \bigl( \lambda^{2\expoP} r^{n-2p} \del_i \del_j u \bigr)
+ (n-2) \lambda^{-2\expoP} r^{-n+2p} \, \Delta \bigl( \lambda^{2\expoP} r^{n-2p} \Delta u \bigr), 
\endaligned
\ee
in which $\omega_p = r^{n/2-p} \, \lambda^\expoP$. In order to analyze the harmonic decay, we rely on the decomposition stated now, which is checked by a routine calculation (details are given in \cite[Appendix B]{LL-PoincareKornHardy}).
We use here the notation $a_{n,p}$, $b_{n,p}$ and $c_{n,p}$ given in~\eqref{equa-our-parame-00}, with $c_{n,p}/a_{n,p}=(n-2)(n-2-a_{n,p}) + 1$.

\begin{lemma}[Harmonic-spherical decomposition of the localized Hamiltonian operator]
\label{lem:sph-Ham}
For every $u\in C^4(\Omega_R)$, the fourth-order squared localized Hamiltonian operator around Euclidean data (using $\vartheta\lambda=0$) enjoys the decomposition\footnote{The operator $\Arr$ does not depend upon $\lambda$.} 
\bse
\label{equa--488-2}
\be
\aligned
r^4 \notreH^\lambda[u] 
& = \Arr[u] + \Ars^\lambda[u] + \ssA^\lambda[u],
\endaligned
\ee
in which  
\bel{equa-62b} 
\aligned
\Arr[u] & \coloneqq (n-1) \vartheta (\vartheta+a_{n,p}) \bigl( \vartheta^2 + a_{n,p} \vartheta - b_{n,p} \bigr) u,
\\ 
\Ars^\lambda[u]
& \coloneqq
\lambda^{-2\expoP} (\vartheta+a_{n,p}) \Bigl(
2 \vartheta \nablaslash \cdot ( \lambda^{2\expoP} \nablaslash u)
+ (n-2) \vartheta\bigl( \lambda^{2\expoP}\Deltaslash u+ \Deltaslash( \lambda^{2\expoP} u) \bigr) \\
& \qquad\qquad\qquad\qquad\
+ \frac{c_{n,p}}{a_{n,p}} \bigl( \Deltaslash( \lambda^{2\expoP} u) -  \lambda^{2\expoP} \Deltaslash u\bigr)
\Bigr) ,
\\
\ssA^\lambda[u] & \coloneqq \lambda^{-2\expoP} \Bigl( (n-2) \, \Deltaslash ( \lambda^{2\expoP} \Deltaslash u) + \nablaslash^a\nablaslash^b\bigl(  \lambda^{2\expoP} \nablaslash_a\nablaslash_b u\bigr) \\
& \qquad\qquad
- 2 (a_{n,p}+1) \nablaslash \cdot (  \lambda^{2\expoP} \nablaslash u) - c_{n,p} \Deltaslash( \lambda^{2\expoP} u) \Bigr).
\endaligned
\ee
\ese
Here, $a,b$ are abstract Penrose indices on the unit sphere, and $\nablaslash$ denotes the Levi--Civita connection of the induced metric~$\gslash$ on the sphere. In particular, $\ssA^\lambda[\nu] = r^{4+a_{n,p}} \notreH^\lambda[\nu \, r^{-a_{n,p}}]$ for every angular function $\nu=\nu(\xh)$ or, equivalently, for any function satisfying $\vartheta\nu=0$.
\end{lemma}  


\subsection{Construction of the silhouette function} 
\label{section=6.2}

\paragraph{Dimension of the kernel and cokernel.}

Our first task is to investigate the implications of the harmonic stability condition~\eqref{equa-stable-H-414} and, especially, explore whether the equation $\notreH^\lambda[u] = 0$ admits non-trivial solutions of the form $u=r^{-a_{n,p}}\nu$ in which $\nu=\nu(\xh)$ is~a function on the $(n-1)$-sphere.  In other words, we consider the operator $\ssA^\lambda[\nu] = r^{4+a_{n,p}} \notreH^\lambda[r^{-a_{n,p}}\nu]$ and establish the following result. 

\begin{proposition}[Kernel properties of the harmonic Hamiltonian operator]
\label{prop:Ham-kernel}
Consider a conical domain $\Omega_R = K \cap {}^{\complement} \Ball_R \subset \RR^n$ together with a localization function $\lambda\colon \Lambda \to (0, \lambda_0]$ with connected support $\Lambda \subset \Sphe^{n-1}$ and some ${\expoP \geq 2}$, such that the harmonic stability condition~\eqref{equa-stable-H-414} holds (cf.~\autoref{def-harmonic-Hstab}). Then the operator $\ssA^\lambda$ and its adjoint  
have one-dimensional kernels
\be
\ker(\ssA^{\lambda*}) = \RR \, 1,
\qquad 
\ker\ssA^\lambda = \RR \, \nu^\normal, 
\ee
consisting of constant functions, and of constant multiples of a (silhouette) function denoted by~$\nu^\normal$, respectively. Moreover, under the radial stability condition~\eqref{equa-b2-positive} (cf.~\autoref{def-radial-Hstab}), the average $\la - \Deltaslash \nu + d_{n,p} \nu \ra$ is non-vanishing for all (non-trivial) elements $\nu \in \ker\ssA^\lambda$, allowing the silhouette function to be normalized by the condition $\bigl\la - \Deltaslash\nu^\normal + d_{n,p} \nu^\normal \bigr\ra = \theta^\lambda$ in agreement with \autoref{def:normalized-kernel-basis}  ---the constant $\theta^\lambda$ being defined in~\eqref{equa-thetalambda}. 
\end{proposition}

 
\paragraph{Asymptotic variational formulation.}

A variant of the Lax--Milgram theorem (Banach--Ne\v{c}as--Babu\v{s}ka theorem) can be applied and provides us with a unique variational solution to the equation $\ssA^\lambda[\nu] = \varphi$ in the domain $\Lambda$, provided we restrict attention to data and solutions with vanishing average. More precisely, since the bilinear form $\ssrmA^\lambda[\nu,\mu]=\fint_{\Lambda}\mu\ssA^\lambda[\nu]d\chi$ is non-symmetric and admits a non-trivial kernel, we rely on a variant~\cite{Babuska}, as follows. 
We use here the normalized dual space $\unH^{2*}_{-\expoP}(\Lambda)$ defined after~\eqref{equa-norm-poids} and, for $\varphi\in\unH^{2*}_{-\expoP}(\Lambda)$ and
$\mu\in\unH^2_{-\expoP}(\Lambda)$,  we set
\be
\langle\varphi,\mu\rangle_\Lambda
\coloneqq
\aire[\Lambda,\lambda]^{-1}
\langle\varphi,\mu\rangle_{d\chi}. 
\ee
In particular, $\la\varphi\ra=\langle\varphi,1\rangle_\Lambda$.

\begin{lemma}[Variational formulation for the asymptotic localized Hamiltonian]
\label{propo-existenceH}
Consider a conical domain $\Omega_R = K \cap {}^{\complement} \Ball_R \subset \RR^n$ together with a localization function $\lambda\colon \Lambda \to (0, \lambda_0]$ with connected support $\Lambda \subset \Sphe^{n-1}$ and some ${\expoP \geq 2}$, such that the harmonic stability condition in \autoref{def-harmonic-Hstab} holds. Then, for any $\varphi \in \unH^{2*}_{-\expoP}(\Lambda)$ satisfying $\la \varphi \ra= 0$, there exists a unique solution $\nu \in \unH^2_{-\expoP}(\Lambda)$ with a vanishing average $\la\nu\ra= 0$ to  
\bse
\bel{Aweak-mu} 
\ssrmA^\lambda[\nu,\mu]
= \langle\varphi,\mu\rangle_\Lambda,
\qquad \mu \in \unH^2_{-\expoP}(\Lambda),
\ee 
which, moreover, is bounded in terms of the data, that is, 
\bel{nuvarphi}
\|\nu\|_{\unH^2_{-\expoP}(\Lambda)}\lesssim\|\varphi\|_{\unH^{2*}_{-\expoP}(\Lambda)}.
\ee
\ese
\end{lemma}

\begin{proof}
\bse
We work in the space $W \coloneqq \{ \nu \in \unH^2_{-\expoP}(\Lambda), \, \la \nu \ra = 0 \}$ of functions with a vanishing average, endowed with its induced norm~$\|{\,\cdot\,}\|_W$.
The bilinear form $\ssrmA^\lambda$ is (bounded and) weakly coercive on~$W$ in the sense that, for some $c>0$, 
\be
\text{(i)}\,\, \sup_{\| \mu \|_W= 1} \ssrmA^\lambda[\nu, \mu] \geq c \, \|\nu\|_W, \quad \nu \in W, 
\qquad\quad
\text{(ii)}\,\, \sup_{\| \nu \|_W = 1} \ssrmA^\lambda[\nu, \mu] \geq c \, \|\mu\|_W , \quad \mu \in W.
\ee
Indeed, (i)~holds since $\mu = \nu/\|\nu\|_W$ has unit norm and obeys $\ssrmA^\lambda[\nu, \mu] = \ssrmA^\lambda[\nu, \nu] / \|\nu\|_W \gtrsim \| \nu \|_W$ by the harmonic stability condition.
The condition~(ii) is established likewise by picking $\nu = \mu/\|\mu\|_W$.
The linear form $\varphi\in \unH^{2*}_{-\expoP}(\Lambda)$ is bounded hence its restriction $\varphi|_W\in W'$ also is, with $\bigl\|\varphi|_W\bigr\|_{W'}\leq\|\varphi\|_{\unH^{2*}_{-\expoP}(\Lambda)}$.
By applying the Lax--Milgram theorem to $\varphi|_W$, there exists a unique $\nu\in W$ such that $\ssrmA^\lambda[\nu,\,\cdot\,]|_W=\varphi|_W$, namely such that~\eqref{Aweak-mu} holds for all $\mu\in W$.
To extend~\eqref{Aweak-mu} to all $\mu$ it suffices to show it for $\mu=1$.
By assumption, $\langle\varphi,1\rangle_\Lambda = \la\varphi\ra=0$, 
and on the other hand, $\lambda^{2\expoP}\ssA^\lambda[\nu]$ is a divergence, by inspection of~\eqref{equa-62b}, and therefore has a vanishing integral
\be
\ssrmA^\lambda[\nu,1] = \fint_\Lambda \ssA^\lambda[\nu]d\chi = \fint_\Lambda \lambda^{-2\expoP} \nablaslash\cdot(\text{vector field})\, d\chi = 0 .
\ee
This establishes~\eqref{Aweak-mu}.
Finally, the bound $\|\nu\|_W\lesssim \bigl\|\varphi|_W\bigr\|_{W'}$ provided by the Lax--Milgram theorem yields~\eqref{nuvarphi} as
\be
\|\nu\|_{\unH^2_{-\expoP}(\Lambda)} = \|\nu\|_W \lesssim \bigl\|\varphi|_W\bigr\|_{W'} \leq \|\varphi\|_{\unH^{2*}_{-\expoP}(\Lambda)} .
\qedhere
\ee
\ese
\end{proof} 

 
\paragraph{Proof of \autoref{prop:Ham-kernel}.}

{\it 1. Kernel dimension.} 
In view of the identity $\la \ssA^\lambda[\nu] \ra = \fint_{\Lambda} \nu \ssA^{\lambda *}[1] d\chi = 0$ for all $\nu\in \unH^2_{-\expoP}(\Lambda)$, 
we see that the image of the operator $\ssA^\lambda$ is contained in the subspace of distributions $\varphi \in \unH^{2*}_{-\expoP}(\Lambda)$ satisfying $\la\varphi\ra=0$. Moreover, by \autoref{propo-existenceH} the image is exactly equal to that subspace, so that $\ker\ssA^{\lambda*} = \RR \, 1$. On the other hand, consider next an element $\nu\in\ker\ssA^\lambda$ that belongs to the hyperplane $\la\nu\ra=0$. By the harmonic stability condition~\eqref{equa-stable-H-414} we have $\| \nu \|_{\unH^2_{-\expoP}(\Lambda)}^2 \lesssim \ssrmA^\lambda[\nu,\nu]$, which vanishes since $\nu$ is in the kernel of $\ssA^\lambda$. Thus, $\nu$ vanishes. We deduce that the linear map $\nu\mapsto\la\nu\ra$ from $\ker\ssA^\lambda$ to~$\RR$ is injective, hence that the kernel is {\sl at most} one-dimensional. To construct a non-trivial element of the kernel, consider the unique solution (by \autoref{propo-existenceH}) of $\ssA^\lambda[\nu] = \ssA^\lambda[1]$ satisfying $\la\nu\ra = 0$. By construction, $1-\nu\in\ker\ssA^\lambda$ and $\la 1-\nu\ra = 1$.

\medskip

\noindent{\it 2. Normalization.} The radial stability condition states that the product of $\bnotreH_{0}$ and~$\bnotreH_{1}$ is positive, where these two constants are linear combinations of $\la\Deltaslash\nu^\normal\ra$ and $\la\nu^\normal\ra$.  This positivity is of course independent of the normalization of $\nu^\normal$, so that it can be expressed in terms of any element $\nu=\gamma\nu^\normal$ ($\gamma\neq 0$) of the kernel as, with $\nut\coloneqq\nu-\la\nu\ra$, 
\be
\aligned
0 < \gamma^2 \bnotreH_{1} \bnotreH_{0}
= (n^2-4n+5) \frac{c_{n,p}}{a_{n,p}} \bigl\la (n-1) c_{n,p} \nu - (n-2)^2 \Deltaslash\nut \bigr\ra \bigl\la - \Deltaslash\nut + d_{n,p} \nu\bigr\ra .
\endaligned
\ee
In particular the last factor is non-vanishing, which concludes the proof of \autoref{prop:Ham-kernel}.

 
\paragraph{Normalization for the kernel.}

Finally, we show that our normalization leads to the following property. 

\begin{lemma}[ADM energy of the metric modulator]
\label{lem-ADMenergymod} 
The modulated metric\footnote{Note that $\gmodu$ is not positive-definite; it is strictly speaking a difference of metrics, and correspondingly $\mbb(\Omega_R, \gmodu)$ is a relative ADM energy in the sense of \autoref{def:relative-ADM}.} defined in $\Omega_R \subset \RR^n$ 
\[
(\gmodu)_{ij} \coloneqq \lambda^{2\expoP} r^{n-2p} \bigl(\del_i \del_j \umodu - \delta_{ij} \Delta \umodu \bigr), 
\qquad 
\umodu \coloneqq \mmodu \, {\nu^\normal(x/r) \over r^{a_{n,p}}}, 
\]
has ADM energy $\mbb(\Omega_R, \gmodu) = \mmodu$.
\end{lemma}

\begin{proof}
\bse
We begin with a preliminary observation. To any function $u$, we associate the vector field 
\be
V^\lambda_j[u] \coloneqq \del_i(g^\lambda_{ij}[u]) - \del_j (g^\lambda_{ii}[u]) ,
\qquad
g^\lambda_{ij}[u]=\lambda^{2\expoP}r^{n-2p}(\del_i\del_j u - \delta_{ij}\Delta u)
\ee
whose divergence is the fourth-order Hamiltonian operator
\bel{divV-notreH}
\del_j V_j^\lambda[u] = \lambda^{2\expoP}r^{n-2p} \notreH^\lambda[u] .
\ee
Up to numerical factors, the ADM energy we seek is the integral of $r^{n-1}\xh_j V_j^\lambda[\umodu]$ measuring the flux of $V_j^\lambda[\umodu]$ through a sphere of radius~$r$, which is independent of~$r$ because its expression is homogeneous of degree~$0$ under dilations.

For now, we continue with a general function~$u$. Decomposing the divergence into a radial and an angular parts yields
\bel{intrH-Vr}
\int_{\Lambda} r^4 \notreH^\lambda[u] d\chi
= r^{-a_{n,p}} \int_{\Lambda} \bigl( r \del_j V_j^\lambda[u] \bigr) r^{n-1} d\xh
= (\vartheta+a_{n,p}) \biggl(r^{-a_{n,p}}\int_{\Lambda} \xh_j V_j^\lambda[u] r^{n-1} d\xh\biggr) .
\ee 
Next, we note that the scalar differential operator $u\mapsto r^{n-1-a_{n,p}}\xh_j V_j^\lambda[u]$ commutes with $\vartheta$ because of its homogeneity in~$r$ and in derivatives.
Thus, the right-hand side is equal to the same expression with $(\vartheta+a_{n,p})$ acting on the argument~$u$ of $V_j^\lambda$ instead of the whole expression.
Up to an overall power of~$r$, the resulting integral is the flux of interest to us provided $(\vartheta+a_{n,p})u = \umodu$.

The choice $u=\umodu\log r$ obeys, for any polynomial~$\polyP$ with constant coefficients,
\bel{Pvarthetau-umodu}
\polyP(\vartheta) u
= (\log r) \polyP(\vartheta) \umodu + \polyP'(\vartheta) \umodu
= (\log r) \polyP(-a_{n,p}) \umodu + \polyP'(-a_{n,p}) \umodu, 
\ee
where $\polyP'$ denotes the derivative of the polynomial.  In particular, $(\vartheta+a_{n,p})u = \umodu$ as desired.
To finish the proof, we evaluate the left-hand side of~\eqref{intrH-Vr} in terms of $\Arr$, $\Ars^\lambda$, $\ssA^\lambda$ given in~\eqref{equa--488-2}.
Because we integrate with the measure $d\chi$, and $\lambda^{2\expoP}\ssA^\lambda[u]$ is a divergence, its contribution vanishes.
The contribution of $\Ars^\lambda[u]$ simplifies for the same reason.
Using~\eqref{Pvarthetau-umodu} for $\Arr[u]$ and the remaining part of~$\Ars^\lambda[u]$ eliminates all radial derivatives.
We arrive at a formula for the ADM energy, where $C_n \coloneqq (n-1) \, |\Sphe^{n-1}|$,
\bel{equa-form-energy} 
\aligned
2 C_n \, \mbb(\Omega_R, \gmodu)  
& = r^{n-1} \int_{\Lambda} \xh_j V_j^\lambda[\umodu] \, d\xh
= r^{a_{n,p}} (\vartheta+a_{n,p}) \biggl( r^{n-1-a_{n,p}} \int_{\Lambda} \xh_j V_j^\lambda[u] \, d\xh \biggr) 
\\
& = r^{a_{n,p}} \int_{\Lambda} \Bigl( (n-1)a_{n,p}b_{n,p} \umodu - \Bigl((n-2) a_{n,p} + \frac{c_{n,p}}{a_{n,p}}\Bigr) \Deltaslash \umodu \Bigr) \, d\chi
\\
& = \mmodu \aire[\Lambda,\lambda] \Bigl( (n-1) a_{n,p} b_{n,p} \la\nu^\normal\ra - (n^2 -4n+5) \la\Deltaslash\nu^\normal\ra \Bigr)
\\
& = \mmodu \aire[\Lambda,\lambda] (n^2 -4n+5) \theta^\lambda = 2 (n-1) |\Sphe^{n-1}| \mmodu ,
\endaligned
\ee
where the latter expression arises in view of our normalization in~\eqref{equa-norm-normal} and~\eqref{equa-thetalambda}.
\ese
\end{proof} 


\subsection{Radial evolution of spherical averages}
\label{section=6.3}

\paragraph{Assumptions.}

We now explain the derivation of the second-order differential equation~\eqref{equa-defb-b-b} (in~\eqref{equa-defb-b-b-bis}, below) and establish a sharp decay estimate for the spherical averages of a solution (in \autoref{prop-66}, below). As pointed out in \autoref{section=3.2}, in contrast to the isotropic case $(\Lambda,\lambda)=(\Sphe^{n-1},1)$, the evolution equation satisfied by the average $\la u(r) \ra$ also involves the average $\la \Deltaslash u(r) \ra$ and, therefore, cannot be solved independently of the fluctuations of the solution. Furthermore, the equation satisfied by $\la \Deltaslash u(r) \ra$ itself involves further weighted averages, so that there exists no closed system of differential equations for the averages. In view of this structure, we analyze the system satisfied by the above two averages while allowing for a source term denoted by $\Kappa^\notreH[u]$, eventually controlled by the shell dissipation and radial Hardy inequalities.

We consider the equation $\notreH^\lambda[u] = E$ satisfied by a solution $u:\Omega_R \to \RR$ for some given source term $E:\Omega_R \to \RR$. Since a basic decay arises from our variational formulation and for the sake of generality in our presentation, at this stage we only assume that the source and solution decay at least like $r^{-a_{n,p}/2}$ in a (dual) Sobolev sense\footnote{The regularity of $\vartheta u$ actually implies an $H^{2*}_{n+2-p,-\expoP}$ regularity of $\vartheta E$, but only away from $r=R$ in a suitable sense.  This information on $\vartheta E$ is not necessary in this section.}, 
\bel{equa-varia-bound}
\|E\|_{H^{2*}_{n+2-p,-\expoP}(\Omega_R)} + \|u\|_{H^2_{n-2-p,-\expoP}(\Omega_R)} + \|\vartheta u\|_{H^2_{n-2-p,-\expoP}(\Omega_R)}
< +\infty,
\ee
where the dual Sobolev space $H^{2*}_{n+2-p,-\expoP}(\Omega_R)$ is defined in~\eqref{H2star-def}.
These assumptions suppress certain constant or growing modes in the following.

The radial differential equations in this section are to be understood in the sense of distributions $\Dcal'((R,+\infty))$, 
namely with test functions being compactly supported in~$(R,+\infty)$ and, therefore, vanishing in a neighborhood of the boundary $r=R$.  As explained in \appref{appendix=F.3}, this ensures that $\vartheta=r\del_r$ defined in $\Dcal'((R,+\infty))$ extends the usual derivative on functions, with no boundary term at $r=R$.
The unknown and source terms in these equations are (restrictions to $(R,+\infty)$ of) elements of suitable dual Sobolev spaces with variational decay $H^{k*}_{n-2-p}([R,+\infty))$ (for some integer $k\geq 0$), which are defined in~\eqref{Hkstar-def}.  This regularity allows for an explicit integration in \autoref{prop:integrate-ODE-distrib}. All angular integrations by parts in this section are justified by Proposition~\ref{prop-density-H2}: since $\expoP\geq2$, smooth functions supported away from the lateral boundary are dense in the relevant weighted $H^2$~space.


\paragraph{Contracting with an element of the co-kernel.}

We integrate the equation $\notreH^\lambda[u] = E$, multiplied by either $1$ or $\nut^\normal=\nu^\normal- \la\nu^\normal\ra$, on each sphere of radius $r \geq R$. We find it convenient to  decompose a solution $u \coloneqq \la u \ra  + \ut$ into its average and its fluctuations on the sphere $\Sphe_r$. Integrating with the weight~$1$ and using $\la\ssA^\lambda[u] \ra=0$, we find ($\iota$ being suppressed throughout)
\bel{equa-average-one-H}
\compresseq{.89}
\aligned
r^4 \la E \ra = \Bigl\la\Arr[u] + \Ars^\lambda[u] + \ssA^\lambda[u] \Bigr\ra 
& = (n-1) \vartheta (\vartheta+a_{n,p}) \bigl( \vartheta^2 + a_{n,p} \vartheta - b_{n,p} \bigr) \la u \ra 
\\
& \quad
 + (\vartheta+a_{n,p}) \bigl((n-2) \vartheta+(n-2)a_{n,p} - n^2 + 4n - 5 \bigr) \la \Deltaslash \ut \ra,
\endaligned
\ee 
which can be regarded as a fourth-order relation for $\la u \ra$ but also involves  $\la \Deltaslash \ut \ra$. After two radial integrations, it reduces to the second-order equation used below once the variationally excluded homogeneous modes have been eliminated.

\begin{remark}
The variational formulation~\eqref{equa-variH} with $w=r^{a_{n,p}}\varphi(r)$ provides a meaning to~\eqref{equa-average-one-H} in $H^{2*}_{n-2-p}([R,+\infty))$.  For instance, in the notation~\eqref{vartheta-star-def}, the highest-order radial derivative is $\vartheta_*^2 \vartheta^2\la u\ra$, defined by its duality bracket $\la\vartheta_*^2\vartheta^2\la u\ra,\varphi\ra=\int_R^{+\infty}\vartheta^2\la u\ra\vartheta^2\varphi dr/r$ for $\varphi\in H^2_{-n+2+p}([R,+\infty))$.  Compared to $\vartheta^4\la u\ra$ (defined for smooth enough~$\la u\ra$) this differs by boundary terms.  To avoid such boundary terms in later calculations, we restrict to $\Dcal'((R,+\infty))$, namely impose that $\varphi$ has compact support in $(R,+\infty)$.
\end{remark}


\paragraph{Contracting with an element of the kernel.}
 
Next, we proceed similarly with the weight $\nut^\normal=\nu^\normal- \la\nu^\normal\ra$, and deduce an expression of the adjoint operator $\ssA^\lambda$. Specifically, we compute 
\[
\aligned
\fint_{\Lambda} \nut^\normal \ssA^\lambda[u] \, d\chi
& = \fint_{\Lambda} u \, \ssA^{\lambda *}\bigl[\nu^\normal- \la\nu^\normal\ra\bigr] \, d\chi
= \fint_{\Lambda} u \, \Bigl( \ssA^{\lambda *}- \ssA^\lambda \Bigr)[\nu^\normal] \, d\chi
\\
& = c_{n,p} \Bigl(
\la\nu^\normal\ra \la\Deltaslash\ut\ra - \la\Deltaslash\nut^\normal\ra \la u\ra
+ \fint_{\Lambda} \Bigl( \nut^\normal \Deltaslash\ut - \ut \, \Deltaslash\nut^\normal\Bigr) d\chi
\Bigr). 
\endaligned
\]
Hence, we have (with $\Kappa^\notreH$ defined next)
\[
\aligned
r^4 \fint_{\Lambda} \nut^\normal \, E d\chi
& = \fint_{\Lambda} \nut^\normal \Bigl(\Arr[u] + \Ars^\lambda[u] + \ssA^\lambda[u] \Bigr) d\chi
\\
& = c_{n,p} \la\nu^\normal\ra \la\Deltaslash\ut\ra
+ \la\Deltaslash\nut^\normal\ra \bigl((n-2) \vartheta + (n-2)^2+1\bigr) \vartheta \la u\ra
+ \vartheta\Kappa^\notreH[\ut] .
\endaligned
\]

\begin{definition}\label{def:KappaH}
The {\bf fluctuation operator} $\Kappa^\notreH$ is a linear functional of the fluctuation~$\ut$ and its derivatives given by
\be
\aligned
\Kappa^\notreH[\ut]  
& \coloneqq \fint_{\Lambda} \Bigl(
(n-1)(\vartheta+a_{n,p}) \bigl(\vartheta^2+a_{n,p}\vartheta-b_{n,p}\bigr) \nut^\normal \ut
\\[-1ex]
& \qquad\quad + (\vartheta+a_{n,p})
\bigl(
- 2 \, \nablaslash \ut \cdot \nablaslash \nut^\normal
+ (n-2) ( \nut^\normal \Deltaslash \ut + \ut \, \Deltaslash\nut^\normal)
\bigr)
\\[-.5ex]
& \qquad\quad + \bigl( (n-2)(n-2-a_{n,p}) + 1\bigr) \bigl( \ut \, \Deltaslash\nut^\normal - \nut^\normal \, \Deltaslash\ut\bigr) \Bigr) d\chi. 
\endaligned
\ee
\end{definition}

In turn, the desired equation for the average of the Laplacian can be cast in the form (cf. \autoref{rem-foot194}) 
\bel{equa-Delta--u}
\aligned
\la\Deltaslash\ut\ra
= \frac{1}{c_{n,p} \la\nu^\normal\ra} \biggl( r^4 \fint_{\Lambda} \nut^\normal \, E d\chi - \vartheta\Kappa^\notreH[\ut] \biggr)
- \frac{\la\Deltaslash\nut^\normal\ra}{c_{n,p} \la\nu^\normal\ra} \bigl((n-2) \vartheta + (n-2)^2+1\bigr) \vartheta \la u\ra, 
\endaligned
\ee
whose right-hand side involves up to two radial derivatives of $\la u\ra$, and four derivatives of~$\ut$.


\paragraph{Fourth-order equation for the average.}

Inserting (radial derivatives of) the expression~\eqref{equa-Delta--u} of $\la\Deltaslash\ut\ra$ into the equation~\eqref{equa-average-one-H} for $\la u\ra$ yields a fourth-order differential equation\footnote{The restriction to $(R,+\infty)$ becomes unavoidable here, as~\eqref{equa-calcul-equa} involves radial derivatives beyond the variational formulation.} in $\Dcal'((R,+\infty))$, namely  
\bse\label{equa-612}
\bel{equa-calcul-equa}
 - \vartheta (\vartheta+a_{n,p}) \Bigl( Q_2(\vartheta) \la u \ra
+ \bigl( (n-2) \vartheta - {c_{n,p} / a_{n,p}} \bigr) \Kappa^\notreH[\ut] 
\Bigr) = \Nuh^\notreH[E] .
\ee
Here, we introduced the second-order operator 
\bel{equa-Ptwo-0} 
\aligned
Q_2(\vartheta)
& \coloneqq - \bnotreH_{1} \, (\vartheta^2 + a_{n,p} \vartheta) +  \bnotreH_{0}, 
\\
\bnotreH_{1}& \coloneqq (n-1) c_{n,p}\la\nu^\normal\ra - (n-2)^2\la\Deltaslash\nut^\normal\ra, 
\\
\bnotreH_{0} &\coloneqq (n-1)b_{n,p} c_{n,p} \la\nu^\normal\ra - (n^2 -4n+5){c_{n,p} \over a_{n,p}} \la\Deltaslash\nut^\normal\ra, 
\endaligned
\ee 
together with the following contribution from the source term 
\bel{equa-la-sourceH}
\aligned
\Nuh^\notreH[E]
& \coloneqq c_{n,p}\la\nu^\normal\ra r^4 \la E \ra
- (\vartheta+a_{n,p}) \Bigl((n-2) \vartheta - {c_{n,p} / a_{n,p}} \Bigr) \biggl( r^4 \fint_{\Lambda} \nut^\normal \, E d\chi \biggr) .
\endaligned
\ee
\ese


\paragraph{Control of spherical averages.}

The operator $Q_2(\vartheta) = - \bnotreH_{1} (\vartheta+\beta_-) (\vartheta+\beta_+)$ admits characteristic exponents
\bel{betapm-def-text}
\beta_{\pm} = a_{n,p} / 2 \pm \sqrt{a_{n,p}^2 / 4 + \bnotreH_{0} / \bnotreH_{1}}
\ee
that obey $\beta_++\beta_-=a_{n,p}$.
Provided the radial stability condition $\bnotreH_0 \bnotreH_1 > 0$ holds,
the exponents $\beta_{\pm}$ are real and {\it outside} the interval $[0,a_{n,p}]$ and we order them so that $\beta_-<0<a_{n,p}<\beta_+$.
The exponents $\beta_-$ and $\beta_+$ play a key role in our analysis, and correspond to a {\it growing} mode (excluded by variational bounds) and a {\it super-harmonic} mode, respectively.
By integrating radially the differential equation~\eqref{equa-612}, we establish next that $\la u\ra$ is given by suitable radial integrals of $\Kappa^\notreH[\ut]$, up to error terms involving the source term~$E$.

To express the radial integrals we introduce the solution operators
$I_\beta$ and $J_\beta$ acting on functions $f:[R,+\infty) \to\RR$ in the radial variable and any exponent $\beta\in\RR$:
\bel{IJdef-first}
I_\beta[f](r) \coloneqq r^{-\beta} \int_R^r f(s) s^\beta \frac{ds}{s},
\qquad
J_\beta[f](r) \coloneqq r^{-\beta} \int_r^{+\infty} f(s) s^\beta \frac{ds}{s} 
\ee
(under obvious integrability conditions).  These solution operators enjoy various positivity and continuity properties stated and proven in \autoref{appendix=F}. In particular, as stated in~\eqref{IJprop} we have $(\vartheta+\beta)I_\beta[f] = f$ as well as $(\vartheta+\beta)J_\beta[f] = -f$.
They are also defined for $f$ in suitable dual Sobolev spaces in \appref{appendix=F.3}.

\begin{proposition}[Spherical averages associated with the localized Hamiltonian operator]
\label{prop-66}
\bse
In the setup of \autoref{thm-sharp-h-localized-vari} (variational formulation), assume that the localization domain $(\Lambda,d\chi)$ obeys the harmonic and radial Hamiltonian stability conditions in \refwithname{Definitions}{def-harmonic-Hstab} \refwithname{and}{def-radial-Hstab} (but not necessarily shell stability), that the variational solution $u \in H^2_{n-2-p, -\expoP}(\Omega_R)$ obeys the additional decay
\[
\vartheta u \in H^2_{n-2-p,-\expoP}(\Omega_R) ,
\]
and that the source~$E$ obeys the decay~\eqref{equa-hyo-EEE} for some $\astar\in[a_{n,p}/2,\beta_+)$, namely $E,\vartheta_*E\in H^{2*}_{\astar+4,-\expoP}(\Omega_R)$ and $\mmodu(E)$ is well-defined if $\astar\geq a_{n,p}$.
Then there exist constants $\Chu,C_+^u\in\RR$ such that the averages over shells~$\Lambda_r$ satisfy\footnote{All terms are functions in $L^2_{n-2-p}([R,+\infty))$ or more regular.}
\bel{equa-main-u-aver} 
\aligned
\la u(r)\ra
& = \Chu r^{-a_{n,p}} + C_+^u r^{-\beta_+} + c_- J_{\beta_-}\bigl[\Kappa^\notreH[\ut]\bigr](r) + c_+ I_{\beta_+}\bigl[\Kappa^\notreH[\ut]\bigr](r) + \Theta^E(r) ,
\endaligned
\ee
with a source term $\Theta^E$ given as integrals of the source~$E$ in \eqref{ThetaE-def} and~\eqref{equa-hereOmegaHcompact}, below.
Moreover,
\be
|\Chu| + |C_+^u|
\lesssim \|E\|_{H^{2*}_{\astar+4,-\expoP}(\Omega_R)} + \cutoff_{\pstar} \mmax(E) + \|\vartheta u\|_{H^2_{n-2-p,-\expoP}(\Omega_R)} ,
\ee
where $\cutoff_{\pstar}$ is $1$ for $\astar\geq a_{n,p}$ and otherwise vanishes, and $\mmax$ is defined in~\eqref{mmodu-def}, and
one can decompose $\Theta^E = \Theta^E_1+\Theta^E_2$ with
\bel{estimp65}
\|\Theta^E_1\|_{L^2_{\astar}([R,+\infty))} \lesssim \|E\|_{H^{2*}_{\astar+4,-\expoP}(\Omega_R)} ,
\qquad
\Theta^E_2 = \begin{cases}
  0 & \text{if } \astar\neq a_{n,p} , \\
  \la\nu^\normal\ra \mmax(E) r^{-a_{n,p}} o(1) & \text{if } \astar=a_{n,p}
\end{cases}
\ee
where $o(1)$~is a function of radius that tends to zero.
\ese
\end{proposition}

We first establish an intermediate result obtained by integrating the fourth-order differential equation~\eqref{equa-612} twice to obtain a lower-order equation that was announced in~\eqref{equa-defb-b-b}.

\begin{lemma}
\bse
In the setup of \autoref{prop-66},
the averages $\la u\ra$ satisfy the second-order differential equation
\bel{equa-defb-b-b-bis}
- \bnotreH_{1} (\vartheta + \beta_-) (\vartheta+\beta_+) \, \la u\ra 
= \bigl( -(n-2) \vartheta +{c_{n,p} / a_{n,p}} \bigr) \Kappa^\notreH[\ut]
+ \Nu^\notreH[E] + C_1^u r^{-a_{n,p}}
\ee
in $\Dcal'((R,+\infty))$,
for some constant~$C_1^u$, where 
the operator $\Kappa^\notreH[\ut] $ is defined in \autoref{def:KappaH}, the constants $\bnotreH_{1}, \beta_-, \beta_+$  are given in~\eqref{equa-Ptwo-0-struct}, and
\bel{equa-hereOmegaHcompact}
\Nu^\notreH[E]
= (n-2) \la \nut^\normal \, r^4 E \ra + \frac{c_{n,p}}{a_{n,p}} J_0\bigl[\la \nu^\normal \, r^4 E \ra \bigr]
+ \frac{c_{n,p}}{a_{n,p}} \la\nu^\normal\ra I_{a_{n,p}} \bigl[ \la r^4 E \ra \bigr].
\ee
\ese
\end{lemma}

\begin{proof}
The equation~\eqref{equa-612} reads $- \vartheta (\vartheta + a_{n,p}) \barg_2 = \barg_0 + (\vartheta+a_{n,p}) \barg_1$ with $\barg_i=g_i|_{(R,+\infty)}\in\Dcal'((R,+\infty))$ and
\be
\aligned
g_2 & \coloneqq Q_2(\vartheta) \la u \ra + \bigl( (n-2) \vartheta_* - {c_{n,p} / a_{n,p}} \bigr) \Kappa^\notreH[\ut] - (n-2) r^4 \la \nut^\normal \, E \ra ,
\\
g_1 & \coloneqq {c_{n,p} \over a_{n,p}} \, r^4 \la \nut^\normal \, E \ra ,
\qquad
g_0 \coloneqq c_{n,p}\la\nu^\normal\ra \, r^4 \la E \ra .
\endaligned
\ee
Here, $\vartheta_*\colon L^2_{n-2-p}([R,+\infty))\to H^{1*}_{n-2-p}([R,+\infty))$ is the derivative operator defined in~\eqref{vartheta-star-def}, which reduces to the usual distributional derivative $\vartheta=r\del_r$ on the interval $(R,+\infty)$ of interest.
Writing $\vartheta_*$ rather than restricting to $(R,+\infty)$ before taking the derivative ensures that each $g_i\in H^{2*}_{n-2-p,-\expoP}(\Omega_R)$.
As stated in~\eqref{equa-formuleODE}, the general solution with this decay does not include a homogeneous $r^0$ term, and reads
\[
g_2(r) = \frac{1}{a_{n,p}} J_0[g_0+a_{n,p}g_1](r) + \frac{1}{a_{n,p}} I_{a_{n,p}}[g_0](r) + C_1^u r^{-a_{n,p}} \quad \text{in } (R,+\infty)
\]
for some constant~$C_1^u$ that depends on boundary conditions on~$\barg_2$.
This coincides with the desired second-order equation~\eqref{equa-defb-b-b-bis} on~$\la u\ra$ thanks to $\la\nut^\normal E\ra+\la\nu^\normal\ra\la E\ra=\la\nu^\normal E\ra$.
\end{proof}

\paragraph{Proof of \autoref{prop-66}}
\bse
{\it 1. Effect of fluctuations.}
The general solution~\eqref{equa-formuleODE} of~\eqref{equa-defb-b-b-bis} with variational decay reads, on the interval $(R,+\infty)$,
\be
\aligned
\la u\ra & = C_+^u r^{-\beta_+} + J_{\beta_-}[h_-] + I_{\beta_+}[h_+] ,
\\
h_{\pm} \, & \! \coloneqq
\frac{1}{(\beta_+-\beta_-)\bnotreH_{1}} \Bigl( \bigl( (n-2) \beta_{\pm} + c_{n,p} / a_{n,p} \bigr) \Kappa^\notreH[\ut]
+ \Nu^\notreH[E] + C_1^u r^{-a_{n,p}} \Bigr)
\endaligned
\ee
for some constant~$C_+^u$, where the absence of homogeneous term $r^{-\beta_-}$ is due to the variational decay (and to $\beta_-<a_{n,p}/2$).
Splitting $\Nu^\notreH[E] + C_1^u r^{-a_{n,p}} = (\Nu^\notreH[E] - \ChNu r^{-a_{n,p}}) + C_2^u r^{-a_{n,p}}$ for some constant~$\ChNu$ chosen below, and
separating out the various contributions, one gets
\bel{equa-main-u-average} 
\aligned
\la u(r)\ra
& = \Chu r^{-a_{n,p}} + C_+^u r^{-\beta_+} + c_- J_{\beta_-}\bigl[\Kappa^\notreH[\ut]\bigr](r) + c_+ I_{\beta_+}\bigl[\Kappa^\notreH[\ut]\bigr](r) + \Theta^E(r) ,
\endaligned
\ee
with $\Chu$ given by a combination of exponents times~$C_2^u$, a source contribution~$\Theta^E$ explicited below, and the constants
\bel{equa-cpcm}
c_\pm \coloneqq {1 \over \bnotreH_{1} } {(n-2) \beta_{\pm} + {c_{n,p} / a_{n,p}} \over \beta_+ - \beta_-} .
\ee
The formula shows that any bound of the form $|\Kappa^\notreH[\ut]|\leq C \, r^{-\gamma}$ for $\beta_-<\gamma<\beta_+$ translates to a similar bound $|\la u \ra-\Chu r^{-a_{n,p}}|\leq C' r^{-\gamma}$, provided that the source term~$\Theta^E$ obeys this bound.
\ese

\medskip

\bse
\noindent{\it 2. Source term.}
Up to a homogeneous solution $\ChE r^{-a_{n,p}}$ that we choose freely to include, the contribution to~\eqref{equa-main-u-average} from the source in~\eqref{equa-hereOmegaHcompact} is
\bel{ThetaE-def}
\aligned
\Theta^E
& = {1 \over \bnotreH_{1}(\beta_+ - \beta_-)} \Bigl( J_{\beta_-}\bigl[\Nu^\notreH[E]\bigr] + I_{\beta_+}\bigl[\Nu^\notreH[E] \bigr] \Bigr)
- \ChE r^{-a_{n,p}} ,
\qquad
\\
\ChE & = \cutoff_{\pstar} \la\nu^\normal\ra \mmodu(E) r^{-a_{n,p}} ,
\endaligned
\ee
which we now bound in a suitable weighted Lebesgue space.
Consider first $\Nu^\notreH[E]$.
The angular integrals of~$r^4 E$ that appear in its expression~\eqref{equa-hereOmegaHcompact} are bounded in $H^{2*}_{\astar}([R,+\infty))$ by the $H^{2*}_{\astar+4,-\expoP}(\Omega_R)$ norm of~$E$.

In the sub-harmonic regime $\astar<a_{n,p}$, \autoref{lem:Ibeta-norm} then bounds the $J_0$ and $I_{a_{n,p}}$ terms in~\eqref{equa-hereOmegaHcompact} in the norm $H^{1*}_{\astar}([R,+\infty))$, hence we have a bound in the weaker norm (because of the term $\la\nut^\normal r^4 E\ra$ that is not integrated radially)
\be
\|\Nu^\notreH[E]\|_{H^{2*}_{\astar}([R,+\infty))} \lesssim \|E\|_{H^{2*}_{\astar+4,-\expoP}(\Omega_R)} ,
\qquad \astar < a_{n,p} ,
\ee
where the implicit constant depends on the localization domain $(\Lambda,d\chi)$ and exponents $\astar$ and~$a_{n,p}$.
In the super-harmonic regime $\astar>a_{n,p}$, the term $I_{a_{n,p}}[\la r^4E\ra]$ does not decay with radial exponent~$\astar$.  We use the identity~\eqref{Ibeta-plus-Jbeta}, whose terms are in $H^{1*}_\gamma([R,+\infty))$ for some irrelevant $\gamma<a_{n,p}$,
\bel{Ianp-to-Janp}
I_{a_{n,p}}[\la r^4E\ra]
= \frac{2(n-1)|\Sphe^{n-1}| \mmodu(E)}{\aire[\Lambda,\lambda]} r^{-a_{n,p}} - J_{a_{n,p}}[\la r^4 E\ra] ,
\ee
to obtain, with $\ChNu= 2(n-1)|\Sphe^{n-1}| (c_{n,p}/a_{n,p})\la\nu^\normal\ra \mmodu(E)/\aire[\Lambda,\lambda]$,
\bel{ChNu}
\bigl\| \Nu^\notreH[E] - \ChNu r^{-a_{n,p}} \bigr\|_{H^{2*}_{\astar}([R,+\infty))}
\lesssim \|E\|_{H^{2*}_{\astar+4,-\expoP}(\Omega_R)} ,
\qquad a_{n,p} < \astar .
\ee
Then, by \autoref{lem:Ibeta-norm}, an expression of the form $J_{\beta_0}[f] + I_{\beta_1}[f]$ ``gains'' two derivatives compared to~$f$.
Taking into account the harmonic term that we subtracted from $\Nu^\notreH[E]$, we get (for $\astar\neq a_{n,p}$)
\bel{ThetaE-nonharmonic}
\aligned
& \biggl\| \Theta^E
+ \cutoff_{\pstar} \la\nu^\normal\ra \mmodu(E) \Bigl(1 - \frac{2(n-1)|\Sphe^{n-1}| c_{n,p}}{a_{n,p}\aire[\Lambda,\lambda] \bnotreH_{1} (a_{n,p} - \beta_-) (\beta_+ - a_{n,p})} \Bigr) r^{-a_{n,p}}
\biggr\|_{L^2_{\astar}([R,+\infty))}
\\
& \quad \lesssim \|E\|_{H^{2*}_{\astar+4,-\expoP}(\Omega_R)} ,
\endaligned
\ee
with implicit constants depending on the geometry and on exponents $a_{n,p},\astar,\beta_{\pm}$.
To eliminate the harmonic term above, we use that $\beta_+ + \beta_- = a_{n,p}$ and $\beta_- \beta_+ = - \bnotreH_{0}/\bnotreH_{1}$ (cf.~\eqref{betapm-def-text}), together with expressions \eqref{equa-thetalambda} and~\eqref{bH10} of $\theta^\lambda$ and~$\bnotreH_{0}$ to see that
\be
\frac{2(n-1)|\Sphe^{n-1}| c_{n,p}}{a_{n,p}\aire[\Lambda,\lambda] \bnotreH_{1} (a_{n,p} - \beta_-) (\beta_+ - a_{n,p})}
= \frac{2(n-1)|\Sphe^{n-1}| c_{n,p}}{a_{n,p}\aire[\Lambda,\lambda] \bnotreH_{0}}
= \frac{\theta^\lambda}{d_{n,p}  - \la\Deltaslash\nut^\normal\ra} = 1 ,
\ee
where the last simplification arises thanks to the normalization $d_{n,p} \la\nu^\normal\ra - \la\Deltaslash\nut^\normal\ra = \theta^\lambda$ in \autoref{def:normalized-kernel-basis}.
\ese

\medskip

\bse
\noindent{\it 3. Source term in the harmonic regime.}
Finally, for $\astar=a_{n,p}$, we begin with the straightforward bound
\be
\Bigl\| \Nu^\notreH[E] - \frac{c_{n,p}}{a_{n,p}} \la\nu^\normal\ra I_{a_{n,p}} \bigl[ \la r^4 E \ra \bigr] \Bigr\|_{H^{2*}_{a_{n,p}}([R,+\infty))}
\lesssim \|E\|_{H^{2*}_{a_{n,p}+4,-\expoP}(\Omega_R)} ,
\qquad \astar = a_{n,p} .
\ee
The decay assumption~\eqref{equa-hyo-EEE} 
(see also \autoref{rem:decay-reg}) states that
 the integral $r^{a_{n,p}}I_{a_{n,p}}[\la r^4E\ra]$ on the interval $[R,r]$ is bounded pointwise in~$r$, and tends to $2(n-1)|\Sphe^{n-1}|\mmodu(E) / \aire[\Lambda,\lambda]$ as $r\to+\infty$.
Thus, the same identity~\eqref{Ianp-to-Janp} holds now as a pointwise identity, and $r^{a_{n,p}} J_{a_{n,p}}[\la r^4 E\ra]$ tends to zero.  In addition to tending to zero, it is bounded by the truncated energy $\mmax(E)$ defined in~\eqref{mmodu-def},
\bel{raJarE}
\bigl| r^{a_{n,p}} J_{a_{n,p}}[\la r^4 E\ra] \bigr|
= \biggl| \int_r^{+\infty} \la s^4 E\ra s^{a_{n,p}-1} ds \biggr|
\leq \frac{2(n-1)|\Sphe^{n-1}|}{\aire[\Lambda,\lambda]} \mmax(E) .
\ee
Thus, $\Nu^\notreH[E]$ is the sum of a distribution in $H^{2*}_{a_{n,p}}([R,+\infty))$ and a harmonic term~$r^{-a_{n,p}}$ whose contributions~$\Theta^E_1$ to~$\Theta^E$ are treated as in the non-harmonic case~\eqref{ThetaE-nonharmonic}, and of a term $J_{a_{n,p}}[\la r^4 E\ra]$ that is (pointwise) $o(r^{-a_{n,p}})$.
By \autoref{lem:appF-pointwise}, its contribution to~$\Theta^E$ (which we denote as~$\Theta^E_2$) is bounded in the same way:
\be
\Theta^E_2 = \frac{1}{\bnotreH_{1}(\beta_+ - \beta_-)} (J_{\beta_-}+I_{\beta_+})\bigl[J_{a_{n,p}}[\la r^4 E\ra]\bigr] = o(1) r^{-a_{n,p}} .
\ee
In addition, taking into account the upper bound~\eqref{raJarE} shows that the $o(1)$ factor here is bounded by $\mmax(E)$ up to factors depending on the geometry and exponents.
\ese

\medskip

\noindent{\it 4. Homogeneous terms.}
Finally, we consider the $L^2_{n-2-p}([R,+\infty))$ norm (with the variational radial exponent) of every term in~\eqref{equa-main-u-average}.  The source term $\Theta^E$ is bounded by $\mmax(E)$ and the weighted $H^{2*}$ norm of~$E$.  The average $\la u\ra$, as well as contributions of $\Kappa^\notreH[\ut]$ are bounded by weighted $H^2$ norms of $u$ and~$\vartheta u$.  In addition, the weighted $H^2$ norm of~$u$ is controlled by the weighted $H^{2*}$ norm of the source~$E$ thanks to the variational existence theorem.  Altogether,
\be
\bigl\|\Chu r^{-a_{n,p}} + C_+^u r^{-\beta_+} \bigr\|_{L^2_{n-2-p}([R,+\infty))}
\lesssim \|E\|_{H^{2*}_{\astar+4,-\expoP}(\Omega_R)} + \cutoff_{\pstar} \mmax(E) + \|\vartheta u\|_{H^2_{n-2-p,-\expoP}(\Omega_R)} ,
\ee
which immediately gives a bound on $|\Chu|$ and~$|C_+^u|$ separately. This completes the proof of~\autoref{prop-66}.


\begin{remark}
\label {rem-foot194} 
The additional regularity $\vartheta u\in H^2_{n-2-p,-\expoP}(\Omega_R)$ in~\eqref{equa-varia-bound} ensures that \autoref{def:KappaH} defines a function $\Kappa^\notreH[\ut]\in L^2_{n-2-p}([R,+\infty))$, which is such that~\eqref{equa-Delta--u} holds in the sense of distributions on $(R,+\infty)$.  All terms of~\eqref{equa-Delta--u} are meaningful in $H^{2*}_{n-2-p}([R,+\infty))$, provided $\vartheta\Kappa^\notreH[\ut]$ is understood as $\vartheta_*\Kappa^\notreH[\ut]$, but the identity~\eqref{equa-Delta--u} fails due to boundary terms.  Indeed, the variational formulation with $w=r^{a_{n,p}}\varphi(r)\nut^\normal$ yields a similar identity but with some radial derivatives $\vartheta$ in \autoref{def:KappaH} replaced by the distributional derivative $\vartheta_*$, in effect adding to $\Kappa^\notreH[\ut]$ a Dirac distribution at $r=R$, which belongs to $H^{1*}_{n-2-p}([R,+\infty))$.
\end{remark}


\section{Integral estimates for the Hamiltonian (\autoref{thm:decayL2})}
\label{section=7}

\subsection{Preliminaries for the integral estimates}
\label{section=7.1}

\paragraph{Range of super-harmonic exponents.}

We now construct the positive super-harmonic margin~$\deltaH$ announced in \autoref{thm:decayL2}, and then formulate the Hamiltonian shell identity in the sense of distributions.
In \autoref{section=7.2} we prove \autoref{thm:decayL2} by solving this shell identity and exhibiting suitable positivity properties, modulo a contribution of shell averages, which are controlled in \autoref{section=7.3}.

The range of sharp decay exponents $\astar\in [a_{n,p}/2,a_{n,p}+\deltaH)$ in  \autoref{thm:decayL2}
for which integral estimates can be obtained is limited by details of the localization domain, such as the implicit constants in the Hamiltonian shell stability condition~\eqref{equa-last-twoH}.
Recall that the exponents $\beta_{\pm}$ are defined in~\eqref{betapm-def-text} in terms of the silhouette function.
We recall the constants~$c_\pm$ in~\eqref{equa-cpcm} and introduce the functions $g_J,g_I\colon(\beta_-,\beta_+)\to(0,+\infty)$,
\bel{equa-cplusmoins}
\aligned
c_\pm & \coloneqq {1 \over \bnotreH_{1}} {(n-2) \beta_{\pm} + c_{n,p}/a_{n,p} \over \beta_+ - \beta_-} ,
\\
g_J(\bstar) & \coloneqq \frac{16 c_+^2}{(2 \beta_+ - a_{n,p})^2} + \frac{8 c_-^2}{(a_{n,p} - 2\beta_-)^2} + \frac{2 c_-^2}{(a_{n,p}-2\beta_-) (\bstar-\beta_-)} ,
\\
g_I(\bstar) & \coloneqq \frac{4 c_-^2}{(\bstar-\beta_-)^2} + \frac{2 c_+^2}{(\beta_+-\bstar)^2} + \frac{2 c_+^2 }{(2 \beta_+ - a_{n,p})(\beta_+ - \bstar)}.
\endaligned
\ee
These expressions arise in \autoref{lem:gnotreH} as a result of radial Hardy inequalities.
The values $g^\notreH_I=g_I(a_{n,p})$ and $g^\notreH_J=g_J(a_{n,p})$ are particularly important (cf.~\eqref{equa--517-repeat}) since we define the Hardy constant appearing in the Hamiltonian shell stability condition to be $\cradialH=\max(g^\notreH_I,g^\notreH_J)$, which depends on $n$, $p$, $\la\nu^\normal\ra$ and $\la\Deltaslash\nu^\normal\ra$.

\bse
Constructing this margin~$\deltaH$ relies on three bounds, whose importance will be explained in due course.
First, as stated in~\eqref{equa-conditionH}, the shell functional~$\Phi^\notreH[v]$ and dissipation~$\Psi^\notreH[v]$, hence the radial integration functional $\Upsilon^\notreH[v]=\frac{1}{a_{n,p}}(\Psi^\notreH_{a_{n,p}}[v]-\Psi^\notreH_{2a_{n,p}}[v])$, are controlled by suitable norms of~$v$,
with radius-independent constants $\gamma_\Phi,\gamma_\Psi>0$ chosen so that
\bel{dl2-Phi}
\aligned
\bigl| \Phi^\notreH[v]  \bigr| 
& \leq \gamma_\Phi \bigl(\norm{v}^\notreH\bigr)^2 , \qquad
\\ 
\bigl| \Psi^\notreH_{a_{n,p}}[v]  \bigr| + \bigl| \Psi^\notreH_{2a_{n,p}}[v]  \bigr| 
& \leq \gamma_\Psi \bigl(\norm{v}^\notreH\bigr)^2 + \gamma_\Psi \bigl(\norm{\vartheta v}^\notreH\bigr)^2 ,
\\
\bigl| \Upsilon^\notreH[v]  \bigr| 
& \leq \frac{\gamma_\Psi}{a_{n,p}} \bigl(\norm{v}^\notreH\bigr)^2 + \frac{\gamma_\Psi}{a_{n,p}} \bigl(\norm{\vartheta v}^\notreH\bigr)^2
\endaligned
\ee
on every spherical shell $\Lambda_r$, $r\geq R$ and every scalar field $v\colon\Omega_R\to\RR$.
The chosen constants $\gamma_\Phi,\gamma_\Psi>0$ 
depend only on $n,p,\cstun,\cstdeux$; thus their possible dependence on the localization geometry occurs solely through the fixed admissible choice of~$\cstun,\cstdeux$. 

Second, the explicit expression of the fluctuation operator $\Kappa^\notreH[\vt]$ in \autoref{def:KappaH} yields a bound with a constant $\gamma_\Kappa>0$, depending on $p$ and on the geometry through the silhouette function~$\nu^\normal$, such that
\be
\bigl| \Kappa^\notreH[\vt]  \bigr|^2 \leq \gamma_\Kappa \bigl(\norm{v}^\notreH \bigr)^2 + \gamma_\Kappa \bigl(\norm{\vartheta v}^\notreH\bigr)^2
\ee
fon every $\Lambda_r$. 
Third, the shell stability condition~\eqref{equa-conditionH2} for $\beta=a_{n,p}$ and $\beta=2a_{n,p}$ states that for some $\gammash^0>0$,
\bel{dl2-shell}
\min\bigl(\Psi^\notreH_{a_{n,p}}[v], \Psi^\notreH_{2a_{n,p}}[v]\bigr)
\geq - C \, \la v\ra^2 + C \, \cradialH \bigl( \Kappa^\notreH[\vt] \bigr)^2
+ \gammash^0 \bigl( \norm{\vartheta v}^\notreH \bigr)^2 + \gammash^0 \bigl( \norm{v}^\notreH \bigr)^2 .
\ee
\ese
\bse
Then (for purposes that will become clear as we proceed) we define the continuous function
\bel{gammaPsipbstar}
\aligned
\gammash(\bstar)
= \gammash^0 & - \max\bigl(0, g_J(\bstar)-\cradialH, g_I(\bstar)-\cradialH\bigr) C\gamma_\Kappa 
\\
& - |2\bstar-2a_{n,p}| \Bigl( (2\bstar-a_{n,p})\gamma_\Phi+\frac{\gamma_\Psi}{a_{n,p}} \Bigr)
\endaligned
\ee
for $\bstar\in(a_{n,p}/2,\beta_+)$.
It is positive at $\bstar=a_{n,p}$ since $\gammash(a_{n,p})=\gammash^0>0$, and we consider the largest $0 < \deltaH \leq\beta_+-a_{n,p}$ and $0<\delta^-\leq a_{n,p}/2$ such that it remains positive on $(a_{n,p}- \delta^-,a_{n,p}+\deltaH)$, that is,
\bel{delta-interval}
\gammash(\bstar)>0 \text{ for all } \bstar\in (a_{n,p}- \delta^-,a_{n,p}+\deltaH) \subset (a_{n,p}/2,\beta_+) .
\ee
Since $g_I(\bstar)$ given in~\eqref{equa-cplusmoins} diverges as $\bstar\to\beta_+$, we learn that $\beta_+ - (a_{n,p}+ \deltaH)$ is bounded below by a positive function of $\gamma_\Kappa,\gamma_\Psi,\gammash^0,C$ and of the exponents $n,p,\beta_{\pm}$.  In particular, $1/(\beta_+-\bstar)$ is uniformly bounded above.
\ese


\paragraph{The Hamiltonian shell identity in the sense of distributions.}

Our main tool in this section is the Hamiltonian shell identity~\eqref{main-func-identity},
\bel{shell-again}
- (\vartheta + a_{n,p}) (\vartheta + 2a_{n,p}) \Phi^\notreH[u]
= - \Chi^\notreH[u] + \Mu^\notreH[u,E] .
\ee
This identity is straightforward to derive for smooth $E$ and~$u$ (cf.~\cite{LL-PoincareKornHardy} for explicit calculations), while we are interested here in a source term $E,\vartheta_* E\in H^{2*}_{n+2-p,-\expoP}(\Omega_R)$ and solution $u,\vartheta u\in H^2_{n-2-p,-\expoP}(\Omega_R)$.
The key step in deriving the shell identity is to multiply the equation $\notreH^\lambda[u]=E$ by $(\vartheta^2 u+a_{n,p}\vartheta u-\cstun u)$ and integrate on a spherical shell~$\Lambda_r$.
In our low-regularity setting, these manipulations can be performed in the sense of distributions on the \emph{open interval} $(R,+\infty)$ as follows (cf.~\appref{appendix=F.3}).

Consider a radial regularization such as $u_\eps(r,\,{\cdot}\,) = \frac{1}{\eps} \int_r^{r+\eps} u(s,\,{\cdot}\,) ds$ for $\eps>0$, with $u_\eps\to u$ and $\vartheta u_\eps\to\vartheta u$ in $H^2_{n-2-p,-\expoP}(\Omega_R)$.
Fix some smooth radial function $\varphi=\varphi(r)$ with compact support in $(R,+\infty)$ (and in particular with $\varphi$ vanishing in a neighborhood of $r=R$).
Given the regularity $\vartheta u\in H^2_{n-2-p,-\expoP}(\Omega_R)$ and the identity $\vartheta u_\eps=r(u(r+\eps) - u(r))/\eps$, one has a convenient regularization of $(\vartheta^2 u+a_{n,p}\vartheta u-\cstun u)$,
\bse
\be
w_\eps \coloneqq r^{-a_{n,p}} \varphi(\vartheta^2 u_\eps+a_{n,p}\vartheta u_\eps-\cstun u_\eps) \in H^2_{n-2-p,-\expoP}(\Omega_R) .
\ee
The variational formulation~\eqref{equa-variH} (cf.~also \eqref{eq-521}) then makes sense for $w=w_\eps$.
One gets the following identity, where derivatives in $\notreH^\lambda[u]$ are understood in the sense of distributions (with no boundary term since $\varphi$ has compact support in~$r$),
\be
0 = \int_R^{+\infty} \fint_{\Lambda_r} \varphi(\vartheta^2 u_\eps+a_{n,p}\vartheta u_\eps-\cstun u_\eps) \bigl( r^4 \notreH^\lambda[u] - r^4 E\bigr) d\chi \, r^{-1} dr .
\ee
\ese
Here, the integral denotes the variational duality pairing with the admissible test field~$w_\eps$, rather than a pointwise product of a distribution with~$w_\eps$. 
When expressed in the variational formulation, the contribution of $\notreH^\lambda[u]$ involves terms $\vartheta^{j_0}\varphi\vartheta^{j_1}\nablaslash^{k_1}u_\eps\vartheta^{j_2}\nablaslash^{k_2}u$ with $j_1+k_1\leq 3$, $j_2+k_2\leq 3$, $k_1\leq 2$, $k_2\leq 2$, and $j_0\leq 1$.  Hence, in the $\eps\to 0$ limit, $(\vartheta^2 u_\eps+a_{n,p}\vartheta u_\eps-\cstun u_\eps)\notreH^\lambda[u]$ converges.  It gives rise to the $\Phi^\notreH$ and $\Chi^\notreH$ terms in~\eqref{shell-again}.
On the other hand, the contribution of~$E$ can be written as $\Mu^\notreH[u_\eps,E]$ in terms of the bilinear operator $\Mu^\notreH$ defined by
\bse\label{MuvE-def}
\be
\la \Mu^\notreH[v,E] , \varphi \ra
\coloneqq \frac{1}{(n-1)\aire[\Lambda,\lambda]} \Bigl(
\bigl\la \vartheta_*(r^4 E), \varphi \vartheta v\bigr\ra
+ \bigl\la r^4 E , (\vartheta v\vartheta\varphi + \cstun v \varphi -a_{n,p}\varphi\vartheta v )\bigr\ra
\Bigr), 
\ee
where $\la\,{\cdot}\,,\,{\cdot}\,\ra$ is the duality pairing of $H^{2*}_{n+2-p,-\expoP}(\Omega_R)$ with $H^2_{-n-2+p,-\expoP}(\Omega_R)$, based on the measure $\lambda^{2\expoP} r^{-n}dx=d\chi\,dr/r$ on~$\Omega_R$.
This definition extends to $\varphi$ with suitably bounded derivatives, thanks to the bound
\bel{MuvEphiW3}
\bigl| \la \Mu^\notreH[v,E] , \varphi \ra \bigr|
\lesssim \|E,\vartheta_* E\|_{H^{2*}_{\astar+4,-\expoP}(\Omega_R)} \|v,\vartheta v\|_{H^2_{n-2-p,-\expoP}(\Omega_R)} \|\varphi\|_{W^{3,\infty}_{-a_{n,p}/2-\astar}(\Omega_R)}
\ee
\ese
involving a Sobolev space whose functions obey $|\vartheta^k\varphi|\lesssim r^{\astar+a_{n,p}/2}$ for $0\leq k\leq 3$.
A consequence of~\eqref{MuvEphiW3} is that $\Mu^\notreH[u_\eps,E]\to\Mu^\notreH[u,E]$ in the sense of distributions, as $\eps\to 0$.
This convergence justifies in our present low-regularity setting the manipulations used to derive the shell identity~\eqref{shell-again}.

 
\subsection{Derivation of the localized integral estimates}
\label{section=7.2}

\paragraph{Strategy.}

We now prove \autoref{thm:decayL2} in multiple steps.
After subtracting a harmonic term from~$u$, we integrate the shell identity to express $\Phi^\notreH[u]$ as the sum of a source term~$\Omega^\Mu$ and of radial integrals of the functionals~$\Psi^\notreH_\beta[u]$ (and less important terms).  The dissipation due to $\Psi^\notreH_\beta[u]$ has a favorable sign, except for a term $\la u\ra^2$ whose effect we control later on in \autoref{lem:gnotreH} (\autoref{section=7.3}).  The source~$\Omega^\Mu$ is controlled in a mollified sense by the source~$E$ (which has fast decay) and the solution itself.
This leads to a control of weighted integrals of $(\norm{u}^\notreH+\norm{\vartheta u}^\notreH)^2$ in a mollified sense, which is equivalent to controlling the integrals themselves, and by integration a control of $\norm{u}^\notreH$ pointwise in~$r$, both strictly better than the power law~$r^{-2\astar}$ as stated in the theorem.  Finally, there remains to identify the harmonic term subtracted from~$u$ with the expected energy term involving~$\mmodu(E)$, and this concludes the proof.

To begin, given the sharp decay exponent $\astar\in[a_{n,p}/2,a_{n,p}+\deltaH)$, we select a larger exponent $\bstar$ within the interval~\eqref{delta-interval} and in the same regime (sub-harmonic or not sub-harmonic) as~$\astar$, namely in the interval $\bigl(\max(a_{n,p}-\delta^-,\astar),a_{n,p}+\cutoff_{\pstar}\deltaH\bigr)$.
In detail, we select $\bstar$ such that $a_{n,p}/2\leq\astar < \bstar < a_{n,p}$ and $a_{n,p}-\delta^- < \bstar$, or such that $a_{n,p} \leq \astar < \bstar < a_{n,p} + \deltaH$.


\paragraph{Subtracting the harmonic contribution.}
At the outset, let us subtract from~$u$ a harmonic term that will eventually be identified with~$\umodu$ when $\astar\geq a_{n,p}$, and consider the unknown
\be
v = u - \frac{1}{\la\nu^\normal\ra} \Chu r^{-a_{n,p}} \nu^\normal .
\ee
At this stage, $\Chu$ is an arbitrary constant.  It arises as a homogeneous term when solving the differential equation obeyed by $\la u\ra$ in \autoref{section=6.3}, and appears as a result in the estimates of certain integrals in \autoref{section=7.3} which we use in the course of our analysis of the shell identity in the present section.  Eventually, equipped with estimates of the solution~$v$ we will identify $\Chu=\mmodu(E)\la\nu^\normal\ra$.
The new unknown obeys $v,\vartheta v\in H^2_{n-2-p,-\expoP}(\Omega_R)$, and since $\nu^\normal\in\ker\ssA^\lambda$, the unknown~$v$ solves $\notreH^\lambda[v]=\notreH^\lambda[u]=E$ with boundary conditions that differ from the Neumann boundary conditions used in the variational formulation~\eqref{equa-opera5}.
It is also straightforward to check that $\Kappa^\notreH[\vt]=\Kappa^\notreH[\ut]$ from the explicit expression of this functional (\autoref{def:KappaH}).
Our main tool to proceed is the shell identity applied to~$v$, combined with the Hamiltonian shell stability condition on~$\lambda$.

\bse\label{Integrationoftheshellidentity}
\paragraph{Integration of the shell identity.}
Our first ingredient for the proof is the shell identity~\eqref{main-func-identity} which we write with critical exponents $\{-a_{n,p},-2\bstar\}$ instead of $\{-a_{n,p},-2a_{n,p}\}$,
\bel{vtvtPhi}
\aligned
& - (\vartheta + a_{n,p}) (\vartheta + 2\bstar) \Phi^\notreH[v]
\\
& = - \Chi^\notreH[v]
+ \Mu^\notreH[v,E]
- (2\bstar - 2a_{n,p}) (\vartheta + a_{n,p}) \Phi^\notreH[v] ,
\quad \text{on } (R,+\infty) ,
\endaligned
\ee
in the sense of distributions (cf.~\appref{appendix=F.3}) on the interval $(R,+\infty)$.
Integrating this equation is subtle.
Firstly, $\Mu^\notreH[v,E]$ was only defined in the sense of distributions in~\eqref{MuvE-def} above.  In \eqref{OmegaMu-def}, below, we construct a distribution $\Omega^\Mu\in\Dcal'((R,+\infty))$ that solves
\bel{abstract-def-OmegaMu}
-(\vartheta+a_{n,p})(\vartheta+2\bstar)\Omega^\Mu = \Mu^\notreH[v,E]
\ee
and we analyze its radial decay.
Secondly, $\Chi^\notreH$~involves fourth derivatives of~$u$, which are only defined in the sense of distributions in a suitable dual Sobolev space.
This functional can be rewritten in terms of the dissipation functionals $\Psi^\notreH_\beta[v] = \Chi^\notreH[v] - (\vartheta+\beta) \Upsilon^\notreH[v]$, which are well-defined functions in $L^1_{a_{n,p}}([R,+\infty))$ under our variational decay and regularity assumptions.\footnote{For linear or quadratic functionals, variational terms are $r^{-a_{n,p}/2}$ and $r^{-a_{n,p}}$, respectively, while harmonic terms are $r^{-a_{n,p}}$ or $r^{-2a_{n,p}}$, respectively.}
With these considerations, the shell identity~\eqref{vtvtPhi} admits as its general solution on $(R,+\infty)$ the distribution (with the notation $I_\beta[f](r)$ and $J_\beta[f](r)$ in~\eqref{IJdef-first}) for some constants $C_0,\Cstar\in\RR$: 
\bel{equa-shell-here}
\aligned
\Phi^\notreH[v] = {} & \Omega^\Phi + \Omega^\Psi + \Omega^\Mu + C_0 r^{-a_{n,p}} + \Cstar r^{-2\bstar} ,
\\
\Omega^\Phi \coloneqq {} & (2\bstar-2a_{n,p}) I_{2\bstar}\bigl[\Phi^\notreH[v]\bigr] , &
\Omega^\Psi & \coloneqq - \frac{J_{a_{n,p}}\bigl[\Psi^\notreH_{a_{n,p}}[v]\bigr] + I_{2\bstar}\bigl[\Psi^\notreH_{2\bstar}[v]\bigr]}{2\bstar-a_{n,p}}
\endaligned
\ee
This representation involves the integral operators $J_{a_{n,p}}$ and $I_{2\bstar}$ applied to $L^1_{a_{n,p}}([R,+\infty))$ functions.
Using the identities~\eqref{IJprop}, that is, $(\vartheta+\beta)I_\beta[f](r) = - (\vartheta+\beta)J_\beta[f](r) = f(r)$, one checks that it is a solution; it is the most general one because the kernel of $(\vartheta+a_{n,p})(\vartheta+2\bstar)$ (applied to distributions) consists of the homogeneous solutions $C_0 r^{-a_{n,p}}+\Cstar r^{-2\bstar}$ only.
\ese

\bse
In view of the term $\Omega^\Psi$ and the fact that $\Psi$, up to a $\la u\ra^2$ contribution, controls the norms $(\norm{v,\vartheta v}^\notreH)^2$, it is natural to introduce the dissipation functional
\bel{equa-Dcal}
\aligned
\Dcal(r) \, & \! \coloneqq
J_{a_{n,p}} \Bigl[ (\norm{v}^\notreH)^2 + (\norm{\vartheta v}^\notreH)^2 \Bigr](r)
+ I_{2\bstar} \Bigl[ (\norm{v}^\notreH)^2 + (\norm{\vartheta v}^\notreH)^2 \Bigr](r)
\\
& = \int_R^{+\infty} \min\Bigl( (s/r)^{a_{n,p}} , (s/r)^{2\bstar} \Bigr) \Bigl( \bigl(\norm{v}^\notreH(s)\bigr)^2 + \bigl(\norm{\vartheta v}^\notreH(s)\bigr)^2 \Bigr) \frac{ds}{s} .
\endaligned
\ee
The kernel in~\eqref{equa-Dcal} is monotone in its second argument, more precisely
\be
2^{-2\bstar}\Dcal(r)\leq\Dcal(r')\leq\Dcal(r),
\qquad r\leq r'\leq2r.
\ee

The next few steps of the proof consist in bounding the terms $\Omega^\Mu$, $\Omega^\Phi$, and~$\Omega^\Psi$ to deduce from~\eqref{equa-shell-here} a bound of the form $\gamma_2\Dcal(r) \leq \gamma_1\Dcal(r)+\#\Dcal(r)^{1/2}r^{-\astar}o(1)+\#r^{-2\astar}o(1)$ where the implicit constants~$\#$ depend on the source term, and where $\gamma_1<\gamma_2$.  This immediately implies the desired conclusion $\Dcal(r)\lesssim \#r^{-2\astar}o(1)$.

Before we reach this fast decay, observe that bounding the weight $\min\bigl( (s/r)^{a_{n,p}} , (s/r)^{2\bstar} \bigr)$ by $(s/r)^{a_{n,p}}$ yields a variational bound
\bel{Dcal-bound}
\Dcal(r) \leq r^{-a_{n,p}} \|v,\vartheta v\|_{H^2_{a_{n,p}/2,-\expoP}(\Omega_R)}^2
\lesssim r^{-a_{n,p}} \Bigl( \|u,\vartheta u\|_{H^2_{a_{n,p}/2,-\expoP}(\Omega_R)}^2 + |\Chu|^2 \Bigr) < +\infty .
\ee
This can be improved to
\bel{Dcalr-smallo}
\lim_{r\to+\infty} r^{a_{n,p}} \Dcal(r) = 0 ,
\ee
by splitting the integral into small and large arguments, for instance as $[R,\sqrt{Rr}]\cup[\sqrt{Rr},+\infty)$.  On the first interval the weight is $(s/r)^{2\bstar}\leq(s/r)^{a_{n,p}}(r/R)^{a_{n,p}/2-\bstar}=o((s/r)^{a_{n,p}})$.  The second interval involves the $H^2_{a_{n,p}/2,-\expoP}$ norms of $v$ and $\vartheta v$ restricted to $\Omega_{\sqrt{Rr}}$, which tend to zero as $r\to+\infty$.
\ese


\bse\label{OmegaPhiPsi-onesided}
\paragraph{Dealing with the terms $\Omega^\Phi$ and $\Omega^\Psi$.}
Given the bound~\eqref{dl2-Phi} on $\Phi^\notreH[v]$, we have
\be
\Omega^\Phi
\leq |2\bstar - 2a_{n,p}| I_{2\bstar} \bigl[ \gamma_\Phi (\norm{v}^\notreH)^2 \bigr]
\leq \gamma_\Phi |2\bstar - 2a_{n,p}| \Dcal(r) .
\ee
For~$\Omega^\Psi$, the shell stability condition~\eqref{equa-conditionH2} restated in~\eqref{dl2-shell} yields an upper bound
\be
\aligned
(2\bstar-a_{n,p}) \Omega^\Psi = - J_{a_{n,p}}\bigl[\Psi^\notreH_{a_{n,p}}[v]\bigr] - I_{2\bstar}\bigl[\Psi^\notreH_{2a_{n,p}}[v]\bigr]
& + (2\bstar-2a_{n,p}) I_{2\bstar}\bigl[\Upsilon^\notreH[v]\bigr]
\\
\leq C (J_{a_{n,p}}+I_{2\bstar})\Bigl[ \la v\ra^2 - \cradialH \bigl( \Kappa^\notreH[v] \bigr)^2 \Bigr]
& + (2\bstar-2a_{n,p}) I_{2\bstar}\bigl[\Upsilon^\notreH[v]\bigr]
- \gammash^0 \Dcal(r) .
\endaligned
\ee
We arrive at a key feature of our argument.
The average $\la v\ra$ obeys a differential equation with $E$ and $\Kappa^\notreH[\vt]=\Kappa^\notreH[\ut]$ as source terms, and we prove in \autoref{lem:gnotreH} (below)\footnote{Compared to the lemma, we found it convenient to replace the separate coefficients $g_J(\bstar)$, $g_I(\bstar)$ by their maximum, and we used that $1/(\beta_+-\bstar)<1/(\beta_+-a_{n,p}-\deltaH)$ is bounded above to control part of the $r^{-2\bstar}$ term by the $r^{-2\astar}$~term.} an integral bound: 
for any $c\in(0,1)$, there exist a constant $\gamma_{JI}(c)>0$, depending on the geometry, the exponents other than~$\bstar$, and~$c$, and a decaying function $o(1)$ bounded by~$1$, such that 
\bel{gnotreH-eq-v}
\aligned
\quad & \unquad (J_{a_{n,p}} + I_{2\bstar})\Bigl[\la v\ra^2 - \frac{1}{1-c} \max(g_J(\bstar), g_I(\bstar)) (\Kappa^\notreH[\vt])^2 \Bigr]
\\
& \leq \gamma_{JI}(c) \Bigl( \bigl( \|E\|_{H^{2*}_{\astar+4,-\expoP}(\Omega_R)} + \cutoff_{\pstar} \mmax(E) \bigr)^2 r^{-2\astar} o(1) + \|\vartheta v\|_{H^2_{n-2-p,-\expoP}(\Omega_R)}^2 r^{-2\bstar} \Bigr)
\endaligned
\ee
where we recall that $g_J$ and $g_I$ were defined in~\eqref{equa-cplusmoins}. 
In terms of $\gammash(\bstar)$ introduced in~\eqref{gammaPsipbstar}, we deduce an upper bound on $\Omega^\Phi+\Omega^\Psi$,
\bel{OmegaPhiPsi-bound}
\aligned
\quad & \unquad (2\bstar-a_{n,p}) (\Omega^\Phi + \Omega^\Psi) + \bigl( \gammash(\bstar) - Cc\bigr) \Dcal(r)
\\
& \leq C \gamma_{JI}(c) \Bigl( \bigl( \|E\|_{H^{2*}_{\astar+4,-\expoP}(\Omega_R)} + \cutoff_{\pstar} \mmax(E) \bigr)^2 r^{-2\astar} o(1) + \|\vartheta v\|_{H^2_{n-2-p,-\expoP}(\Omega_R)}^2 r^{-2\bstar} \Bigr) .
\endaligned
\ee
The dissipation term~$\Dcal(r)$ has a favorable (positive) sign for small enough $c\in(0,1)$ since $\gammash(\bstar)>0$ by construction of the interval $(a_{n,p}-\delta^-,a_{n,p}+\deltaH)$ defined in~\eqref{delta-interval}.
\ese


\bse
\paragraph{Dealing with the source term $\Omega^\Mu$.}
We must construct $\Omega^\Mu$ as a distribution, and bound it in terms of the dissipation~$\Dcal$ and source~$E$.
Let us define $\la \Omega^\Mu , \psi \ra$ for smooth test functions~$\psi$ with compact support in $(R,+\infty)$.
The natural formula arising from integrating by parts the formal solution to~\eqref{abstract-def-OmegaMu} is
\bel{OmegaMu-def}
\la \Omega^\Mu , \psi \ra
\coloneqq \frac{1}{2\bstar-a_{n,p}} \bigl\la \Mu^\notreH[v,E] , \varphi \bigr\ra , \qquad
\varphi \coloneqq J_{-2\bstar}[\psi] + I_{-a_{n,p}}[\psi] .
\ee
Observe that $J_{-2\bstar}[\psi]$ and $I_{-a_{n,p}}[\psi]$ are both smooth, with the first one vanishing in a neighborhood of infinity, and the second one behaving as~$r^{a_{n,p}}$.
In particular, $\varphi$ belongs to $W^{3,\infty}_{-a_{n,p}/2-\astar}(\Omega_R)$ for any $\astar\geq a_{n,p}/2$, which is the space of test functions on which $\Mu^\notreH[v,E]$ was defined in~\eqref{MuvE-def}, above.
Therefore, \eqref{OmegaMu-def} defines a distribution $\Omega^\Mu\in\Dcal'((R,+\infty))$.

Fix once and for all a smooth function $\psi_1:\RR\to[0,+\infty)$ supported on the interval $[1,2]$ and normalized as
\be
\int_1^2 \psi_1(\rho) \frac{d\rho}{\rho} = 1 .
\ee
For any radius $r\in(R,+\infty)$, define a test function $\psi_r:[R,+\infty)\to[0,+\infty)$ supported on $[r,2r]$ by
\be
\psi_r(s)=\psi_1(s/r) , \qquad s\in [R,+\infty) .
\ee
We now bound $\la\Omega^\Mu,\psi_r\ra$ by some (solution-dependent) multiple of $\Dcal(r)^{1/2}r^{-\astar}$.
The definition~\eqref{MuvE-def} of $\Mu^\notreH$ implies
\bel{OmMupsir10}
|\la \Omega^\Mu , \psi_r \ra|
= \frac{| \la \Mu^\notreH[v,E] , \varphi \ra|}{2\bstar-a_{n,p}}
\lesssim \bigl\|E,\vartheta_* E\bigr\|_{H^{2*}_{\astar+4,-\expoP}(\Omega_R)}
\bigl\|\vartheta\varphi\vartheta v, \; \varphi \vartheta v, \; \varphi v\bigr\|_{H^2_{-\astar,-\expoP}(\Omega_R)} .
\ee
Observe that
\be
\aligned
\bigl\|\vartheta\varphi\vartheta v, \; \varphi \vartheta v, \; \varphi v\bigr\|_{H^2_{-\astar,-\expoP}(\Omega_R)}^2
& \lesssim \sum_{k=0}^3 \int_R^{+\infty} |\vartheta^k\varphi|^2 \bigl( (\norm{v}^\notreH)^2 + (\norm{\vartheta v}^\notreH)^2 \bigr) \, s^{-2\astar-1} ds
\\
& \lesssim \Dcal(r) r^{-2\astar} \sup_{s\in[R,+\infty)} \frac{(|\varphi|^2+|\vartheta\varphi|^2+|\vartheta^2\varphi|^2+|\vartheta^3\varphi|^2)(s)}{\min((s/r)^{a_{n,p}+2\astar},(s/r)^{2\bstar+2\astar})} .
\endaligned
\ee
Since $\vartheta\varphi=2\bstar J_{-2\bstar}[\psi_r]+a_{n,p}I_{-a_{n,p}}[\psi_r]$, the derivatives $\vartheta^k\varphi$, $0\leq k\leq 3$, are linear combinations of $J_{-2\bstar}[\psi_r]$, $I_{-a_{n,p}}[\psi_r]$, $\psi_r$, and~$\vartheta\psi_r$.
We bound each of these contributions in turn.  First,
\be
\aligned
\frac{|J_{-2\bstar}[\psi_r](s)|}{\min((s/r)^{a_{n,p}/2+\astar},(s/r)^{\bstar+\astar})}
& \leq \max\Bigl(\!\Bigl(\frac{s}{r}\Bigr)^{2\bstar-a_{n,p}/2-\astar}\!,\Bigl(\frac{s}{r}\Bigr)^{\bstar-\astar}\Bigr) \! \int_s^{+\infty} \! |\psi_r(t)| (t/r)^{-2\bstar} \frac{dt}{t}
\\
& \leq 2^{2\bstar-a_{n,p}/2-\astar} \int_r^{2r} |\psi_r(t)| \frac{dt}{t} \lesssim 1 ,
\endaligned
\ee
since the integral vanishes unless $s\leq 2r$, and since $\psi_r(t)$ vanishes unless $t\in[r,2r]$.
Controlling the remaining contribution in the same way, we deduce that
\bel{OmegaMupsi-bound}
|\la \Omega^\Mu , \psi_r \ra|
\lesssim r^{-\astar} \Dcal(r)^{1/2} \bigl\|E,\vartheta_* E\bigr\|_{H^{2*}_{\astar+4,-\expoP}(\Omega_R)} \|1,\psi_1,\del\psi_1\|_{L^\infty} .
\ee
\ese
This shows that $\Omega^\Mu$ decays as $r^{-\astar}\Dcal(r)^{1/2}$ in a mollified sense.

\bse\label{ImprovedboundonthesourceOmegaMu}
\paragraph{Improved bound on the source~$\Omega^\Mu$.}
The bound~\eqref{OmegaMupsi-bound} can be improved by an $o(1)$ factor as follows.
For $\rho\in(R,r/2)$ we consider a smooth monotonically-increasing radial cutoff function $\kappa_\rho\colon[R,+\infty)\to[0,1]$ with $\kappa_\rho=0$ on $[R,\rho]$ and $\kappa_\rho=1$ on $[2\rho,+\infty)$, and bounded derivatives $\vartheta^k\kappa_\rho$ for each $k\geq 0$.
Given the test function $\psi_r$ defined above, we split $\varphi_r=(J_{-2\bstar}+I_{-a_{n,p}})[\psi_r]$ into the sum of two contributions
\be
\varphi_1 = (1-\kappa_\rho) \varphi_r , \qquad \varphi_2 = \kappa_\rho \varphi_r ,
\ee
which depend on the two parameters $r$ and~$\rho$, as well as their argument in $[R,+\infty)$.
They are such that $\varphi_1$ and $\varphi_2$ are supported on $[R,2\rho]$ and $[\rho,+\infty)$ respectively, and Sobolev norms of $\varphi_1$ and~$\varphi_2$ are bounded by those of~$\varphi_r$.
Given that $2\rho<r$, one has, for $s\in[R,2\rho]$, the bound
\be
\varphi_1(s) = (1-\kappa_\rho(s)) (s/r)^{2\bstar} \int_r^{2r} (r/t)^{2\bstar} \psi_r(t) \frac{dt}{t}
\lesssim (\rho/r)^{2\bstar} ,
\ee
thanks to the favorable power of $s/r$ (the integral is bounded).  Derivatives of $\varphi_1$ are subject to the same bound, hence the contribution of $\varphi_1$ to the bounds \eqref{OmMupsir10}--\eqref{OmegaMupsi-bound} is suppressed by a factor of $(\rho/r)^{2\bstar}$.

The contribution of $\varphi_2$ is suppressed via a different mechanism based on the lack of concentration at infinity in dual Sobolev spaces.
We return to the definition~\eqref{MuvE-def} of~$\Mu^\notreH$.  The products $\varphi v$, $\varphi \vartheta v$, and $\vartheta\varphi\vartheta v$ are all supported on $[\rho,+\infty)$, therefore the bound on $\la \Mu^\notreH[v,E] , \varphi \ra$ used in~\eqref{OmMupsir10} is improved by \autoref{lem:Hstar-cutoff} to any small $\eps>0$ times the $H^2_{-\astar,-\expoP}(\Omega_R)$ norm of these products by taking $\rho=\rho_\eps$ large enough.
The bound \eqref{OmegaMupsi-bound} is then improved to
\be
|\la \Omega^\Mu , \psi_r \ra|
\lesssim r^{-\astar} \Dcal(r)^{1/2} \bigl( (\rho_\eps/r)^{2\bstar} + \eps) = r^{-\astar} \Dcal(r)^{1/2} o(1) ,
\qquad r\to+\infty ,
\ee
where we used that for $r\to+\infty$ one can take $\eps\to 0$ and $\rho_\eps\to+\infty$ while ensuring that $\rho_\eps/r\to 0$.
\ese

\bse
\paragraph{Dealing with the homogeneous terms.}
We now show $C_0=0$ and we bound~$\Cstar$.
Given the bound~\eqref{dl2-Phi} on $\Phi^\notreH$ and~$\Psi^\notreH_\beta$, one has
\bel{OmegaPhiPsi-twosided}
|\Omega^\Phi(r) + \Omega^\Psi(r)|
\leq \Bigl( |2\bstar-2a_{n,p}| \gamma_\Phi + \frac{\bigl(1+|2\bstar-2a_{n,p}|/a_{n,p}\bigr)}{2\bstar-a_{n,p}} \gamma_\Psi \Bigr) \Dcal(r),
\ee
which is a coarse two-sided bound, in contrast to the refined one-sided bound~\eqref{OmegaPhiPsi-onesided}.
Using the monotonicity $\Dcal(r')\leq\Dcal(r)$ for $r'\geq r$, the same bound holds for $|\la\Omega^\Phi+\Omega^\Psi,\psi_r\ra|$.
Furthermore, we have 
\bel{bound-Phi-moll}
\la\Phi^\notreH[v],\psi_r\ra
\leq \gamma_\Phi \|\psi_1\|_{L^\infty} \int_r^{2r} (\norm{v}^\notreH(s))^2 \frac{ds}{s}
\leq \gamma_\Phi \|\psi_1\|_{L^\infty} \Dcal(r) .
\ee
Given that $\Dcal(r)=o(r^{-a_{n,p}})$ (cf.~\eqref{Dcalr-smallo}) and $\astar\geq a_{n,p}/2$, we learn that
\be
\lim_{r\to+\infty} r^{a_{n,p}} \bigl|\bigl\la\Phi^\notreH[v]-\Omega^\Phi-\Omega^\Psi-\Omega^\Mu,\psi_r\bigr\ra\bigr| = 0 .
\ee
Returning to~\eqref{equa-shell-here}, mollified with $\psi_r$, namely
\bel{equa-shell-here-2}
r^{a_{n,p}} \la s^{-a_{n,p}},\psi_r\ra C_0 + r^{a_{n,p}} \la s^{-2\bstar},\psi_r\ra \Cstar
= r^{a_{n,p}} \bigl\la\Phi^\notreH[v](s)-\Omega^\Phi-\Omega^\Psi-\Omega^\Mu,\psi_r\bigr\ra,
\ee
we observe that the first term lies in the interval $[2^{-a_{n,p}},1]C_0$, and the second one decays as $r^{a_{n,p}-2\bstar}$, hence their sum can only tend to zero if $C_0=0$.

To control~$\Cstar$, we simply apply~\eqref{equa-shell-here-2} to some arbitrary radius, say $r=2R$, and use the available bounds \eqref{Dcal-bound}, \eqref{OmegaMupsi-bound}, \eqref{OmegaPhiPsi-twosided}, and~\eqref{bound-Phi-moll}
to get
\bel{Cstar-bound}
|\Cstar| \lesssim \|1,\psi_1,\del\psi_1\|_{L^\infty}^2 \Bigl( \|v,\vartheta v\|_{H^2_{a_{n,p}/2,-\expoP}(\Omega_R)}^2 + \bigl\|E,\vartheta_* E\bigr\|_{H^{2*}_{\astar+4,-\expoP}(\Omega_R)}^2 \Bigr) .
\ee
\ese

\bse\label{Controlofthedissipation}
\paragraph{Control of the dissipation.}
By collecting all of the error terms from~\eqref{equa-shell-here}, using in particular the one-sided bound~\eqref{OmegaPhiPsi-bound}, and letting $\Cstar'>0$ denote a constant arising from~\eqref{Cstar-bound} and the $r^{-2\bstar}$ term in~\eqref{OmegaPhiPsi-bound}, we arrive at
\bel{PhiDcal-bound-1}
\aligned
\quad & \unquad \Phi^\notreH[v] + \frac{1}{2\bstar-a_{n,p}} \bigl( \gammash(\bstar) - Cc \bigr) \Dcal(r)
\\
& \leq \Omega^\Mu + \frac{C \gamma_{JI}(c)}{2\bstar-a_{n,p}} \Bigl( \|E\|_{H^{2*}_{\astar+4,-\expoP}(\Omega_R)} + \cutoff_{\pstar} \mmax(E) \Bigr)^2 r^{-2\astar} o(1)
+ \Cstar' r^{-2\bstar}.
\endaligned
\ee 
Since $\gammash(\bstar)$ is positive by assumption on~$\bstar$, for $c\in(0,1)$ sufficiently small, namely $c<\gammash(\bstar)/C$, the left-hand side consists of two positive terms.
The upper bound involves~$\Omega^\Mu$, which we control only in a mollified sense~\eqref{OmegaMupsi-bound}.
 There exists a constant $\gamma_\Mu>0$, depending on the fixed non-negative test function $\psi_1$ supported on $[1,2]$ used to define $\psi_r(s)=\psi_1(s/r)$, such that
\be
\la \Omega^\Mu , \psi_r \ra
\leq \gamma_\Mu \bigl\|E,\vartheta_* E\bigr\|_{H^{2*}_{\astar+4,-\expoP}(\Omega_R)} \Dcal(r)^{1/2} r^{-\astar} o(1) ,
\ee
for every $r\in(R,+\infty)$.
Using the lower bound $\Dcal(r')\geq2^{-2\bstar}\Dcal(r)$ for $r'\in[r,2r]$, and using the normalization $\int_1^2\psi_1(\rho)d\rho/\rho=1$, the mollified version of~\eqref{PhiDcal-bound-1} is
\be
\aligned
\quad & \unquad \bigl\la\Phi^\notreH[v], \psi_r\bigr\ra
+ \frac{2^{-2\bstar}}{2\bstar-a_{n,p}} \bigl( \gammash(\bstar) - Cc \bigr) \Dcal(r)
\\
& \leq \gamma_\Mu \bigl\|E,\vartheta_* E\bigr\|_{H^{2*}_{\astar+4,-\expoP}(\Omega_R)} \Dcal(r)^{1/2} r^{-\astar} o(1)
\\
& \quad + \frac{C \gamma_{JI}(c)}{2\bstar-a_{n,p}} \Bigl( \|E\|_{H^{2*}_{\astar+4,-\expoP}(\Omega_R)} + \cutoff_{\pstar} \mmax(E) \Bigr)^2 r^{-2\astar} o(1)
+ \Cstar' r^{-2\bstar} .
\endaligned
\ee
Since the shell functional $\Phi^\notreH[v]$ and the test function $\psi_r$ are non-negative, their pairing is non-negative and may be discarded from the left-hand side.
From the remaining terms, we can deduce a bound on dissipation by bounding the term involving $\Dcal(r)^{1/2}$ by a small multiple of~$\Dcal(r)$ and a correspondingly large multiple of $r^{-2\astar}\bigl\|E,\vartheta_* E\bigr\|_{H^{2*}_{\astar+4,-\expoP}(\Omega_R)}^2$.  This yields
\bel{Dcal-decay}
\Dcal(r) \lesssim \Bigl( \|E,\vartheta_* E\|_{H^{2*}_{\astar+4,-\expoP}(\Omega_R)}^2 + \cutoff_{\pstar} \mmax(E)^2 \Bigr) r^{-2\astar} o(1) + \|\vartheta v\|_{H^2_{n-2-p,-\expoP}(\Omega_R)}^2 r^{-2\bstar}
\ee
where the implicit constant depends\footnote{The dependence on~$\psi_1$ can henceforth be ignored since it no longer appears on the left-hand side.} on $c$, on the geometry, and on exponents including~$\bstar$.
\ese


\paragraph{Radially-pointwise control.}

The bound on~$\Dcal(r)$ implies a bound on a simpler notion of dissipation, without weight, thanks to\footnote{In fact, the resulting bound on~$\Dcal_2(r)$, combined with $\Dcal(r) \lesssim \int_R^{+\infty} \min( (s/r)^{a_{n,p}} , (s/r)^{2\bstar}) \Dcal_2(s) \frac{ds}{s}$ implies~\eqref{Dcal-decay} with $r^{-2\bstar}$ replaced by the very slightly weaker bound $r^{-2\bstar}\log r$.}
\bel{simpler-dissip}
\Dcal_2(r)\coloneqq\int_r^{2r} \Bigl( (\norm{v}^\notreH(s))^2 + (\norm{\vartheta v}^\notreH(s))^2 \Bigr) \frac{ds}{s} \leq \Dcal(r) .
\ee
Since $\vartheta$ commutes with the derivatives appearing in the definition of~$\norm{v}^\notreH$, the derivative~$\vartheta\norm{v}^\notreH$ is controlled by~$\norm{\vartheta v}^\notreH$.
Thus, by integrating $f(r)=f(s)+\int_r^s\vartheta f(t)dt/t$ over $s\in[r,2r]$ with $f=(\norm{v}^\notreH)^2$, we obtain
\be
(\norm{v}^\notreH(r))^2
\lesssim \int_r^{2r} \Bigl( \bigl(\norm{v}^\notreH(s)\bigr)^2 + \bigl(\vartheta\norm{v}^\notreH(s)\bigr)^2 \Bigr) \frac{ds}{s} \lesssim \Dcal_2(r) .
\ee

\paragraph{Identifying the harmonic term.}
Altogether, we have established the desired bound~\eqref{equa-bound-int} and its $o(1)$ improvement, but expressed on the shifted solution $v=u-\Chu r^{-a_{n,p}}\nu^\normal/\la\nu^\normal\ra$ instead of $u-\cutoff_{\pstar}\umodu$.
In the sub-harmonic case $\astar<\bstar<a_{n,p}$, we are done since these two functions differ by a harmonic term $\Chu r^{-a_{n,p}}\nu^\normal/\la\nu^\normal\ra$ whose quadratic functionals decay as $r^{-2a_{n,p}}$, faster than both terms on the right-hand side of~\eqref{equa-bound-int}.
In the harmonic and super-harmonic cases we must establish $\Chu=\la\nu^\normal\ra \mmodu(E)$.
For this we integrate the equation $\notreH^\lambda[v]=E$ with the weight $r^{3+a_{n,p}} dr\,d\chi$ on~$\Omega_R$.
The angular integral is performed in~\eqref{equa-average-one-H} above using explicit expressions of the Hamiltonian operator, and reads
\be
\aligned
r^4 \la E \ra = (\vartheta+a_{n,p})\Bigl( (n-1) \vartheta \bigl( \vartheta^2 + a_{n,p} \vartheta - b_{n,p} \bigr) \la u \ra
+ \bigl((n-2) \vartheta-c_{n,p}/a_{n,p}\bigr) \la \Deltaslash \ut \ra \Bigr) .
\endaligned
\ee 
Integrating on $[R,r]$ with the weight $r^{a_{n,p}-1}dr$ yields boundary terms.
The boundary term at~$R$ is simply the average $R^{a_{n,p}} \la\Abb_3^\lambda[u](R)\ra$ of the Hamiltonian boundary operator~$\Abb_3^\lambda$ given in~\eqref{equa-bound-ope-2-lambda}, which vanishes as the natural weak boundary functional encoded by~\eqref{equa-opera5}.

The boundary term at~$r$ is useful rewritten in terms of~$v$,
with the coefficient of $\Chu$ simplifying thanks to the normalization~\eqref{equa-norm-normal} of the silhouette function,
\bel{RrlaE-Chu}
\aligned
\int_R^r \la E(s)\ra s^{3+a_{n,p}} ds
& = r^{a_{n,p}} \Bigl( (n-1) \vartheta \bigl( \vartheta^2 + a_{n,p} \vartheta - b_{n,p} \bigr) \la v\ra + \Bigl((n-2) \vartheta-\frac{c_{n,p}}{a_{n,p}}\Bigr) \la\Deltaslash\vt\ra \Bigr) \\
& \quad + \frac{2(n-1)|\Sphe^{n-1}|}{\aire[\Lambda,\lambda]\la\nu^\normal\ra} \Chu .
\endaligned
\ee
By assumption on the source term, the left-hand side is bounded and as $r\to+\infty$ it tends to $\mmodu(E)$ up to normalization factors.  We learn that the combination of third derivatives of $v$ is bounded and has a finite limit, say~$\ell$.
On the other hand, by \eqref{Dcal-decay}--\eqref{simpler-dissip} and the $o(1)$ improvement mentioned there, all second derivatives of $v$ and $\vartheta v$ decay as $o(r^{-a_{n,p}})$ in a weighted~$L^2$ sense, and more precisely, for any $r_1\in[R,+\infty)$,
\be
\aligned
\quad & \unquad \int_{r_1}^{2r_1} r^{2a_{n,p}} \Bigl( (n-1) \vartheta \bigl( \vartheta^2 + a_{n,p} \vartheta - b_{n,p} \bigr) \la v \ra + \bigl((n-2) \vartheta-c_{n,p}/a_{n,p}\bigr) \la\Deltaslash\vt\ra \Bigr)^2 \frac{dr}{r}
\\
& \lesssim \int_{r_1}^{2r_1} r^{2a_{n,p}} \bigl(\norm{v,\vartheta v}^\notreH\bigr)^2 \frac{dr}{r} \lesssim r_1^{2a_{n,p}} \Dcal_2(r_1) \longrightarrow 0 , \qquad r_1\to+\infty .
\endaligned
\ee
On the other hand, we have just argued that the integrand tends to $\ell^2$ hence the integral tends to $\ell^2\log 2$.  This implies $\ell=0$.
Taking the $r\to+\infty$ limit in~\eqref{RrlaE-Chu} then identifies $\mmodu(E)=\Chu/\la\nu^\normal\ra$ as desired, provided one accounts for the standard normalization of the energy (cf.~\eqref{mmodu-def}).

This concludes the proof of \autoref{thm:decayL2}, modulo the upcoming \autoref{lem:gnotreH} and calculations in \autoref{section=6}.


\subsection{Contribution of the average}
\label{section=7.3}

The average $\la u\ra$ obeys an ordinary differential equation~\eqref{equa-defb-b-b}, derived and solved in \autoref{section=6.3}.
The solution is given in \autoref{prop-66} in the form
\bse\label{avu-early}
\be
\aligned
\la u(r)\ra
& = \Chu r^{-a_{n,p}} + C_+^u r^{-\beta_+} + c_- J_{\beta_-}\bigl[\Kappa^\notreH[\ut]\bigr](r) + c_+ I_{\beta_+}\bigl[\Kappa^\notreH[\ut]\bigr](r) + \Theta^E_1(r) + \Theta^E_2(r) ,
\endaligned
\ee
where $\Kappa^\notreH$ is a linear functional of~$\ut$ involving the derivatives $\vartheta^j \nablaslash^k \ut$ with $j+k\leq 3$ and $k\leq 2$, and
\bel{preJaTEIbTE}
\|\Theta^E_1\|_{L^2_{\astar}([R,+\infty))} \lesssim \|E\|_{H^{2*}_{\astar+4,-\expoP}(\Omega_R)} ,
\qquad
\Theta^E_2 = \begin{cases}
  0 & \text{if } \astar\neq a_{n,p} , \\
  \la\nu^\normal\ra \mmax(E) r^{-a_{n,p}} o(1) & \text{if } \astar=a_{n,p} .
\end{cases}
\ee
\ese

To complete the proof of \autoref{thm:decayL2}, we establish the following technical estimate, which is used in~\eqref{OmegaPhiPsi-onesided} to place a (negative) lower bound on the integrated effect of the dissipation functionals~$\Psi^\notreH_\beta[u]$.

\begin{lemma}[Lower bound on the dissipation]\label{lem:gnotreH}
Fix a localization domain $(\Lambda,\lambda)$ and projection exponent $p\in(p_n^\flat,n-2)$ as in \autoref{thm:decayL2} and assume Hamiltonian harmonic and radial stability (\refwithname{Definitions}{def-harmonic-Hstab} \refwithname{and}{def-radial-Hstab}), but not shell stability.
Fix sharp decay exponents $\astar,\bstar$ such that either $a_{n,p}/2\leq\astar<\bstar<a_{n,p}$ or $a_{n,p}\leq\astar<\bstar<\beta_+$,
consider the variational solution $u \in H^2_{n-2-p, -\expoP}(\Omega_R)$ to the localized Hamiltonian equation~\eqref{equa-opera5} associated with a source term~$E$ enjoying the radial decay~\eqref{equa-hyo-EEE},
and assume the additional regularity $\vartheta u\in H^2_{n-2-p,-\expoP}(\Omega_R)$ as in~\eqref{equa-hyo-EEu}.
Then, there exists a constant $\Chu$ (depending on the solution~$u$) such that for any $c\in(0,1)$,
\bel{gnotreH-eq}
\aligned
\quad & \unquad
(1-c) (J_{a_{n,p}}+I_{2\bstar}) \Bigl[ \bigr(\la u\ra-\Chu r^{-a_{n,p}}\bigr)^2 \Bigr]
- g_J(\bstar) J_{a_{n,p}} \Bigl[ \bigl( \Kappa^\notreH[\ut] \bigr)^2 \Bigr]
- g_I(\bstar) I_{2\bstar} \Bigl[ \bigl( \Kappa^\notreH[\ut] \bigr)^2 \Bigr] \\
& \lesssim \bigl( \|E\|_{H^{2*}_{\astar+4,-\expoP}(\Omega_R)} + \cutoff_{\pstar} \mmax(E) \bigr)^2 r^{-2\astar} o(1) \\
& \quad + \frac{1}{\beta_+-\bstar} \Bigl( \|E\|_{H^{2*}_{\astar+4,-\expoP}(\Omega_R)} + \cutoff_{\pstar} \mmax(E) + \|\vartheta u\|_{H^2_{n-2-p,-\expoP}(\Omega_R)} \Bigr)^2 r^{-2\bstar} ,
\endaligned
\ee
with the shorthand $(J_{a_{n,p}}+I_{2\bstar})[f]\coloneqq J_{a_{n,p}}[f]+I_{2\bstar}[f]$.
The left-hand side here may be (arbitrarily) negative and the implicit constant in the upper bound depends on the localization domain, the constant~$c$, the radius $R$, and the exponents $a_{n,p},\beta_{\pm},\astar$ but not~$\bstar$, and finally $o(1)$ is bounded by~$1$ and decays to zero as $r\to+\infty$.
\end{lemma}

\begin{proof}
\noindent{\it 1. Integrability of the fluctuation operator.}
Throughout the proof we use the short-hand notation $\Kappa \coloneqq \Kappa^\notreH[\ut]$.
As described in~\eqref{equa-defb-b-b} and given explicitly in \autoref{def:KappaH}, $\Kappa$~is linear in the derivatives $\vartheta^j \nablaslash^k \ut$ with $j+k\leq 3$ and $k\leq 2$, with coefficients involving the silhouette function~$\nu^\normal$.
Thus, the additional regularity assumption on $\vartheta u$ ensures that the $\Kappa$~contributions on the left-hand side of~\eqref{gnotreH-eq} are well-defined: first of all,
\bel{JaKappa2defined}
\aligned
J_{a_{n,p}}\bigl[\Kappa^2\bigr](r)
& = r^{-a_{n,p}} \! \int_r^{+\infty} \!\! \Kappa(s)^2 s^{a_{n,p}-1} ds
\\
& \lesssim r^{-a_{n,p}} \Bigl( \|u\|_{H^2_{n-2-p,-\expoP}(\Omega_R)}^2 + \|\vartheta u\|_{H^2_{n-2-p,-\expoP}(\Omega_R)}^2 \Bigr) ,
\endaligned
\ee
with implicit constants that depend on the localization domain $(\Lambda,d\chi)$ and, second, $I_{2\bstar}[\Kappa^2]$ is well-defined since~\eqref{JaKappa2defined} implies local integrability of~$\Kappa^2$.
The inequality~\eqref{gnotreH-eq} can thus be seen as an upper bound on integrals of $(\la u(r)\ra - \Chu r^{-a_{n,p}})^2$ with a precise coefficient for $\Kappa$ contributions and unimportant constants for the $r^{-2\astar}$ and $r^{-2\bstar}$ terms.

\medskip

\bse
\noindent{\it 2. Expression of the average.}
We rely on the evolution equation~\eqref{equa-defb-b-b} given by
\be
- (\vartheta+\beta_-)(\vartheta+\beta_+) \la u\ra = {1 \over \bnotreH_{1}} \Bigl( -(n-2) \vartheta +{c_{n,p} \over a_{n,p}} \Bigr)  \Kappa
+ {1 \over \bnotreH_{1}} \, \Nu^\notreH[E] + {1 \over \bnotreH_{1}} C_1^u r^{-a_{n,p}} ,
\ee
in which $\beta_- < 0 < a_{n,p} <\beta_+$, as stated near~\eqref{betapm-def-text}.
By solving this equation with the absence of a growing contribution $r^{-\beta_-}$ enforced by our variational bounds, we can derive that
(cf.~\eqref{avu-early}, above)
\bel{equa-598}
\la u(r)\ra - \Chu r^{-a_{n,p}} = C_+^u r^{-\beta_+} + \Theta^E + c_- J_{\beta_-}[\Kappa](r) + c_+ I_{\beta_+}[\Kappa](r)
\ee
in terms of the constants $c_\pm$ given in~\eqref{equa-cplusmoins} and some constant $\Chu$, together with a suitable super-harmonic term $C_+^u r^{-\beta_+}$ and source term $\Theta^E=\Theta^E_1+\Theta^E_2$ with 
\be
\|\Theta^E_1\|_{L^2_{\astar}([R,+\infty))} \lesssim \|E\|_{H^{2*}_{\astar+4,-\expoP}(\Omega_R)}
\ee
 and $\Theta^E_2 = \Oneone_{\astar=a_{n,p}} \la\nu^\normal\ra \mmax(E) r^{-a_{n,p}} o(1)$.
Then we deduce a bound on $\la u\ra-\Chu r^{-a_{n,p}}$ by four contributions, namely 
\bel{laurra-up}
\bigl(\la u(r)\ra - \Chu r^{-a_{n,p}}\bigr)^2
\leq
\frac{3}{c} (C_+^u r^{-\beta_+})^2
+ \frac{3}{c} (\Theta^E_1)^2
+ \frac{3}{c} (\Theta^E_2)^2
+ \frac{1}{1-c} \bigl( c_- J_{\beta_-}[\Kappa](r) + c_+ I_{\beta_+}[\Kappa](r) \bigr)^2 ,
\ee
and we are interested in their effect in the left-hand side of the inequality~\eqref{gnotreH-eq}.
\ese

\medskip

\noindent{\it 3. Super-harmonic term contribution.}
The term $C_+^u r^{-\beta_+}$ has super-harmonic decay, and its coefficient is controlled by \autoref{prop-66}.
In addition, a direct computation using that $a_{n,p}/2,\bstar<\beta_+$ shows that
\be
J_{a_{n,p}}[r^{-2\beta_+}] + I_{2\bstar}[r^{-2\beta_+}]
\leq \frac{R^{2\bstar-2\beta_+}}{2(\beta_+-\bstar)} r^{-2\bstar}
+ \frac{1}{2\beta_+ - a_{n,p}} r^{-2\beta_+}
\lesssim \frac{1}{\beta_+-\bstar} r^{-2\bstar}
\ee
since $\bstar<\beta_+$.  The contribution of the super-harmonic term to the left-hand side of~\eqref{gnotreH-eq} is thus bounded by $(\beta_+-\bstar)^{-1} r^{-2\bstar}(C_+^u)^2$. 

\medskip

\bse\label{integralThetaE1}
\noindent{\it 4. Source term contribution.}
Next, we wish to bound $(J_{a_{n,p}}+I_{2\bstar})[(\Theta^E_i)^2]$ for $i=1,2$.  For $\Theta^E_1$, we use $a_{n,p}\leq 2\astar<2\bstar$ to find
\be
\aligned
r^{2\astar} (J_{a_{n,p}} + I_{2\bstar})[(\Theta^E_1)^2]
&= \int_R^{+\infty} \!\!\!\min\bigl((s/r)^{a_{n,p}-2\astar}, (s/r)^{2\bstar-2\astar}\bigr) (\Theta^E_1)^2 s^{2\astar-1} ds
\\
& 
\leq \|\Theta^E_1\|_{L^2_{\astar}([R,+\infty))}^2 ,
\endaligned
\ee
itself bounded by~\eqref{preJaTEIbTE}.  Furthermore, since the weight is bounded above by a function $\min(1,(s/r)^{2\bstar-2\astar})$ that tends pointwise monotonously towards zero, the integral tends to zero, hence gaining an $o(1)$ factor.

In the harmonic case $\astar=a_{n,p}$ we must also consider 
\be
(\Theta^E_2)^2=\la\nu^\normal\ra^2\mmax(E)^2 r^{-2a_{n,p}}o(1).
\ee
By \autoref{lem:appF-pointwise}, which applies thanks to $a_{n,p}<2a_{n,p}<2\bstar$, applying $J_{a_{n,p}}+I_{2\bstar}$ preserves this pointwise decay, namely
\be
r^{2\astar} (J_{a_{n,p}} + I_{2\bstar})[(\Theta^E_2)^2] = \la\nu^\normal\ra^2\mmax(E)^2 o(1) , \qquad \astar=a_{n,p} .
\ee
Altogether (absorbing $\la\nu^\normal\ra$ into the implicit constant), in all regimes we get the following bound, where $o(1)$ is bounded independently of the source term (and decays to zero),
\be
\sum_{i=1,2} \Bigl( J_{a_{n,p}}[(\Theta^E_i)^2] + I_{2\bstar}[(\Theta^E_i)^2] \Bigr)
\lesssim \bigl( \|E\|_{H^{2*}_{\astar+4,-\expoP}(\Omega_R)} + \cutoff_{\pstar} \mmax(E) \bigr)^2 r^{-2a_{n,p}} o(1) .
\ee
\ese

\medskip

\bse
\noindent{\it 5. Contribution of fluctuations.}
We now turn to the last term in~\eqref{laurra-up}.
We apply the Hardy-type inequalities in \refwithname{Lemmas}{lem:rad-Hardy} \refwithname{and}{lem:rad-mixed}. Here, we must pay attention to the {\it specific constants derived in these lemmas.} That is, by recalling $\beta_- < 0 < a_{n,p} < \beta_+$ and  $\bstar < \beta_+$,
we write 
\be
\aligned
& J_{a_{n,p}} \Bigl[ \bigl( c_- J_{\beta_-}[\Kappa] + c_+ I_{\beta_+}[\Kappa] \bigr)^2 \Bigr]
\leq 2 c_-^2 J_{a_{n,p}} \Bigl[\bigl(J_{\beta_-}[\Kappa]\bigr)^2 \Bigr] + 2 c_+^2 J_{a_{n,p}} \Bigl[ \bigl(I_{\beta_+}[\Kappa]\bigr)^2 \Bigr]
\\
& \leq \Big(
{8 c_-^2 \over (a_{n,p} - 2\beta_-)^2} +   \frac{16 c_+^2}{(2 \beta_+ - a_{n,p})^2}  \Big) J_{a_{n,p}}[\Kappa^2]
+  \frac{2 c_+^2 }{(2 \beta_+ - a_{n,p})(\beta_+ - \bstar)} I_{2\bstar}[\Kappa^2]. 
\endaligned
\ee
Treating similarly the $I_{2\bstar}$ term we get a bound that involves the constants defined in~\eqref{equa-cplusmoins},
\be
\aligned
&
J_{a_{n,p}} \Bigl[ \bigl( c_- J_{\beta_-}[\Kappa] + c_+ I_{\beta_+}[\Kappa] \bigr)^2 \Bigr] + I_{2\bstar}\Bigl[ \bigl( c_- J_{\beta_-}[\Kappa] + c_+ I_{\beta_+}[\Kappa] \bigr)^2 \Bigr]
\\
& \leq g_J(\bstar) J_{a_{n,p}}[\Kappa^2] + g_I(\bstar) I_{2\bstar}[\Kappa^2] .
\endaligned
\ee
We recognize the terms subtracted on the left-hand side of~\eqref{gnotreH-eq}, which completes the proof.
\ese
\end{proof} 


\section{Linear estimates for the squared localized momentum operator}
\label{section=8}

\subsection{Variational formulation}
\label{section=8.1}

In this section and the next one, we continue to work within the setup of \autoref{section=5.1} in a truncated cone of the $n$-dimensional Euclidean space (with $n \geq 3$), but we now turn our attention to the linearization of the (squared) localized momentum operator $\notreM^\lambda$. As before, a basic variational formulation (cf.~\autoref{thm-sharp-m-localized-vari}, below) provides us with a preliminary control of the decay of solutions in a mild integral sense, and next we seek sharp estimates in (integral, pointwise) weighted norms. To the operator $\notreM^\lambda$ given in~\eqref{equa:acalew0-deux} we associate the following \textbf{momentum boundary operator} 
\bel{equa-bound-ope-2-lambda-M}
\mathbb{B}^\lambda[Z]_i
\coloneqq 2 \xh_i\vartheta\Zperp + \nablaslash_i\Zperp + \vartheta\Zpar_i - \Zpar_i ,
\ee
expressed in terms of the orthogonal and parallel components defined by $Z_i = \xh_i\Zperp+ \Zpar_i$ with $\xh_i\Zpar_i=0$ (cf.~\eqref{Zi-split}). This is a first-order operator defined on any sphere~$\Sphe_r$ (for $r \geq R$, say) which provides us with a vector-valued analogue of the Neumann boundary operator for the Laplace operator. We point out that, as was the case for the Hamiltonian operator, the variational decay we achieve at this stage is much weaker than the \mbox{(super-)harmonic} decay that we will establish later on in this section.
We define $H^{1*}_{\astar+2,-\expoP}(\Omega_R,\RR^n)$ as the dual of $H^1_{-\astar-2,-\expoP}(\Omega_R,\RR^n)$ under the measure $\lambda^{2\expoP} r^{-n}dx=d\chi\,dr/r$ in the same way as~\eqref{H2star-def}.

\begin{theorem}[Variational formulation for the localized momentum operator]
\label{thm-sharp-m-localized-vari}
Consider a conical domain $\Omega_R = K \cap {}^{\complement} \Ball_R \subset \RR^n$ together with a localization function $\lambda\colon \Lambda \to (0, \lambda_0]$ with connected support  $\Lambda \subset \Sphe^{n-1}$. Fix some arbitrary localization exponent~$\expoP \geq 2$ and consider a projection exponent\footnote{The restriction based on the lower bound $p_n^\flat$ is unnecessary here.} $p \in (0, n-2)$. Given any vector field $F \in H^{1*}_{n-p,-\expoP}(\Omega_R,\RR^n)$, there exists a unique variational solution $Z \in H^1_{n-2-p, -\expoP}(\Omega_R,\RR^n)$ to the localized second-order momentum system
\bel{equa-opera5-M} 
\aligned
\notreM^\lambda[Z]^i 
& = F^i  \quad &&(1 \leq i \leq n) 
&& \text{ in the exterior domain } \Omega_R = {}^\complement \Ball_R \cap K,
\qquad 
\\
\mathbb{B}^\lambda [Z]^i 
& = 0 \quad &&(1 \leq i \leq n) 
&& \text{ on the subset of the sphere } \Lambda_R = \Sphe_R \cap K.  
\endaligned
\ee
Moreover, $\|Z\|_{H^1_{n-2-p, -\expoP}(\Omega_R,\RR^n)} \lesssim \|F\|_{H^{1*}_{n-p,-\expoP}(\Omega_R,\RR^n)}$.
\end{theorem}

\begin{proof}
  We determine the variational formulation of interest by considering a sufficiently regular solution $Z$ to~\eqref{equa-opera5-M} and a smooth test vector field~$W$ on~$\Omega_R$.  We write
\[
\aligned 
& \int_{\Omega_R} W_j \, F^j \, r^{n-2-2p} \lambda^{2\expoP} dx 
= \int_{\Omega_R} W_j \, \notreM^\lambda_{n,p}[Z]^j  \, r^{n-2-2p} \lambda^{2\expoP} dx 
\\ 
& = - \frac{1}{2} \int_{\Omega_R} W_i \Bigl( \del_j \bigl( r^{n-2-2p} \lambda^{2\expoP}  \del_j Z_i \bigr) + \del_j \bigl( r^{n-2-2p} \lambda^{2\expoP}  \del_i Z_j \bigr) \Bigr) \, dx 
\\ 
& = \frac{1}{2}  \int_{\Omega_R} ( \del_j W_i ) \bigl( \del_j Z_i + \del_i Z_j \bigr) \, r^{n-2-2p} \lambda^{2\expoP} dx 
+ \frac{1}{2} R^{n-3-2p} \int_{\Lambda_R} W_i x_j \bigl(\del_j Z_i + \del_i Z_j \bigr) \lambda^{2\expoP} dx. 
\endaligned
\]
Hence the variational formulation reads 
\bel{equa-variM}
\aligned 
\frac{1}{4}  \int_{\Omega_R}\big(\del_j W_i + \del_i W_j  \bigr)  \bigl(\del_j Z^i + \del_i Z^j  \bigr) \, r^{n-2-2p} \lambda^{2\expoP} dx 
= \la F, r^{2n-2-2p}W\ra ,
\endaligned
\ee 
provided we impose the following boundary condition at $r=R$,
\[
\aligned
0 & = x_j(\del_j Z_i + \del_i Z_j)
= \vartheta Z_i - Z_i + \del_i(r\Zperp)
\\
& = (\vartheta - 1) (\xh_i\Zperp + \Zpar{}_i) + \bigl(\xh_i\Zperp + (\xh_i\vartheta + \nablaslash_i)\Zperp\bigr)
= 2 \xh_i\vartheta\Zperp + \nablaslash_i\Zperp + \vartheta\Zpar_i - \Zpar_i ,
\endaligned
\]
which leads us to the boundary operator~\eqref{equa-bound-ope-2-lambda-M}. Consequently, for $F\in H^{1*}_{n-p,-\expoP}(\Omega_R,\RR^n)$, the variational solution $Z \in H^1_{n-2-p, -\expoP}(\Omega_R,\RR^n)$ is defined by requiring that~\eqref{equa-variM} holds for all $W \in H^1_{n-2-p, -\expoP}(\Omega_R,\RR^n)$, where the right-hand side of~\eqref{equa-variM} is the duality bracket between $F$ and $r^{2n-2-2p}W\in H^1_{-n+p,-\expoP}(\Omega_R,\RR^n)$. Standard continuity and coercivity properties for the linearized momentum then allow us to establish the existence of the variational solution to the problem~\eqref{equa-opera5-M}. The coercivity is a consequence of the localized weighted Korn inequality in~$\Omega_R$ (cf.~\autoref{appendix=D}).
\end{proof}
 

\subsection{Localized integral estimates}
\label{section=8.2}

From our shell identity and momentum stability conditions, we derive Sobolev bounds on the solution~$Z$, involving the exponent $\astar = n - 2 - 2p + \pstar \geq a_{n,p}/2$
and the notation $\cutoff_{\pstar}=1$ when $\astar\geq a_{n,p}$ and otherwise $\cutoff_{\pstar}=0$.
For suitable sources~$F$, the momentum modulator $\Jmodu(F)$ and $\Jmax(F)$ may be defined as
\bel{Jmodu-def}
\compresseq{.3}
\aligned
\Jmodu(F;r_1,r_2) & \coloneqq
\frac{1}{(n-1)|\Sphe^{n-1}|} \int_{\Omega_{r_1}\setminus\Omega_{r_2}} \!\!\! F \, r^{n-2p-2} \lambda^{2\expoP} dx
= \frac{\aire[\Lambda,\lambda]}{(n-1)|\Sphe^{n-1}|} \int_{r_1}^{r_2} \la F(s)\ra \, s^{1+a_{n,p}} ds ,
\\[1ex]
\Jmodu(F) & \coloneqq \lim_{r_2 \to +\infty} \Jmodu(F;R,r_2) ,
\qquad\qquad
\Jmax(F) \coloneqq \sup_{r_1\geq R} \Bigl| \lim_{r_2\to+\infty} \Jmodu(F;r_1,r_2) \Bigr| .
\endaligned
\ee

\begin{theorem}[Integral estimates for the localized momentum operator]
\label{thm:decayL2-M}
Consider a conical domain $\Omega_R = K \cap {}^{\complement} \Ball_R \subset \RR^n$ together with a localization function $\lambda\colon \Lambda \to (0, \lambda_0]$ with connected support $\Lambda \subset \Sphe^{n-1}$ and some ${\expoP \geq 2}$. Fix a projection exponent\footnote{In contrast to \autoref{thm:decayL2}, the lower bound $p>p_n^\flat$ is unnecessary since the momentum operator does not feature~$c_{n,p}$.} $p\in (0,n-2)$ and assume that the localization function is momentum-stable  (cf.~\autoref{def-shell-Mstab}). For some sufficiently small $\deltaM>0$, given a sharp decay exponent $\astar \in [a_{n,p}/2, a_{n,p}+\deltaM)$, consider the variational solution $Z \in H^1_{n-2-p, -\expoP}(\Omega_R,\RR^n)$ to the localized momentum equation~\eqref{equa-opera5-M} with a source term $F \in H^{1*}_{n-p,-\expoP}(\Omega_R,\RR^n)$ satisfying the radial decay
\bel{defJmoduZ}
\aligned
& F \in H^{1*}_{\astar+2,-\expoP}(\Omega_R,\RR^n) ,
\\
\text{when $\astar\geq a_{n,p}$,} \quad & \Jmodu(F;R,r) \text{ is bounded for $r\geq R$, and has an $r\to+\infty$ limit $\Jmodu(F)$.}
\endaligned
\ee
With the notation $\Zmodu\coloneqq \Jmodu(F)_j \xi^{\normal (j)} r^{-a_{n,p}}$ when $\astar\geq a_{n,p}$,
the variational solution enjoys pointwise and integral radial decay
\bel{equa-bound-int-M}
\aligned
\quad & \unquad \| Z - \cutoff_{\pstar} \Zmodu\|_{\unL^2_{-\expoP}(\Lambda,\RR^n)}^2
+ \int_r^{2r} \Bigl( \|\vartheta(Z - \cutoff_{\pstar} \Zmodu)\|_{\unL^2_{-\expoP}(\Lambda,\RR^n)}^2 + \|Z - \cutoff_{\pstar} \Zmodu\|_{\unH^1_{-\expoP}(\Lambda,\RR^n)}^2 \Bigr) \frac{ds}{s}
\\
& = o(1) \bigl( \|F\|_{H^{1*}_{\astar+2,-\expoP}(\Omega_R,\RR^n)} + \cutoff_{\pstar} \Jmax(F) \bigr)^2 r^{-2\astar} , \qquad r \geq R,
\endaligned
\ee
where $o(1)$ is bounded by constants depending on the exponents and the geometry (and not the source~$F$) and tends to zero as $r\to+\infty$ at a rate that may depend on~$F$.
\end{theorem} 

\begin{remark}
As explained for the energy in \autoref{rem:decay-reg}, the angular average $r^2\la F\ra$ belongs to $H^{1*}_{\astar}([R,+\infty))$ hence its radial integral $\Jmodu(F;R,{\,\cdot\,})$ is in $L^2_\gamma([R,+\infty))$ for all $\gamma<\astar-a_{n,p}$.  In the (super-)harmonic regime, \eqref{defJmoduZ} imposes that $\Jmodu(F;R,{\,\cdot\,})$ is bounded (regularity condition) and has a finite limit (decay condition).
\end{remark}

\paragraph{Proof of \autoref{thm:decayL2-M}: outline.}

This theorem is proven in the rest of the section, following mostly the same strategy as in \autoref{section=7.2} for the Hamiltonian operator with two main differences.
Firstly, averages of~$Z$ obey an algebraic (instead of differential) equation.  As in \autoref{section=7.2}, we shift~$Z$ by silhouette vector fields to eliminate an undetermined harmonic term.  The new unknown~$Y$ has averages expressed in terms of~$F$ and of fluctuation operators $\Kappa^{\notreM(j)}[Y]$.
Secondly, in contrast with the Hamiltonian where an additional regularity condition $\vartheta u\in H^2$ was necessary, we observe that all terms in the momentum shell identity are easily defined as distributions, and we integrate this identity applied to~$Y$.
Thanks to shell stability, we deduce an inequality of the form
\be
\Dsym(r) \lesssim |\bstar-a_{n,p}| \Dtot(r) + \Dtot(r)^{1/2} \|F\| r^{-\astar} o(1) + \|F\|^2 r^{-2\astar} o(1) ,
\ee
for a suitable dual Sobolev norm of the source term~$F$, where $\Dsym(r)$ and $\Dtot(r)$ are two \emph{different} notions of dissipation, defined momentarily: a weighted $L^2$~norm of $\Sym(\nabla Y)$, and a weighted $H^1$~norm of~$Y$.
Namely, the dissipation functionals $\Psi^\notreM_\beta[Y]$ only control some of the derivatives of~$Y$, leading to~$\Dsym(r)$, while error terms involve a larger dissipation~$\Dtot(r)$.
This complication compared to the Hamiltonian case is dealt with thanks to a Korn--Hardy inequality, which we explain before starting the proof.

At the outset, we choose an exponent $\bstar>\astar$ in the same regime (sub-harmonic or not) as~$\astar$.  In addition, we require in~\eqref{ineq-Dsymtot-F} below that $|\bstar-a_{n,p}|<\deltaM$ for some sufficiently small $\deltaM>0$: this can be done provided $\astar<a_{n,p}+\deltaM$.  Explicitly, we choose
\be
\aligned
\bstar& \in\bigl(\max(\astar,a_{n,p}-\deltaM), \; a_{n,p}\bigr)
& & \text{if } \astar < a_{n,p} ,
\\
\bstar& \in\bigl(\astar, \; a_{n,p}+\deltaM)\bigr)
& & \text{if } \astar \geq a_{n,p} .
\endaligned
\ee


\paragraph{Two equivalent dissipations.}

For a vector field $Y\in H^1_{n-2-p,-\expoP}(\Omega_R,\RR^n)$, later identified with a shift of~$Z$, we introduce the dissipations
\bel{equa-Dcal-M}\compresseq{.1}
\aligned
\Dtot & \coloneqq (J_{a_{n,p}} + I_{2\bstar})[\dtot] ,
& \dtot(r) & \coloneqq \|\vartheta Y,\nablaslash Y\|_{\unL^2_{-\expoP}(\Lambda_r)}^2 + \|Y\|_{\unL^2_{-\expoP+1}(\Lambda_r)}^2,
\\
\Dsym(r) & \coloneqq (J_{a_{n,p}} + I_{2\bstar})[\dsym] ,
& \dsym(r) & \coloneqq \bigl\|\vartheta\Yperp\!,\vartheta\Ypar+\nablaslash\Yperp\!,\Sym(\nablaslash\Ypar),\Yperp\!,\Ypar\bigr\|_{\unL^2_{-\expoP}(\Lambda_r)}^2 .
\endaligned
\ee
(Observe that in $\Dtot$ the $\unL^2_{-\expoP+1}$~norm is stronger than~$\unL^2_{-\expoP}$.)
It is also useful to note that
\bel{JI8-observation}
(J_{a_{n,p}} + I_{2\bstar})[f](r)
= \int_R^{+\infty} \min\bigl( (s/r)^{a_{n,p}} , (s/r)^{2\bstar} \bigr) f(s) \frac{ds}{s} .
\ee
The symmetric dissipation is bounded as $\Dsym(r) \leq r^{-a_{n,p}} \|Y\|_{H^1_{n-2-p,-\expoP}(\Omega_R,\RR^n)}^2$ by the variational norm, and as explained for the Hamiltonian in~\eqref{Dcalr-smallo}, splitting the integral into small and large~$s$ improves the decay rate to
\bel{Dsym-o1}
\Dsym(r) = r^{-a_{n,p}} \|Y\|_{H^1_{n-2-p,-\expoP}(\Omega_R,\RR^n)}^2 o(1) , \qquad r\to+\infty .
\ee
The following result shows that the total dissipation $\Dtot(r)$ is well-defined and admits the same decay.

\begin{lemma}[Weighted Korn--Hardy inequality]
\label{lem:Korn7}
The dissipations obey $\Dsym(r) \simeq \Dtot(r)$, with equivalence constants that are uniform for $\bstar$ in a compact subinterval of $(0,+\infty)$.
\end{lemma}

\begin{proof}
\noindent{\it 1. Global Korn--Hardy inequality.}
The bound $\Dsym(r)\leq 2\Dtot(r)$ is obvious, with a factor of~$2$ due to $|\vartheta\Ypar+\nablaslash\Yperp|^2\leq 2|\vartheta\Ypar|^2+2|\nablaslash\Yperp|^2$.
For the other direction, the Korn--Hardy inequality in \autoref{lem:PoincareKornHardyD} controls in an integral sense $\lambda^{-1}Y$ and $r\del Y$ (namely all of~$\dtot$) by the symmetrized gradient $\Sym(\del Y)_{ij}=\frac{1}{2}(\del_i Y_j+\del_j Y_i)$ in~$\RR^n$.  The latter decomposes as
\be
r^2 |\Sym(\del Y)|^2
= (\vartheta\Yperp)^2
+ (1/2) \bigl|(\vartheta-1)\Ypar+\nablaslash\Yperp\bigr|^2
+ \bigl|\Sym(\nablaslash\Ypar)+\gslash\Yperp\bigr|^2
\lesssim \dsym .
\ee
Thus in our context the Korn--Hardy inequality reads as follows for any $q>0$ (provided $\expoP>1$),
\bel{dtotL2q}
\int_R^{+\infty} \dtot[Y](s) s^{2q} \frac{ds}{s}
\lesssim \int_R^{+\infty} \dsym[Y](s) s^{2q} \frac{ds}{s} .
\ee

\noindent{\it 2. Localizing the dissipation bound.}
We now wish to splice together the bounds for $q=a_{n,p}/2$ and $q=\bstar$ to get the weight in~\eqref{JI8-observation}.
Consider a smooth monotonic function $\kappa\colon(0,+\infty)\to[0,1]$ with $\kappa(\rho)=0$ for $\rho\leq 1$ and $\kappa(\rho)=1$ for $\rho\geq 2$.  Given a vector field $Y$ and a radius $r\geq R$, we define $Y_1(s)=\kappa(s/r)Y(s)$, which is supported on $[r,+\infty)$ and coincides with $Y$ on $[2r,+\infty)$.  By the triangle inequality the dissipation of~$Y$ is bounded by that of $Y_1$ and $Y_2\coloneqq Y-Y_1$, which we bound using~\eqref{dtotL2q} and the fact that in the interval $s\in[r,2r]$ the factors $(s/r)^{a_{n,p}}$ and $(s/r)^{2\bstar}$ are equivalent:
\be
\aligned
\Dtot[Y] & \leq 2\Dtot[Y_1] + 2\Dtot[Y_2]
\\
& \leq 2 r^{-a_{n,p}} \int_R^{+\infty} \dtot[Y_1](s) s^{a_{n,p}} \frac{ds}{s}
+ 2 r^{-2\bstar} \int_R^{+\infty} \dtot[Y_2](s) s^{2\bstar} \frac{ds}{s}
\\
& \lesssim r^{-a_{n,p}} \int_R^{+\infty} \dsym[Y_1](s) s^{a_{n,p}} \frac{ds}{s} + r^{-2\bstar} \int_R^{+\infty} \dsym[Y_2](s) s^{2\bstar} \frac{ds}{s}
\\
& \lesssim \Dsym[Y_1] + \Dsym[Y_2] .
\endaligned
\ee
To finish the proof, we note that the derivatives $s\del_s\kappa(s/r)=(\rho\del_\rho\kappa)|_{\rho=s/r}$ and $\nablaslash\kappa(s/r)=0$ are uniformly bounded, and therefore $\dsym[Y_1]+\dsym[Y_2]\lesssim\dsym[Y]$ and likewise for the integrated dissipations~$\Dsym$.
All constants used above are locally uniform in the exponent $q>0$, leading to the stated uniformity in~$\bstar$.
\end{proof}


\paragraph{Shifting the solution by a harmonic term.}

Recall the positive-definite structure matrix with components $T_{kl}=\delta_{kl}+\la\xh_k\xh_l\ra$.
Jumping ahead somewhat, we introduce the following notation for any vector field (cf.~\eqref{equa-fluct-vectors}--\eqref{equa--813} below for a detailed motivation)
\be
\projP^k(Z) \coloneqq (T^{-1})^{kl} \, \bigl\la 2 \xh_l \Zperp + \Zpar{}_l \bigr\ra,
\qquad
\matQ^{(j)k} \coloneqq \projP^k(\xi^{\normal(j)}) = (T^{-1})^{kl} \, \bigl\la 2 \xh_l \xi^{\normal (j) \perp} + \xi^{\normal (j) \parallel}{}_l \bigr\ra .
\ee
In \autoref{section=9.3}, below, we integrate the differential equation $\notreM^\lambda[Z]=F$ against elements of the co-kernel and kernel of~$\ssB^\lambda$ and deduce a second-order differential equation~\eqref{Zave-diffeq} for the averages~$\projP^k(Z)$.  Its solution in \autoref{prop-87-moment},
\bel{sol-Z}
(\Xi^\notreM)^{(j)}{}_{k} \projP^k(Z)
- \Kappa^{\notreM (j)}[Z^\fluc]
= \Theta^{F(j)} + \ChZ{}^{(j)} r^{-a_{n,p}} ,
\ee
involves homogeneous harmonic terms with $n$ integration constants~$\ChZ{}^{(j)}$.
The same identity holds with $Z$~replaced by one of the harmonic kernels $\xi^{\normal(i)} r^{-a_{n,p}}$, without source term since $\notreM^\lambda[\xi^{\normal(i)}r^{-a_{n,p}}]=0$, and we evaluate in~\eqref{aveKappa-silhouette} the homogeneous term.  This results in
\be
(\Xi^\notreM)^{(j)}{}_{k} \matQ^{(i)k} r^{-a_{n,p}} - \Kappa^{\notreM (j)}[\xi^{\normal(i)\fluc}r^{-a_{n,p}}]
= (\Xi^\notreM)^{(i)}{}_{k} \matQ^{(j)k} r^{-a_{n,p}} .
\ee
Upon subtracting from~\eqref{sol-Z} a linear combination of this identity with $n$~constants~$v_{(i)}$ we find that
\bel{Y-of-Z-def}
Y \coloneqq Z - \bigl(\ChZ(\Xi^\notreM\matQ^t)^{-1}\bigr)_{(i)} \xi^{\normal(i)} r^{-a_{n,p}}
\quad \text{obeys} \quad
(\Xi^\notreM)^{(j)}{}_{k} \projP^k(Y) - \Kappa^{\notreM (j)}[Y] = \Theta^{F(j)} .
\ee
The choice of $v_{(i)}$ involves inverses of the matrices $\Xi^\notreM$ and~$\matQ$.
Radial stability states that $\Xi^\notreM$ is invertible.
Harmonic stability~\eqref{equa-stable-M-414} ensures that the matrix with components $\matQ^{(i)l}T_{kl}=\la 2\xh_k\xi^{\normal(i)\perp}+\xi^{\normal(i)\parallel}{}_k\ra$ is also invertible, as otherwise a non-zero linear combination~$\xi$ of the silhouette vector fields would have vanishing averages $\la 2\xh_k\xiperp+\xipar{}_k\ra$ and vanishing $\ssrmB^\lambda[\xi]$, yet non-zero norm, in contradiction with~\eqref{equa-stable-M-414}.


\paragraph{Integrating the shell identity.}

We consider next the momentum shell identity~\eqref{main-func-identity-MM} for the unknown $Y\in H^1_{n-2-p,-\expoP}(\Omega_R,\RR^n)$, with a source $F\in H^{1*}_{\astar+2,-\expoP}(\Omega_R,\RR^n)$, and integrate it following~\eqref{Integrationoftheshellidentity}.  There exist constants $C_0,\Cstar\in\RR$ such that
\bel{equa-shell-mom}\compresseq{.7}
\aligned
\Phi^\notreM[Y] = {} & \Omega^\Phi + \Omega^\Psi + \Omega^\Mu + C_0 r^{-a_{n,p}} + \Cstar r^{-2\bstar} ,
\\
\Omega^\Phi \coloneqq {} & (2\bstar-2a_{n,p}) I_{2\bstar}\bigl[\Phi^\notreM[Y]\bigr] ,
\endaligned
\quad\ \
\Omega^\Psi \coloneqq - \frac{J_{a_{n,p}}\bigl[\Psi^\notreM_{a_{n,p}}[Y]\bigr] + I_{2\bstar}\bigl[\Psi^\notreM_{2\bstar}[Y]\bigr]}{2\bstar-a_{n,p}} ,
\ee
with $\Omega^\Mu$~defined as follows.
The source functional $\Mu^\notreM[Y,F]=2r^2\la Y\cdot F\ra$ is a distribution, and we apply the solution operator $J_{a_{n,p}}+I_{2\bstar}$ to it by treating the weight in~\eqref{JI8-observation} as a test function: we set
\be
\aligned
\Omega^\Mu(r) & = \la \Mu^\notreM[Y,F], \varphi_r \ra = 2 \la F , \uY \ra ,
\\
\varphi_r(s) & = \frac{\min( (s/r)^{a_{n,p}} , (s/r)^{2\bstar})}{2\bstar-a_{n,p}} ,
\qquad
\uY(x) = |x|^2 \varphi_r(|x|) Y(x) .
\endaligned
\ee
The vector field $\uY$ is in an appropriate~$H^1$ space since
\be
\aligned
\|\uY\|_{H^1_{-\astar-2,-\expoP}(\Omega_R,\RR^n)}^2
& \lesssim \int_R^{+\infty} \bigl(\varphi_r(s)^2 + (s\del_s\varphi_r(s))^2\bigr) \dtot(s) s^{-2\astar-1} ds
\\
& \lesssim r^{-2\astar} \int_R^{+\infty} \min((s/r)^{2a_{n,p}-2\astar}, (s/r)^{4\bstar-2\astar}) \dtot(s) s^{-1} ds
\\
& \lesssim r^{-2\astar} \Dtot(r) ,
\endaligned
\ee
where we used that $2a_{n,p}-2\astar\leq a_{n,p}<2\bstar<4\bstar-2\astar$.
Therefore
\bel{mom-OmegaMu}
|\Omega^\Mu(r)| \lesssim r^{-\astar} \Dtot(r)^{1/2} \|F\|_{H^{1*}_{\astar+2,-\expoP}(\Omega_R,\RR^n)} .
\ee
This is easily improved by an $o(1)$ factor following the same steps as in the Hamiltonian case~\eqref{ImprovedboundonthesourceOmegaMu}, because of lack of concentration at infinity in $H^{1*}$ and because $\bstar>\astar$.


\paragraph{Controlling the homogeneous terms.}

By momentum shell stability (\autoref{def-shell-Mstab}), there exist constants $\gamma_\Phi,\gamma_\Psi>0$ such that for $\beta\in\{a_{n,p},2a_{n,p}\}$,
\be
0 \leq \Phi^\notreM[Y] \leq \gamma_\Phi \|Y\|_{\unL^2_{-\expoP}(\Lambda)}^2 , \qquad
\bigl|\Psi^\notreM_{a_{n,p}}[Y]\bigr| + \bigl|\Upsilon^\notreM[Y]\bigr| \leq \gamma_\Psi \bigl( \|\vartheta Y\|_{\unL^2_{-\expoP}(\Lambda)}^2 + \|Y\|_{\unH^1_{-\expoP}(\Lambda)}^2 \bigr) .
\ee
In particular,
\bel{mom-OmegaPhi}
\Omega^\Phi \leq |2\bstar-2a_{n,p}| \gamma_\Phi \Dtot , \qquad
|\Omega^\Phi| + |\Omega^\Psi| \lesssim \Dtot .
\ee
Thanks to \eqref{mom-OmegaMu}, \eqref{mom-OmegaPhi}, $\astar\geq a_{n,p}/2$ and $\Dtot=o(r^{-a_{n,p}})$ (from \eqref{Dsym-o1} and \autoref{lem:Korn7}), all terms in~\eqref{equa-shell-mom} decay pointwise faster than $r^{-a_{n,p}}$ except $C_0 r^{-a_{n,p}}$, therefore $r^{a_{n,p}}\Phi^\notreM[Y]\to C_0$.  The integrability of $\Phi^\notreM[Y]r^{a_{n,p}-1}dr$ then forces $C_0$ to vanish.
Evaluating \eqref{equa-shell-mom} at some arbitrary radius, say $r=2R$ and using these same coarse bounds leads to a control of the other homogeneous term~$\Cstar$ by $\Dtot(2R)$ hence by $\Dsym(2R)$, controlled by~\eqref{Dsym-o1} by the variational norm of~$Y$ and by the source~$F$:
\be
C_0 = 0 , \qquad |\Cstar| \lesssim \|Y\|_{H^1_{a_{n,p}/2,-\expoP}(\Omega_R,\RR^n)}^2 + \|F\|_{H^{1*}_{\astar+2,-\expoP}(\Omega_R,\RR^n)}^2 .
\ee

\paragraph{Shell Korn functional.}

We will use momentarily the shell stability condition~\eqref{equa-conditionM2}, which involves an additional term $\shellKorn[Z]$, whose radial integral $(J_{a_{n,p}}+I_{2\bstar})\bigl[\shellKorn[Z]\bigr]$ we consider now.  Specifically we seek a lower bound.
We split the radial integral into dyadic intervals $[2^jR,2^{j+1}R)$ for $j\geq 0$. Denoting by $\mu_j^{\pm}$ the maximum and minimum of the radial weight $\min( (s/r)^{a_{n,p}} , (s/r)^{2\bstar})$ for $s$ in the dyadic interval, observe that $\mu_j^+/\mu_j^-\leq 2^{2\bstar}$.
Recall then that $\shellKorn[Z]$ is the difference of two terms~\eqref{def-Korn-shell}.
The negative contribution is controlled by noting first that the minimum over ambient Killing vectors~$\zeta_s$ independently for each~$s$ is bounded above by the minimum over a single~$\zeta$ throughout the dyadic interval, and then applying the ambient Korn inequality~\eqref{ambient-Korn}:
\be
\aligned
& 2^{-2a_{n,p}} (\CKornM)^{-2} \int_{2^j R}^{2^{j+1} R} \min_{\zeta_s} \|\Zpar - \zeta_s^\parallel\|_{\unH^1_{-\expoP}(\Lambda_s)}^2 \, \min( (s/r)^{a_{n,p}} , (s/r)^{2\bstar}) \frac{ds}{s}
\\
& \leq 2^{-2a_{n,p}} (\CKornM)^{-2} \mu_j^+ \min_\zeta \biggl( \int_{2^j R}^{2^{j+1} R} \|\Zpar - \zeta^\parallel\|_{\unH^1_{-\expoP}(\Lambda_s)}^2 \, \frac{ds}{s} \biggr)
\\
& \leq 2^{-2a_{n,p}} \mu_j^+
\int_{2^j R}^{2^{j+1} R} \bigl\|\Sym(\del Z)\bigr\|_{\unL^2_{-\expoP}(\Lambda_s)}^2 \, \frac{ds}{s}
\\
& \leq 2^{2\bstar-2a_{n,p}}
\int_{2^j R}^{2^{j+1} R} \bigl\|\Sym(\del Z)\bigr\|_{\unL^2_{-\expoP}(\Lambda_s)}^2 \, \min( (s/r)^{a_{n,p}} , (s/r)^{2\bstar}) \frac{ds}{s} .
\endaligned
\ee
Summing over $j\geq 0$, this results in the lower bound
\be
\aligned
\quad & \unquad (J_{a_{n,p}}+I_{2\bstar})\bigl[\shellKorn[Z]\bigr](r)
\\
& \geq (1 - 2^{2\bstar-2a_{n,p}}) \int_R^{+\infty} \bigl\|\Sym(\del Z)\bigr\|_{\unL^2_{-\expoP}(\Lambda_s)}^2 \, \min( (s/r)^{a_{n,p}} , (s/r)^{2\bstar}) \frac{ds}{s}
\\
& \geq - \max(0, \, 2^{2\bstar-2a_{n,p}} - 1) \Dsym(r) .
\endaligned
\ee
In particular, for $\bstar<a_{n,p}$ the lower bound is simply zero.

\paragraph{Shell stability.}

The momentum shell stability condition states that the functionals $\Psi^\notreM_\beta[Y]$, $\beta\in\{a_{n,p},2a_{n,p}\}$, are coercive modulo the Korn remainder functional $\shellKorn[Z]$, and modulo one term that combines averages of~$Y$ and the fluctuation operator~$\Kappa^{\notreM}$.  The latter term is expressed in~\eqref{Y-of-Z-def} in terms of the source contribution~$\Theta^F$.
This gives a one-sided bound on~$\Omega^\Psi$, for some constants $C,\gammash^0>0$,
\be
\aligned
(2\bstar-a_{n,p}) \Omega^\Psi
& = - J_{a_{n,p}}[\Psi^\notreM_{a_{n,p}}[Y]] - I_{2\bstar}[\Psi^\notreM_{2a_{n,p}}[Y]] + (2\bstar-2a_{n,p}) I_{2\bstar}[\Upsilon^\notreM[Y]]
\\
& \leq - \gammash^0 \Dsym + C (J_{a_{n,p}}+I_{2\bstar})\bigl[|T(\Xi^\notreM)^{-1}\Theta^F|^2\bigr] \\
& \quad
- \gammaKornM (J_{a_{n,p}}+I_{2\bstar})\bigl[\shellKorn[Z]\bigr]
+ |2\bstar-2a_{n,p}| \gamma_\Psi \Dtot .
\endaligned
\ee
As part of analyzing the equation for averages, we show in particular in \autoref{prop-87-moment} that $\Theta^F=\Theta^F_1+\Theta^F_2$ where $\Theta^F_1\in L^2_{\astar}([R,+\infty))$ and $\Theta^F_2=o(r^{-\astar})$ pointwise, with norms controlled by the $H^{1*}_{\astar+2,-\expoP}$ norm of~$F$ together with $\cutoff_{\pstar}\Jmax(F)$.
Analogously to~\eqref{integralThetaE1} in the Hamiltonian case, the $J_{a_{n,p}}+I_{2\bstar}$ integral of $|T(\Xi^\notreM)^{-1}\Theta^F|^2$ is thus bounded pointwise as $o(r^{-2\astar})$.
Combining our bounds into~\eqref{equa-shell-mom} yields
\bel{ineq-Dsymtot-F}
\aligned
\quad & \unquad
\frac{\gammash^0 - \gammaKornM \max(0, 2^{2\bstar-2a_{n,p}}-1))}{2\bstar-a_{n,p}} \Dsym(r)
- |2\bstar-2a_{n,p}| \Bigl( \gamma_\Phi + \frac{\gamma_\Psi}{2\bstar-a_{n,p}} \Bigr) \Dtot(r)
\\
& \leq \Dtot(r)^{1/2} \|F\|_{H^{1*}_{\astar+2,-\expoP}(\Omega_R,\RR^n)} r^{-\astar} o(1) \\
& \quad
+ \Bigl( \|F\|_{H^{1*}_{\astar+2,-\expoP}(\Omega_R,\RR^n)} + \cutoff_{\pstar} \Jmax(F) \Bigr)^2 r^{-2\astar} o(1) ,
\endaligned
\ee
where as usual $o(1)$ are bounded independently of the source term~$F$ and tend to zero as $r\to+\infty$ at some $F$-dependent rate.
Thanks to \autoref{lem:Korn7}, the left-hand side controls $\Dtot(r)$ for $\bstar$ sufficiently close to~$a_{n,p}$.  We thus reach a quadratic inequality on $\Dtot(r)^{1/2}$, which bounds $\Dtot(r)$ by $o(r^{-\astar})$ times some norm of~$F$.  This immediately gives the bounds~\eqref{equa-bound-int-M} with $Z-\cutoff_{\pstar}\Zmodu$ replaced by~$Y$ since the dissipation~$\Dtot(r)$ controls both terms in the left-hand side: $\Dtot(r)$ manifestly controls the $H^1$ norm on $\Omega_r\setminus\Omega_{2r}$ as stated, and controls the pointwise norm through an easy radial integration, explained below~\eqref{simpler-dissip} in the Hamiltonian case.

\paragraph{Identifying the harmonic term.}
For $\astar<a_{n,p}$ we are done since $Z-\cutoff_{\pstar}\Zmodu$ and $Y$ differ by an $r^{-a_{n,p}}$ term, which decays faster than the bounds we established.
For $\astar\geq a_{n,p}$ we must show $\ChZ$ is given by $\Jmodu(F)$.
We integrate $\notreM^\lambda[Z]=F$ against $s^{a_{n,p}+1}ds\,d\chi$.  The angular integral is evaluated in~\eqref{equa-818bis} below (as part of deriving an equation for averages of~$Z$) and we get
\be
\aligned
& 2 \int_R^r \la F_l(s) \ra s^{a_{n,p}+1} ds
= - \int_R^r (\vartheta+a_{n,p}) \Bigl( \vartheta \la 2 \xh_l  \Zperp + \Zpar_l \ra + \la \nablaslash_l \Zperp -  \Zpar_l \ra \Bigr) (s)\, s^{a_{n,p}-1} ds
\\
& \qquad = - \bigl( \vartheta \la 2 \xh_l  \Yperp + \Ypar_l \ra + \la \nablaslash_l \Yperp -  \Ypar_l \ra \bigr) r^{a_{n,p}}
+ (\ChZ(\matQ^t)^{-1})_l + \la \mathbb{B}^\lambda[Z]_l(R) \ra R^{a_{n,p}}
\endaligned
\ee
where we have written the boundary term at~$R$ as the momentum boundary operator~\eqref{equa-bound-ope-2-lambda-M}, which vanishes by definition of variational solutions, and the boundary term at $r$ by splitting $Z = Y + (\ChZ(\matQ^t)^{-1}(\Xi^\notreM)^{-1})_{(i)} \xi^{\normal(i)} r^{-a_{n,p}}$.
By assumption on the source term the left-hand side has a limit, specifically $\eta^\lambda\Jmodu(F)_l$ with the constant $\eta^\lambda=2(n-1)|\Sphe^{n-1}|/\aire[\Lambda,\lambda]$ given in~\eqref{equa-thetalambda-M}.
The term involving~$Y$ thus also has a finite $r\to+\infty$ limit~$\ell_l$.
On the other hand, its square, averaged over $[r,2r]$, decays to zero by the dissipation bounds we have established, and therefore $\ell_l=0$.
We deduce $\eta^\lambda \Jmodu(F) = \ChZ(\matQ^t)^{-1}$.
The normalization~\eqref{equa-norm-normal} of the silhouette vectors, namely $\Xi^{\notreM(i)}{}_k = \eta^\lambda\delta^{(i)}_k$ then implies
\be
Y = Z - \bigl(\ChZ(\matQ^t)^{-1}(\Xi^\notreM)^{-1}\bigr)_{(i)} \xi^{\normal(i)} r^{-a_{n,p}}
= Z - \Jmodu(F)_{(i)} \xi^{\normal(i)} r^{-a_{n,p}} = Z - \Zmodu .
\ee
This concludes the proof of \autoref{thm:decayL2-M}, modulo some calculations in \autoref{section=9}, and especially \autoref{prop-87-moment} solving the equation for averages.


\subsection{Localized pointwise estimates}
\label{section=8.3} 

The analogue of \autoref{thm-sharp-h-localized} (concerning the Hamiltonian operator) is now stated for the momentum operator. The proof is essentially identical to the one in \autoref{section=5.4} for the Hamiltonian operator and, therefore, is omitted. 
We recall that to a decay exponent $\pstar  \geq p$, one associates $\astar = \pstar - p + a_{n,p}/2$.

\begin{theorem}[Pointwise estimates for the localized momentum operator]
\label{thm-sharp-m-localized}
In the setup of \autoref{thm:decayL2-M}, with exponents $\expoP\geq 2$ and $p\in(0,n-2)$, a momentum-stable localization function~$\lambda$ on a connected domain~$\Lambda$ (cf.~\autoref{def-shell-Mstab}), and a sharp decay exponent $\astar\in[a_{n,p}/2,a_{n,p}+\deltaM)$ for some small enough $\deltaM>0$, consider the variational solution $Z\in H^1_{n-2-p,-\expoP}(\Omega_R,\RR^n)$ to~\eqref{equa-opera5-M} with a source term~$F$ with finite norm
\bel{NbbastarF-def}
\Nbb_{\astar}^{\notreM} \coloneqq \|F\|^{N-1,\alpha}_{\Omega_R,\astar+2,-\expoPp-1}+\|F\|_{H^{1*}_{\astar+2,-\expoP}(\Omega_R,\RR^n)} + \cutoff_{\pstar} \Jmax(F) < +\infty
\ee
for some exponent $\expoPp \geq \expoP + (n+1)/2$.
In the \mbox{(super-)harmonic} case, $\Jmax(F)$~is defined in~\eqref{Jmodu-def}, and $\Jmodu(F)$ and $\Zmodu=\Jmodu(F)_j\xi^{\normal(j)}r^{-a_{n,p}}$ as in~\eqref{defJmoduZ}.  Fix some radius $R'> R$.  Then\footnote{We emphasize that the norms are over $\Omega_R$ or $\Omega_{R'}$, as specified.}
\bel{harmonic-regime-713}
\bigl\| Z - \cutoff_{\pstar} \Zmodu \bigr\|^{N+1,\alpha}_{\Omega_{R'}, \astar, -\expoPp+1}
+ \sup_{\beta\in[a_{n,p}/2,\astar)} \Bigl( (\astar-\beta)^{1/2} \bigl\| Z - \cutoff_{\pstar} \Zmodu \bigr\|_{H^1_{\beta,-\expoP}(\Omega_{R'},\RR^n)} \Bigr)
\lesssim  \Nbb_{\astar}^{\notreM} .
\ee
For $\astar=a_{n,p}/2$, the interval in the supremum is empty and the supremum term is omitted.
Furthermore, the left-hand side, with the $C^{N+1,\alpha}$ H\"older norm replaced by a $C^1$~norm, tends to zero in the limit $R'\to+\infty$.
\end{theorem}


\section{Consequences of momentum harmonic and radial stability}
\label{section=9}

\subsection{The harmonic-spherical decomposition}
\label{section=9.1}

We now analyze the structure of the localized momentum equations. The relevant decomposition is presented first and we then proceed with the analysis of the kernel of the harmonic momentum operator (in \autoref{section=9.2}) and next the evolution of the spherical averages of the solutions in terms of fluctuations (in \autoref{section=9.3}). Specifically, we consider the second-order operator (with implicit summation in $j$) 
\be
\aligned
\notreM^\lambda[Z]^i 
& \coloneqq - \frac{1}{2} \Big(
(\del_j \del_j Z^i + \del_j \del_i Z^j) + \big(\del_j \log ( r^{n-2-2p} \lambda^{2\expoP} ) \big) \bigl(\del_j Z^i + \del_i Z^j \bigr)
\Big).
\endaligned
\ee
In our decomposition below, observe that $\Brr^{\bullet\bullet}$ contains radial derivatives only (and does not depend upon $\lambda$), while $\ssB^{\lambda\bullet\bullet}[Z]$ contains tangential derivatives only. When writing $\xi=(\xiperp,\xipar)$  we can either regard $\xipar$ as a vector in $\RR^n$ with components $\xipar_i$ (with $i=1,\dots,n$) or as a tangent vector to the sphere with components denoted by $\xipar_a$, using abstract Penrose indices $a,b,\dots$ on the unit sphere~$\Sphe^{n-1}$. In the calculations below it is convenient to use both standpoints. Recall also that $\nablaslash$ is the Levi--Civita connection of the round metric~$\gslash$ on the sphere, and that $\Sym(\nablaslash \xipar)_{ab} \coloneqq \frac{1}{2} \bigl( \nablaslash_a \xipar_b + \nablaslash_b \xipar_a\bigr)$.

\begin{lemma}[Harmonic-spherical decomposition of the localized momentum operator]
\label{lem:sph-mom}
The second-order elliptic operator obtained by composing the linearized momentum constraint (around the Euclidean data set) and its formal adjoint, together with a weight $(\lambda^\expoP r^{n/2-1-p})^2$ (with $\vartheta\lambda= 0$) enjoys the decomposition 
\be
\aligned 
r^2 \notreM^{\lambda\bullet\bullet}[Z]
& = \Brr^{\bullet\bullet}[Z] + \Brs^{\lambda\bullet\bullet}[Z] + \ssB^{\lambda\bullet\bullet}[Z],
\endaligned
\ee
in which both $Z_i = \xh_i\Zperp+ \Zpar_i$ and the operator itself are decomposed into a component $\Zperp$ orthogonal to the sphere (a scalar field) and a component $\Zpar$ parallel to it (a vector field), with 
\bse
\bel{Brs-expr-main}
\aligned
\ssB^{\lambda\perp\perp}[\xiperp]
& = (n-1) \xiperp - \frac{1}{2} \lambda^{-2\expoP} \nablaslash \cdot \bigl( \lambda^{2\expoP} \nablaslash \xiperp \bigr), 
\\
\ssB^{\lambda\parallel\perp}[\xiperp]_a & = - \frac{1}{2} \nablaslash_a\xiperp - \lambda^{-2\expoP} \nablaslash_a\bigl(\lambda^{2\expoP} \xiperp\bigr),
\\
\ssB^{\lambda\perp\parallel}[\xipar]
& = \frac{1+a_{n,p}}{2} \lambda^{-2\expoP} \nablaslash \cdot \bigl(\lambda^{2\expoP} \xipar \bigr) + \nablaslash \cdot \xipar, 
\\
\ssB^{\lambda\parallel\parallel}[\xipar]_a  & = \frac{1+a_{n,p}}{2} \xipar_a
-  \lambda^{-2\expoP} \nablaslash^b\Bigl(\lambda^{2\expoP} \Sym(\nablaslash \xipar)_{ab} \Bigr), 
\endaligned
\ee
and (with all other terms taken to vanish identically)
\bel{equa-Bmoment} 
\aligned
\Brr^{\perp\perp}[\Zperp]
& = - \vartheta(\vartheta+a_{n,p}) \Zperp,
& 
\Brs^{\lambda\parallel\perp}[\Zperp]_a & = - \frac{1}{2} (\vartheta+a_{n,p}) \nablaslash_a\Zperp,
\\
\Brs^{\lambda\perp\parallel}[\Zpar]
& = - \frac{1}{2} \lambda^{-2\expoP} (\vartheta+a_{n,p}) \nablaslash \cdot \bigl(\lambda^{2\expoP} \Zpar \bigr),
&
\Brr^{\parallel\parallel}[\Zpar]_a & = - \frac{1}{2} \vartheta(\vartheta+a_{n,p}) \Zpar_a .
\endaligned
\ee
\ese
\end{lemma}


\subsection{Construction of the silhouette vectors} 
\label{section=9.2}

\paragraph{Preliminary on the kernel of the adjoint.}

We follow the same strategy as we did for the Hamiltonian operator in \autoref{section=6.2} to investigate consequences of the momentum harmonic stability condition~\eqref{equa-stable-M-414} in \autoref{def-harmonic-Mstab}.
We show that $\notreM^\lambda[Z] = 0$ admits an $n$-dimensional basis of non-trivial solutions of the form $Z=r^{-a_{n,p}}\xi$ in which $\xi=\xi(\xh)$ is~a vector field on the $(n-1)$-sphere, in other words that the operator $\ssB^\lambda[\xi] = r^{2+a_{n,p}} \notreM^\lambda[r^{-a_{n,p}}\xi]$ has an $n$-dimensional kernel.
As a first step in \autoref{section=6.2}, we observed that the adjoint Hamiltonian admits constants as solutions.  The analogous solutions of the adjoint momentum are constant vectors (somewhat obscured by splitting them into orthogonal and parallel components). We thus seek solutions to the coupled system
\be 
(\ssB^\lambda)^{*\perp\perp}[\xiperp] + (\ssB^\lambda)^{*\perp\parallel}[\xipar] = 0,
\qquad\quad 
(\ssB^\lambda)^{*\parallel\perp}[\xiperp] + (\ssB^\lambda)^{*\parallel\parallel}[\xipar] = 0, 
\ee 
in which $\xiperp$ is a scalar field and $\xipar$ is a vector field. 

\begin{lemma} 
\label{lem-kernaladjointM} 
The kernel of the adjoint $(\ssB^\lambda)^*$ of the harmonic momentum operator has dimension at least $n$, since for each $l=1,\dots,n$ 
\bel{equa-moment-adjoint-expl}
\text{the pair } \, \xistar_l \coloneqq (\xiperp, \xipar) = (\xh_l, \nablaslash \xh_l) \, \text{ belongs to } \ker\bigl((\ssB^\lambda)^*\bigr).
\ee 
\end{lemma} 

\begin{proof} Using $\xipar = \nablaslash\xiperp$ we have 
$(\ssB^\lambda)^{*\perp\perp}[\xiperp] + (\ssB^\lambda)^{*\perp\parallel}[\xipar]
= (n-1) \xiperp + \Deltaslash\xiperp$, which vanishes since $\xh_l$ are eigenfunctions of the spherical Laplacian $-\Deltaslash$ with eigenvalue $(n-1)$. Next, using $\xipar = \nablaslash\xiperp$ we find 
\be
\aligned
& (\ssB^\lambda)^{*\parallel\perp}[\xiperp]_a + (\ssB^\lambda)^{*\parallel\parallel}[\xipar]_a \\
&  = - \lambda^{-2\expoP} \nablaslash_a(\lambda^{2\expoP}) \xiperp
 - \lambda^{-2\expoP} \nablaslash^b(\lambda^{2\expoP}) \nablaslash_a\nablaslash_b\xiperp
- \nablaslash_a\xiperp - \nablaslash^b\nablaslash_a\nablaslash_b\xiperp.
\endaligned
\ee
To see that this vanishes, we note that the function $\xiperp=\xh_l$ obeys $\nablaslash_a\nablaslash_b\xiperp = - \gslash_{ab}\xiperp$.
\end{proof} 


\paragraph{Dimension of the kernel and cokernel.}

We can now state our main result concerning the operator $\ssB^\lambda$, proven in the rest of the section.

\begin{proposition}[Kernel properties of the harmonic momentum operator]
\label{prop:Mom-kernel}
Consider a conical domain $\Omega_R = K \cap {}^{\complement} \Ball_R \subset \RR^n$ together with a localization function $\lambda\colon \Lambda \to (0, \lambda_0]$ with connected support $\Lambda \subset \Sphe^{n-1}$ and some ${\expoP \geq 2}$, such that the harmonic stability condition~\eqref{equa-stable-M-414} in \autoref{def-harmonic-Mstab} holds. Then the operator~$\ssB^\lambda$ and its adjoint  
have $n$-dimensional kernels, that is, 
\bel{equa-coker-ker-moment}
\ker\ssB^{\lambda *} = \Span\big\{ \xistar_l, 1\leq l\leq n \big\} ,
\qquad
\ker\ssB^\lambda = \Span\big\{ \xi^{\normal(l)}, 1\leq l\leq n \big\} ,
\ee
consisting, on the one hand, of linear combinations of vector fields $\xistar_l$ given explicitly in~\eqref{equa-moment-adjoint-expl} and, on the other hand, of linear combinations of $n$~silhouette vector fields $\xi^{\normal(l)}\colon\Lambda\to\RR^n$, $1\leq l\leq n$, which are linearly independent as vector fields.\footnote{The vectors $\xi^{\normal(l)}(x)$ at a given $x\in\Lambda$ are not necessarily linearly independent.}
Moreover, under the radial stability condition~\eqref{equa-Xi-invertible} in \autoref{def-radial-Mstab}, the basis of silhouette vector fields 
$\xi^{\normal (j)} \in \ker(\ssB^\lambda)$ can be characterized by the average conditions in \autoref{def:normalized-kernel-basis}, namely \textup{(}the constant~$\eta^\lambda$ being given in~\eqref{equa-thetalambda-M}\textup{)}
\bel{equa-normal-moment} 
\bigl\la - \nablaslash_l \xi^{\normal (j) \perp}  + 2 a_{n,p} \xh_l\, \xi^{\normal (j) \perp} \bigr \ra 
 + (a_{n,p}+1) \la \xi_l^{\normal (j) \parallel} \ra = \eta^\lambda \, \delta^{(j)}_{l}. 
\ee
\end{proposition}

 
\paragraph{Asymptotic variational formulation.}

\bse
The bilinear form $\ssrmB^\lambda[\xi,\rho]=\fint_{\Lambda}\rho \ssB^\lambda[\xi] \, d\chi$ is non-symmetric, with
\be
\ssrmB^\lambda[\xi,\rho] - \ssrmB^\lambda[\rho,\xi]
= \frac{1}{2} a_{n,p} \fint_{\Lambda} (\rhopar \cdot \nablaslash\xiperp - \xipar \cdot \nablaslash \rhoperp) d\chi .
\ee
In addition, given the non-trivial cokernel identified in \autoref{lem-kernaladjointM},
the corresponding quadratic functional $\ssrmB^\lambda[\xi,\xi]$ (cf.~the notation\footnote{Here, we prefer to write $\ssrmB^\lambda[\xi,\xi]$ rather than $\ssrmB^\lambda[\xi]$ to emphasize quadraticity.}~\eqref{ssBalpha-quaform-0})
\bel{ssBalpha-quaform-1}
\aligned 
\ssrmB^\lambda[\xi,\xi]
= \fint_{\Lambda} \Bigl(
&  \frac{1}{2} |\nablaslash\xiperp|^2
- \frac{1}{2} (a_{n,p}+2) \xipar \cdot \nablaslash \xiperp 
+ \frac{1}{2} \bigl(a_{n,p}+1 \bigr) |\xipar|^2
\\
& 
+ \bigl|\Sym(\nablaslash \xipar) \bigr|^2 
+ 2 \xiperp \nablaslash \cdot \xipar 
+ ( n-1) (\xiperp)^2 \Bigr) \, d\chi
\endaligned
\ee
\ese
is not positive-definite since it vanishes for $\xi=\xistar_l$.
Observe that the second line in \eqref{ssBalpha-quaform-1} is non-negative since
$\bigl|\Sym(\nablaslash \xipar) \bigr|^2 = \bigl|\Sym(\nablaslash \xipar)^\circ\bigr|^2 +  |\nablaslash \cdot \xipar|^2/(n-1)$, but that the first is only bounded below by $- (a_{n,p}^2/8) |\xipar|^2$.
Nevertheless, a variant of the Lax--Milgram theorem now provides us with a unique variational solution to the equation $\ssB^\lambda[\xi] = \psi$ in the domain $\Lambda$, provided we restrict attention to data and solutions satisfying suitable average constraints, as stated now.\footnote{The source term obeys 
$\langle\psi,\xistar_l\rangle_\Lambda = \la \xh_l\, \psi^{\perp} + \psi_{l}^{\parallel}  \ra=0$, $1 \leq l \leq n$, while the solution obeys $\la 2\xh_l\, \xi^{\perp} + \xi_{l}^{\parallel}  \ra=0$, with a factor of~$2$ difference that can be tracked down to the difference between $\Brr^{\perp\perp}$ and~$\Brr^{\parallel\parallel}$ in~\eqref{equa-Bmoment}.}
Here, the weighted dual Sobolev space $\unH^{1*}_{-\expoP}(\Lambda,\RR^n)$ is defined in~\eqref{dual-Sobolev-Lambda} by duality.
 
\begin{lemma}[Variational formulation for the asymptotic localized momentum]
\label{propo-existenceM}
Consider a conical domain $\Omega_R = K \cap {}^{\complement} \Ball_R \subset \RR^n$ together with a localization function $\lambda\colon \Lambda \to (0, \lambda_0]$ with connected support $\Lambda \subset \Sphe^{n-1}$ and some ${\expoP \geq 2}$. Assume that the harmonic stability condition in~\eqref{equa-stable-M-414} in \autoref{def-harmonic-Mstab} holds. Then, for any vector field $\psi \in \unH^{1*}_{-\expoP}(\Lambda,\RR^n)$ satisfying $\langle\psi,\xistar_l\rangle_\Lambda=0$
($1 \leq l \leq n$), there exists a unique solution $\xi \in \unH^1_{-\expoP}(\Lambda,\RR^n)$ satisfying
$\bigl\la 2 \xh_l\, \xiperp + \xipar_l \bigr\ra=0$ $(1 \leq l \leq n)$ to the variational problem 
\bse
\bel{varBrho}
\ssrmB^\lambda[\xi, \rho] = \langle\psi,\rho\rangle_\Lambda
\qquad  
\rho \in \unH^1_{-\expoP}(\Lambda,\RR^n),
\ee 
which, moreover, is bounded in terms of the data, that is, 
\bel{varBrho-2}
\|\xi\|_{\unH^1_{-\expoP}(\Lambda,\RR^n)}\lesssim\|\psi\|_{\unH^{1*}_{-\expoP}(\Lambda,\RR^n)}.
\ee
\ese
\end{lemma}

\begin{proof}
The proof is essentially identical to that of \autoref{propo-existenceH} for the Hamiltonian operator.
We work with the codimension~$n$ subspace $W \coloneqq \{ \xi \in \unH^1_{-\expoP}(\Lambda,\RR^n), \, \bigl\la 2 \xh_l\, \xi^{\perp} + \xi_{l}^{\parallel} \bigr\ra  =0 \ \text{for} \allowbreak \ l=1,\dots,n \}$.
By the harmonic stability condition~\eqref{equa-stable-M-414}, the bilinear form~$\ssrmB^\lambda$ is bounded and weakly coercive in the sense that $\sup_{\| \rho \|_W= 1} \ssrmB^\lambda[\xi, \rho] \gtrsim \| \xi \|_W$ and likewise with the arguments interchanged.
Thus, by the (generalized) Lax--Milgram theorem \cite[Theorem 2.1]{Babuska}, the restriction of $\psi$ to~$W$ can be written as $\ssrmB^\lambda[\xi,{\,\cdot\,}]$ for some unique $\xi\in W$, namely \eqref{varBrho}~holds for $\rho\in W$.
By \autoref{lem-kernaladjointM}, the relation~\eqref{varBrho} holds for each of the $n$ vectors $\rho=\xistar_k$ ($1\leq k\leq n$).  These span a subspace of $\unH^1_{-\expoP}(\Lambda,\RR^n)$ complementary to~$W$ since the matrix of averages $\la 2\xh_l \xistarperp_k + \xistarpar_k\ra = \delta_{kl} + \la\xh_k\xh_l\ra = T_{kl}$ is invertible.
Finally, $\|\xi\|_{\unH^1_{-\expoP}(\Lambda,\RR^n)}=\|\xi\|_W\lesssim\|\psi|_W\|_{W'}\leq\|\psi\|_{\unH^{1*}_{-\expoP}(\Lambda,\RR^n)}$.
\end{proof} 

 
\paragraph{Proof of the first part of \autoref{prop:Mom-kernel}.}

In view of the identities ($1 \leq l \leq n$) 
\be
\langle\ssB^\lambda[\xi],\xistar_l\rangle_\Lambda
= \fint_{\Lambda} \xi \cdot \ssB^{\lambda *}[ \xistar_l ] d\chi = 0 , \qquad \xi \in \unH^1_{-\expoP}(\Lambda,\RR^n), 
\ee
we see that the image of the operator $\ssB^\lambda$ is contained in the subspace of distributions $\psi \in \unH^{1*}_{-\expoP}(\Lambda,\RR^n)$ satisfying $\langle\psi,\xistar_l\rangle_\Lambda=0$.
 Moreover, by \autoref{propo-existenceM} the image is exactly equal to that subspace, so that
$\ker\ssB^{\lambda *} = \Span\big\{ \xistar_l, 1 \leq l \leq n\bigr\}$.
On the other hand, consider next an element $\xi \in\ker\ssB^\lambda$ that satisfies the $n$ conditions 
$\la 2\xh_l\xiperp+\xipar{}_l \ra =0$.
By the harmonic stability condition~\eqref{equa-stable-M-414} we have $\| \xi \|_{\unH^1_{-\expoP}(\Lambda,\RR^n)}^2 \lesssim \ssrmB^\lambda[\xi,\xi]$, which vanishes since $\xi$ is in the kernel of $\ssB^\lambda$. Thus, $\xi$ vanishes identically. We deduce that the linear map $\xi \mapsto \bigl(\la 2\xh_l\xiperp+\xipar{}_l \ra \bigr)_{1 \leq l \leq n}$ from $\ker\ssB^\lambda$ to~$\RR^n$ is injective, hence that the kernel is {\sl at most} $n$-dimensional. To construct a basis of the kernel, for \emph{each $j$} we consider the unique solution $\xi=\xi^{(j)}$ (by \autoref{propo-existenceM}) of
\be
\ssB^\lambda[\xi^{(j)}] = \ssB^\lambda[\xistar_j] , 
\qquad 
\bigl\la 2 \xh_l\, \xi^{(j)\perp} + \xi^{(j)\parallel}{}_l \bigr\ra  = 0 
\quad (1 \leq l \leq n).
\ee
Indeed, \autoref{propo-existenceM} applies since the source satisfies $\la  \xistar_k \cdot \ssB^\lambda[\xistar_j] \ra = \la  \ssB^{\lambda *}[\xistar_k]\cdot \xistar_j \ra = 0$ ($1 \leq k \leq n$). It follows that, by construction, $\xi^{(j)} - \xistar_j\in \ker\ssB^\lambda$ and
\be
\aligned
& \bigl\la 2 \xh_k \, (\xi^{(j)} - \xistar_j)^{\perp} + (\xi^{(j)} - \xistar_j)_k^{\parallel} \bigr\ra
= - \bigl\la 2 \xh_k \, \xistar_j{}^{\perp} + \xistar_j{}_k^{\parallel} \bigr\ra
\\
& = - \bigl\la 2 \xh_k \, \xh_j + \delta_{kj} - \xh_k \xh_j  \bigr\ra
= - \bigl( \delta_{kj} + \bigl\la \xh_k \, \xh_j \bigr\ra \bigr) = - T_{kj}.
\endaligned
\ee
The matrix $(T_{kj})$ is invertible, thus we have found $n$ linearly independent elements of the kernel. This establishes \eqref{equa-coker-ker-moment}.

 
\paragraph{Proof of the second part of \autoref{prop:Mom-kernel}.} 

By the radial stability condition, there exists a basis of elements of the momentum kernel so that the matrix 
$
 \bigl\la - \nablaslash_l \xi^{\normal (j) \perp} + 2 a_{n,p} \xh_l\, \xi^{\normal (j) \perp} \bigr\ra
 +(1+a_{n,p}) \la \xi^{\normal (j) \parallel}{}_l \ra 
 $
 is invertible. Clearly, this property is independent of the specific choice of the basis. We normalize according to  \autoref{def:normalized-kernel-basis}, so that \eqref{equa-normal-moment} holds. This concludes the proof of \autoref{prop:Mom-kernel}.
 

\paragraph{Normalization for the kernel.}

Finally, we check that our normalization of the silhouette vector fields $\xi^{\normal (j)}$ in \autoref{def:normalized-kernel-basis} leads to the following property. 

\begin{lemma}[ADM momentum of the modulator]
\label{lem-ADMmomentmod}
The two-tensor field $\hmodu$ defined in $\Omega_R \subset \RR^n$ by
\be
\aligned 
\hmodu_{jk} 
& = - {1 \over 2} \, \lambda ^{2\expoP} r^{ - (n - 2) + a_{n,p}} \bigl( \del_j \Zmodu_{k} + \del_k \Zmodu_j \bigr), 
\qquad    
\Zmodu_k = J_{(j)}^\modu \, \xi^{\normal (j)}{}_k (\xh) r^{-a_{n,p}},
\endaligned
\ee
with $\xh=x/r$, has ADM momentum $\Jbb(\Omega_R, \hmodu) = \Jmodu$. 
\end{lemma}

\begin{proof}
\bse
 In view of~\eqref{equa-def-energy-momentum}, with $C_n \coloneqq  (n-1) \, |\Sphe^{n-1}|$ we compute  
\be
\aligned
& \Jbb(\Omega_R, \hmodu)_{l} 
= {1 \over C_n} \lim_{r \to +\infty} r^{1 + a_{n,p}}
\int_\Lambda \xh_k (-1/2) \bigl( \del_k \Zmodu_l + \del_l \Zmodu_k \bigr) \Big|_{|x|=r} \lambda^{2 \expoP} d \xh
\\
& \quad = - {J_{(j)}^\modu \over 2 C_n}  \, 
\lim_{r \to +\infty}
\int_\Lambda \Bigl( r^{a_{n,p}} x_k \del_k \bigl( \xi^{\normal (j)}{}_l (\xh) \, r^{-a_{n,p}} \bigr)
+ r^{a_{n,p}} x_k \del_l \bigl( \xi^{\normal (j)}{}_k (\xh) \, r^{-a_{n,p}}\bigr) \Bigr) \,  \lambda^{2 \expoP} d \xh .
\endaligned
\ee
We then use $x_k \del_k = \vartheta$ in the first term in the integrand, that is,
\[
r^{a_{n,p}} x_k \del_k \big( \xi^{\normal (j)}{}_l  \, r^{-a_{n,p}} \big) 
= r^{a_{n,p}} \vartheta \big( \xi^{\normal (j)}{}_l \, r^{-a_{n,p}} \big) 
= - a_{n,p} \xi^{\normal (j)}{}_l.
\]
For the second term in the integrand, use that $Y_k=\xh_k\Yperp+\Ypar{}_k$ for any vector field $Y$ to get
$x_k\del_lY_k = \del_l(r\Yperp)-Y_l = r^{-1} (\xh_l \vartheta(r\Yperp) + \nablaslash_l(r\Yperp)) - \xh_l\Yperp - \Ypar{}_l = \xh_l\vartheta\Yperp + \nablaslash_l\Yperp - \Ypar{}_l$ and deduce
\[
r^{a_{n,p}} x_k\del_l\bigl( \xi^{\normal (j)}{}_k (\xh) \, r^{-a_{n,p}}\bigr)
= - a_{n,p} \xh_l \xi^{\normal(j)\perp} + \nablaslash_l\xi^{\normal(j)\perp} - \xi^{\normal(j)\parallel}{}_l .
\]
Consequently,
\be
\aligned
\Jbb(\Omega_R, \hmodu)_{l}
& = {1 \over 2 C_n} J_{(j)}^\modu \,
\int_\Lambda \Bigl(  - \nablaslash_l \xi^{\normal (j) \perp}
+ 2 a_{n,p} \xh_l\, \xi^{\normal (j) \perp}
+ (1+a_{n,p}) \xi^{\normal (j) \parallel}{}_l 
\Bigr) d\chi
\\
& = {\aire[\Lambda,\lambda] \over 2 C_n} J_{(j)}^\modu \,
\Bigl(  - \la \nablaslash_l \xi^{\normal (j) \perp}  \ra 
+ 2 a_{n,p} \la \xh_l\, \xi^{\normal (j) \perp} \ra 
+ (1+a_{n,p}) \la \xi^{\normal (j) \parallel}{}_l \ra 
\Bigr).
\endaligned
\ee
This suggests the following normalization (which, in view of the notation in~\eqref{equa-thetalambda-M}, is consistent with~\eqref{equa-norm-normal}): 
\be
 \bigl\la - \nablaslash_l \xi^{\normal (j) \perp} + 2 a_{n,p} \xh_l\, \xi^{\normal (j) \perp} \bigr\ra 
 +(1+a_{n,p}) \la \xi^{\normal (j) \parallel}{}_l \ra 
 = {2(n-1) |\Sphe^{n-1}| \over \aire[\Lambda,\lambda]} \delta_{jl}. \qedhere 
\ee
\ese
\end{proof} 


\subsection{Radial evolution of spherical averages}
\label{section=9.3}  

\paragraph{Contracting with an element of the co-kernel.}

We now turn our attention to averages of solutions $Z$ to the momentum operator $\notreM^\lambda$ in~\eqref{equa-opera5-M}. We work with variational solutions enjoying a basic integral control and, as pointed out for the Hamiltonian equation, this integrability allows us to suppress the worst decay terms that could arise in our computation, below. We proceed by using, first, the elements of the co-kernel and, next, the elements of the kernel of the harmonic momentum operator. By combining these two sets of equations we arrive at a coupled system of second-order differential equations for a set of $n$ weighted spherical averages of~$Z$, which are easily integrated to get non-differential equations.

We begin by integrating the equations $\notreM^\lambda[Z]^j = F^j$ on a sphere of radius $r \geq R$ after multiplication by one of the vector fields 
$\xistar_l = (\xh_l, \nablaslash \xh_l)$ derived in~\eqref{equa-moment-adjoint-expl}. Using $\la \xistar_l \cdot \ssB^\lambda[Z] \ra=0$, we find ($l=1,2, \ldots, n$) 
\[
\aligned
r^2 \la F_l\ra
& = \fint_\Lambda \xistar_l \cdot F \, r^2 \, d\chi 
= \fint_\Lambda \xistar_l \cdot \notreM^\lambda[Z]  r^2 \, d\chi 
= \fint_\Lambda \xistar_l \cdot \Big(
\Brr^{}[Z] + \Brs^\lambda[Z] 
\Big) \, d\chi  
\\
& =
\fint_\Lambda  \Big(
 - \xh_l \vartheta(\vartheta+a_{n,p}) \Zperp
  - \frac{1}{2} (\vartheta+a_{n,p}) \nablaslash_l \Zperp \\
& \qquad\qquad
  - \xh_l \frac{1}{2} \lambda^{-2\expoP} (\vartheta+a_{n,p}) \nablaslash \cdot \bigl(\lambda^{2\expoP} \Zpar \bigr)
 - \frac{1}{2} \vartheta(\vartheta+a_{n,p}) \Zpar_l
\Big) \, d\chi
\\
& =   
- \vartheta(\vartheta+a_{n,p}) \la \xh_l  \Zperp\ra
- \frac{1}{2} (\vartheta+a_{n,p}) \la \nablaslash_l \Zperp\ra
+ \frac{1}{2} (\vartheta+a_{n,p}) \la \Zpar_l \ra 
- \frac{1}{2} \vartheta(\vartheta+a_{n,p}) \la \Zpar_l\ra .
\endaligned
\]
Hence, we arrive at the second-order equations ($l=1,2, \ldots, n$) 
\bel{equa-818bis} 
- \vartheta(\vartheta+a_{n,p}) \la 2 \xh_l  \Zperp + \Zpar_l \ra
= (\vartheta+a_{n,p}) \la \nablaslash_l \Zperp -  \Zpar_l \ra  
+ 2 r^2  \, \la F_l \ra.
\ee
We now seek a differential equation for $\la \nablaslash_l \Zperp -  \Zpar_l \ra$, which may involve $\la 2 \xh_l  \Zperp + \Zpar_l \ra$ and its derivatives.


\paragraph{On fluctuations and averages.}

To better separate the contribution of co-kernel elements~$\xistar_l$ to various expressions, we derive the formula, already announced in \autoref{section=3.2}, for the fluctuation of a vector field~$\xi$, expressed as follows (with implied sum over $k$): 
\be
\xi^{\fluc} \coloneqq \xi - \projP^k(\xi) \xistar_k,
\qquad 
\xi^{\fluc\perp} = \xi^\perp - \projP^k(\xi)  \xh_k ,
\qquad
\xi^{\fluc\parallel} = \xipar - \projP^k(\xi)  \nablaslash \xh_k ,
\ee
in which the constants $\projP^k(\xi)$ are chosen so that the following averages vanish:
\be
\aligned
\la 2\xh_l\xi^{\fluc\perp}+\xi^{\fluc\parallel}{}_l\ra
& = \la 2\xh_l\xi^\perp+\xi^\parallel{}_l\ra - \la 2\xh_l \projP^k(\xi) \xh_k + \projP^k(\xi) \nablaslash_l\xh_k\ra
\\
& = \la 2\xh_l\xi^\perp+\xi^\parallel{}_l\ra - (\delta_{kl} + \la \xh_k\xh_l\ra) \projP^k(\xi). 
\endaligned
\ee
This gives the following formal restatement of the earlier notation.

\begin{definition}
\label{DEF-fluct-vectors} 
The {\bf fluctuation of a vector field} $\xi$ is the vector field $\xi^{\fluc}$ defined by 
\bel{equa-fluct-vectors} 
\xi^{\fluc} \coloneqq \xi - \projP^k(\xi) \xistar_k,
\qquad
\projP^k(\xi) \coloneqq (T^{-1})^{kl} \, \bigl\la 2 \xh_l \xiperp + \xipar{}_l \bigr\ra, 
\ee
in terms of the positive-definite matrix $(T_{kl}) = (\delta_{kl} + \la \xh_k\xh_l\ra)$.
\end{definition}

Observe that $\projP^k(\xistar_l) = \delta^k_l$ hence $\xistar_l{}^\fluc=0$ and the averages $\projP^k(\xi^\fluc)=0$ vanish for any vector field~$\xi$.
We express fluctuations of the silhouette vector fields in terms of a matrix~$\matQ^{(j)k}$, with
\bel{equa--813}
\aligned 
\xi^{\normal (j) \fluc} & = \xi^{\normal (j)} - \matQ^{(j)k} \xistar_k,
\qquad 
\matQ^{(j)k}  = (T^{-1})^{kl} \, \bigl\la 2 \xh_l \xi^{\normal (j) \perp} + \xi^{\normal (j) \parallel}{}_l \bigr\ra, 
\\
\xi^{\normal (j) \fluc\perp} 
& = \xi^{\normal (j) \perp} - \matQ^{(j)k} \xh_k ,
\qquad
\xi^{\normal (j) \fluc\parallel} = \xi^{\normal (j) \parallel} - \matQ^{(j)k}  \nablaslash \xh_k. 
\endaligned
\ee


\paragraph{Contracting with an element of the kernel.}
Next, we contract the momentum operator with the silhouette vector fields  in \autoref{prop:Mom-kernel}, namely the vector fields $\xi^{\normal (j)} \in \ker(\ssB^\lambda)$ in~\eqref{equa-normal-moment}.  More precisely, we integrate the equation $\notreM^\lambda[Z]=F$ against the fluctuation~$\xi^{\normal(j)\fluc}$.
We use that $\ssB^{\lambda *}[\xistar_l]=0$ and $\ssB^\lambda[\xi^{\normal (j)}]=0$ to compute
\bel{xiBZ}
\aligned
&
\Bigl\la \xi^{\normal (j) \fluc} \cdot \ssB^\lambda[Z] \Bigr\ra
= \Bigl\la Z \cdot \ssB^{\lambda *}\bigl[ \xi^{\normal (j) \fluc}  \bigr] \Bigr\ra
= \Bigl\la Z \cdot \ssB^{\lambda *}\bigl[ \xi^{\normal (j)}  \bigr] \Bigr\ra
= \Bigl\la Z \cdot \Bigl( \ssB^{\lambda *}- \ssB^\lambda \Bigr)[\xi^{\normal (j)} ] \Bigr\ra
\\
&  = \frac{a_{n,p}}{2} \Bigl\la \nablaslash \Zperp  \cdot \xi^{\normal (j) \parallel} - \Zpar \cdot \nablaslash \xi^{\normal (j) \perp} \Bigr\ra
\\
&  = \frac{a_{n,p}}{2} \Bigl\la \nablaslash \Zperp  \cdot \matQ^{(j)k}  \nablaslash \xh_k
- \Zpar \cdot \matQ^{(j)k} \nablaslash \xh_k \Bigr\ra
+ \frac{a_{n,p}}{2} \la {\nablaslash \Zperp}  \cdot {\xi^{\normal (j) \fluc\parallel}} - {\Zpar} \cdot {\nablaslash \xi^{\normal (j) \fluc\perp}} \ra
\\ 
&  = \frac{a_{n,p}}{2} \matQ^{(j)k} \Bigl\la \nablaslash_k\Zperp - \Zpar_k \Bigr\ra
+ \frac{a_{n,p}}{2} \la {\nablaslash \Zperp}  \cdot {\xi^{\normal (j) \fluc\parallel}} - {\Zpar} \cdot {\nablaslash \xi^{\normal (j) \fluc\perp}} \ra.
\endaligned
\ee

Recalling the operator expressions~\eqref{equa-Bmoment} we deduce that  the solution to $\notreM^\lambda[Z]^j = F^j$ satisfies 
\[
\aligned 
\bigl\la \xi^{\normal (j) \fluc}\cdot F \, r^2 \bigr\ra
& = \Bigl\la \xi^{\normal (j) \fluc} \cdot \big( \Brr^{}[Z] + \Brs^\lambda[Z] + \ssB^\lambda[Z] \big) \Bigr\ra
\\
& = \frac{a_{n,p}}{2} \matQ^{(j)k} \Bigl\la \nablaslash_k\Zperp - \Zpar_k \Bigr\ra
+ \frac{a_{n,p}}{2} \la {\nablaslash \Zperp}  \cdot {\xi^{\normal (j) \fluc\parallel}}
 - {\Zpar} \cdot {\nablaslash \xi^{\normal (j) \fluc\perp}} \ra
\\
& \quad - \vartheta(\vartheta+a_{n,p}) \la \xi^{\normal (j) \fluc\perp} \Zperp\ra
 - \frac{1}{2} (\vartheta+a_{n,p}) \la \xi^{\normal (j) \fluc\parallel} \cdot \nablaslash \Zperp \ra
 \\
 & \quad + \frac{1}{2} ( \vartheta + a_{n,p}) \la \nablaslash \xi^{\normal (j) \fluc\perp} \cdot \Zpar \ra
  - \frac{1}{2} \vartheta(\vartheta+a_{n,p}) \la \xi^{\normal (j) \fluc\parallel} \cdot \Zpar \ra .
\endaligned
\]
After some obvious cancellation, we have 
\be
\aligned 
a_{n,p} \matQ^{(j)k} \Bigl\la \nablaslash_k\Zperp - \Zpar_k \Bigr\ra
& = \vartheta\Bigl(
(\vartheta+a_{n,p}) \bigl\la 2 \xi^{\normal (j) \fluc\perp} \Zperp + \xi^{\normal (j) \fluc\parallel} \cdot \Zpar\bigr\ra \\
& \qquad
+ \bigl\la \xi^{\normal (j) \fluc\parallel} \cdot \nablaslash \Zperp - \nablaslash \xi^{\normal (j) \fluc\perp} \cdot \Zpar \bigr\ra \Bigr)
+ 2 \bigl\la \xi^{\normal (j) \fluc}\cdot F \, r^2 \bigr\ra .
\endaligned
\ee

Next, we decompose $Z=Z^{\fluc}+ \projP^k(Z) \xistar_k$ into its fluctuation and average with the notation in~\eqref{equa-fluct-vectors}, and we use $\bigl\la 2 \xh_k \xi^{\normal (j) \fluc\perp} + \xi^{\normal (j) \fluc\parallel}{}_k \bigr\ra=0$ for a cancellation.
 
\begin{definition}\label{def:KappaM}
For the normalized kernel basis, the \textbf{fluctuation operators} introduced in \autoref{section=3.2} are the linear functionals of the fluctuation~$Z^\fluc$ and its derivatives given by the convention $\Kappa^{\notreM(j)}[Z]\coloneqq\Kappa^{\notreM(j)}[Z^\fluc]$ and by
\be
\aligned
\Kappa^{\notreM(j)}[Z^\fluc]
= \bigl\la \nablaslash \xi^{\normal (j) \fluc\perp} \cdot Z^{\fluc\parallel} - \xi^{\normal (j) \fluc\parallel} \cdot \nablaslash Z^{\fluc\perp} \bigr\ra 
- (\vartheta+a_{n,p}) \bigl\la 2 \xi^{\normal (j) \fluc\perp} Z^{\fluc\perp} + \xi^{\normal (j) \fluc\parallel} \cdot Z^{\fluc\parallel}\bigr\ra. 
\endaligned
\ee
\end{definition}

This is the normalized-basis form of the functional announced in \autoref{section=3.2}; with it, we obtain
\bel{equa--823}
a_{n,p} \matQ^{(j)k} \Bigl\la \nablaslash_k\Zperp - \Zpar_k \Bigr\ra
= \vartheta\Bigl(
\bigl\la \xi^{\normal (j) \fluc\parallel}{}_k - \nablaslash_k \xi^{\normal (j) \fluc\perp} \bigr\ra \projP^k(Z) 
- \Kappa^{\notreM(j)}[Z^\fluc] \Bigr) 
+ 2 \bigl\la \xi^{\normal (j) \fluc}\cdot F \, r^2 \bigr\ra .
\ee

\paragraph{Observations on $\Kappa^{\notreM(j)}$.}
In the second average (acted upon by $\vartheta+a_{n,p}$) the vector $Z^\fluc$ can be replaced by $Z$ since $\xi^{\normal(j)\fluc}$ has vanishing averages.
Moreover, in the coefficient of $\projP^k(Z)$ in~\eqref{equa--823}, the vector $\xi^{\normal(j) \fluc}$ can be replaced by $\xi^{\normal(j)}$ since $\bigl\la \xipar{}_k - \nablaslash_k \xiperp \bigr\ra$ vanishes for $\xi=\xistar_l$.
It will prove useful to evaluate $\Kappa^{\notreM (j)}[\xi^{\normal(i)\fluc}r^{-a_{n,p}}]$.
For this, we return to~\eqref{xiBZ} with $Z=\xi^{\normal(i)\fluc}r^{-a_{n,p}}$, which yields (after multiplying by $2r^{a_{n,p}}/a_{n,p}$)
\be
\aligned
0 & = \matQ^{(j)k} \la \nablaslash_k \xi^{\normal(i)\perp} - \xi^{\normal(i)\parallel}{}_k \ra
+ \la \nablaslash\xi^{\normal(i)\perp}  \cdot \xi^{\normal (j) \fluc\parallel} - \xi^{\normal(i)\parallel} \cdot {\nablaslash \xi^{\normal (j) \fluc\perp}} \ra
\\
& =
\matQ^{(j)k} \la \nablaslash_k \xi^{\normal(i)\perp} - \xi^{\normal(i)\parallel}{}_k \ra
+ \la \nablaslash\xi^{\normal(i)\fluc\perp}  \cdot \xi^{\normal (j) \fluc\parallel} - \xi^{\normal(i)\fluc\parallel} \cdot {\nablaslash \xi^{\normal (j) \fluc\perp}} \ra \\
& \quad + \matQ^{(i)k} \la \xi^{\normal (j) \fluc\parallel}{}_k - \nablaslash_k \xi^{\normal (j) \fluc\perp} \ra ,
\endaligned
\ee
where $\xi^{\normal(j)\fluc}$ can be replaced by $\xi^{\normal(j)}$ in the last average since the vector $\xistar_l{}^\parallel - \nablaslash\xistar_l{}^\perp$ vanishes for all~$l$.
Therefore,
\bel{aveKappa-silhouette}
\aligned
\Kappa^{\notreM (j)}[\xi^{\normal(i)\fluc}r^{-a_{n,p}}]
& = \bigl\la \nablaslash \xi^{\normal (j) \fluc\perp} \cdot \xi^{\normal(i)\fluc\parallel} - \xi^{\normal (j) \fluc\parallel} \cdot \nablaslash \xi^{\normal(i)\fluc\perp} \bigr\ra r^{-a_{n,p}}
\\
& = \bigl( (\Xi^\notreM)^{(j)}{}_{k} \matQ^{(i)k} - (\Xi^\notreM)^{(i)}{}_{k} \matQ^{(j)k} \bigr) r^{-a_{n,p}}
\endaligned
\ee
in terms of the matrix $\Xi^{\notreM(j)}{}_k = \bigl\la - \nablaslash_k \xi^{\normal (j) \perp} + 2 a_{n,p} \xh_k\, \xi^{\normal (j) \perp} + (1+a_{n,p}) \xi^{\normal (j) \parallel}{}_k \bigr\ra$ introduced earlier in~\eqref{equa-the-matrix}.


\paragraph{Conclusion.}

Inserting \eqref{equa--823} into~\eqref{equa-818bis}, we find 
\[
\aligned
\quad & \unquad \vartheta(\vartheta+a_{n,p}) \Bigl(
a_{n,p} \matQ^{(j)l} \la 2 \xh_l  \Zperp + \Zpar_l \ra
+ \bigl\la \xi^{\normal (j) \parallel}{}_k - \nablaslash_k \xi^{\normal (j) \perp} \bigr\ra \projP^k(Z)
-\Kappa^{\notreM(j)}[Z^{\fluc}]
 \Bigr)
\\
& = - 2 r^2 a_{n,p} \matQ^{(j)k} \la F_k \ra
- 2 (\vartheta+a_{n,p}) \bigl\la \xi^{\normal (j) \fluc}\cdot F \, r^2 \bigr\ra .
\endaligned
\]
Using that $\matQ^{(j)l} \la 2 \xh_l  \Zperp + \Zpar_l \ra = \la 2 \xh_k \xi^{\normal (j) \perp} + \xi^{\normal (j) \parallel}{}_k \ra\projP^k(Z)$, we recognize that the coefficient of $\projP^k(Z)$ is the structure matrix~$\Xi^\notreM$ and, by using $\projP^k(Z) = (T^{-1})^{kl} \bigl\la 2 \xh_l \Zperp + \Zpar_l \bigr\ra$, we arrive at 
\bel{Zave-diffeq}
\aligned
& \vartheta(\vartheta+a_{n,p}) \Bigl(
(\Xi^\notreM)^{(j)}{}_{k} \, (T^{-1})^{kl} \bigl\la 2 \xh_l \Zperp + \Zpar_l \bigr\ra
- \Kappa^{\notreM(j)}[Z^\fluc] 
\Bigr)
\\
& \quad = - 2 r^2 a_{n,p} \matQ^{(j)k} \la F_k \ra - 2 (\vartheta+a_{n,p}) \la \xi^{\normal (j) \fluc}\cdot r^2 F\ra .
\endaligned
\ee
Taking the inverse of the operator $\vartheta(\vartheta+a_{n,p})$ thanks to~\eqref{equa-formuleODE} (and \eqref{Ibeta-plus-Jbeta} for $\astar\geq a_{n,p}$), we arrive at \eqref{equa-defb-b-b-MM} and, specifically, the following conclusion.
 
\begin{proposition}[Spherical averages associated with the localized momentum operator]
\label{prop-87-moment} 
In the setup of \autoref{thm-sharp-m-localized-vari}, assume that the localization domain $(\Lambda,d\chi)$ obeys the harmonic and radial momentum stability conditions in \refwithname{Definitions}{def-harmonic-Mstab} \refwithname{and}{def-radial-Mstab} (but not necessarily shell stability) and that the source~$F$ obeys the decay~\eqref{defJmoduZ} for some $\astar\geq a_{n,p}/2$, namely $F\in H^{1*}_{\astar+2,-\expoP}(\Omega_R,\RR^n)$ and, if $\astar\geq a_{n,p}$, the limits in~\eqref{Jmodu-def} exist and $\Jmax(F)<+\infty$. Here $\cutoff_{\pstar}$ has the convention introduced immediately before~\eqref{Jmodu-def}. Then there exists a constant vector $\ChZ\in\RR^n$ such that the averages $\la 2 \xh_l \Zperp + \Zpar_l \ra$ satisfy the identity
\bse
\bel{equa-defb-b-b-MM-3} 
\aligned 
( \Xi^\notreM)^{(j)}{}_{k} (T^{-1})^{kl} \bigl\la 2 \, \xh_l \Zperp + \Zpar_l \bigr\ra
& = \Kappa^{\notreM (j)}[Z^\fluc] + \Theta^{F(j)} + \ChZ{}^{(j)} \, r^{-a_{n,p}}
\endaligned
\ee
with a source term
\be
\Theta^{F(j)} = 2J_0[\la \xi^{\normal (j)}\cdot r^2 F\ra] + 2 \matQ^{(j)k} \begin{cases} I_{a_{n,p}}[\la r^2 F_k \ra] & \text{if } \astar<a_{n,p} , \\ - J_{a_{n,p}}[\la r^2 F_k \ra] & \text{if } \astar\geq a_{n,p} . \end{cases}
\ee
In addition, upon decomposing $\Theta^F=\Theta^F_1+\Theta^F_2$ with $\Theta^F_2 = -2\Oneone_{\astar=a_{n,p}}\matQ^{(j)k} J_{a_{n,p}}[\la r^2 F_k \ra]$, one has the bounds
\be
\aligned
\|\Theta^F_1\|_{L^2_{\astar}([R,+\infty))} & \lesssim \|F\|_{H^{1*}_{\astar+2,-\expoP}(\Omega_R,\RR^n)} ,
\\
|\Theta^F_2(r)| & \lesssim\Oneone_{\astar=a_{n,p}}\Jmax(F)r^{-a_{n,p}}o_F(1),
\\
|\ChZ| & \lesssim \|F\|_{H^{1*}_{\astar+2,-\expoP}(\Omega_R,\RR^n)} + \cutoff_{\pstar} \Jmax(F) .
\endaligned
\ee
\ese
\end{proposition}


\section{Nonlinear analysis of the squared localized Einstein constraints}
\label{section=10}

\subsection{Formulation of the estimates}
\label{section=10.1}

\paragraph{Setting up the stage.}

Building upon the linear estimates established in \refwithname{Sections}{section=5} \refwithname{to}{section=9}, we are now in a position to give a proof of \autoref{theo--beyond-harmonic} (and \autoref{theo--beyond-harmonic-II}, its sub-harmonic analogue).
At this stage, the seed-to-solution operator is known to exist by \autoref{thm:sts-existence}, whose construction scheme was described in \autoref{section=2.1}. 
The resulting solutions $(g,h)$ of the Einstein constraint equations, expressed in terms of auxiliary fields $(u,Z)$, satisfy the sub-harmonic estimates stated in \autoref{thm:sts-existence} (and \autoref{thm:sts-Sobolev}) with radial exponent~$p$.
After some set-up and reduction to an asymptotic end, we present \autoref{prop:improve-radial} which improves slightly the radial control of the solution (proven in \autoref{section=11}).  By iterating it finitely many times we reach the desired sharp decay exponent~$\pstar$.
As pointed out earlier, one of our challenges is dealing with geometries with (possibly) very low decay exponent $p_G>0$, which increases the number of iteration steps.  Consequences for the ADM energy and momentum modulators are specified in \autoref{section=10.4}.

We follow the notation in \refwithname{Sections}{section=2} \refwithname{to}{section=4}, and we recall now the hypotheses of \autoref{theo--beyond-harmonic} and \autoref{theo--beyond-harmonic-II} (which differ only in whether $\pstar\geq n-2$ or $\pstar<n-2$).
Let $(\Mbf, \Omega, g_0,h_0, \wtrr, \lambdabf)$ be a conical localization data set in the sense of \autoref{def-conical}. We assume that $\lambdabf$ is a stable localization function (cf.~\autoref{def-41-stable}), so that the pointwise decay estimates for the Euclidean linear operators (from earlier sections) hold up to 
the common super-harmonic upper bound $p^{\lambdabf}_{n,p}>n-2$ defined in~\eqref{p-lambda-def}, and a set of exponents $p,p_G, p_A,\pstar$ is given so that \eqref{exponent-range} holds:\footnote{To be precise, stability conditions involve both $\lambdabf^{2\expoP}$ and the projection exponent~$p$, so that they may hold only on a subinterval of $(p_n^\flat, n-2)$.  Alternatively, one selects the exponents $\expoPm,\expoP,\expoPp,p,p_G,p_A$ before~$\lambdabf$, and lastly~$\pstar$.}
\bel{condi--1-repeat-55}
\alignedat{2}
& \text{projection exponent: } \, & p & \in (p_n^\flat, n-2),
\\
& \text{geometry exponent: } & p_G & > 0, \\
& \text{accuracy exponent: } & p_A & \geq \max(p_G,p),
\\
& \text{sharp decay exponent: } & \pstar & \in [p,p_A] \cap \bigl[p, \; \min\bigl( p^{\lambdabf}_{n,p}, n-2+p_G \bigr)\bigr).
\endalignedat
\ee 
Recall that the exponent $p$ arises in the definition of the (integral) variational projection~\eqref{equa--221}. The exponent $p_G$ determines the pointwise decay of the solutions and $p_A$ the pointwise decay of the Einstein operators, as stated in~\eqref{equa-near-refe} and~\eqref{equa-nearEins}, respectively.
The upper bounds on the sharp decay exponent are explained in due course.
In addition, our results involve an admissible set $(\expoPm,\expoP,\expoPp)$ of localization exponents satisfying~\eqref{equa-expoPPP}, namely
\bel{expoP-sec10-cond}
\max(N+\alpha,1+\expoP/2) < \expoPm \leq \expoP - \frac{n+4}{2} , \qquad \expoPp = 2 \expoP - \expoPm .
\ee


\paragraph{Formulation of the equations.}

By definition of the seed-to-solution map whose existence is established in \autoref{thm:sts-existence}, under suitable smallness assumptions on $(\seedg-g_0,\seedh-h_0)$ and $\Gcal(\seedg,\seedh)$, there exist a function~$u$ and vector field~$Z$ that yield a solution $(g,h)$ of the Einstein constraints,
\bel{equa-EEHM-33}
\alignedat{2}
(\Hcal,\Mcal)[g,h] & = (0,0), \qquad
& g & = \seedg + \gdiff, \qquad \gdiff = \omegabf_p^2 \, d\Hcal^{*\flat\flat}_{(g_0,h_0)}[u,Z], 
\\
&& h & = \seedh + \hdiff, \qquad \hdiff = \omegabf_{p+1}^2 \, d\Mcal^{*\sharp\sharp}_{(g_0,h_0)}[u,Z], 
\endalignedat
\ee
such that $g,\seedg,g_0$ are uniformly equivalent Riemannian metrics.
At this stage of our analysis, the solutions $(\gdiff,\hdiff)$ and $(u,Z)$ are known only to enjoy integral and pointwise estimates with radial decay exponents $p$ and $a_{n,p}/2=n-2-p$, respectively.
As we quantify in \autoref{section=10.2}, below, these sub-harmonic estimates provide sufficient control of the solution in a large bounded region, which allows us to concentrate on the asymptotic ends.


\paragraph{Strategy of the proof.}
In each asymptotic end~$\Omega_\iota$, we proceed by iteratively improving the radial control of the solution from the exponent~$p$ to an {\bf iteration exponent} $\pstar'\in[p,\pstar]$, all the way to the sharp decay exponent~$\pstar$.
This procedure takes finitely many steps to reach any exponent $\pstar$ in the range~\eqref{condi--1-repeat-55}.

We express the Hamiltonian and momentum constraints in the schematic form:
\bel{linearEuclidenoperatorsource}
\text{linear Euclidean operator }
= \text{ linear curved terms } + \text{ nonlinearities } + \text{ source,}
\ee
obtained by expanding the linearized operators and nonlinearities in terms of objects associated with the Euclidean data set $(\delta,0)$.
We combine together
\bei
\item integral and pointwise estimates with a decay exponent $\pstar'$ from earlier steps in the iteration,

\item estimates of the linear curved terms with a decay exponent $\pstar'+p_G$ due to the decay of differences between the data sets $(g_0,h_0)$, $(\seedg,\seedh)$, and $(\delta,0)$,

\item estimates of nonlinearities of the Einstein constraints, with a decay exponent $2\pstar'$,

\item our analysis of the linear Euclidean operator in \refwithname{Sections}{section=5} \refwithname{to}{section=9}, which yields estimates on the solution with a new improved exponent bounded above by $p^{\lambdabf}_{n,p}$ and by the exponents of all terms in~\eqref{linearEuclidenoperatorsource}.
\eei

The estimates on the linearized Euclidean operators concern variational solutions, namely solutions obeying Neumann-type boundary conditions on the radial boundary of~$\Omega_\iota$.  We thus apply the above operators to the unknowns $\kappa_\iota u$ and $\kappa_\iota Z$ which are supported in~$\Omega_\iota$ and vanish in a neighborhood of its radial boundary, by the construction~\eqref{equa-partition} of~$\kappa_\iota$.
Derivatives of the cutoff function~$\kappa_\iota$ introduce additional sources in~\eqref{linearEuclidenoperatorsource}, which pose no difficulty since they are all supported in the bounded region.

Before we can start the iteration, we must contend with the fact that the integral estimates provided by \autoref{thm:sts-existence} only control $(u,Z)$ in a weighted $H^2\times H^1$ space, and that the resulting ``linear curved terms'' in~\eqref{linearEuclidenoperatorsource} are not regular enough to apply the results of \autoref{thm-sharp-h-localized} on the linearized Hamiltonian operator.
Likewise, during the iteration, the integral estimates provided by \refwithname{Sections}{section=5} \refwithname{to}{section=9} control $u$ and $\vartheta u$ in a weighted $H^2$ space, but still only control $Z$ in~$H^1$.
To gain the required radial derivative, we observe that $\vartheta(\kappa_\iota u,\kappa_\iota Z)$ is the solution of a (curved-space) linear elliptic equation with source terms in suitable dual Sobolev spaces, such that variational estimates provide the weighted $H^2\times H^1$ control of $\vartheta(\kappa_\iota u,\kappa_\iota Z)$ necessary to continue the iteration.  Crucially, the elliptic operator used at this step differs from the Euclidean operator to which the results of \refwithname{Sections}{section=5} \refwithname{to}{section=9} apply, so we use \emph{two different} elliptic operators in the course of each iteration.


\paragraph{Family of estimates for the induction argument.}

Throughout the iteration, we keep track of how the solution $(u,Z,\gdiff,\hdiff)$ is controlled thanks to a family of estimates that depend on the iteration exponent $\pstar'\in[p,\pstar]$.
To accommodate for \mbox{(super-)}harmonic exponents $\pstar'\geq n-2$ for which modulator terms must be subtracted from the solution, as well as sub-harmonic exponents $\pstar'<n-2$ for which such terms are ill-defined, it is convenient to introduce a {\it cut-off indicator} $\cutoff_{\pstar'}\in\{0,1\}$,
\be
\cutoff_{\pstar'} = \Oneone_{\pstar'\geq n-2} = \bigl( 0 \text{ when } \pstar' < (n-2) \text{ and $1$ otherwise} \bigr) .
\ee
We also denote by $\astar'=n-2-2p+\pstar'$ the radial exponent relevant for the decay of $(u,Z)$ instead of $(\gdiff,\hdiff)$.

In agreement with \autoref{def-317} the modulated objects are given by~\eqref{equa-2p8} with the additional factor of~$\cutoff_{\pstar'}$, namely
\bel{equa-2p8-repeat} 
\aligned
\seedmodgpstarp
& = \seedg + \cutoff_{\pstar'} \sum_\iota \kappa_\iota \, \gmodu_\iota, 
\qquad & \gmodu_\iota & = \omegabf_p^2 d\Hcal^{*\flat\flat}_{(\delta,0)}[\umodu_\iota] ,
\qquad & \umodu_\iota & = \mmodu_\iota \nu^\normal(x/r) r^{-a_{n,p}},
\\
\seedmodhpstarp & = \seedh + \cutoff_{\pstar'} \sum_\iota \kappa_\iota \, \hmodu_\iota,
\qquad & \hmodu_\iota & = \omegabf_{p+1}^2 d\Mcal^{*\sharp\sharp}_{(\delta,0)}[\Zmodu_\iota] ,
\qquad & \Zmodu_\iota & = \Jmodu_\iota \cdot \xi^\normal(x/r) r^{-a_{n,p}} ,
\endaligned
\ee
where $\delta$, $x$, and~$r$ are the Euclidean metric, Cartesian coordinates, and radius $r=|x|$, derived from the identification $\Omega_\iota\simeq\Omega_R=K_\iota\cap{}^\complement\Ball_R\subset\RR^n$.
The modulators $(\gmodu_\iota,\hmodu_\iota,\umodu_\iota,\Zmodu_\iota)$ in each asymptotic end depend on $1+n$ parameters $(\mmodu_\iota, \Jmodu_\iota)$ that are not immediately determined, and only become so once the iteration exponent reaches $\pstar'\geq n-2$.  We will not need the coordinate expressions~\eqref{equa-2p8-correct} of the adjoint Hamiltonian and momentum constraints.

We denote by $\Est(\Omega_\iota,\pstar',\mmodu_\iota,\Jmodu_\iota)$ the following list of inequalities on an asymptotic end~$\Omega_\iota$ for $\pstar'\geq p$,
\bel{Est-def}
\Est(\Omega_\iota,\pstar',\mmodu_\iota,\Jmodu_\iota) \colon 
\aligned
\| u - \cutoff_{\pstar'} \umodu_\iota \|^{N+2, \alpha}_{\Omega_\iota, \astar', -\expoPp+2}
+ \| u - \cutoff_{\pstar'} \umodu_\iota \|_{H^2_{\astar', -\expoP}(\Omega_\iota)}
& \lesssim \Err^+_{\pstar}[\seedg,\seedh] ,
\\
\| Z - \cutoff_{\pstar'} \Zmodu_\iota \|^{N+1, \alpha}_{\Omega_\iota, \astar', -\expoPp+1}
+ \| Z - \cutoff_{\pstar'} \Zmodu_\iota \|_{H^1_{\astar', -\expoP}(\Omega_\iota)}
& \lesssim \Err^+_{\pstar}[\seedg,\seedh] .
\endaligned
\ee
For $\pstar'\in[p, n-2)$ we shall take vanishing $(\mmodu_\iota,\Jmodu_\iota)$ and omit them from the notation.
Note that for any $\pstar'\geq p$, these estimates imply any $\Est(\Omega_\iota,\tilde{p},\mmodu_\iota,\Jmodu_\iota)$ for $\tilde{p}\in[p,\pstar')$ provided $(\mmodu_\iota,\Jmodu_\iota)$ are themselves controlled by~$\Err^+_{\pstar}[\seedg,\seedh]$.
The right-hand side is a norm already introduced in~\eqref{equa-Ecal-pstar}
\bel{equa-Ecal-pq-rep}
\aligned
& \Err^+_{\pstar}[\seedg,\seedh]
= \bigl\| \Hcal(\seedg,\seedh) \bigr\|^{N-2,\alpha}_{\Omega, \pstar+2, \expoPm-2}
+ \bigl\| \Hcal(\seedg,\seedh) \bigr\|_{L^2_{\pstar+2,\expoP}(\Omega)}
\\
& \quad + \bigl\| \Mcal(\seedg,\seedh) \bigr\|^{N-1,\alpha}_{\Omega, \pstar+2, \expoPm-1}
+ \bigl\| \Mcal(\seedg,\seedh) \bigr\|_{L^2_{\pstar+2,\expoP}(\Omega)}
\\
& \quad
+ \cutoff_{\pstar} \sum_\iota \sup_{r\geq R} \biggl( \biggl| \int_{K_\iota\cap {}^\complement \Ball_r} \phi_{\iota\pushforward} \bigl(\Hcal( \seedg, \seedh) \, \dVol_{\seedg}\bigr) \biggr|
+ \sum_{1\leq j\leq n} \biggl| \int_{K_\iota\cap {}^\complement \Ball_r} \bigl(\phi_{\iota\pushforward}\Mcal( \seedg, \seedh)\bigr)_j \phi_{\iota\pushforward} \dVol_{\seedg} \biggr| \biggr) ,
\endaligned
\ee
which ranges over the whole domain~$\Omega$ and involves the final exponent~$\pstar$.  Since $\pstar\in[p,p_A]$, the norm controls $\Err_p[\seedg,\seedh]$ used in \autoref{section=2.3} when proving existence of the seed-to-solution map, and the norm is defined thanks to the decay of $\Gcal(\seedg,\seedh)$ with exponent~$p_A+2$.

In our notation of the norms we need not specify the choice of the metric ---which can be taken to be the reference $g_0$ or, at each asymptotic end,~$\delta$.
Technical calculations on the expansion of the Einstein constraints are collected in \autoref{appendix=C} and will be applied in the present section.


\subsection{Reduction to an asymptotically Euclidean end}
\label{section=10.2}

\paragraph{Dealing with the bounded region.}

Our aim in this section is to kick-start the iteration by reducing to one asymptotic end and observing that the estimates of~\eqref{Est-def} hold for the initial value $\pstar'=p$.  Then, \autoref{section=10.3} will be devoted to the iteration itself.

The construction of the seed-to-solution map in \autoref{appendix=E} provides us with useful bounds on $(u,Z)$ and $(\gdiff,\hdiff)=(g-\seedg,h-\seedh)$, stated in \autoref{thm:sts-Sobolev} (based on the variational formulation and interior elliptic regularity):
\bse
\label{equa-bounds-var-p}
\begin{align}
\| u \|^{N+2, \alpha}_{\Omega, a_{n,p}/2, -\expoPp+2}
+ \| u \|_{H^2_{a_{n,p}/2, -\expoP}(\Omega)} 
+ \| Z \|^{N+1, \alpha}_{\Omega, a_{n,p}/2, -\expoPp+1}
+ \| Z \|_{H^1_{a_{n,p}/2, -\expoP}(\Omega)}
& \lesssim \Err^+_{\pstar}[\seedg,\seedh] ,
\\
\| \gdiff \|^{N, \alpha}_{\Omega, p,\expoPm}
+ \| \gdiff \|_{L^2_{p,\expoP}(\Omega)}
+ \| \hdiff \|^{N,\alpha}_{\Omega, p+1, \expoPm}
+ \| \hdiff \|_{L^2_{p+1,\expoP}(\Omega)}
& \lesssim \Err^+_{\pstar}[\seedg,\seedh] .
\end{align}
\ese
Inside the large bounded subset~$\Omega_0$ (introduced in \autoref{def-conical}), the decay at infinity is irrelevant and norms with different radial exponents are all equivalent.  Explicitly, one obtains bounds on both weighted H\"older and Sobolev norms,
\bse\label{estim-Omega0}
\begin{align}
\| u \|^{N+2, \alpha}_{\Omega_0, \astar', -\expoPp+2}
+ \| u \|_{H^2_{\astar', -\expoP}(\Omega_0)} 
+ \| Z \|^{N+1, \alpha}_{\Omega_0, \astar', -\expoPp+1}
+ \| Z \|_{H^1_{\astar', -\expoP}(\Omega_0)}
& \lesssim \Err^+_{\pstar}[\seedg,\seedh], 
\\
\| \gdiff \|^{N, \alpha}_{\Omega_0, \pstar',\expoPm}
+ \| \gdiff \|_{L^2_{\pstar',\expoP}(\Omega_0)}
+ \| \hdiff \|^{N,\alpha}_{\Omega_0, \pstar'+1, \expoPm}
+ \| \hdiff \|_{L^2_{\pstar'+1,\expoP}(\Omega_0)}
& \lesssim \Err^+_{\pstar}[\seedg,\seedh],  
\end{align}
\ese
for all $\pstar'\in[p,\pstar]$.
The implied constants depend upon $\pstar$ and $\Omega_0$, now fixed once and for all.
Importantly, these estimates compare $(g,h)$ to the \emph{uncorrected seed data} $(\seedg,\seedh)$ rather than the modulated data $(\seedmodgpstarp,\seedmodhpstarp)$ because at this stage the modulators are not known to be well-defined.  Correspondingly, the bounds involve $(u,Z)$ rather than $(u-\cutoff_{\pstar'}\umodu,Z-\cutoff_{\pstar'}\Zmodu)$.
We prove suitable bounds on the modulators later on.


\paragraph{The initial estimates hold in an asymptotic end.}

We now pick and focus on an asymptotic end $\Omega_\iota$ for a given~$\iota$.  We observe that the estimates $\Est(\Omega_\iota,p)$ hold, namely \eqref{Est-def} with $\pstar'=p$ and $\astar'=n-2-p=a_{n,p}/2$.  There is no modulator for $p<n-2$, so $\Est(\Omega_\iota,p)$ is simply \eqref{equa-bounds-var-p} restricted to~$\Omega_\iota$.  Explicitly,
\bel{Est-p}
\aligned
\Est(\Omega_\iota,p) \colon \quad
& \| u \|^{N+2, \alpha}_{\Omega_\iota, n-2-p, -\expoPp+2}
+ \| u \|_{H^2_{n-2-p, -\expoP}(\Omega_\iota)} \\
& \quad + \| Z \|^{N+1, \alpha}_{\Omega_\iota, n-2-p, -\expoPp+1}
+ \| Z \|_{H^1_{n-2-p, -\expoP}(\Omega_\iota)}
\lesssim \Err^+_{\pstar}[\seedg,\seedh] .
\endaligned
\ee


\paragraph{Notation at an asymptotic end.}

According to our notation, in the asymptotic domain~$\Omega_\iota$, identified with $\Omega_R \subset \RR^n$, we have coordinates $(x^i)$ and, with the subscript~$\iota$ suppressed (when no confusion is possible) for the rest of this section,
\bel{sec10-omegabfp-omegap}
\omegabf_p = \lambdabf_\iota^{\expoP} \wtrr^{n/2-p}
= \lambda^{\expoP} r^{n/2-p} = \omega_p,
\ee
where $\lambda$ is solely a function of $\xh$ and $r^2=\sum (x^i)^2$. In our notation we identify the asymptotic end $\Omega_\iota$ with the subset $\Omega_R \subset \RR^n$ for some $R>1$.  As introduced near~\eqref{equa-partition}, the cutoff function $\kappa=\kappa_\iota\colon\Omega_R\to[0,1]$ vanishes for $r\in[R,R_2]$ and is identically~$1$ for $r\in[R_3,+\infty)$.  The bounded region~$\Omega_0$ contains the ball $\Ball_{R_4}\cap\Omega_R$ (or rather, its image under the identification $\Omega_R \simeq \Omega_\iota$), so that the estimates~\eqref{estim-Omega0} control $(u,Z,g-\seedg,h-\seedh)$ throughout the region where $\kappa\neq 1$.
In agreement with~\eqref{equa-2p8-repeat} (restricted to $\Omega_\iota$ and with $\iota$ suppressed) we write
\be
\seedmodgpstarp = \seedg + \cutoff_{\pstar'} \kappa \gmodu, 
\qquad \seedmodhpstarp = \seedh + \cutoff_{\pstar'} \kappa \hmodu  
\qquad \text{ in the domain } \Omega_R \subset \RR^n .
\ee


\subsection{Reaching a sharp radial decay}
\label{section=10.3}

\paragraph{Statement of the radial decay improvement.}
We now state the result that enables us to go step by step from a radial exponent~$p$ to the final exponent~$\pstar$.  To streamline the main argument, we postpone the proof of this proposition to \autoref{section=11}.  We recall that $\Est(\Omega_\iota,\pstar')$ are estimates~\eqref{Est-def} on $(u,Z)$ with exponent $\astar'=n-2-2p+\pstar'$.

\begin{proposition}[Improvement of the radial decay]\label{prop:improve-radial}
Assume that $\Est(\Omega_\iota,\pstar')$ holds for a decay exponent $\pstar'\in[p,n-2)$, and let 
\be 
\pstar''\coloneqq\min(\pstar,\pstar'+p_G,2\pstar'),
\qquad
\astar''\coloneqq n-2-2p+\pstar''.
\ee 
Then there exist parameters $(\mmodu, \Jmodu)$ with
\be
\aligned
& (\mmodu, \Jmodu) = (0, 0) & \text{for } \pstar''<n-2 , 
\\
& |\mmodu| + |\Jmodu| \lesssim \Err^+_{\pstar}[\seedg,\seedh] & \text{for } \pstar''\geq n-2 ,
\endaligned
\ee
such that $\Est(\Omega_\iota,q,\mmodu,\Jmodu)$ holds for each $q\in(\pstar',\pstar'')$. Furthermore, provided $\pstar''\neq n-2$ or $\pstar''=\pstar=n-2<\min(\pstar'+p_G, 2\pstar')$, one has a pointwise control that decays faster than~$r^{-\astar''}$, namely
\bse
\label{improve-radial-top}
\be
\aligned
& \| u - \cutoff_{\pstar''} \umodu \|^{N+2, \alpha}_{\Omega_\iota, \astar'', -\expoPp+2}
+ \| Z - \cutoff_{\pstar''} \Zmodu \|^{N+1, \alpha}_{\Omega_\iota, \astar'', -\expoPp+1}
\lesssim \Err^+_{\pstar}[\seedg,\seedh] ,
\endaligned
\ee
and 
\be
\aligned
\lim_{R'\to+\infty} \Bigl( & \bigl\| u - \cutoff_{\pstar''} \umodu \bigr\|_{C^3_{\astar'', -\expoPp+2}(\Omega_{R'})} + \bigl\| Z - \cutoff_{\pstar''} \Zmodu \bigr\|_{C^1_{\astar'', -\expoPp+1}(\Omega_{R'})} \Bigr) = 0 .
\endaligned
\ee
\ese
\end{proposition}

\begin{remark}
  1.~Sources for the linear Euclidean operator are controlled with a radial decay exponent~$\pstar''$. The Sobolev estimates in \refwithname{Sections}{section=5} \refwithname{to}{section=9} have a loss of exponent, hence the bound $q<\pstar''$, while the H\"older estimates reach this exponent, as stated in~\eqref{improve-radial-top}.
  \quad 2.~The condition $\pstar'<n-2$ ensures the absence of modulators in $\Est(\Omega_\iota,\pstar')$, which would otherwise contribute additional sources with $n-2+p_G$ radial decay.  This difficulty explains the upper bound $\pstar<n-2+p_G$ in~\eqref{exponent-range}.
  \quad 3.~The harmonic case $\pstar''=n-2$ is excluded when $\min(\pstar'+p_G, 2\pstar')=n-2$ because nonlinearities may have uncontrolled energy-momentum modulators.  In contrast, H\"older estimates do hold in the harmonic case when $\min(\pstar'+p_G, 2\pstar')>n-2$.
\end{remark}


\paragraph{Iteration.}
We are now ready to improve the estimates $\Est(\Omega_\iota,p)$ in~\eqref{Est-p} from radial exponent~$p$ to~$\pstar$.
If $\pstar=p$, the conclusion of \autoref{theo--beyond-harmonic-II} follows directly from~\eqref{equa-bounds-var-p}. Hence, throughout the remainder of this paragraph, we assume that $\pstar>p$.
As intermediate steps, we introduce a (finite) sequence of exponents $\pstar^{(i)}$ that are sufficiently close to each other in the sense that
\bel{ppstar-interpolation}
\aligned
p = \pstar^{(0)} & < \pstar^{(1)} < \dots < \pstar^{(k)} < \pstar , \\
\pstar^{(i+1)} & < \min\bigl(\pstar^{(i)}+p_G, 2\pstar^{(i)}\bigr) , && 0 \leq i\leq k-1 , \\
\pstar & < \min\bigl(\pstar^{(k)}+p_G, 2\pstar^{(k)}\bigr) .
\endaligned
\ee
For $\pstar\leq n-2$ the conditions~\eqref{ppstar-interpolation} are easily achieved by setting $\pstar^{(i)}=p+(\pstar-p)i/(k+1)$ for an integer $k > (\pstar - p) / \min(p, p_G)$.
For a sharp decay exponent $n-2<\pstar<n-2+p_G$, we first select $\pstar^{(k)}$ in the range $\max(\pstar-p_G,\pstar/2)<\pstar^{(k)}<n-2$, then construct other $\pstar^{(i)}$ as in the sub-harmonic case.
This interval is nonempty since $\pstar<\pmax\leq\min(2n-4,n-2+p_G)$ by~\eqref{p-lambda-def}.

In either case, $\pstar^{(k)}<n-2$ so an immediate induction based on \autoref{prop:improve-radial} shows that $\Est(\Omega_\iota,\pstar^{(i)})$ holds for $0\leq i\leq k$, with no energy and momentum modulators.  Applying \autoref{prop:improve-radial} one last time, one gets H\"older estimates with radial exponent exactly~$\pstar$:
\be
\aligned
\| u - \cutoff_{\pstar} \umodu \|^{N+2, \alpha}_{\Omega_\iota, \astar, -\expoPp+2}
+ \| Z - \cutoff_{\pstar} \Zmodu \|^{N+1, \alpha}_{\Omega_\iota, \astar, -\expoPp+1}
& \lesssim \Err^+_{\pstar}[\seedg,\seedh] ,
\\
\lim_{R'\to+\infty} \Bigl( \bigl\| u - \cutoff_{\pstar} \umodu \bigr\|_{C^3_{\astar, -\expoPp+2}(\Omega_{R'})} + \bigl\| Z - \cutoff_{\pstar} \Zmodu \bigr\|_{C^1_{\astar, -\expoPp+1}(\Omega_{R'})} \Bigr) & = 0 .
\endaligned
\ee


\paragraph{Pointwise estimates on the geometry.}
We now deduce H\"older estimates on $(g,h)$.  Within the asymptotic end~$\Omega_\iota\simeq\Omega_R$ we consider $R_3>R$ as specified in~\eqref{sec10-omegabfp-omegap}, such that the cutoff function $\kappa=\kappa_\iota$ is identically~$1$ 
on the domain $\Omega_{R_3} = K_\iota\cap{}^\complement\Ball_{R_3}$.  Then, in $\Omega_{R_3}$, we have 
\bel{final-gh-split}
\aligned
g-\seedmodg = \gdiff - \cutoff_{\pstar} \gmodu
& = \omega_p^2 d\Hcal^{*\flat\flat}_{(g_0,h_0)}[u - \cutoff_{\pstar}  \umodu,\,Z-\cutoff_{\pstar}\Zmodu] \\
& \quad + \cutoff_{\pstar} \omega_p^2 \, \Bigl( d\Hcal^{*\flat\flat}_{(g_0,h_0)}[\umodu,\Zmodu] - d\Hcal^{*\flat\flat}_{(\delta,0)}[\umodu] \Bigr) ,
\\
h-\seedmodh = \hdiff - \cutoff_{\pstar} \hmodu
& = \omega_{p+1}^2 d\Mcal^{*\sharp\sharp}_{(g_0,h_0)}[u - \cutoff_{\pstar}  \umodu, Z - \cutoff_{\pstar} \Zmodu] \\
& \quad + \cutoff_{\pstar} \omega_{p+1}^2 \, \Bigl( d\Mcal^{*\sharp\sharp}_{(g_0,h_0)}[\umodu,\Zmodu] - d\Mcal^{*\sharp\sharp}_{(\delta,0)}[\Zmodu] \Bigr) .
\endaligned
\ee
For all terms we rely on \autoref{prop:adjoint-constr-coord}, which expands the adjoint constraints in an asymptotically Euclidean end as
\be
\aligned
d\Hcal^{*\flat\flat}_{(g_0,h_0)}[u,Z]
& = d\Hcal^{*\flat\flat}_{(\delta,0)}[u] +(g_0-\delta)*\del\del u + \del*(\del g_0 * u) + h_0 * h_0 * u + \del * (h_0 * Z) ,
\\
d\Mcal^{*\sharp\sharp}_{(g_0,h_0)}[u,Z]
& = d\Mcal^{*\sharp\sharp}_{(\delta,0)}[Z] + (g_0-\delta) * \del Z + \del g_0 * Z + h_0 * u ,
\endaligned
\ee
with $d\Hcal^{*\flat\flat}_{(\delta,0)}[u]_{ij} = \del_i\del_j u - \delta_{ij} \Delta u$ and $d\Mcal^{*\sharp\sharp}_{(\delta,0)}[Z]^{ij} = -\frac{1}{2} (\del^i Z^j + \del^j Z^i)$.
From the bounds on $u-\cutoff_{\pstar}\umodu$ and $Z-\cutoff_{\pstar}\Zmodu$ with radial exponent~$\astar$ we deduce decay of the first terms in~\eqref{final-gh-split} in H\"older norms with radial exponents $\pstar$ and $\pstar+1$.
The $\umodu$ and $\Zmodu$ modulators both scale as~$r^{-a_{n,p}}$, and the differences $d\Hcal^{*\flat\flat}_{(g_0,h_0)} - d\Hcal^{*\flat\flat}_{(\delta,0)}$ and $d\Mcal^{*\sharp\sharp}_{(g_0,h_0)} - d\Mcal^{*\sharp\sharp}_{(\delta,0)}$ improve the exponent by~$p_G$ owing to $g_0-\delta$, $\del g_0$, $\del\del g_0$, $h_0$, or~$\del h_0$ factors.  The second terms in~\eqref{final-gh-split} are thus bounded with exponents $n-2+p_G$ and $n-1+p_G$, respectively, which are strictly better than the exponents for the first terms.

By combining with the estimates~\eqref{estim-Omega0} in the bounded region~$\Omega_0$, we reach the following conclusion:
\bel{bound-proof-gseedmodg}
\aligned
\|g-\seedmodg\|^{N,\alpha}_{\Omega_\iota,\pstar,\expoPm}
+ \|h-\seedmodh\|^{N,\alpha}_{\Omega_\iota,\pstar+1,\expoPm}
& \lesssim \Err^+_{\pstar}[\seedg,\seedh]
\\
\lim_{R'\to+\infty} \Bigl( \|g-\seedmodg\|_{C^1_{\pstar,\expoPm}(\Omega_{R'})} + \|h-\seedmodh\|_{C^0_{\pstar+1,\expoPm}(\Omega_{R'})} \Bigr)
& = 0 .
\endaligned
\ee
Since the argument applies to each of the finitely many ends, and \eqref{estim-Omega0} applies to~$\Omega_0$, the estimates hold on~$\Omega$.
These are the H\"older estimates stated as~\eqref{pstar-less2} in \autoref{theo--beyond-harmonic} and as~\eqref{pstar-less2-SUB} in \autoref{theo--beyond-harmonic-II}.  The latter theorem is therefore established---conditional on \autoref{prop:improve-radial} which is proven in \autoref{section=11}.


\subsection{ADM energy and momentum for the modulators}
\label{section=10.4}

\paragraph{ADM energy.}

In the super-harmonic case $\pstar\geq n-2$ considered in \autoref{theo--beyond-harmonic}, there remains to check~\eqref{energyg2b}.  The coarse bounds $|\mmodu_\iota|+|\Jmodu_\iota|\lesssim\Err^+_{\pstar}[\seedg,\seedh]$ are simply those given in \autoref{prop:improve-radial}.
For the remainder of this subsection, we impose the additional hypothesis under which~\eqref{energyg2b-2} is stated, namely that $g_0=\delta$ globally on $\Omega\subset\RR^n$.
We are interested in the integral of the source term over~$\Omega$, which can be defined as a limit over bounded domains by restricting to a large ball as was done in~\eqref{equa-mstarJstar},
\be
2(n-1) |\Sphe^{n-1}| \mseed \coloneqq - \int_{\Omega} \Hcal(\seedg, \seedh) \,d^nx
= - \lim_{r\to+\infty} \int_{\Omega\cap \Ball_r} \Hcal( \seedg, \seedh) \,d^nx .
\ee
This limit exists because the integrals in each asymptotic end are assumed to converge in \autoref{theo--beyond-harmonic}.
The pointwise bounds~\eqref{bound-proof-gseedmodg} we have obtained (in the super-harmonic case) imply in each asymptotic end $\Omega_\iota$
\bel{bound-proof-gseedmodg-bis}
\lim_{r\to+\infty} \Bigl( \|\gdiff-\gmodu_\iota\|_{C^1_{n-2,\expoPm}(\Omega_{\iota,r})} + \|\hdiff-\hmodu_\iota\|_{C^0_{n-1,\expoPm}(\Omega_{\iota,r})} \Bigr)
= 0 ,
\ee
and in particular $|\gdiff|\lesssim r^{-n+2}$ and $|\hdiff|\lesssim r^{-n+1}$ pointwise.
We decompose $(g,h) = (\seedg,\seedh) + (\gdiff,\hdiff)$ and expand the exact Hamiltonian constraint $\Hcal(g,h)=0$ near $(\seedg,\seedh)$ using~\autoref{lem:lin-constr-coord}, and specifically~\eqref{lin-constr-coord-2}
\bel{minusHcal-seednabla}
\aligned
- \Hcal(\seedg,\seedh)
& = \del_i \del_j \gdiff_{ij} - \del_i \del_i \gdiff_{jj} + \Rcal^\Hcal ,
\\
\Rcal^\Hcal \, & \! \coloneqq
(\seedg- \delta) * \del\del\gdiff
+ \del\seedg * \del\gdiff
+ \del\del\seedg * \gdiff
+ \del\seedg * \del\seedg * \gdiff
+ \seedh * \hdiff
+ \seedh * \seedh * \gdiff
\\
& \quad
+ \del\gdiff * \del\gdiff
+ \gdiff * \del\del\gdiff
+ \hdiff*\hdiff .
\endaligned
\ee
The terms containing a background factor are $\Obig(r^{-n-p_G})$, whereas the terms quadratic in the deformation are $\Obig(r^{-2n+2})$. Thus, $\Rcal^\Hcal$ is integrable ($p_G>0$, $n>2$). It obeys the bound
\bel{int-Rcal-seedRic}
\biggl| \int_\Omega \Rcal^\Hcal \, d^nx \biggr|
\lesssim \Err^+_{\pstar}[\seedg,\seedh] \Bigl( \|\seedg - \delta\|_{C^2_{p_G}(\Omega)} + \|\seedh\|_{C^0_{p_G+1}(\Omega)} + \Err^+_{\pstar}[\seedg,\seedh] \Bigr) ,
\ee
where implicit constants depend on H\"older norms of $(g_0,h_0)$ and $(\seedg,\seedh)$, and on the exponents.

On the other hand, the integral of the first two terms in~\eqref{minusHcal-seednabla} yields a boundary term along a radial shell~$\Lambda_{\iota,r}$ of each asymptotic end (with no contribution from other boundaries thanks to the H\"older estimates $\gdiff,\hdiff=\Obig(\lambda^{\expoPm})$ and $\del\gdiff=\Obig(\lambda^{\expoPm-1})$ and $\expoPm>1$),
\be
\int_{\Omega\cap \Ball_r} \bigl( \del_i \del_j \gdiff_{ij} - \del_i\del_i \gdiff_{jj} \bigr) d^nx
= \sum_{\iota=1}^I \int_{\Lambda_{\iota,r}} \bigl( \del_j \gdiff_{ij} - \del_i \gdiff_{jj} \bigr) \frac{x_i}{r} r^{n-1} d\xh ,
\ee
where $x_i/r$ is the outwards normal and $r^{n-1}d\xh$ is the volume form on $\Lambda_{\iota,r}$ induced by~$\delta$.
In the $r\to+\infty$ limit the integral tends to the relative ADM energy $\mbb(\Omega_\iota, \gdiff) = \mbb(\Omega_\iota, g - \seedg)$ defined in \autoref{def:relative-ADM}.  Thanks to the faster than $r^{-n+2}$ decay of $\gdiff-\gmodu_\iota$ (in $C^1$~norm) in~\eqref{bound-proof-gseedmodg-bis}, we have
\be
\lim_{r\to+\infty} r^{n-1} \int_{\Lambda_{\iota,r}} \frac{x_i}{r} \bigl(\del_j \gdiff_{ij} - \del_i \gdiff_{jj} \bigr) d\xh
= \lim_{r\to+\infty} r^{n-1} \int_{\Lambda_{\iota,r}} \frac{x_i}{r} \bigl(\del_j \gmodu_{\iota ij} - \del_i \gmodu_{\iota jj} \bigr) d\xh
\ee
hence $\mbb(\Omega_\iota, \gdiff) = \mbb(\Omega_\iota, \gmodu_\iota)$.  We have computed explicitly the latter to be equal to $\mmodu_\iota$ in \autoref{lem-ADMenergymod}.
Overall, we find
\be
\mseed - \sum_{\iota=1}^I \mmodu_\iota
= \frac{1}{2(n-1) |\Sphe^{n-1}|} \int_\Omega \Rcal^\Hcal\,d^nx ,
\ee
which we have bounded in~\eqref{int-Rcal-seedRic} as needed.


\paragraph{ADM momentum.}

The computation for the momentum is similar, but the constraint needs to be integrated component-wise.  Applying~\eqref{lin-constr-coord-2} to $(g,h)$ and $(\seedg,\seedh)$ yields
\be
- \Mcal[\seedg,\seedh]_j = \del_i \hdiff_{ij} + \Rcal^\Mcal_j ,
\qquad
\Rcal^\Mcal = (\seedg-\delta)*\del\hdiff + \del\seedg * \hdiff + \seedh * \del\gdiff + \seedh * \del\seedg * \gdiff + \hdiff * \del\gdiff.
\ee
We find
\be
(n-1) |\Sphe^{n-1}| \Jseed_j
= - \int_\Omega \Mcal(\seedg, \seedh)_j\, d^nx
= \lim_{r\to+\infty} \int_{\Lambda_r} \frac{x_i}{r} \hdiff_{ij} r^{n-1} d\xh
+ \int_\Omega \Rcal^\Mcal_j\, d^nx .
\ee
The integral over each connected component of~$\Lambda_r$ converges to the relative ADM momentum $\Jbb(\Omega_\iota,\hdiff)_j$ of \autoref{def:relative-ADM}.  The radial decay (faster than~$r^{-n+1}$) of $\hdiff-\hmodu = h-\seedmodh$ in~\eqref{bound-proof-gseedmodg} ensures that the relative ADM momenta of $\hdiff$ and of~$\hmodu$ coincide.  The latter, in turn was explicitly calculated in~\autoref{lem-ADMmomentmod} to be~$\Jmodu_\iota$.  We conclude that
\be
|\Jmodu - \Jseed| \lesssim \int_\Omega |\Rcal^\Mcal|\, d^nx
\lesssim \Err^+_{\pstar}[\seedg,\seedh] \Bigl( \|\seedg - \delta\|_{C^1_{p_G}(\Omega)} + \|\seedh\|_{C^0_{p_G+1}(\Omega)} + \Err^+_{\pstar}[\seedg,\seedh] \Bigr) .
\ee
This completes the proof of \autoref{theo--beyond-harmonic}, up to establishing \autoref{prop:improve-radial}, which we do next.


\section{Proof of the radial decay improvement (\autoref{prop:improve-radial})}
\label{section=11}

\subsection{Organization of the proof}
\label{section=11.1}

\paragraph{Aim and strategy.}
This section proves \autoref{prop:improve-radial}: assuming the $\Est(\Omega_\iota,\pstar')$ with $\pstar'<n-2$ it establishes Sobolev estimates with radial exponent $q<\pstar''\coloneqq\min(\pstar,\pstar'+p_G,2\pstar')$ and H\"older estimates with $q\leq\pstar''$, except in the case $\min(\pstar'+p_G,2\pstar')=n-2\leq\pstar$ where we get no bounds with the harmonic exponent $q=\pstar''=n-2$.

As explained in \autoref{section=10.1}, the conditions of \refwithname{Theorems}{thm-sharp-h-localized} \refwithname{and}{thm-sharp-m-localized} on the linear Euclidean operators require a control of $\vartheta(\kappa u)$ and $\vartheta(\kappa Z)$ that is not included in $\Est(\Omega_\iota,\pstar')$.
The proof of \autoref{prop:improve-radial} thus proceeds along the following steps, illustrated in \autoref{fig:estimates}.
Our starting point is the weighted H\"older and $H^2\times H^1$ estimates $\Est(\Omega_\iota,\pstar')$ on $(u,Z)$.
\bei

\item[(1)] We easily deduce H\"older and $L^2$ estimates on $(\gdiff,\hdiff)$ from explicit expressions in terms of $(u,Z)$.

\item[(2)] Quadratic terms $\Qcal\Hcal_{(\seedg,\seedh)}[\gdiff,\hdiff]$ in the constraints are then controlled in suitable dual Sobolev norms~$H^{k*}$ by products of H\"older and $L^2$ norms of $(\gdiff,\hdiff)$.

\item[(3)] For the rescaling 
$(u',Z')= r^{\pstar'-p} \kappa(u,Z)$,  the radial derivatives $(\vartheta u',\vartheta Z')$ are variational solutions of an elliptic equation $\Jcal_{(\seedg,\seedh;g_0,h_0)}[\vartheta u',\vartheta Z']=(\text{sources})$, where the sources are bounded in dual Sobolev norms.  This provides $H^2\times H^1$ estimates on $(\vartheta u', \vartheta Z')$ with variational decay exponent~$p$.

\item[(4)] Accounting for the weights, this is a control of $(\kappa\vartheta u,\kappa\vartheta Z)$ with exponent~$\pstar'$.

\item[(5)] We easily deduce~$L^2$ estimates on $(\kappa\vartheta\gdiff,\kappa\vartheta\hdiff)$.

\item[(6)] We then control in H\"older and dual Sobolev norms all sources in~\eqref{linearEuclidenoperatorsource} (and their radial derivative), which include the quadratic terms $\Qcal\Hcal_{(\seedg,\seedh)}$ as well as curved linear terms $d\Hcal_{(\seedg,\seedh)}-d\Hcal_{(\delta,0)}$.  These estimates involve the improved exponent~$\pstar''$.

\item[(7)] The (Euclidean) \refwithname{Theorems}{thm-sharp-h-localized} \refwithname{and}{thm-sharp-m-localized} yield a decay of $(\kappa u,\kappa Z)$ with improved radial exponent.

\item[(8)] Finally, we account for terms with bounded support to bound $(u,Z)$ and conclude.
\eei

\begin{figure}\centering
\begin{tikzpicture}
\node(kuZ) at (0,0) {$\kappa_\iota u,\kappa_\iota Z$};
\node(uZ) at (0,-1) {$u,Z$};
\node(ghdiff) at (0,-2) {$\gdiff,\hdiff$};
\node(QG) at (0,-3.3) {$\vartheta\bigl(\kappa_\iota r^{\pstar'-p}\Qcal\Gcal_{(\seedg,\seedh)}[\gdiff,\hdiff]\bigr)$};
\node(thkuZ) at (9.5,-3.3) {$\vartheta u',\vartheta Z'$};
\node(thuZ) at (9.5,-2.3) {$\kappa_\iota\vartheta u,\kappa_\iota\vartheta Z$};
\node(thghdiff) at (9.5,-1.3) {$\kappa_\iota\vartheta\gdiff,\kappa_\iota\vartheta\hdiff$};
\node(lincurve) at (9.5,0) {$\aligned \kappa_\iota\,\text{sources} & \\[-.5ex] \vartheta(\kappa_\iota\,\text{sources} & )\endaligned$};
\draw[-{stealth}](kuZ) -- (uZ);
\node[shape=circle,draw,inner sep=.5pt] at (.2,-.5){\tiny 8};
\draw[-{stealth}](uZ) -- (ghdiff);
\node[shape=circle,draw,inner sep=.5pt] at (.2,-1.5){\tiny 1};
\draw[-{stealth}](ghdiff) -- (QG);
\node[shape=circle,draw,inner sep=.5pt] at (.2,-2.5){\tiny 2};
\draw[-{stealth}](QG) -- (thkuZ) node [midway,above] {\autoref{appendix=E} variational bounds} node [midway,below] {for curved $\Jcal_{(\seedg,\seedh; g_0,h_0)}$};
\node[shape=circle,draw,inner sep=.5pt] at (2.5,-3.5){\tiny 3};
\draw[-{stealth}](thkuZ) -- (thuZ);
\node[shape=circle,draw,inner sep=.5pt] at (9.3,-2.8){\tiny 4};
\draw[-{stealth}](thuZ) -- (thghdiff);
\node[shape=circle,draw,inner sep=.5pt] at (9.3,-1.8){\tiny 5};
\draw[-{stealth}](thghdiff) -- (lincurve);
\node[shape=circle,draw,inner sep=.5pt] at (9.3,-.8){\tiny 6};
\draw[-{stealth}](lincurve) -- (kuZ) node [midway,above] {\refwithname{Sections}{section=5} \refwithname{to}{section=9} on $\notreH^\lambda,\notreM^\lambda$} node [midway,below] {improved $\pstar'$};
\node[shape=circle,draw,inner sep=.5pt] at (7.5,-.2){\tiny 7};
\node(s102) at (-2.7,-1) {\autoref{section=10.2}, $\pstar'=p$}; \draw[-{stealth}] (s102)--(uZ);
\node(done) at (-2.7,-2) {Done once $\pstar'=\pstar$}; \draw[-{stealth}] (ghdiff)--(done);
\end{tikzpicture}
\caption{\label{fig:estimates}Intermediate steps to improve the decay exponent~$\pstar'$ of the solution $(u,Z)$ in an asymptotic end~$\Omega_\iota$.  Vertical arrows are straightforward.  Horizontal arrows rely on ellipticity of the curved operator $\Jcal_{(\seedg,\seedh; g_0,h_0)}$ and its Euclidean analogue $(\notreH^\lambda,\notreM^\lambda)$, with sources being either quadratic terms $\Qcal\Gcal_{(\seedg,\seedh)}$, or a sum of such terms and linear curved terms outlined in~\eqref{linearEuclidenoperatorsource}.  We enter the loop with bounds on $u,Z$, and eventually exit with bounds on the geometry $(\gdiff,\hdiff)$.  \autoref{prop:improve-radial} captures one iteration of steps 1--8.}
\end{figure}


\paragraph{Convention and notation.}

In the upcoming iteration, we build upon the derivation in \autoref{appendix=C} which led us to various expansions of the operators and nonlinearities associated with the Einstein constraints.  Each asymptotic end~$\Omega_\iota$ is equipped with the reference Euclidean metric~$\delta$, with Levi-Civita connection~$\del$.  Indices are raised and lowered using the metric~$\delta$ and its inverse.  For instance, the tensor fields $g_i{}^j = g_{ik}\delta^{kj}$ and $g^j{}_i = \delta^{jk} g_{ki}$ coincide by virtue of the symmetry of $g$ and of $\delta$, which allows us to denote both of them simply as~$g_i^j$. By convention, for any pair of tensors $A,B$ the product $A * B$ denotes arbitrary index contractions using any of the metrics $\delta,g_0,\seedg,g$ (or their inverse). In addition, we find it convenient to set 
\be
\aligned
A^{*n} & \coloneqq A * \dots * A \quad \text{(with $n$ factors),} 
\\
{\del*}A & \coloneqq \del A + \del g_0 * A + \del\seedg * A + \del g * A ,
\\
{\del*}(A B) \coloneqq {\del*}(A * B) & \coloneqq \del A * B + A * \del B + \del g_0 * A * B + \del\seedg * A * B + \del g * A * B ,
\endaligned
\ee
omitting the terms involving $g_0,\seedg,g$ if the problem at hand does not involve these metrics at all.
It should be emphasized that $\del g= \del(g- \delta)$ is small for metrics close to~$\delta$. With this convention, from our decay and regularity assumptions on the objects $g, h, u, Z$, etc. we are able to deduce the desired decay and regularity properties enjoyed by the operators of interest.
The Sobolev and H\"older norms on $\Omega_\iota$ using the metric $g_0$ or the Euclidean metric~$\delta$ are equivalent; we shall use the latter for convenience.


\subsection{Steps 1--4: control of radial derivatives}
\label{section=11.2}

\paragraph{Towards a control of radial derivatives: set-up for steps 1--4.}

The first half of our work here aims to control radial derivatives $(\vartheta u, \vartheta Z)$ in weighted $H^2\times H^1$ norms.
To begin with, the solution $(u,Z)$ obeys~\eqref{equa-EEHM-33}, namely
\be
\Gcal[\seedg+\gdiff, \seedh+\hdiff] = 0 , \qquad
\gdiff = \omegabf_p^2 \, d\Hcal_{(g_0,h_0)}^{*\flat\flat}[u,Z], \qquad
\hdiff = \omegabf_{p+1}^2  \, d\Mcal_{(g_0,h_0)}^{*\sharp\sharp}[u,Z] .
\ee
From the expansion of the constraints in \autoref{lem:constr-expand}, we deduce
\bel{JuQG}
\Jcal_{(\seedg,\seedh;g_0,h_0)}[u,Z]
= - \Gcal[\seedg,\seedh] - \Qcal\Gcal_{(\seedg,\seedh)}[\gdiff,\hdiff] ,
\ee
with a linear operator $\Jcal_{(\seedg,\seedh;g_0,h_0)}[u,Z]
= d\Gcal_{(\seedg,\seedh)}\bigl[ \omegabf_p^2 \, d\Hcal_{(g_0,h_0)}^{*\flat\flat}[u,Z], \omegabf_{p+1}^2  \, d\Mcal_{(g_0,h_0)}^{*\sharp\sharp}[u,Z] \bigr]$ that already appeared in the fixed-point construction of $(u,Z)$.
It is shown to be (non-self-adjoint) elliptic in \autoref{appendix=E}.

We convert $(u,Z)$ to unknowns that satisfy Neumann boundary conditions on the radial boundary of $\Omega_\iota\simeq\Omega_R$ and have the variational decay $r^{-a_{n,p}/2}$ by defining (note that $\pstar'-p=\astar'-a_{n,p}/2$)
\be
u' = r^{\pstar'-p} \kappa u , \qquad Z' = r^{\pstar'-p} \kappa Z .
\ee
Applying to~\eqref{JuQG} the first-order differential operator $\Dcal_{\pstar'}\coloneqq\vartheta(r^{\pstar'-p}\kappa\ \cdot\ )$ gives an equation of the form $\Jcal_{(\seedg,\seedh;g_0,h_0)}[\vartheta u',\vartheta Z']=\text{source}$, where the left-hand side is understood in the sense of distributions, and the source can be shown to lie in weighted $H^{2*}\times H^{1*}$ spaces.  The invertibility of $\Jcal_{(\seedg,\seedh;g_0,h_0)}$ provides a unique solution in $H^2\times H^1$, but this does not constrain $(\vartheta u',\vartheta Z')$, which only lie in $H^1\times L^2$.
Instead, we rely on a regularized version of the radial derivative~$\vartheta$ and of~$\Dcal_{\pstar'}$: define $\vartheta_0=\vartheta$ and $\Dcal_{\pstar',0}=\Dcal_{\pstar'}$ and, for $\sigma>0$,
\be
\aligned
(\vartheta_\sigma f)(x) & = \frac{1}{\sigma} \int_0^{\sigma} \vartheta f(e^\tau x) d\tau = \frac{f(e^\sigma x) - f(x)}{\sigma} , \qquad x\in \Omega_R ,
\\
\Dcal_{\pstar',\sigma} f & = \vartheta_\sigma(r^{\pstar'-p}\kappa f) .
\endaligned
\ee
One has $\vartheta_\sigma f\to\vartheta f$ and $\Dcal_{\pstar',\sigma}f\to\Dcal_{\pstar'}f$ as $\sigma\to 0$ in the sense of distributions.
Then the following identity holds,
\bse\label{Jthu}
\be
\aligned
\Jcal_{(\seedg,\seedh;g_0,h_0)}[\vartheta_\sigma u',\vartheta_\sigma Z']
& = - \Dcal_{\pstar',\sigma} \bigl(\Gcal[\seedg,\seedh]\bigr) - \Dcal_{\pstar',\sigma} \bigl(\Qcal\Gcal_{(\seedg,\seedh)}[\gdiff,\hdiff]\bigr) + \Rcal ,
\endaligned
\ee
where the remainder
\be
\aligned
\Rcal \, & \! \coloneqq \bigl(\Jcal_{(\seedg,\seedh;g_0,h_0)} \circ \Dcal_{\pstar',\sigma} - \Dcal_{\pstar',\sigma} \circ \Jcal_{(\seedg,\seedh;g_0,h_0)}\bigr)[u,Z] \\
& = \Jcal_{(\seedg,\seedh;g_0,h_0)}[\vartheta_\sigma(r^{\pstar'-p} \kappa u),\vartheta_\sigma(r^{\pstar'-p} \kappa Z)] - \vartheta_\sigma\bigl(r^{\pstar'-p} \kappa\Jcal_{(\seedg,\seedh;g_0,h_0)}[u,Z]\bigr)
\endaligned
\ee
\ese
captures the effect of commuting $\Dcal_{\pstar',\sigma}$ through the operator~$\Jcal$.
As stated in \autoref{lem:appE-invertibility}, the operator~$\Jcal$ from $H^2_{n-2-p,-\expoP}\times H^1_{n-2-p,-\expoP}$ to $H^{2*}_{p+2,\expoP}\times H^{1*}_{p+2,\expoP}$ has a bounded inverse.
Our aim in the upcoming few steps is thus to control the sources of~\eqref{Jthu} in $H^{2*}_{p+2,\expoP}\times H^{1*}_{p+2,\expoP}$ norm, so as to deduce a weighted $H^2\times H^1$ control of $(\vartheta_\sigma u',\vartheta_\sigma Z')$.  These bounds are uniform in~$\sigma$ and yield bounds on the $\sigma\to 0$ limit $(\vartheta u',\vartheta Z')$ in these spaces.


\paragraph{Step 1: Estimates on the geometry.}

We now estimate $\gdiff = \omega_p^2 \, d\Hcal^{*\flat\flat}_{(g_0,h_0)}[u,Z]$ and $\hdiff = \omega_{p+1}^2 \, d\Mcal^{*\sharp\sharp}_{(g_0,h_0)}[u,Z]$ (see~\eqref{equa-EEHM-33}) in an asymptotic end using the near-Euclidean expansions of the adjoint linearized constraints from \autoref{prop:adjoint-constr-coord}, further simplified since we do not need the detailed tensor structure presently,
\bel{omegapgdiff-uZexpr}
\aligned
\omega_p^{-2} \gdiff = d\Hcal_{(g_0,h_0)}^{*\flat\flat}[u,Z]
& = {\del*} {\del*} u + h_0 * h_0 * u + {\del*}(h_0 * Z) ,
\\
\omega_{p+1}^{-2} \hdiff = d\Mcal_{(g_0,h_0)}^{*\sharp\sharp}[u,Z] & = {\del*}Z + h_0 * u .
\endaligned
\ee
It is routine to extract weighted H\"older and $L^2$ estimates on $(\gdiff,\hdiff)$.

\begin{lemma}\label{lem:sec11-estimates-geom}
One has
\[
\aligned
\|\gdiff\|^{N,\alpha}_{\Omega_\iota,\pstar',\expoPm} + \|\hdiff\|^{N,\alpha}_{\Omega_\iota,\pstar'+1,\expoPm}
& \simeq \bigl\|d\Gcal_{(g_0,h_0)}^*[u,Z]\bigr\|_{C^{N,\alpha}_{\astar'+2,-\expoPp}(\Omega_\iota)\times C^{N,\alpha}_{\astar'+1,-\expoPp}(\Omega_\iota)}
\\
& \lesssim
\| u \|^{N+2, \alpha}_{\Omega_\iota, \astar', -\expoPp+2}
+ \| Z \|^{N+1, \alpha}_{\Omega_\iota, \astar', -\expoPp+1}
\endaligned
\]
and
\[
\aligned
\|\gdiff\|_{L^2_{\pstar',\expoP}(\Omega_\iota)} {+} \|\hdiff\|_{L^2_{\pstar'+1,\expoP}(\Omega_\iota)}
& \simeq \bigl\|d\Gcal_{(g_0,h_0)}^*[u,Z]\bigr\|_{L^2_{\astar'+2,-\expoP}(\Omega_\iota)\times L^2_{\astar'+1,-\expoP}(\Omega_\iota)}
\\
& \lesssim \| u \|_{H^2_{\astar', -\expoP}(\Omega_\iota)}
+ \| Z \|_{H^1_{\astar', -\expoP}(\Omega_\iota)} ,
\endaligned
\]
with implicit constants depending on the unweighted $C^{N+2,\alpha}(\Omega_\iota)$ norm of $g_0$ and weighted $C^{N+1,\alpha}_1(\Omega_\iota)$ norm of~$h_0$, controlled respectively by the $C^{N+2,\alpha}_{p_G}(\Omega_\iota)$ and $C^{N+1,\alpha}_{p_G+1}(\Omega_\iota)$ norms in the definition of conical localization data set in \autoref{def-conical}.
\end{lemma}


\paragraph{Step 2: Controlling nonlinear terms.}

We turn to the quadratic terms $\Qcal\Gcal_{(\seedg,\seedh)}[\gdiff,\hdiff]$.
The bounds we establish on radial derivatives involve the radial exponent $\pstar'+p+2$, which is larger than the exponent $p+2$ that we need.  Note also that the H\"older regularity required for $(\seedg,\seedh)$ is lower than that of $(g_0,h_0)$.

\begin{lemma}\label{lem:step2-nonlinear}
  Under the smallness conditions of \autoref{thm:sts-existence},  one has, uniformly in $\sigma\geq 0$,
  \be
  \aligned
 & \bigl\| \Qcal\Hcal_{(\seedg,\seedh)}[\gdiff,\hdiff] \bigr\|_{H^{1*}_{2\pstar'+2,\expoP}(\Omega_\iota)}
  + \bigl\| \Dcal_{\pstar',\sigma}\bigl(\Qcal\Hcal_{(\seedg,\seedh)}[\gdiff,\hdiff]\bigr) \bigr\|_{H^{2*}_{\pstar'+p+2,\expoP}(\Omega_\iota)}
  \\
  & + \bigl\| \Qcal\Mcal_{(\seedg,\seedh)}[\gdiff,\hdiff] \bigr\|_{L^2_{2\pstar'+2,\expoP}(\Omega_\iota)}
  + \bigl\| \Dcal_{\pstar',\sigma}\bigl(\Qcal\Mcal_{(\seedg,\seedh)}[\gdiff,\hdiff]\bigr) \bigr\|_{H^{1*}_{\pstar'+p+2,\expoP}(\Omega_\iota)}
  \\
  & \qquad \lesssim
  \bigl(\|\gdiff\|^{N,\alpha}_{\Omega_\iota,\pstar',\expoPm} + \|\hdiff\|^{N,\alpha}_{\Omega_\iota,\pstar'+1,\expoPm}\bigr)
  \bigl(\|\gdiff\|_{L^2_{\pstar',\expoP}(\Omega_\iota)} + \|\hdiff\|_{L^2_{\pstar'+1,\expoP}(\Omega_\iota)}\bigr) ,
  \endaligned
  \ee
  with implicit constants depending on norms $\|g_0\|_{C^{N+2,\alpha}}$, $\|\seedg\|_{C^{N,\alpha}}$, $\|h_0\|_{C^{N+1,\alpha}_1}$, $\|\seedh\|_{C^{N,\alpha}_1}$ on~$\Omega_\iota$ and uniform-equivalence constants of $g,\seedg,g_0$, controlled by the norms in the definition of localized seed data set in \autoref{def-aset} and conical localization data set in \autoref{def-conical}.
\end{lemma}

\begin{proof}
We rely on the elliptic estimates to control one factor in each nonlinear term given in \autoref{lem:constr-expand}.  Importantly, while the localization exponent~$\expoPm$ in elliptic estimates is worse than~$\expoP$, it still remains positive, namely $(\gdiff,\hdiff)$ decay pointwise at the angular boundary.
The momentum nonlinearity has the form
$\Qcal\Mcal_{(\seedg,\seedh)}[\gdiff,\hdiff]
=\seedh*\gdiff*\nabla^{\seedg}\gdiff+\hdiff*\nabla^{\seedg}\gdiff$.
Its $L^2$ norm is bounded by the $L^2$ norm of $\gdiff$ or~$\hdiff$ and the supremum norm of $\nabla^{\seedg}\gdiff$, itself bounded by the H\"older norm thanks to $\expoPm>1$ and $N\geq 1$.
Next, we consider the quadratic Hamiltonian terms
\be
\aligned
\Qcal\Hcal[\gdiff,\hdiff]
= \nabla^{\seedg}\gdiff * \nabla^{\seedg}\gdiff
+ \gdiff * \nabla^{\seedg}\nabla^{\seedg}\gdiff
+ \Ric^{\seedg} * \gdiff * \gdiff
+ \seedh*\seedh*\gdiff*\gdiff + \seedh*\gdiff*\hdiff + \hdiff*\hdiff.
\endaligned
\ee
Apart from the $\nabla^{\seedg}\gdiff * \nabla^{\seedg}\gdiff$ term, all others can be bounded in $L^2$ by bounding a factor $\gdiff$ or~$\hdiff$ in $L^2$ and the other pointwise (using $\expoPm>2$ and $N\geq 2$).  This first term can be rewritten, up to a term $\gdiff * \nabla^{\seedg}\nabla^{\seedg}\gdiff$ controlled in~$L^2$, as a derivative $\nabla^{\seedg}(\gdiff * \nabla^{\seedg}\gdiff)$ of a product in~$L^2$, hence it is controlled in $H^{1*}$~norm as needed.
The discrete radial derivative of $\Dcal_{\pstar',\sigma}\Qcal\Gcal_{(\seedg,\seedh)}[\gdiff,\hdiff]$ for $\sigma>0$ is controlled in the same weighted $H^{1*}\times L^2$ spaces, but non-uniformly in~$\sigma$.  To get a uniform bound in weighted $H^{2*}\times H^{1*}$ spaces, observe that~$\Dcal_{\pstar'}\Qcal\Gcal_{(\seedg,\seedh)}[\gdiff,\hdiff]$ is bounded in these spaces, that radial dilations are operators of norm $\leq 1$ on these spaces, and that $\Dcal_{\pstar',\sigma}$ is obtained as an average of $\Dcal_{\pstar'}$ under a set of radial dilations.
\end{proof}


\paragraph{Step 3: Radial derivatives of the solution.}

We consider the source terms in~\eqref{Jthu}.
We have just provided bounds on the nonlinear sources~$\Qcal\Gcal$.
Next, the $L^2_{\pstar'+2,\expoP}$ norm of $\Gcal[\seedg,\seedh]$ is controlled by the $L^2_{\pstar+2,\expoP}$ norm included in $\Err^+_{\pstar}[\seedg,\seedh]$, hence the source term $- \Dcal_{\pstar',\sigma}(\Gcal[\seedg,\seedh])$ is bounded uniformly in $\sigma\geq 0$ in $H^{1*}$ norm (which we can freely weaken to~$H^{2*}$ in the case of the Hamiltonian),
\bel{Gcalprime-ctrl}
\bigl\| - \vartheta \bigl( r^{\pstar'-p} \kappa \Gcal[\seedg,\seedh] \bigr) \bigr\|_{H^{2*}_{p+2,\expoP}(\Omega_\iota)\times H^{1*}_{p+2,\expoP}(\Omega_\iota)}
\lesssim \Err^+_{\pstar}[\seedg,\seedh] .
\ee
The remaining term~$\Rcal$ is a commutator of the operator $\Dcal_{\pstar',\sigma}$ and of~$\Jcal$, hence it has the same order as~$\Jcal$, which is sufficient for our needs, as stated next.  Its components corresponding to the Hamiltonian and momentum constraints are bounded in $H^{2*}_{p+2,\expoP}$ and $H^{1*}_{p+2,\expoP}$ norms, respectively.

\begin{lemma}\label{lem:sec10-commute}
  The remainder $\Rcal$ in~\eqref{Jthu} obeys, uniformly in $\sigma\geq 0$,
  \be
  \aligned
  \|\Rcal\|_{H^{2*}_{p+2,\expoP}(\Omega_\iota)\times H^{1*}_{p+2,\expoP}(\Omega_\iota)}
  \lesssim \Err^+_{\pstar}[\seedg,\seedh] ,
  \endaligned
  \ee
  where implicit constants depend on $\|g_0,\seedg\|_{C^3(\Omega_\iota)}$ and $\|h_0,\seedh\|_{C^2_1(\Omega_\iota)}$, which are controlled by the available $C^{N,\alpha}_{p_G}$ and $C^{N,\alpha}_{p_G+1}$ norms for $N\geq 3$.
\end{lemma}

\begin{proof}
We present the proof for $\sigma=0$, as $\sigma>0$ is obtained by averaging radially.

The linear operator is $\Jcal_{(\seedg,\seedh;g_0,h_0)} = d\Gcal_{(\seedg,\seedh)} \circ \diag(\omega_p^2,\omega_{p+1}^2) \circ (d\Hcal_{(g_0,h_0)}^{*\flat\flat},d\Mcal_{(g_0,h_0)}^{*\sharp\sharp})$, in which $\diag(\omega_p^2,\omega_{p+1}^2)$ stands for multiplication of the two components by these two weights.  The expressions of $d\Gcal$ and $d\Gcal^*$ are given in \autoref{lem:constr-expand} and \autoref{lem:lin-constr}; each of these operators maps a pair of tensors to another pair of tensors hence can be written in matrix form,
\be\compresseq{.5}\setlength\arraycolsep{0pt}
d\Gcal_{(\seedg,\seedh)} =
\begin{pmatrix}
  {\del*}{\del*} + {\seedh*}{\seedh*} \quad & \seedh* \\
  {\seedh*}{\del*} & {\del*}
\end{pmatrix} ,
\quad\ \
(d\Hcal_{(g_0,h_0)}^{*\flat\flat},d\Mcal_{(g_0,h_0)}^{*\sharp\sharp}) =
\begin{pmatrix}
  {\del*}{\del*} + {h_0*}{h_0*} \quad & {\del*}{h_0*} \\
  {h_0*} & {\del*}
\end{pmatrix} .
\ee
Commutators with $\Dcal_{\pstar'} = \vartheta(r^{\pstar'-p}\kappa\ \cdot\ )$ are differential operators of the same order, and with radial decay exponents shifted by $\pstar'-p$: for instance, using $[\vartheta,\del_i]=-\del_i$, and using that $\del_i\kappa$ is compactly supported, we get
\bel{deli-Dcalpstarp-comm}
\gathered
\bigl[\del_i, \Dcal_{\pstar'}\bigr](u)
= r^{\pstar'-p}\kappa \del_i u + \del_i(r^{\pstar'-p}\kappa) \vartheta u + \del_i \vartheta(r^{\pstar'-p}\kappa) u
\\
\bigl\| \bigl[\del_i, \Dcal_{\pstar'}\bigr](u) \bigr\|_{H^1_{n-1-p,-\expoP}(\Omega_\iota)}
\lesssim \| u \|_{H^2_{\astar',-\expoP}(\Omega_\iota)} .
\endgathered
\ee
Similar calculations yield
\be
\aligned
\bigl[d\Hcal_{(g_0,h_0)}^{*\flat\flat}, \Dcal_{\pstar'}\bigr][u,Z]
& =
[{\del*}{\del*}, \Dcal_{\pstar'}](u) + [{h_0*}{h_0*}, \Dcal_{\pstar'}](u) + [{\del*}{h_0*},\Dcal_{\pstar'}](Z) ,
\\
\bigl[d\Mcal_{(g_0,h_0)}^{*\sharp\sharp}, \Dcal_{\pstar'}\bigr][u,Z]
& = [{h_0*}, \Dcal_{\pstar'}](u) + [{\del*},\Dcal_{\pstar'}](Z) ,
\endaligned
\ee
which involve derivatives of the metrics $\delta$ and $g_0$ that are implicit in the notation~$*$, and (derivatives of)~$h_0$.  We deduce
\bel{dGcal-Dcalpstarp-comm}
\aligned
& \Bigl\|\Bigl[(d\Hcal_{(g_0,h_0)}^{*\flat\flat},d\Mcal_{(g_0,h_0)}^{*\sharp\sharp}), \Dcal_{\pstar'}\Bigr][u,Z]\Bigr\|_{L^2_{n-p,-\expoP}(\Omega_\iota)\times L^2_{n-p-1,-\expoP}(\Omega_\iota)}
\\
& \quad \lesssim \| u \|_{H^2_{\astar',-\expoP}(\Omega_\iota)} + \| Z \|_{H^1_{\astar',-\expoP}(\Omega_\iota)} ,
\endaligned
\ee
and, after acting with $\diag(\omega_p^2,\omega_{p+1}^2)$ and $d\Gcal_{(\seedg,\seedh)}$,
\be
\aligned
\quad & \unquad \Bigl\|d\Gcal_{(\seedg,\seedh)}\circ\diag(\omega_p^2,\omega_{p+1}^2)\circ\bigl[(d\Hcal_{(g_0,h_0)}^{*\flat\flat},d\Mcal_{(g_0,h_0)}^{*\sharp\sharp}), \Dcal_{\pstar'}\bigr][u,Z]\Bigr\|_{H^{2*}_{p+2,\expoP}(\Omega_\iota)\times H^{1*}_{p+2,\expoP}(\Omega_\iota)} 
\\
& \lesssim \| u \|_{H^2_{\astar',-\expoP}(\Omega_\iota)} + \| Z \|_{H^1_{\astar',-\expoP}(\Omega_\iota)} ,
\endaligned
\ee
The commutator $[\omega_q^2,\Dcal_{\pstar'}]=-(n-2q)\omega_q^2 r^{\pstar'-p}\kappa$ for $q=p$ and $p+1$ adds innocuous terms that are bounded in the same way.
We are left with the commutator $[d\Gcal_{(\seedg,\seedh)},\Dcal_{\pstar'}](\gdiff,\hdiff)$,
\be
\aligned\relax
[d\Hcal_{(\seedg,\seedh)} , \Dcal_{\pstar'}](\gdiff,\hdiff) & = [{\del*}{\del*},\Dcal_{\pstar'}](\gdiff) + [{\seedh*}{\seedh*},\Dcal_{\pstar'}](\gdiff) + [\seedh*,\Dcal_{\pstar'}](\hdiff) ,
\\
[d\Mcal_{(\seedg,\seedh)} , \Dcal_{\pstar'}](\gdiff,\hdiff) & = [{\seedh*}{\del*},\Dcal_{\pstar'}](\gdiff) + [{\del*},\Dcal_{\pstar'}](\hdiff) ,
\endaligned
\ee
which we bound using the control of $\gdiff$ and $\hdiff$ in \autoref{lem:sec11-estimates-geom},
\be
\aligned
\bigl\| [d\Gcal_{(\seedg,\seedh)} , \Dcal_{\pstar'}](\gdiff,\hdiff) \bigr\|_{H^{2*}_{p+2,\expoP}(\Omega_\iota)\times H^{1*}_{p+2,\expoP}(\Omega_\iota)}
& \lesssim \|\gdiff\|_{L^2_{\pstar',\expoP}(\Omega_\iota)} + \|\hdiff\|_{L^2_{\pstar'+1,\expoP}(\Omega_\iota)}
\\
& \lesssim \| u \|_{H^2_{\astar',-\expoP}(\Omega_\iota)} + \| Z \|_{H^1_{\astar',-\expoP}(\Omega_\iota)} ,
\endaligned
\ee
where the implicit constant depends on $\|\seedg\|_{C^3(\Omega_\iota)}$ and $\|\seedh\|_{C^2_1(\Omega_\iota)}$.  The need for three derivatives of $\seedg$ (hence the condition $N\geq 3$ in our main theorem) comes solely from the term
$\Dcal_{\pstar'}(\del\del\seedg*\gdiff)$ in $\Dcal_{\pstar'}\circ d\Hcal_{(\seedg,\seedh)}[\gdiff,\hdiff]$.
Overall, each term in
\be
\aligned
\Rcal & = \bigl[d\Gcal_{(\seedg,\seedh)} , \Dcal_{\pstar'} \bigr](\gdiff,\hdiff)
- d\Gcal_{(\seedg,\seedh)}\bigl[(n-2p)r^{\pstar'-p}\kappa\gdiff,(n-2-2p)r^{\pstar'-p}\kappa\hdiff\bigr] \\
& \quad + d\Gcal_{(\seedg,\seedh)} \circ \diag(\omega_p^2,\omega_{p+1}^2) \circ \bigl[(d\Hcal_{(g_0,h_0)}^{*\flat\flat},d\Mcal_{(g_0,h_0)}^{*\sharp\sharp}), \Dcal_{\pstar'}\bigr](u,Z)
\endaligned
\ee
is bounded by $\| u \|_{H^2_{\astar',-\expoP}(\Omega_\iota)} + \| Z \|_{H^1_{\astar',-\expoP}(\Omega_\iota)}$, itself bounded by $\Err^+_{\pstar}[\seedg,\seedh]$ by $\Est(\Omega_\iota,\pstar')$ given in~\eqref{Est-def}.
\end{proof}


\paragraph{Step 4: Summarizing bounds on the solution.}

We have obtained bounds on the sources in~\eqref{Jthu}, from which we now deduce bounds on $(\vartheta_\sigma u',\vartheta_\sigma Z')$ for $\sigma>0$, then on $(\vartheta u',\vartheta Z')$ and finally $(\vartheta u,\vartheta Z)$.

\begin{lemma}\label{lem:kappavarthetau}
The radial derivatives of the solution obey
\be
\aligned
\| \kappa\vartheta u \|_{H^2_{\astar', -\expoP}(\Omega_\iota)}
+ \| \kappa\vartheta Z \|_{H^1_{\astar', -\expoP}(\Omega_\iota)}
& \lesssim \Err^+_{\pstar}[\seedg,\seedh]
\endaligned
\ee
where implicit constants depend on $\|g_0,\seedg\|_{C^3(\Omega_\iota)}$ and $\|h_0,\seedh\|_{C^2_1(\Omega_\iota)}$, as well as on the fixed localization geometry and exponents of the end~$\Omega_\iota$.
\end{lemma}

\begin{proof}
  Thanks to~\eqref{Gcalprime-ctrl}, \autoref{lem:step2-nonlinear} and~\autoref{lem:sec10-commute}, the  $H^{2*}_{p+2,\expoP}(\Omega_\iota)\times H^{1*}_{p+2,\expoP}(\Omega_\iota)$ norm of the source terms in \eqref{Jthu} is bounded by $\Err^+_{\pstar}[\seedg,\seedh]$.
  The estimates on the elliptic operator $\Jcal_{(\seedg,\seedh;g_0,h_0)}$ in~\autoref{lem:appE-invertibility} then imply that, uniformly in $\sigma>0$,
  \bel{varthetauprime-bnd-sigma}
  \| \vartheta_\sigma u' \|_{H^2_{a_{n,p}/2,-\expoP}(\Omega_\iota)}
  + \| \vartheta_\sigma Z' \|_{H^1_{a_{n,p}/2,-\expoP}(\Omega_\iota)}
  \lesssim \Err^+_{\pstar}[\seedg,\seedh] .
  \ee
  By weak compactness, there exists a sequence $\sigma_j\to 0$ such that $(\vartheta_{\sigma_j} u',\vartheta_{\sigma_j} Z')$ converges weakly in these spaces.  The weak limit must coincide with the distributional limit $(\vartheta u',\vartheta Z')$ and lower semi-continuity then gives a control of the radial derivative
  \bel{varthetauprime-bnd}
  \| \vartheta u' \|_{H^2_{a_{n,p}/2,-\expoP}(\Omega_\iota)}
  + \| \vartheta Z' \|_{H^1_{a_{n,p}/2,-\expoP}(\Omega_\iota)}
  \lesssim \Err^+_{\pstar}[\seedg,\seedh] .
  \ee
  Then $\kappa\vartheta u = r^{-\pstar'+p}\vartheta u' - ((\pstar' - p) \kappa + \vartheta \kappa) u$ and likewise for~$Z$ implies
  \be
  \aligned
  \| \kappa \vartheta u \|_{H^2_{\astar',-\expoP}(\Omega_\iota)}
  & \lesssim \| \vartheta u' \|_{H^2_{a_{n,p}/2,-\expoP}(\Omega_\iota)} + \| u \|_{H^2_{\astar',-\expoP}(\Omega_\iota)} ,
  \\
  \| \kappa \vartheta Z \|_{H^1_{\astar',-\expoP}(\Omega_\iota)}
  & \lesssim
  \| \vartheta Z' \|_{H^1_{a_{n,p}/2,-\expoP}(\Omega_\iota)}
  + \| Z \|_{H^1_{\astar',-\expoP}(\Omega_\iota)} .
  \endaligned
  \ee
  Here $\kappa$ is a fixed smooth radial cutoff with bounded scale-invariant derivatives; all of its nonzero derivatives are supported in a fixed annulus.
  The terms on both right-hand sides are bounded by $\Err^+_{\pstar}[\seedg,\seedh]$ thanks to~\eqref{varthetauprime-bnd} and~\eqref{Est-def}, respectively.  (The harmonic shift in~\eqref{Est-def} is absent here since $\pstar'<n-2$.)
\end{proof}


\subsection{Steps 5--8: proof of the radial decay improvement}
\label{section=11.3}

\paragraph{Towards a better radial decay: set-up for steps 5--8.}

The Euclidean space version of the squared constraint operator~$\Jcal$ is $\Jcal_{(\delta,0;\delta,0)}[u,Z] = (\omega_p^2\notreH[u], \omega_{p+1}^2\notreM[Z])$ in terms of the operators studied in \refwithname{Sections}{section=5} \refwithname{to}{section=9}.
In this second stage of the proof, we consider the constraints~\eqref{JuQG} in the form
\bel{eq-kappau}
\notreH[\kappa u] = E_\kappa , \qquad \notreM[\kappa Z] = F_\kappa
\ee
with sources specified shortly.
For the present discussion, in analogy to $\cutoff_{\pstar''}$, we introduce the exceptional indicator
\bel{exceptional-indicator}
\cutoff_{\pstar''}^{\textnormal{ex}} = \begin{cases}
0 & \text{if } \pstar'' < n-2 \text{ or } \pstar''= n-2 = \min(\pstar'+p_G,2\pstar')  , \\
1 & \text{if } \pstar''> n-2 \text{ or } \pstar''= n-2 < \min(\pstar'+p_G,2\pstar') ,
\end{cases}
\ee
where we recall that $\pstar''=\min(\pstar,\pstar'+p_G,2\pstar')$.
In other words, $\cutoff_{\pstar''}^{\textnormal{ex}}$ coincides with $\cutoff_{\pstar''}$ except in the special case $\pstar''=\min(\pstar'+p_G,2\pstar')=n-2$ singled out in \autoref{prop:improve-radial} as lacking the pointwise control with exponent~$\pstar''$.
Just as for the indicator $\cutoff_{\pstar''}$, we use the convention that products of the form $\cutoff_{\pstar''}^{\textnormal{ex}}A$ vanish when the indicator vanishes, regardless of whether~$A$ is defined.

\refwithname{Theorems}{thm-sharp-h-localized} \refwithname{and}{thm-sharp-m-localized} imply that solutions decay with improved radial exponent, provided the sources obey suitable H\"older and dual Sobolev bounds, specifically\footnote{Since $\kappa$ (hence~$E_\kappa$) vanishes in a neighborhood of the radial boundary of $\Omega_\iota$, no distinction is needed between $\vartheta E_\kappa$ and the distributional derivative~$\vartheta_* E_\kappa$.}. 
\bel{neededEFbounds}
\aligned
&
\| E_\kappa \|_{\Omega_\iota,\astar''+4,-\expoPp-2}^{N-2,\alpha} + \|F_\kappa\|^{N-1,\alpha}_{\Omega_\iota,\astar''+2,-\expoPp-1}
 + \| E_\kappa,\vartheta E_\kappa \|_{H^{2*}_{\astar''+4,-\expoP}(\Omega_\iota)}
 + \|F_\kappa\|_{H^{1*}_{\astar''+2,-\expoP}(\Omega_\iota,\RR^n)}
\\
& \quad + \cutoff_{\pstar''}^{\textnormal{ex}} \mmax(E_\kappa)
+ \cutoff_{\pstar''}^{\textnormal{ex}} \Jmax(F_\kappa)
+ \|\vartheta(\kappa u)\|_{H^2_{n-2-p,-\expoP}(\Omega_\iota)}
\lesssim \Err^+_{\pstar}[\seedg,\seedh] .
\endaligned
\ee
The last term is controlled by \autoref{lem:kappavarthetau} and $\Est(\Omega_\iota,\pstar')$ since $n-2-p\leq\astar'$.
The other terms are bounded in \refwithname{Lemmas}{lem:EF-Holder}\refwithname{,}{lem:EF-energymom} \refwithname{and}{lem:EF-Sobolev} below.

Let us write down the sources explicitly.
From the nonlinear curved equation~\eqref{JuQG} and from the expression of the operators~$\Jcal$ as compositions of linearized and adjoint constraints, we get
\bse\label{EF-curved}
\be
(E_\kappa,F_\kappa) = T^\Gcal_\kappa[u, Z] + T^\Gcal_\adj[u, Z] + T^\Gcal_\lin[\gdiff,\hdiff] + T^\Gcal_\qua[\gdiff,\hdiff] + T^\Gcal_\seed ,
\ee
with
\be
\aligned
T^\Gcal_\kappa[u,Z] & \coloneqq \bigl(\notreH[\kappa u]-\kappa\notreH[u], \ \notreM[\kappa Z]-\kappa\notreM[Z]\bigr) ,
\\
T^\Gcal_\adj[u,Z] & \coloneqq
- \diag(\omega_p^{-2},\omega_{p+1}^{-2}) \kappa d\Gcal_{(\delta,0)}\Bigl[ \diag(\omega_p^2,\omega_{p+1}^2) \\[-.5ex]
& \qquad\qquad \bigl( d\Hcal_{(g_0,h_0)}^{*\flat\flat} - d\Hcal_{(\delta,0)}^{*\flat\flat}, d\Mcal_{(g_0,h_0)}^{*\sharp\sharp} - d\Mcal_{(\delta,0)}^{*\sharp\sharp}\bigr)[u,Z] \Bigr] ,
\\
T^\Gcal_\lin[\gdiff,\hdiff] & \coloneqq - \diag(\omega_p^{-2},\omega_{p+1}^{-2}) \kappa (d\Gcal_{(\seedg,\seedh)} - d\Gcal_{(\delta,0)})[\gdiff,\hdiff] ,
\\
T^\Gcal_\qua[\gdiff,\hdiff] & \coloneqq - \diag(\omega_p^{-2},\omega_{p+1}^{-2}) \kappa \Qcal\Gcal_{(\seedg,\seedh)}[\gdiff,\hdiff] ,
\\
T^\Gcal_\seed & \coloneqq - \diag(\omega_p^{-2},\omega_{p+1}^{-2}) \kappa \Gcal[\seedg,\seedh] .
\endaligned
\ee
\ese
The contribution of the seed term~$T^\Gcal_\seed$ to the norms~\eqref{neededEFbounds} is controlled by $\Err^+_{\pstar''}[\seedg,\seedh]$ (hence by $\Err^+_{\pstar}[\seedg,\seedh]$) since $H^{2*}$ norms of $E_\kappa$ and $\vartheta E_\kappa$ and the $H^{1*}$ norm of~$F_\kappa$ are bounded by the $L^2$ norms of~$E_\kappa$ and~$F_\kappa$.  For the previous two terms we need a control of $\gdiff$ and~$\hdiff$.


\paragraph{Step 5: Radial derivatives of the geometry.}

We begin with $L^2$ estimates on $(\kappa\vartheta\gdiff,\kappa\vartheta\hdiff)$, based on the expressions $\gdiff = \omega_p^2 ( {\del*} {\del*} u + h_0 * h_0 * u + {\del*}(h_0 * Z))$ and $\hdiff = \omega_{p+1}^2({\del*}Z + h_0 * u)$ in~\eqref{omegapgdiff-uZexpr}.
This involves commuting the first-order operator $\kappa\vartheta$ through the differential operators $d\Hcal_{(g_0,h_0)}^{*\flat\flat}$ and $d\Mcal_{(g_0,h_0)}^{*\sharp\sharp}$ that define $(\gdiff,\hdiff)$ in terms of $(u,Z)$.  This produces lower-order derivative terms, and derivatives of $g_0,h_0,\kappa,\omega_p^2,\omega_{p+1}^2$.  It leads to an analogue of \autoref{lem:sec11-estimates-geom}.

\begin{lemma}\label{lem:kappavarthetagdiff}
The radial derivatives of the metric and extrinsic curvature obey
\be
\aligned
\|\kappa\vartheta\gdiff\|_{L^2_{\pstar',\expoP}(\Omega_\iota)} + \|\kappa\vartheta\hdiff\|_{L^2_{\pstar'+1,\expoP}(\Omega_\iota)}
& \lesssim \| u , \kappa\vartheta u\|_{H^2_{\astar', -\expoP}(\Omega_\iota)}
+ \| Z , \kappa\vartheta Z \|_{H^1_{\astar', -\expoP}(\Omega_\iota)} ,
\endaligned
\ee
where implicit constants depend on $\|g_0\|_{C^3(\Omega_\iota)}$ and $\|h_0\|_{C^2_1(\Omega_\iota)}$.
\end{lemma}

\begin{proof}
Since $\vartheta(\omega_p^2) = (n-2p)\omega_p^2$ one has
\be
\aligned
 &  \|\kappa\vartheta\gdiff\|_{L^2_{\pstar',\expoP}(\Omega_\iota)} + \|\kappa\vartheta\hdiff\|_{L^2_{\pstar'+1,\expoP}(\Omega_\iota)}
\\
& \lesssim \|\kappa\gdiff\|_{L^2_{\pstar',\expoP}(\Omega_\iota)} + \|\kappa\hdiff\|_{L^2_{\pstar'+1,\expoP}(\Omega_\iota)}
+ \Bigl\|\kappa\vartheta(\omega_p^{-2}\gdiff), \ \kappa\vartheta(\omega_{p+1}^{-2}\hdiff)\Bigr\|_{L^2_{\astar'+2,-\expoP}(\Omega_\iota)\times L^2_{\astar'+1,-\expoP}(\Omega_\iota)}
\\
& \lesssim \|\gdiff\|_{L^2_{\pstar',\expoP}(\Omega_\iota)} + \|\hdiff\|_{L^2_{\pstar'+1,\expoP}(\Omega_\iota)}
+ \bigl\|\kappa\vartheta(d\Hcal_{(g_0,h_0)}^{*\flat\flat}[u,Z])\bigr\|_{L^2_{\astar'+2,-\expoP}(\Omega_\iota)} \\
& \quad + \bigl\|\kappa\vartheta(d\Mcal_{(g_0,h_0)}^{*\sharp\sharp}[u,Z])\bigr\|_{L^2_{\astar'+1,-\expoP}(\Omega_\iota)}
\endaligned
\ee
Thanks to \autoref{lem:sec11-estimates-geom}, the first term is controlled by the (weighted) $H^2$ norm of~$u$ and $H^1$ norm of~$Z$.
We split the last term (and the previous one likewise) into two parts.  Firstly, $d\Mcal_{(g_0,h_0)}^{*\sharp\sharp}$ acting on $(\kappa\vartheta u,\kappa\vartheta Z)$, which is bounded by the (weighted) $H^2$ norm of~$\kappa\vartheta u$ and $H^1$ norm of~$\kappa\vartheta Z$ by the last inequality in \autoref{lem:sec11-estimates-geom}, applied to $\kappa\vartheta(u,Z)$.  Secondly, a commutator of $d\Mcal_{(g_0,h_0)}^{*\sharp\sharp}$ and $\kappa\vartheta$ very similar to that of $\Dcal_{\pstar'}$ and the adjoint constraints, which was bounded in the proof of \autoref{lem:sec10-commute} (see \eqref{deli-Dcalpstarp-comm}--\eqref{dGcal-Dcalpstarp-comm}).  Importantly, all terms in the commutator involve $\kappa$ or its derivatives, which are bounded pointwise and vanish outside~$\Omega_\iota$, thus the upper bounds involve norms of $u$ and~$Z$ within $\Omega_\iota$, only.
\end{proof}


\paragraph{Step 6: Controlling sources for the Euclidean operator.}  We proceed to prove three lemmas on the source terms of $\notreH[\kappa u] = E_\kappa$ and $\notreM[\kappa Z] = F_\kappa$, concerning the H\"older norm, energy-momentum terms, and dual Sobolev norm, respectively.  We do not control the energy in the special case $\pstar''=n-2=\min(\pstar'+p_G,2\pstar')\leq\pstar$, and correspondingly the conclusions of \autoref{prop:improve-radial} are slightly weaker in that case.

\begin{lemma}\label{lem:EF-Holder}
The source terms obey H\"older bounds
\be
\aligned
\| E_\kappa \|_{\Omega_\iota,\astar''+4,-\expoPp-2}^{N-2,\alpha} + \|F_\kappa\|^{N-1,\alpha}_{\Omega_\iota,\astar''+2,-\expoPp-1} \lesssim \Err^+_{\pstar}[\seedg,\seedh] .
\endaligned
\ee
\end{lemma}

\begin{proof}
\bse\label{TGcal-Holder-1043}
We treat each term in~\eqref{EF-curved} in turn.
The seed term~$T^\Gcal_\seed$ is controlled, as observed immediately after~\eqref{EF-curved}.
Next, consider~$T^\Gcal_\kappa$.
The difference $\notreH[\kappa u]-\kappa\notreH[u]$ is a sum of terms of the form $\lambda^{-2\expoP}D^{k_1}(\lambda^{2\expoP})\,D^{k_2}u\,D^{k_3}\kappa$, where $D^k$ denotes a differential operator of order~$k$ constructed from $\vartheta$ and $\nablaslash$, and where $k_1+k_2+k_3\leq 4$ and $k_3\geq 1$.  Likewise, the momentum component of $T^\Gcal_\kappa[u,Z]$ consists of $\lambda^{-2\expoP}D^{k_1}(\lambda^{2\expoP})\,D^{k_2}Z\,D^{k_3}\kappa$ with $k_1+k_2+k_3\leq 2$ and $k_3\geq 1$.  Observe that $|\lambda^{-2\expoP}D^{k_1}(\lambda^{2\expoP})|\lesssim\lambda^{-k_1}$ and $D^{k_3}\kappa$ is uniformly bounded for $k_3\geq 1$, and is supported on $\Omega_0\cap\Omega_\iota$.  This yields pointwise bounds
\be
\bigl\| T^\Gcal_\kappa[u,Z] \bigr\|_{C_{\beta+4,-\expoPp-2}^{N-2,\alpha}(\Omega_\iota)\times C_{\beta+2,-\expoPp-1}^{N-1,\alpha}(\Omega_\iota)}
\lesssim \|u\|_{\Omega_0,\beta,-\expoPp+1}^{N+1,\alpha} + \|Z\|_{\Omega_0,\beta,-\expoPp}^{N,\alpha}
\ee
for any given radial exponent $\beta$ (such as the exponent $\beta=\astar''$ of interest) since this exponent is irrelevant in the bounded region~$\Omega_0$.  These norms of $u$ and~$Z$ are controlled by~\eqref{estim-Omega0}.

Consider next $T^\Gcal_\adj[u,Z]$.  Using \autoref{prop:adjoint-constr-coord}, we bound $d\Hcal_{(g_0,h_0)}^{*\flat\flat} - d\Hcal_{(\delta,0)}^{*\flat\flat}$ and $d\Mcal_{(g_0,h_0)}^{*\sharp\sharp} - d\Mcal_{(\delta,0)}^{*\sharp\sharp}$ in terms of products of $(g_0-\delta,h_0)$ and $(u,Z)$,
\be
\aligned
\quad & \unquad \bigl\| T^\Gcal_\adj[u,Z] \bigr\|_{C_{p_G+\astar'+4,-\expoPp-2}^{N-2,\alpha}(\Omega_\iota)\times C_{p_G+\astar'+2,-\expoPp-1}^{N-1,\alpha}(\Omega_\iota)} \\
& \lesssim \bigl\| \bigl(d\Hcal_{(g_0,h_0)}^{*\flat\flat} - d\Hcal_{(\delta,0)}^{*\flat\flat}\bigr)[u,Z] \bigr\|_{C_{p_G+\astar'+2,-\expoPp}^{N,\alpha}(\Omega_\iota)} \\
& \quad + \bigl\| \bigl(d\Mcal_{(g_0,h_0)}^{*\sharp\sharp} - d\Mcal_{(\delta,0)}^{*\sharp\sharp}\bigr)[u,Z] \bigr\|_{C_{p_G+\astar'+1,-\expoPp}^{N,\alpha}(\Omega_\iota)}
\\
& \lesssim
\|u\|_{\Omega_\iota,\astar',-\expoPp+2}^{N+2,\alpha}
+\|Z\|_{\Omega_\iota,\astar',-\expoPp+1}^{N+1,\alpha}, 
\endaligned
\ee
where the implicit constant includes a factor of $\|g_0-\delta\|_{p_G}^{N+2,\alpha} + \|h_0\|_{p_G+1}^{N+1,\alpha}$.
This norm of $u,Z$ is controlled by $\Err^+_{\pstar}[\seedg,\seedh]$ according to $\Est(\Omega_\iota,\pstar')$.
Since $\astar''\leq\astar'+p_G$ (equivalently $\pstar''\leq\pstar'+p_G$), the norm on the left-hand side has a radial exponent at least as good as needed.

Next, using \autoref{lem:lin-constr-coord} to express $d\Gcal_{(\seedg,\seedh)} - d\Gcal_{(\delta,0)}$ as products of $(\seedg-\delta,\seedh)$ and $(\gdiff,\hdiff)$, we get
\be\compresseq{.42}
\aligned
& \bigl\| T^\Gcal_\lin[\gdiff,\hdiff] \bigr\|_{C_{p_G+\astar'+4,-\expoPp-2}^{N-2,\alpha}(\Omega_\iota)\times C_{p_G+\astar'+2,-\expoPp-1}^{N-1,\alpha}(\Omega_\iota)}
\\
& \quad \lesssim \bigl\| (d\Gcal_{(\seedg,\seedh)} - d\Gcal_{(\delta,0)})[\gdiff,\hdiff] \bigr\|_{C_{p_G+\pstar'+2,\expoPm-2}^{N-2,\alpha}(\Omega_\iota)\times C_{p_G+\pstar'+2,\expoPm-1}^{N-1,\alpha}(\Omega_\iota)} 
\lesssim \|\gdiff\|_{\Omega_\iota,\pstar',\expoPm}^{N,\alpha} + \|\hdiff\|_{\Omega_\iota,\pstar'+1,\expoPm}^{N-1,\alpha} ,
\endaligned
\ee
where the implicit constant includes a factor of $\|\seedg-\delta\|_{p_G}^{N,\alpha} + \|\seedh\|_{p_G+1}^{N-1,\alpha}$.
The norms of $\gdiff,\hdiff$ are controlled thanks to~\autoref{lem:sec11-estimates-geom}.

The quadratic terms are described in \autoref{lem:constr-expand}
\bel{T-quadra}
\aligned
& \bigl\| T^\Gcal_\qua[\gdiff,\hdiff] \bigr\|_{C_{\pstar'+\astar'+4,-\expoPp-2}^{N-2,\alpha}(\Omega_\iota)\times C_{\pstar'+\astar'+2,-\expoPp-1}^{N-1,\alpha}(\Omega_\iota)}
\\
& \quad \lesssim \bigl\| \Qcal\Gcal_{(\seedg,\seedh)}[\gdiff,\hdiff] \bigr\|_{C_{2\pstar'+2,\expoPm-2}^{N-2,\alpha}(\Omega_\iota)\times C_{2\pstar'+2,\expoPm-1}^{N-1,\alpha}(\Omega_\iota)}
\lesssim \bigl( \|\gdiff\|_{\Omega_\iota,\pstar',\expoPm/2}^{N,\alpha} + \|\hdiff\|_{\Omega_\iota,\pstar'+1,\expoPm/2}^{N-1,\alpha} \bigr)^2 .
\endaligned
\ee
where the implicit constants depend on the $C_0^{N,\alpha}(\Omega_\iota)$ norm of $\seedg-\delta$ and the $C_1^{N,\alpha}(\Omega_\iota)$ norm of~$\seedh$. The upper bound is controlled through~\autoref{lem:sec11-estimates-geom}.  The left-hand side controls the desired norm since $\pstar''\leq 2\pstar'$.
\ese
\end{proof}


\begin{lemma}\label{lem:EF-energymom}
If $\pstar''\neq n-2$ or $\pstar''=\pstar=n-2<\min(\pstar'+p_G, 2\pstar')$,
the energy and momentum semi-norms of the sources $E_\kappa$ and~$F_\kappa$ are bounded by
\be
\cutoff_{\pstar''}^{\textnormal{ex}} \mmax(E_\kappa)
+ \cutoff_{\pstar''}^{\textnormal{ex}} \Jmax(F_\kappa)
\lesssim \Err^+_{\pstar}[\seedg,\seedh] .
\ee
\end{lemma}

\begin{proof}
The energy and momentum terms are controlled by the H\"older norms for $\pstar''>n-2$ (cf.~\eqref{mmodu-def}), and vanish for $\pstar''<n-2$, so they can be ignored in these cases.  Consider now $\pstar''=\pstar=n-2<\min(\pstar'+p_G, 2\pstar')$.
Per~\eqref{TGcal-Holder-1043}, the first four terms in~\eqref{EF-curved} are controlled with (pointwise) radial decay $\pstar'+p_G$ or $2\pstar'$, both larger than $n-2$, hence sufficient to bound the energy and momentum integrals.
On the other hand, the energy and momentum contributions of~$T^\Gcal_\seed$ are integrals of $\kappa \Gcal[\seedg,\seedh]$, which coincide with those in the definition of~$\Err^+_{\pstar}[\seedg,\seedh]$ in~\eqref{equa-Ecal-pq-rep}, except for terms with $\Obig(r^{-(n-2)-p_G})$ decay due to metric differences and compact terms caused by~$\kappa$.  These in turn are bounded by the H\"older (or Lebesgue) norms of $\Hcal(\seedg,\seedh)$ and $\Mcal(\seedg,\seedh)$ in~$\Err^+_{\pstar}[\seedg,\seedh]$. 
\end{proof}


\begin{lemma}\label{lem:EF-Sobolev}
The source terms obey dual Sobolev bounds
\be
\| E_\kappa,\vartheta E_\kappa \|_{H^{2*}_{\astar''+4,-\expoP}(\Omega_\iota)}
+ \|F_\kappa\|_{H^{1*}_{\astar''+2,-\expoP}(\Omega_\iota)}
\lesssim \Err^+_{\pstar}[\seedg,\seedh] .
\ee
\end{lemma}

\begin{proof}
We fix a test function $\varphi\in H^2_{-\astar''-4,-\expoP}(\Omega_\iota)$ and test vector field $\psi\in H^1_{-\astar''-2,-\expoP}(\Omega_\iota,\RR^n)$ throughout this proof, and seek bounds on
\bel{EFphipsi}
\bigl\la (E_\kappa,F_\kappa),\, (\varphi,\psi)\bigr\ra = \int_{\Omega_\iota} \bigl( E_\kappa \varphi + F_\kappa\cdot\psi \bigr) \lambda^{2\expoP} r^{-n} d^nx , \qquad
\la \vartheta E_\kappa,\varphi\ra = - \int_{\Omega_\iota} E_\kappa\vartheta\varphi \, \lambda^{2\expoP} r^{-n} d^nx
\ee
by weighted $H^2\times H^1$ norms of $(u,Z)$ and $(\varphi,\psi)$.

Let us begin with the term $T^\Gcal_\kappa[u,Z] = (\notreH[\kappa u]-\kappa\notreH[u], \ \notreM[\kappa Z]-\kappa\notreM[Z])$ in~\eqref{EF-curved}.  The expressions~\eqref{equa:acalew0-deux} of $(\notreH,\notreM)$ yield (here and below, the index contractions ``$\cdot$'' do not involve any implicit metric)
\be
\aligned
\bigl\la (T^\Gcal_\kappa[u,Z])^F, \psi \bigr\ra
& = \int_{\Omega_\iota} \Bigl( d\Mcal_{(\delta,0)}^{*\sharp\sharp}[\kappa Z] \cdot d\Mcal_{(\delta,0)}^*[r^{2p+2-2n} \psi] \\[-1ex]
& \qquad\qquad - d\Mcal_{(\delta,0)}^{*\sharp\sharp}[Z] \cdot d\Mcal_{(\delta,0)}^*[\kappa r^{2p+2-2n} \psi] \Bigr) \lambda^{2\expoP} r^{n-2p-2} d^nx .
\endaligned
\ee
Upon expanding derivatives, one gets terms of the form $D\kappa\,D^{k_1}Z\,D^{k_2}\psi$ with $k_1+k_2\leq 1$, where $D^k$ are differential operators of order~$k$ built from scale-invariant derivatives $\nablaslash$ or~$\vartheta$.
These terms are controlled by the $H^1_{-\expoP}$ norms of $Z$ and $\psi$ within the compact set~$\Omega_0$ (within which $D\kappa$ is supported).
The bound on $(T^\Gcal_\kappa[u,Z])^E$ is entirely analogous, so
\bel{TGkappa-Sobolev-bound}
\bigl\|T^\Gcal_\kappa[u,Z]\bigr\|_{H^{2*}_{\astar''+4,-\expoP}(\Omega_\iota)\times H^{1*}_{\astar''+2,-\expoP}(\Omega_\iota)}
\lesssim \|(u,Z)\|_{H^2_{\astar'',-\expoP}(\Omega_0)\times H^1_{\astar'',-\expoP}(\Omega_0)}
\lesssim \Err^+_{\pstar}[\seedg,\seedh] ,
\ee
where the choice of radial exponent~$\astar''$ is irrelevant and the upper bound comes from~\eqref{estim-Omega0}.
For the radial derivative, we commute both $\kappa$ and $\vartheta$ through powers of $r$ and the derivatives in $d\Hcal_{(\delta,0)}^*$ to get
\be
\aligned
& \bigl\la \vartheta(T^\Gcal_\kappa[u,Z])^E, \varphi\bigr\ra
\\
& = \int_{\Omega_\iota} \Bigl( d\Hcal_{(\delta,0)}^{*\flat\flat}[u] \cdot d\Hcal_{(\delta,0)}^*[\kappa r^{2p-2n} \vartheta\varphi]
- d\Hcal_{(\delta,0)}^{*\flat\flat}[\kappa u] \cdot d\Hcal_{(\delta,0)}^*[r^{2p-2n} \vartheta\varphi] \Bigr) \lambda^{2\expoP} r^{n-2p} d^nx
\\
& = \int_{\Omega_\iota} \Bigl( d\Hcal_{(\delta,0)}^{*\flat\flat}[u] \cdot \bigl( \bigl[d\Hcal_{(\delta,0)}^*,\kappa\cdot\ \bigr][r^{2p-2n} \vartheta\varphi] \bigr) \\[-1ex]
& \qquad\qquad - \bigl( \bigl[d\Hcal_{(\delta,0)}^{*\flat\flat},\kappa\cdot\ \bigr][u] \bigr) \cdot d\Hcal_{(\delta,0)}^*[r^{2p-2n} \vartheta\varphi] \Bigr) \lambda^{2\expoP} r^{n-2p} d^nx
\\
& = \int_{\Omega_\iota} \Bigl( d\Hcal_{(\delta,0)}^{*\flat\flat}[u] \cdot \bigl( \bigl[d\Hcal_{(\delta,0)}^*,\kappa\cdot\ \bigr][r^{2p-2n} \vartheta\varphi] \bigr) \\[-1ex]
& \qquad\qquad + \bigl( (\vartheta - 2) \bigl[d\Hcal_{(\delta,0)}^{*\flat\flat},\kappa\cdot\ \bigr][u] \bigr) \cdot d\Hcal_{(\delta,0)}^*[r^{2p-2n}\varphi] \Bigr) \lambda^{2\expoP} r^{n-2p} d^nx .
\endaligned
\ee
Expanding all commutators and derivatives yields terms of the form $D^{k_1}\kappa D^{k_2}u D^{k_3}\varphi$ with $k_1+k_2+k_3\leq 5$ and $k_1\geq 1$ and $k_2,k_3\leq 2$, which are bounded by $H^2_{-\expoP}(\Omega_0)$ norms of $u$ and~$\varphi$.  Thus,
\be
\bigl\|\vartheta(T^\Gcal_\kappa[u,Z])^E\bigr\|_{H^{2*}_{\astar''+4,-\expoP}(\Omega_\iota)}
\lesssim \|u\|_{H^2_{\astar'',-\expoP}(\Omega_0)} \lesssim \Err^+_{\pstar}[\seedg,\seedh] .
\ee

For the remaining terms we only spell out the arguments for the radial derivative $\vartheta E_\kappa$, as the bounds on $E_\kappa,F_\kappa$ involve fewer derivatives and are thus easier.
We proceed to $T^\Gcal_\adj[u,Z]$.
In the following calculations, we commute $\kappa$ through $d\Hcal_{(\delta,0)}^*$, which produces compactly-supported terms with few enough derivatives of $u,Z,\varphi$, bounded by $\Err^+_{\pstar}[\seedg,\seedh]$ as in~\eqref{TGkappa-Sobolev-bound}.
We also use $d\Hcal_{(\delta,0)}^*\circ\vartheta = (\vartheta+2)\circ d\Hcal_{(\delta,0)}^*$, whose radial integration by parts produces the operator $(\vartheta-2)$.
At this stage, \autoref{prop:adjoint-constr-coord} expresses the difference $d\Hcal_{(g_0,h_0)}^{*\flat\flat} - d\Hcal_{(\delta,0)}^{*\flat\flat}$ in terms of
$u,Z,\del u,\del^2u,\del Z$
multiplied by geometrical factors involving $g_0-\delta$, $h_0$ and their derivatives, which yields the bound
\be
\aligned
 &  \bigl| \bigl\la \vartheta(T^\Gcal_\adj[u,Z])^E,\varphi\bigr\ra \bigr|
= \biggl| \int_{\Omega_\iota} \bigl(d\Hcal_{(g_0,h_0)}^{*\flat\flat} - d\Hcal_{(\delta,0)}^{*\flat\flat}\bigr)[u,Z] \cdot d\Hcal_{(\delta,0)}^*[r^{2p-2n}\kappa\vartheta\varphi] \lambda^{2\expoP} r^{n-2p} d^n x \biggr|
\\
& \lesssim \biggl| \int_{\Omega_\iota} (\vartheta-2)\Bigl(\kappa\bigl(d\Hcal_{(g_0,h_0)}^{*\flat\flat} - d\Hcal_{(\delta,0)}^{*\flat\flat}\bigr)[u,Z]\Bigr) \cdot d\Hcal_{(\delta,0)}^*[r^{2p-2n}\varphi] \lambda^{2\expoP} r^{n-2p} d^n x \biggr|
+ \Err^+_{\pstar}[\seedg,\seedh]
\\
& \lesssim
\Bigl(
\|u,\kappa\vartheta u\|_{H^2_{\astar''-p_G,-\expoP}}
+\|Z,\kappa\vartheta Z\|_{H^1_{\astar''-p_G,-\expoP}}
+ \Err^+_{\pstar}[\seedg,\seedh]
\Bigr)
\|\varphi\|_{H^2_{-\astar''-4,-\expoP}(\Omega_\iota)}
\endaligned
\ee
with an implicit constant depending on the $C^3_{p_G}(\Omega_\iota)$ norm of $g_0-\delta$ and $C^2_{p_G+1}(\Omega_\iota)$ norm of~$h_0$.
The norms of $u,Z$ are controlled by the estimate $\Est(\Omega_\iota,\pstar')$ and 
those of the radial derivatives by \autoref{lem:kappavarthetau}.  (We recall $\pstar''-p_G\leq\pstar'$.)

For the control of $\vartheta(T^\Gcal_\lin)^E$, when integrating by parts, one must beware that the volume form used in~\eqref{EFphipsi} is the Euclidean one and not that of~$\seedg$, hence $d\Hcal_{(\seedg,\seedh)}$ gets dualized to $\mu\circ d\Hcal_{(\seedg,\seedh)}^*\circ\mu^{-1}$ where $\mu$ is multiplication by the scalar field $\vol_{\seedg}/d^nx$ (Radon--Nikodym derivative).  A related difficulty is that \autoref{prop:adjoint-constr-coord} controls the difference $d\Hcal^{*\flat\flat}_{(\seedg,\seedh)} - d\Hcal^{*\flat\flat}_{(\delta,0)}$ with indices lowered by the two different metrics $\seedg$ and~$\delta$, respectively (and similarly for $d\Mcal^{*\sharp\sharp}$).  The momentum terms do not need such treatment since (for~$\upsilon$ any scalar field) $d\Mcal_{(\delta,0)}^*[\upsilon,0]=0$ and $d\Mcal_{(\seedg,\seedh)}^*[\upsilon,0] = \seedh*\upsilon$ by \autoref{lem:lin-constr}.
Altogether,
\be
\aligned
& \bigl\la \vartheta(T^\Gcal_\lin[\gdiff,\hdiff])^E , \varphi \bigr\ra
\\
& = \int_{\Omega_\iota} \Bigl( \gdiff \cdot \bigl(\mu\circ d\Hcal_{(\seedg,\seedh)}^*\circ\mu^{-1} - d\Hcal_{(\delta,0)}^*\bigr)[\kappa\vartheta\varphi \, r^{2p-2n}] + \hdiff*\seedh*\kappa\vartheta\varphi \, r^{2p-2n} \Bigr) d^nx
\\
& = \int_{\Omega_\iota} \Bigl( \gdiff \cdot \bigl[\mu,d\Hcal_{(\seedg,\seedh)}^*\bigr][\mu^{-1}\kappa\vartheta\varphi \, r^{2p-2n}]
+ \gdiff_{ij} (\seedg^{ik} \seedg^{jl} - \delta^{ik} \delta^{jl}) d\Hcal_{(\seedg,\seedh)kl}^{*\flat\flat}[\kappa\vartheta\varphi \, r^{2p-2n}] \\[-1ex]
& \qquad\quad
+ \gdiff_{ij} \delta^{ik} \delta^{jl} (d\Hcal_{(\seedg,\seedh)kl}^{*\flat\flat} - d\Hcal_{(\delta,0)kl}^{*\flat\flat})[\kappa\vartheta\varphi \, r^{2p-2n}]
+ \hdiff*\seedh*\kappa\vartheta\varphi \, r^{2p-2n} \Bigr) d^nx ,
\endaligned
\ee
with the short-hand notation $d\Hcal_{(\seedg,\seedh)}^*[\upsilon] = d\Hcal_{(\seedg,\seedh)}^*[\upsilon,0]$ to avoid confusion with the commutator notation~$[\mu,d\Hcal_{(\seedg,\seedh)}^*]$.
This commutator term and the difference of adjoint constraints are easily bounded, while the $\seedg\seedg-\delta\delta$ term requires commuting $\kappa\vartheta$ through the operators.  We get
\be
\aligned
& \bigl| \bigl\la \vartheta(T^\Gcal_\lin[\gdiff,\hdiff])^E , \varphi \bigr\ra \bigr|
\\
& \lesssim \|\gdiff\|_{L^2_{\pstar''-p_G,\expoP}(\Omega_\iota)} \|\varphi\|_{H^2_{-\astar''-4,-\expoP}(\Omega_\iota)}
+ \|\hdiff\|_{L^2_{\pstar''-p_G+1,\expoP}(\Omega_\iota)} \|\varphi\|_{H^1_{-\astar''-4,-\expoP}(\Omega_\iota)} \\
& \quad + \biggl| \int_{\Omega_\iota} \Bigl(
\gdiff_{ij} \Bigl[ (\seedg^{ik} \seedg^{jl} {-} \delta^{ik} \delta^{jl}) d\Hcal_{(\seedg,\seedh)kl}^{*\flat\flat},\,\kappa\vartheta\Bigr] [\varphi \, r^{2p-2n}] \\[-2ex]
& \qquad\qquad\quad
- \vartheta(\kappa \gdiff_{ij}) (\seedg^{ik} \seedg^{jl} {-} \delta^{ik} \delta^{jl}) d\Hcal_{(\seedg,\seedh)}^{*\flat\flat}[\varphi \, r^{2p-2n}]_{kl} \Bigr) d^nx \biggr|
\\
& \lesssim \|\gdiff,\kappa\vartheta\gdiff\|_{L^2_{\pstar''-p_G,\expoP}(\Omega_\iota)} \|\varphi\|_{H^2_{-\astar''-4,-\expoP}(\Omega_\iota)}
+ \|\hdiff\|_{L^2_{\pstar''-p_G+1,\expoP}(\Omega_\iota)} \|\varphi\|_{H^1_{-\astar''-4,-\expoP}(\Omega_\iota)}
\endaligned
\ee
where the implicit constant depends on the $C^3_{p_G}(\Omega_\iota)$ norm of $\seedg-\delta$ and $C^1_{p_G+1}(\Omega_\iota)$ norm of~$\seedh$.
The norms of $\gdiff,\hdiff$ are bounded in \autoref{lem:sec11-estimates-geom} and those of $\kappa\vartheta\gdiff$ in \autoref{lem:kappavarthetagdiff}.

Finally, the required bounds on $T^\Gcal_\qua[\gdiff,\hdiff] = - \diag(\omega_p^{-2},\omega_{p+1}^{-2}) \kappa \Qcal\Gcal_{(\seedg,\seedh)}[\gdiff,\hdiff]$ are immediate from \autoref{lem:step2-nonlinear} (and the control of H\"older norms of $\gdiff,\hdiff$ in \autoref{lem:sec11-estimates-geom}) since commuting the weight $\omega_p^{-2}$ through derivatives only gives lower-derivative terms.  The bounds on~$T^\Gcal_\seed$, and especially on the $H^{2*}$ norm of $\vartheta(T^\Gcal_\seed)^E=-\vartheta(\omega_p^{-2}\kappa\Hcal[\seedg,\seedh])$, are immediate from the $L^2$ norms of $\Gcal[\seedg,\seedh]$ appearing in $\Err^+_{\pstar}[\seedg,\seedh]$.
\end{proof}


\paragraph{Step 7: Applying \refwithname{Sections}{section=5} \refwithname{to}{section=9}.}

We have established dual Sobolev, H\"older, and energy-momentum bounds on the sources $(E_\kappa,F_\kappa)$ of the equation~\eqref{eq-kappau} for the cut-off solutions $(\kappa u,\kappa Z)$.  These are precisely the conditions\footnote{The stability of $\lambdabf$ was assumed at the start of \autoref{section=10}.  The fact that the Euclidean theorems only apply up to some exponent $a_{n,p}+\delta$ for some small $\delta>0$ translates to the upper bound $p^{\lambdabf}_{n,p}$ given in~\eqref{p-lambda-def} on the exponent~$\pstar$ in our exponent ranges~\eqref{exponent-range}.}
to apply \refwithname{Theorems}{thm-sharp-h-localized} \refwithname{and}{thm-sharp-m-localized} with $\pstar$ replaced by~$\pstar''$, except in the special case $\pstar''=n-2=\min(\pstar'+p_G, 2\pstar')\leq\pstar$, where the theorems can only be applied with strictly smaller exponents.  We encapsulate both cases by introducing~$\widetilde{\pstar}$ chosen as follows:
\begin{itemize}
\item $\widetilde{\pstar}=\pstar''$ in the case $\pstar''\neq n-2$ or the case $\pstar''=\pstar=n-2<\min(\pstar'+p_G, 2\pstar')$;
\item $\widetilde{\pstar}\in(q,\pstar'')$ in the special case $\pstar''=n-2=\min(\pstar'+p_G, 2\pstar')\leq\pstar$.  Here, $q\in(\pstar',\pstar'')$ is the exponent in \autoref{prop:improve-radial}.
\end{itemize}
Observe that the indicator $\cutoff_{\widetilde{\pstar}}=\Oneone_{\widetilde{\pstar}\geq n-2}$ coincides with $\cutoff_{\pstar''}^{\textnormal{ex}}$ given in~\eqref{exceptional-indicator}.
Indeed, if $\pstar''\neq n-2$ this is obvious by $\widetilde{\pstar}=\pstar''$, while for the two cases $\pstar''= n-2 = \min(\pstar'+p_G,2\pstar')$ and $\pstar''= n-2 < \min(\pstar'+p_G,2\pstar')$ one has, respectively, $\widetilde{\pstar}<\pstar''$ and $\widetilde{\pstar}=\pstar''$.
As a consequence, the modular parameters $\cutoff_{\widetilde{\pstar}} \mmodu(E_\kappa)$ and $\cutoff_{\widetilde{\pstar}} \Jmodu(F_\kappa)$ that appear below are well-defined and bounded by \autoref{lem:EF-energymom}.

For this exponent~$\widetilde{\pstar}$, and the corresponding $\widetilde{\astar}=n-2-2p+\widetilde{\pstar}$, we obtain Sobolev and H\"older bounds on $\kappa u,\kappa Z$: for any $\beta\in[a_{n,p}/2,\widetilde{\astar})$,
\bel{kappau-step7}
\aligned
& \bigl\| \kappa u - \cutoff_{\widetilde{\pstar}} \mmodu(E_\kappa) r^{-a_{n,p}}\nu^\normal \bigr\|^{N+2,\alpha}_{\Omega_{R'}, \widetilde{\astar}, -\expoPp+2}
+ \sum_{k=0,1} \bigl\| \vartheta^k(\kappa u - \cutoff_{\widetilde{\pstar}}  \mmodu(E_\kappa) r^{-a_{n,p}}\nu^\normal) \bigr\|_{H^2_{\beta,-\expoP}(\Omega_{R'})}
\\
& \quad + \bigl\| \kappa Z - \cutoff_{\widetilde{\pstar}} \Jmodu(F_\kappa)_j \xi^{\normal(j)}r^{-a_{n,p}} \bigr\|^{N+1,\alpha}_{\Omega_{R'}, \widetilde{\astar}, -\expoPp+1}
+ \bigl\| \kappa Z - \cutoff_{\widetilde{\pstar}} \Jmodu(F_\kappa)_j \xi^{\normal(j)}r^{-a_{n,p}} \bigr\|_{H^1_{\beta,-\expoP}(\Omega_{R'})}
\\
& \lesssim \Err^+_{\pstar}[\seedg,\seedh]
\endaligned
\ee
for any given~$R'>R$.  In addition, in the limit $R'\to+\infty$,
\bel{kappau-step7-lim}
\bigl\| \kappa u - \cutoff_{\widetilde{\pstar}}  \mmodu(E_\kappa) r^{-a_{n,p}}\nu^\normal \bigr\|_{C^3_{\!\widetilde{\astar}, -\!\expoPp+2}(\Omega_{R'})}
+ \bigl\| \kappa Z - \cutoff_{\widetilde{\pstar}} 
\Jmodu(F_\kappa)_j \xi^{\normal(j)}r^{-a_{n,p}} \bigr\|_{C^1_{\!\widetilde{\astar}, -\!\expoPp+1}(\Omega_{R'})}
\to 0 .
\ee


\paragraph{Step 8: An improved radial exponent.}

We are ready to conclude the proof of \autoref{prop:improve-radial}, which states H\"older and Sobolev bounds on the shifted solution $(u-\cutoff_{\pstar''}\umodu, Z-\cutoff_{\pstar''}\Zmodu)$.  For this, we select a fixed radius $R_{\mathrm{in}}\in(R,R_4)$
to ensure that (under the identification $\phi_\iota:\Omega_\iota\overset{\sim}\to\Omega_R$) the domain
$\phi_\iota^*(\Omega_R\setminus\Omega_{R_{\mathrm{in}}})$
is contained in the bounded domain~$\Omega_0$, namely $R_{\mathrm{in}}<R_4$ in the notation below~\eqref{sec10-omegabfp-omegap}.  Then the (H\"older and Sobolev) control of $u,Z$ in~$\Omega_0$ given in~\eqref{estim-Omega0} and the control of energy and momentum modulators in \autoref{lem:EF-energymom} suffice to complete the estimates \eqref{kappau-step7} and~\eqref{kappau-step7-lim} into estimates of $(u-\cutoff_{\pstar''}\umodu, Z-\cutoff_{\pstar''}\Zmodu)$ stated in \autoref{prop:improve-radial} with
\be
\mmodu = \mmodu(E_\kappa) , \qquad \Jmodu = \Jmodu(F_\kappa).
\ee
For $\pstar''\neq n-2$ or $\pstar''=\pstar=n-2<\min(\pstar'+p_G,2\pstar')$ the exponent $\widetilde{\pstar}=\pstar''$ is exactly as needed for all of the H\"older and Sobolev bounds in \autoref{prop:improve-radial}.  As we point out in \autoref{section=10.4}, the energy and momentum modulators are in fact uniquely fixed by these bounds if $\pstar''>n-2$ or $\pstar''=\pstar=n-2<\min(\pstar'+p_G,2\pstar')$.

In the remaining case $\pstar''=n-2=\min(\pstar'+p_G,2\pstar')\leq\pstar$, the bounds with radial decay exponent~$\widetilde{\pstar}$ imply $\Est(\Omega_\iota,q)$ since $q<\widetilde{\pstar}$, but the H\"older estimates do not reach the harmonic exponent $\pstar''=n-2$.  There is no harmonic control at all, and no expression for modulator terms.

This completes the proof of \autoref{prop:improve-radial}, hence of the main results of this work, \autoref{theo--beyond-harmonic} and \autoref{theo--beyond-harmonic-II}.


\paragraph*{Acknowledgments.} 

The authors were partially supported by the research project ANR-23-CE40-0010-02: Einstein-PPF, entitled {\it ``Einstein constraints: past, present, and future''} funded by the Agence Nationale de la Recherche (ANR), as well as by the MSCA Staff Exchange Project 101131233: Einstein-Waves, entitled {\it ``Einstein gravity and nonlinear waves: physical models, numerical simulations, and data analysis'',} funded by the European Research Council (ERC). 
 

\begin{spacing}{.92}


\end{spacing}

\appendix

\section{Operator coefficients and structure constants}
\label{appendix=A}

\paragraph{Operator coefficients.}

For convenience, we recall here the operator coefficients introduced in~\eqref{equa-our-parame-00}:
\bel{equa-our-parame-000} 
\aligned
a_{n,p} & \coloneqq 2( n-2-p)
&& \text{(harmonic exponent)},
\\
b_{n,p} & \coloneqq 
2 + (n-3)(n-2-a_{n,p}), 
&&  
\\ 
c_{n,p} & \coloneqq a_{n,p} \bigl( 1 + (n-2) ( n-2- a_{n,p}) \bigr), 
&&  
\\
d_{n,p} & \coloneqq {(n-1) a_{n,p} b_{n,p} \over (n-2)^2+1}
&& \text{(ADM energy coefficient),} 
\endaligned
\ee  
as well as
\be
p^\flat_n \coloneqq {(n-1)(n-3) \over 2(n-2)} < \frac{n-2}{2} .
\ee
For $p\in(p^\flat_n, n-2)$ the coefficients $a_{n,p},b_{n,p},c_{n,p},d_{n,p}$ are all positive.


\paragraph{Structure constants for the Hamiltonian operator.} 

We collect here the expressions of the structure constants that arise in our analysis. Constants that depend upon the area $\aire[\Lambda,\lambda] = \int_{\Lambda} \, d\chi$  of the localization function are needed, specifically 
\bel{equa-thetalambda}
\theta^\lambda \coloneqq \frac{\theta_n}{\aire[\Lambda,\lambda]} \coloneqq {2 (n-1) \over (n-2)^2 + 1} \, {|\Sphe^{n-1}| \over \aire[\Lambda,\lambda]}.
\ee
We also have  introduced 
\bel{bH10}
\bnotreH_{1} \coloneqq (n-1) c_{n,p}\la\nu^\normal\ra - (n-2)^2\la\Deltaslash\nut^\normal\ra, 
\quad\ \
\bnotreH_{0} \coloneqq (n^2 -4n+5){c_{n,p} \over a_{n,p}} \Bigl( d_{n,p} \la\nu^\normal\ra - \la\Deltaslash\nut^\normal\ra \Bigr)
\ee
and 
the second-order operator 
\bel{betapm-def-app}
\aligned
& - \bnotreH_{1}\vartheta(\vartheta+a_{n,p}) + \bnotreH_{0} =: - \bnotreH_{1}  (\vartheta+\beta_-)(\vartheta+\beta_+), 
\qquad
 \beta_- < 0 < a_{n,p} < \astar < \beta_+, 
\endaligned
\ee
namely $\beta_- + \beta_+ = a_{n,p}$ and $\beta_- \beta_+ = -  \bnotreH_{0}/\bnotreH_{1}$. In~\eqref{equa-cplusmoins}, we also set 
\bel{equa-cplusminus} 
c_\pm \coloneqq {1 \over \bnotreH_{1} } {(n-2) \beta_{\pm} + {c_{n,p} / a_{n,p}} \over \beta_+ - \beta_-}. 
\ee 

The radial Hardy constant 
\bel{equa--517-repeat}
\cradialH := \max(  g^\notreH_I,   g^\notreH_J ) 
\ee
is defined from 
\bel{equa-radH}
\aligned 
g^\notreH_I & = \frac{4 c_-^2}{(a_{n,p}-\beta_-)^2} + \frac{2 c_+^2}{(\beta_+-a_{n,p})^2} + \frac{2 c_+^2 }{(2 \beta_+ - a_{n,p})(\beta_+ - a_{n,p})} ,
\\
g^\notreH_J & = \frac{16 c_+^2}{(2 \beta_+ - a_{n,p})^2} + \frac{8 c_-^2}{(a_{n,p} - 2\beta_-)^2} + \frac{2 c_-^2}{(a_{n,p}-2\beta_-) (a_{n,p}-\beta_-)} .
\endaligned
\ee 


\paragraph{Structure constants for the momentum operator.} 

At each end (with $\iota$ suppressed), we have set 
\bel{equa-thetalambda-M}
\eta^\lambda \coloneqq \frac{\eta_n}{\aire[\Lambda,\lambda]} 
\coloneqq \frac{2(n-1) |\Sphe^{n-1}|}{\aire[\Lambda,\lambda]}.
\ee
In \eqref{equa-the-matrix}, \eqref{equa-Poin2}, and~\eqref{equa--813} we defined several interrelated structure matrices for the momentum, with components
\bel{equa-the-matrix-repeat}
\aligned
T_{kl} & \coloneqq  \delta_{kl} + \la \xh_k\xh_l\ra,
&
G_{kl} & \coloneqq \delta_{kl} + \la\xh_k\xh_l\ra - 2 \la\xh_k\ra \la\xh_l\ra ,
\\
\matQ^{(j)k} & = \bigl\la 2 \xh_l \xi^{\normal (j) \perp} + \xi^{\normal (j) \parallel}{}_l\bigr\ra (T^{-1})^{lk} ,
\\
(\Xi^\notreM)^j_{l}
& \coloneqq 
\bigl\la - \nablaslash_l \xi^{\normal (j) \perp} + 2 a_{n,p} \xh_l\, \xi^{\normal (j) \perp} \bigr\ra
+(1+a_{n,p}) \la \xi^{\normal (j) \parallel}{}_l \ra. \mspace{-130mu}
\endaligned
\ee


\section{Explicit shell functionals}
\label{appendix=B}

\paragraph{Preliminary comments.}
We complement here \refwithname{Sections}{section=3.4} \refwithname{and}{section=3.5} with explicit expressions of the Hamiltonian and momentum functionals appearing in the shell identities \eqref{main-func-identity} and~\eqref{main-func-identity-MM}.  The structure of these two identities is the same,
\be
- (\vartheta + a_{n,p}) (\vartheta + 2a_{n,p}) \Phi[u] + \Chi[u] = \Mu[u,E],
\ee
where the shell functional~$\Phi$ and source term~$\Mu$ are given in \refwithname{Sections}{section=3.4} \refwithname{and}{section=3.5} in the main text.  Here we do not give explicitly the dissipation functionals~$\Chi$.  Instead we give the shifted functionals $\Psi_\beta = \Chi - (\vartheta +\beta) \Upsilon$ for $\beta\in\{a_{n,p},2a_{n,p}\}$, from which one can easily extract
\be
\Upsilon = \frac{1}{a_{n,p}} (\Psi_{a_{n,p}}-\Psi_{2a_{n,p}}) ,
\qquad
\Chi = \frac{1}{a_{n,p}} \bigl( (\vartheta+2a_{n,p})\Psi_{a_{n,p}} - (\vartheta+a_{n,p})\Psi_{2a_{n,p}} \bigr) .
\ee
Checking that the functionals $\Phi,\Chi,\Mu$ obey the shell identity is then a tedious but straightforward calculation.


\paragraph{Hamiltonian functionals.}

The Hamiltonian shell functional $\Phi^\notreH_\iota[u]$ is given in~\eqref{PhinotreH-expr}.
Incidentally, this corresponds to the expressions in our work~\cite{LL-PoincareKornHardy}, with constants $c_1,\dots,c_{13}$ chosen to be
\be\compresseq{.65}
\aligned
c_2 & = \cstun , \quad
c_1 = c_5 = c_7 = a_{n,p}, \quad
c_3 = c_4 = c_6 = c_8 = 0, \quad
c_9^2 = \frac{n^2 -3n+3}{(n-1)^2}, \quad
c_{11}^2 = \frac{1}{n-1},
\\
c_{10}
& = - \frac{(n-1)(c_{n,p}+(n-2)\cstun)}{n^2 -3n+3} , \quad
c_{12}^2 = \frac{2(1+a_{n,p}+\cstun)}{n-1} , \quad
c_{13}^2 = \frac{1}{n-1} \cstdeux
- c_9^2 c_{10}^2 .
\endaligned
\ee
The (shifted) shell dissipations (for $\beta\in\{a_{n,p},2a_{n,p}\}$) are expressed in terms of $w=(\vartheta+a_{n,p})u$:
\bel{PsinotreH-expr}\compresseq{.78}
\aligned
\Psi^\notreH_{\beta\iota}[u]
& = \Psi^\notreH_{\beta\iota,0}[u] + \cstun \Psi^\notreH_{\beta\iota,1}[u] + \Bigl( \frac{\cstdeux}{n-1} + \cstun^2 \Bigr) \Psi^\notreH_{\beta\iota,2}[u] ,
\\
\Psi^\notreH_{\beta\iota,0}[u] & = \frac{1}{n-1} \fint_{\Lambda_{\iota,r}} \! \Bigl(
(\vartheta^2 w)^2
+ 2 |\nablaslash\vartheta w|^2
+ |\nablaslash^2 w|^2
+ (n-2) (\vartheta^2 w+\Deltaslash w)^2
+ 2 (1 + a_{n,p}) |\nablaslash w|^2
\\[-1ex]
& \qquad\qquad\qquad\quad
+ (n-1) \bigl(b_{n,p} - a_{n,p}(\beta - a_{n,p}) \bigr) (\vartheta w)^2
- c_{n,p}\vartheta w\Deltaslash u
\\
& \qquad\qquad\qquad\quad
+ 2 \Bigl(\frac{c_{n,p}}{a_{n,p}} + (n - 2) \beta\Bigr) \vartheta w \Deltaslash w
- (2a_{n,p}-\beta)\frac{c_{n,p}}{a_{n,p}} w \Deltaslash w
\Bigr) d\chi_\iota ,
\\
\Psi^\notreH_{\beta\iota,1}[u]
& = \frac{1}{n-1} \ssrmA^{\lambda_\iota}[u]
+ \frac{1}{n-1} \fint_{\Lambda_{\iota,r}} \! \Bigl(
2 |\nablaslash w|^2
+ (n-1) (\vartheta w)^2
- u \Bigl( \frac{c_{n,p}}{a_{n,p}} \Deltaslash w + (n-1) b_{n,p} \vartheta w \Bigr)
\\[-.5ex]
& \qquad\qquad\qquad\qquad\qquad\quad + w \Bigl( 2(n-1) (\beta - a_{n,p}) \vartheta w - 2(n-2) \Deltaslash w + \frac{c_{n,p}}{a_{n,p}} \Deltaslash u \Bigr) \Bigr) d\chi_\iota ,
\\[-1.5ex]
\Psi^\notreH_{\beta\iota,2}[u] & = \fint_{\Lambda_{\iota,r}} \!\! (w^2 + u \vartheta w) d\chi_\iota .
\endaligned
\ee
In $\Psi^\notreH_{\beta\iota,0}[u]$, the terms without definite sign all include a factor of $\vartheta w$ or $w$, hence for large enough $C>0$ one has $\Psi^\notreH_{\beta\iota,0}[u]+C\|\vartheta w,w,\Deltaslash u\|_{\unL^2_{-\expoP}(\Lambda_{\iota,r})}^2\gtrsim(\norm{w}^{\notreH})^2$.
For the same reason, for large enough $\cstun>0$ and $C>0$, one has
$\Psi^\notreH_{\beta\iota,0}[u]+\cstun\Psi^\notreH_{\beta\iota,1}[u]+C\|w,u\|_{\unL^2_{-\expoP}(\Lambda_{\iota,r})}^2\gtrsim(\norm{\vartheta u,u}^{\notreH})^2$ (using that $\ssrmA^{\lambda_\iota}[u]+C\|u\|^2$ controls $\|\Deltaslash u\|^2$).  The $\unL^2$ norm of $w$ is then provided by the last contribution to the dissipation functional, so that overall for large enough $\cstun,\cstdeux,C>0$, one has
\be
\Psi^\notreH_{\beta\iota}[u] + C \|u\|_{\unL^2_{-\expoP}(\Lambda_{\iota,r})}^2 \gtrsim (\norm{\vartheta u,u}^{\notreH})^2 .
\ee
In \cite{LL-PoincareKornHardy}, we then apply the Poincaré inequality to control (in $\unL^2$ norms) fluctuations $\ut=u-\la u\ra_\iota$ by the gradient $\nablaslash u$.  Provided the Poincaré constant is small enough, this implies that $\Psi^\notreH_{\beta\iota}[u] + C \la u\ra_\iota^2$ is coercive for large enough $C>0$.
This is combined with suitable bounds on $\Kappa^\notreH_\iota[\ut]$ to establish \autoref{thm:informal-sufficient-stability}, namely Hamiltonian stability of localization domains $(\Lambda,d\chi_\iota)$ with small enough diameter and Poincaré constant.


\paragraph{Momentum functionals.}
The shell functional is $\Phi^\notreM_\iota[Z] = \fint_{\Lambda_{\iota,r}} ( \Zperp{}^2 + |\Zpar|^2 /2 ) d\chi_\iota$ and the (shifted) shell dissipations are ($\beta\in\{a_{n,p},2a_{n,p}\}$)
\bel{PsinotreM-expr}
\aligned
\Psi^\notreM_{\iota\beta}[Z]
& = \fint_{\Lambda_{\iota,r}} \Bigl( 2 \bigl((\vartheta+a_{n,p}) \Zperp\bigr)^2 + \bigl|(\vartheta+a_{n,p}-1) \Zpar + \nablaslash\Zperp\bigr|^2
\\[-1ex]
& \qquad\qquad + 2 \, \bigl|\Sym(\nablaslash\Zpar) + \Zperp \gslash\bigr|^2
+ (3a_{n,p} - \beta) \bigl( |\Zpar|^2 - \Zpar \cdot \nablaslash\Zperp\bigr)
\Bigr) d\chi_\iota .
\endaligned
\ee
On a given shell $\Lambda_{\iota,r}$ the vector fields $\vartheta Z$ and $Z$ are independent, so for any $Z$ the terms involving $\vartheta Z$ may vanish.
Semi-coercivity of the dissipation functionals (modulo suitable averages) is thus equivalent to semi-coercivity of the last two terms.
The term $\Zpar\cdot\nablaslash\Zperp$ prevents the integrand from being positive, but integrating by parts produces terms $\Zperp\nablaslash\cdot\Zpar$ and $\Zperp\Zpar\cdot\nablaslash\log\lambda_\iota^{2\expoP}$ that can be controlled by suitable Korn and Hardy inequalities on $(\Lambda_\iota,d\chi_\iota)$.
These inequalities are analyzed in~\cite{LL-PoincareKornHardy}, with the result being that for small enough $a_{n,p}>0$ (small $3a_{n,p} - \beta$) the last term in~\eqref{PsinotreM-expr} is under control and the localization function is momentum stable, as stated in \autoref{thm:informal-sufficient-stability-M}.


\section{Expansion of the Einstein constraints and their adjoint}
\label{appendix=C}

\subsection{Expansion of the metric, connection, and curvature}

Throughout the main text, we work with three (asymptotically Euclidean) metrics denoted here by $g_0=\zg$, $\seedg=\yg$, and $g$, since the notation $\zg, \yg$ is more compact in the forthcoming calculations. When working in a given asymptotic end equipped with coordinates, we also consider the Euclidean metric~$\delta$ in these coordinates. For clarity, let us recall some further notation. Starting from an asymptotically Euclidean data set $(\yg,\yh)$ and using suitable deformations, we seek to construct a data set $(g,h)$ such that 
\bel{gh-ygyh-uZ}
g = \yg + \omegabf_p^2 \, \zdH^{*\flat\flat}[u,Z], 
\qquad
h = \yh + \omegabf_{p+1}^2 \, \zdM^{*\sharp\sharp}[u,Z],
\ee
where $\omegabf_p$ and $\omegabf_{p+1}$ are scalar weights, $u$~is a scalar field, and $Z$~is a vector field. Here, the operator $\big( \zdH^{*\flat\flat},\zdM^{*\sharp\sharp} \big)$ constitutes the formal adjoint of the linearized constraint $\zdG$ around an (asymptotically Euclidean) localization  data set $(\zg,\zh)$. The unknowns $u$ and~$Z$ are determined as solutions to elliptic equations that arise by setting the difference $\Gcal[g,h]- \Gcal[\yg,\yh]$ equal to a given source term. In this appendix, we express this difference in coordinates on an asymptotically Euclidean end, and we list all nonlinearities involved.

We denote the Levi-Civita connections, Laplacian, and curvatures of $\zg,\yg,g,\delta$ as $\znabla,\ynabla,\nabla,\del$ and $\zDelta,\yDelta,\Delta,\Delta_\delta$, and $\zR,\yR,R,0$, respectively. In our statements, {\it indices are never raised nor lowered implicitly,} in order to avoid a confusion between the four metrics $\zg$, $\yg$, $g$, and~$\delta$. For a pair of tensors $A,B$ the product $A * B$ denotes arbitrary index contractions using any of the metrics $\zg$, $\yg$, $g$, $\delta$ (or their inverses) involved in the problem at hand.  In addition, we write
\bel{equa--B2} 
\aligned
a^{*n} & \coloneqq a * \dots * a \text{ \ with $n$ factors,} 
\\
{\del*}a & \coloneqq \del a + \del\zg * a + \del\yg * a + \del g * a,
\\
{\del*}(a b) \coloneqq {\del*}(a * b) & \coloneqq \del a * b + a * \del b + \del\zg * a * b + \del\yg * a * b + \del g * a * b
\endaligned
\ee
with arbitrary index contractions, and likewise with $\del$ replaced by $\znabla$ or $\ynabla$, and including derivatives of all the metrics involved in the problem, only.  For instance, in statements that only involve the metrics $g$ and~$\yg$, the notation ${\ynabla*}a$ stands for $\ynabla a+ \ynabla g*a$. It should be emphasized at this stage that $\ynabla g= \ynabla(g- \yg)$ is small for metrics close to~$\yg$.  We will keep the more concise notation $\ynabla g$ in such cases, at the cost of leaving smallness less manifest.

As a starting point, we focus on the two data sets $(g,h)$ and $(\yg,\yh)$, without (yet) making use of the relation~\eqref{gh-ygyh-uZ} between them which involves~$(\zg,\zh)$.  This will provide us with explicit expressions for the difference $\Gcal[g,h]- \Gcal[\yg,\yh]$ in term of the linearized constraints and of quadratic contributions in $(g- \yg,h- \yh)$. We begin with a few objects derived from the metric~$g$.

\begin{lemma}\label{lem:split-basic}
The inverse of a metric~$g$ can be written as
\be
\aligned
(g^{-1})^{ij} & = (\yg^{-1})^{ij} - (\yg^{-1})^{ik} (g - \yg)_{kl} (g^{-1})^{lj} 
\\
& = (\yg^{-1})^{ij} - (\yg^{-1})^{ik} (g - \yg)_{kl} (\yg^{-1})^{lj} + (g - \yg)*(g - \yg).
\endaligned
\ee
The difference of Levi-Civita connections $\nabla$ and $\ynabla$ (associated with $g$ and $\yg$) is the tensor with components
\bel{nabla-ynabla}
\aligned
(\nabla- \ynabla)^i{}_{jk}
& = \frac{1}{2} (g^{-1})^{il} \bigl(\ynabla_j g_{kl} + \ynabla_k g_{jl} - \ynabla_l g_{jk}\bigr)
\\
& = \frac{1}{2} (\yg^{-1})^{il} \bigl(\ynabla_j g_{kl} + \ynabla_k g_{jl} - \ynabla_l g_{jk}\bigr) + (g- \yg)*\ynabla g.
\endaligned
\ee
The Riemann, Ricci, and scalar curvatures of~$g$ are given in terms of those of~$\yg$ by
\bel{Riem-zRiem}
\aligned
R^k{}_{lij}
& = \yR^k{}_{lij} + \frac{1}{2} (\yg^{-1})^{km} \Bigl(\ynabla_i \ynabla_l g_{jm} - \ynabla_i \ynabla_m g_{jl} - \ynabla_j \ynabla_l g_{im} + \ynabla_j \ynabla_m g_{il} + [\ynabla_i,\ynabla_j] g_{lm} \Bigr) \mspace{-30mu} \\
& \quad + \ynabla g * \ynabla g + (g-\yg)*\ynabla\ynabla g ,
\\
R_{lj}
& = \yR_{lj}
+ \frac{1}{2} (\yg^{-1})^{ik} \Bigl(\ynabla_i \ynabla_l g_{jk} - \ynabla_i \ynabla_k g_{jl} - \ynabla_j \ynabla_l g_{ik} + \ynabla_i \ynabla_j g_{kl} \Bigr) \\
& \quad + \ynabla g * \ynabla g + (g-\yg)*\ynabla\ynabla g ,
\\
R
& =
\yR + (\yg^{-1})^{ik} (\yg^{-1})^{jl} \bigl( \ynabla_i \ynabla_j g_{kl} - \ynabla_i \ynabla_k g_{jl} - \yR_{ij} (g- \yg)_{kl} \bigr) \\
& \quad + \ynabla g * \ynabla g + (g-\yg)*\ynabla\ynabla g + \yRic * (g- \yg)^{*2} .
\endaligned
\ee
\end{lemma}

\begin{proof} 
\bse
{\it 1. Metric and connection.} The first expression of the inverse metric is proven by expanding
\be
(\yg^{-1})^{ik} (g- \yg)_{kl} (g^{-1})^{lj}
= (\yg^{-1})^{ik} g_{kl} (g^{-1})^{lj} - (\yg^{-1})^{ik} \yg_{kl} (g^{-1})^{lj}
= (\yg^{-1})^{ij} - (g^{-1})^{ij} .
\ee
Applying this first equality to the term $(g^{-1})^{kj}$ on its right-hand side gives the second expression.

The fact that $\nabla- \ynabla$ is a tensor, and its first expression, are well-known and are easily derived by expanding both sides in terms of Christoffel symbols of $g$ and~$\yg$ in a coordinate chart. Splitting $g^{-1}$ into $\yg^{-1}$ and a term of order $g- \yg$ yields the second expression.

\medskip

\noindent{\it 2. Curvature.} We set $\zeta^i{}_{jk} \coloneqq (\nabla- \ynabla)^i{}_{jk}$.  Acting with the Levi-Civita connection~$\nabla$ of~$g$ on an arbitrary vector field~$v$ yields
\be
\aligned
\nabla_i \nabla_j v^k
& = \ynabla_i (\nabla_j v^k) + \zeta^k{}_{im} \nabla_j v^m - \zeta^m{}_{ij} \nabla_m v^k
\\
& = \ynabla_i (\ynabla_j v^k + \zeta^k{}_{jl} v^l)
+ \zeta^k{}_{im} (\ynabla_j v^m + \zeta^m{}_{jl} v^l)
- \zeta^m{}_{ij} (\ynabla_m v^k + \zeta^k{}_{ml} v^l).
\endaligned
\ee
Taking into account the symmetry $\zeta^m{}_{ij}=\zeta^m{}_{ji}$, the commutator simplifies to
\be
[\nabla_i,\nabla_j] v^k
= [\ynabla_i, \ynabla_j] v^k
+ \bigl( \ynabla_i \zeta^k{}_{jl} - \ynabla_j \zeta^k{}_{il}
+ \zeta^k{}_{im} \zeta^m{}_{jl} - \zeta^k{}_{jm} \zeta^m{}_{il} \bigr) v^l, 
\ee
whence the well-known expression
\be
R^k{}_{lij} = \yR^k{}_{lij} + \ynabla_i \zeta^k{}_{jl} - \ynabla_j \zeta^k{}_{il}
+ \zeta^k{}_{im} \zeta^m{}_{jl} - \zeta^k{}_{jm} \zeta^m{}_{il}.
\ee
The $\zeta*\zeta$ terms can also be written as $\ynabla g*\ynabla g$. The derivatives of $\zeta= \nabla- \ynabla$ are evaluated by using the second expression in~\eqref{nabla-ynabla}, which yields the expression of $R^k{}_{lij}$ stated in~\eqref{Riem-zRiem}.

Consequently, the Ricci curvature, obtained by tracing the indices $i$ and~$k$, has two fewer terms of the form $\ynabla\ynabla g$, since the terms $\ynabla_j \ynabla_m g_{il} - \ynabla_j\ynabla_i g_{lm}$ cancel in this trace. The scalar curvature is obtained by contraction with the inverse metric $(g^{-1})^{jl}= (\yg^{-1})^{jl}- (\yg^{-1})^{jm}(g- \yg)_{mk}(\yg^{-1})^{kl}+(g- \yg)^{*2}$, which is responsible for the linear term $- (\yg^{-1})^{jm} \yR_{lj}(\yg^{-1})^{kl} (g- \yg)_{km}$ and the quadratic term $\yRic*(g- \yg)^{*2}$.
\ese
\end{proof}


\subsection{Expansion of the Einstein constraints}

We turn our attention to the constraint operators, namely 
\be
\Hcal[g,h] = R + \frac{1}{n-1} (\Tr_g h)^2 - |h|_g^2, 
\quad \qquad
\Mcal[g,h]^i = \nabla_j h^{ji},
\ee
where $\nabla$ and $R$ are the Levi-Civita connection and scalar curvature of~$g$, respectively.

\begin{lemma}
\label{lem:constr-expand}
Given data $(\yg,\yh)$ and $(g,h)=(\yg,\yh)+(\gdiff,\hdiff)$, the constraints admit the expansion
\be
\aligned
\Hcal[g,h]
& = \Hcal[\yg,\yh] + \ydH[\gdiff,\hdiff] + \yQH[\gdiff,\hdiff] ,
\\
\yQH[\gdiff,\hdiff] & = \ynabla\gdiff * \ynabla\gdiff + \gdiff * \ynabla\ynabla\gdiff
 + \yRic * \gdiff * \gdiff
+ \yh*\yh*\gdiff*\gdiff + \yh*\gdiff*\hdiff + \hdiff*\hdiff,
\\
\Mcal[g,h]^i & = \Mcal[\yg,\yh]^i + \ydM[\gdiff,\hdiff]^i + \yQM[\gdiff,\hdiff] , \qquad\
\yQM[\gdiff,\hdiff] = \yh*\gdiff*\ynabla\gdiff + \hdiff * \ynabla\gdiff,
\endaligned
\ee
in which the nonlinearities are expressed in the notation~\eqref{equa--B2}, and the linearized constraints read 
\bel{lin-constr}
\aligned
\ydH[\gdiff,\hdiff]
& \coloneqq (\yg^{-1})^{ik} (\yg^{-1})^{jl} \bigl( \ynabla_i \ynabla_j \gdiff_{kl} - \ynabla_i \ynabla_k \gdiff_{jl} - \yR_{ij} \gdiff_{kl} \bigr) 
\\
& \, \quad + \frac{2}{n-1} \yg_{kl} \yh^{kl} \bigl( \yg_{ij} \hdiff^{ij} + \gdiff_{ij} \yh^{ij} \bigr)
- 2 \yg_{kl} \yh^{li} \bigl(\yg_{ij}\hdiff^{jk} + \gdiff_{ij} \yh^{jk}\bigr),
\\
\ydM[\gdiff,\hdiff]^i & \coloneqq \ynabla_j \hdiff^{ji} + \yh^{jk} (\yg^{-1})^{il} \Bigl( \ynabla_j\gdiff_{kl} - \frac{1}{2} \ynabla_l \gdiff_{jk} \Bigr) + \frac{1}{2} \yh^{ij} \ynabla_j\yTr\gdiff.
\endaligned
\ee
\end{lemma}

\begin{proof}
{\it 1. Hamiltonian constraint.} Using that the notation $*$ can absorb factors of $\gdiff$ (since that is a difference of metrics), the expansions of $\Tr_g h$ and $|h|_g^2$ are found to be
\be
\aligned
\Tr_g h & = (\yg_{ij}+ \gdiff_{ij})(\yh^{ij}+ \hdiff^{ij}) = \yTr \yh + (\yg_{ij} \hdiff^{ij} + \gdiff_{ij} \yh^{ij}) + \gdiff*\hdiff, 
\\
(\Tr_g h)^2 & = (\yTr \yh)^2 + 2 (\yTr \yh) (\yg_{ij}\hdiff^{ij} + \gdiff_{ij} \yh^{ij}) + \yh*\yh*\gdiff*\gdiff + \yh*\gdiff*\hdiff + \hdiff*\hdiff, 
\\
g_{ij} h^{jk} & = (\yg_{ij}+ \gdiff_{ij}) (\yh^{jk}+ \hdiff^{jk})
= \yg_{ij} \yh^{jk} + (\yg_{ij}\hdiff^{jk} + \gdiff_{ij} \yh^{jk}) + \gdiff*\hdiff, 
\\
|h|_g^2 & = |\yh|_{\strut\yg}^2 + 2 \yg_{kl} \yh^{li} (\yg_{ij}\hdiff^{jk} + \gdiff_{ij} \yh^{jk}) + \yh*\yh*\gdiff*\gdiff + \yh*\gdiff*\hdiff + \hdiff*\hdiff.
\endaligned
\ee
Together with the expansion of~$R$ given in~\eqref{Riem-zRiem}, this yields the stated expansion of~$\Hcal$.

\medskip

\noindent{\it 2. Momentum constraint.} Using $\zeta= \nabla- \ynabla$ given in~\eqref{nabla-ynabla}, we write
\be
\aligned
\nabla_j h^{ji}
& = \ynabla_j h^{ji} + \zeta^i{}_{jk} h^{jk} + \zeta^j{}_{jk} h^{ki}
\\
& \compressmath{.8}{
= \ynabla_j \yh^{ji} + \ynabla_j \hdiff^{ji} + \yh^{jk} (\yg^{-1})^{il} \Bigl( \ynabla_j\gdiff_{kl} - \frac{1}{2} \ynabla_l \gdiff_{jk} \Bigr)
+ \frac{1}{2} \yh^{ik} \ynabla_k\yTr\gdiff + \yh*\gdiff*\ynabla\gdiff + \hdiff * \ynabla \gdiff, 
}
\endaligned
\ee
which is the desired expansion of the momentum constraint.
\end{proof}


\begin{proposition}[Expansion of the Einstein constraints near a Euclidean end]
\label{lem:lin-constr-coord}
Given asymptotically Euclidean data sets $(\yg,\yh)$ and $(g,h)=(\yg,\yh)+(\gdiff,\hdiff)$ and a coordinate chart $(x^i)$ defined on an asymptotically Euclidean end, the linearized Einstein constraint admits the expansion
\bel{lin-constr-coord}
\aligned
\ydH[\gdiff,\hdiff]
& = \del_i \del_j \gdiff_{ij} - \del_i \del_i \gdiff_{jj} + (\yg- \delta) * \del\del\gdiff
\\
& \quad 
+ \del\yg * \del\gdiff + \del\del\yg * \gdiff + \del\yg * \del\yg * \gdiff + \yh * \hdiff + \yh * \yh * \gdiff,
\\
\ydM[\gdiff,\hdiff]^{\smallbullet} & = \del_j \hdiff^{j\smallbullet} + \del\yg * \hdiff + \yh * \del\gdiff + \yh * \del\yg * \gdiff, 
\endaligned
\ee
where $\del_i=\del/\del x^i$ denotes a partial derivative, indices are raised or lowered using the Euclidean metric $\delta=dx^i\otimes dx^i$ and summation over repeated indices is implicit, regardless of their (upper, lower) position. The nonlinear constraints read
\bel{lin-constr-coord-2}
\aligned
\Hcal[g,h] - \Hcal[\yg,\yh]
& = \del_i \del_j \gdiff_{ij} - \del_i \del_i \gdiff_{jj}
\\
& \quad 
+ (\yg- \delta) * \del\del\gdiff
+ \del\yg * \del\gdiff
+ \del\del\yg * \gdiff
+ \del\yg * \del\yg * \gdiff
\\
& \quad
+ \yh * \hdiff
+ \yh * \yh * \gdiff
+ \del\gdiff * \del\gdiff
+ \gdiff * \del\del\gdiff
+ \hdiff*\hdiff,
\\
\Mcal[g,h]^{\smallbullet} - \Mcal[\yg,\yh]^{\smallbullet}
& = \del_j \hdiff^{j\smallbullet} + \del\yg * \hdiff + \yh * \del\gdiff + \yh * \del\yg * \gdiff + \hdiff * \del\gdiff.
\endaligned
\ee
\end{proposition}

\begin{proof} 
\bse
The main ingredient in expanding the linearized constraints in coordinates is the observation that, for an arbitrary tensor~$T$, 
\bel{ynablaT}
\ynabla_i T^{j\dots}_{k\dots} = \del_i T^{j\dots}_{k\dots} + \del\yg * T. 
\ee
 Thus, we have 
\be
\aligned
\ynabla_i \ynabla_j \gdiff_{kl} - \ynabla_i \ynabla_k \gdiff_{jl} - \yR_{ij} \gdiff_{kl}
& = \del_i \ynabla_j \gdiff_{kl} - \del_i \ynabla_k \gdiff_{jl} + \del\yg * \ynabla\gdiff
+ \del\yg * \del\yg * \gdiff + \del\del\yg * \gdiff
\\
& = \del_i \del_j \gdiff_{kl} - \del_i \del_k \gdiff_{jl} + \del\yg * \del\gdiff
+ \del\yg * \del\yg * \gdiff + \del\del\yg * \gdiff.
\endaligned
\ee
Contracting with the inverse metric as in~\eqref{lin-constr} yields the announced expression of~$\ydH$. The analogous expansion of the linearized momentum constraint is straightforward. Regarding nonlinearities, most of the relevant terms in the coordinate expansion can be absorbed into error terms that are already present in the linear part. For instance, we can write 
\be
\ynabla\gdiff * \ynabla\gdiff
= \del\gdiff * \del\gdiff + \gdiff * \del\yg * \del\gdiff + \gdiff * \gdiff * \del\yg * \del\yg
= \del\gdiff * \del\gdiff + \del\yg * \del\gdiff + \gdiff * \del\yg * \del\yg, 
\ee
since $\gdiff=g- \yg$ can be absorbed into the notation~$*$.
\ese
\end{proof}


\subsection{Expansion of the adjoint constraints}

With an obvious change of notation, \autoref{lem:constr-expand} also provides us with the linearization~\eqref{lin-constr} of the constraints around a data set $(\zg,\zh)$ instead of $(\yg,\yh)$. Here, we determine the adjoint of these linearized constraints. Starting from the linearized Einstein constraint operator $\zdG=(\zdH,\zdM) \colon S^0_2\times S^2_0\to S^0_0\times S^1_0$, we determine its formal adjoint operator $\zdG^*\colon S^0_0\times S^0_1\to S^2_0\times S^0_2$ (defined over the tensor fields $S_m^n$ of the type corresponding to each argument).
With a mild abuse of notation, we write $\zdG^*=(\zdH^*,\zdM^*)$: this notation is motivated by the time-symmetric case $\zh= 0$ since, in that case, $\zdH[\gdiff,\hdiff] =\zdH[\gdiff]$ and $\zdM[\gdiff,\hdiff] =\zdM[\hdiff]$ decouple.
To make sense of~\eqref{gh-ygyh-uZ}, we lower and raise indices using the metric~$\zg$, and obtain an operator $(\zdH^{*\flat\flat},\zdM^{*\sharp\sharp})\colon S^0_0\times S^0_1\to S^0_2\times S^2_0$, for which we occasionally use the same notation~$\zdG^*$ when no confusion may arise.

\begin{lemma}[Adjoint of the linearized Einstein constraints]
\label{lem:lin-constr}
The operator $\zdG^*$ takes the explicit form
\[
\aligned
\zdH^{*\flat\flat}[u,Z]_{ij} & = \znabla_i \znabla_j u - \zR_{ij} u
- \zg_{ij} \zDelta u 
+ \frac{2}{n-1} (\zTr\zh) \zg_{ik} \zh^{kl} \zg_{lj} u - 2 \zg_{ik} \zh^{km} \zg_{mq} \zh^{ql} \zg_{lj} u 
\\
& \quad + \frac{1}{2} \znabla_k \Bigl(\zg_{il}\zh^{lm}\zg_{mj} (\zg^{-1})^{kq} Z_q - \zh^{kl} \zg_{li} Z_j - \zh^{kl} \zg_{lj} Z_i - \zh^{kl} Z_l \zg_{ij}\Bigr),
\\
\zdM^{*\sharp\sharp}[u,Z]^{ij} & = - \frac{1}{2} (\zg^{-1})^{ik} (\Lie_Z \zg)_{kl} (\zg^{-1})^{lj} + \Bigl( \frac{2}{n-1} (\zTr\zh) (\zg^{-1})^{ij} - 2 \zh^{ij} \Bigr) u.
\endaligned
\]
\end{lemma}

\begin{proof} As there is a single metric involved in the problem, we freely raise and lower indices using~$\zg$ in the present proof. The adjoint of the linearized constraints is evaluated by a formal integration by parts starting from the linearized operators $\zdH$ and $\zdM$ derived previously in~\eqref{lin-constr}, that is, 
\[
\aligned
& \int \bigl( \zdH[\gdiff,\hdiff] \, u + \zdM[\gdiff,\hdiff]^i \, Z_i \bigr) \zvol
\\
& = \int \Bigl( \znabla^i \znabla^j \gdiff_{ij} - \zDelta \gdiff^i_i - \zR^{ij} \gdiff_{ij}
+ \frac{2}{n-1} \zh^j_j \bigl( \hdiff^i_i + \zh_i^k \gdiff_k^i \bigr)
- 2 \zh^{ij} \bigl( \hdiff_{ij} + \zh_i^k \gdiff_{kj} \bigr) \Bigr) u \, \zvol \\
& \quad + \int \Bigl( \znabla_j \hdiff^{ji} + \zh^{jk} \Bigl( \znabla_j\gdiff_k^i - \frac{1}{2} \znabla^i \gdiff_{jk} \Bigr) + \frac{1}{2} \zh^{ij} \znabla_j\gdiff_l^l \Bigr) Z_i \, \zvol, 
\endaligned
\]
which equals 
\[
\aligned
& \int \Bigl( \znabla^i \znabla^j u - \zg^{ij} \zDelta u - \zR^{ij} u + \frac{2}{n-1} \zh^k_k \zh^{ij} u - 2 \zh_k^i \zh^{kj} u \\
& \qquad - \znabla_k(Z^i \zh^{jk}) + \frac{1}{2} \znabla_k(Z^k \zh^{ij}) - \frac{1}{2} \znabla_l(Z_k \zh^{kl}) \zg^{ij} \Bigr) \gdiff_{ij} \zvol \\
& \quad + \int \Bigl( \frac{2}{n-1} u \zh^k_k \zg_{ij} - 2 u \zh_{ij} - \znabla_j Z_i \Bigr) \hdiff^{ij} \zvol
\\
& = \int \bigl( \zdH^*[u,Z]^{ij} \gdiff_{ij} + \zdM^*[u,Z]_{ij} \hdiff^{ij} \bigr) \, \zvol.
\endaligned
\]
The symmetry of $\gdiff$ and~$\hdiff$ requires $\zdH^*$ and $\zdM^*$ to be symmetrized, which replaces (for instance) the term $- \znabla_j Z_i$ by its symmetrization $- \frac{1}{2}(\znabla_i Z_j+ \znabla_jZ_i)=- \frac{1}{2}(\Lie_Z\zg)_{ij}$. We finally restore explicitly in the statement of \autoref{lem:lin-constr} the metrics and inverse metrics used to raise or lower indices.
\end{proof}

\begin{proposition}[Expansion of the adjoint constraints near a Euclidean end]
\label{prop:adjoint-constr-coord}
Given asymptotically Euclidean data $(\zg,\zh)$, a scalar field $u$, and vector field $Z$, together with a coordinate chart $(x^i)$ defined on an asymptotically Euclidean end, the adjoint of the linearized Einstein constraint admits the expansion
\[
\aligned
\zdH^{*\flat\flat}[u,Z]_{\smallbullet\smallbullet}
& = \del_{\smallbullet} \del_{\smallbullet} u - \delta_{\smallbullet\smallbullet} \Delta_\delta u
+ (\zg-\delta) * \del\del u
+ \del\zg * \del u
+ \del\del\zg * u + \del\zg * \del\zg * u \\
& \quad + \zh * \zh * u + \del\zg * \zh * Z + \del\zh * Z + \zh * \del Z,
\\
\zdM^{*\sharp\sharp}[u,Z]^{\smallbullet\smallbullet} & =
- \frac{1}{2} (\Lie_Z\delta)^{\smallbullet\smallbullet}
+ (\zg- \delta) * \del Z + \del\zg * Z + \zh * u,
\endaligned
\]
where $\del_i=\del/\del x^i$ denotes a partial derivative, $\Delta_\delta=\del_i\del_i$, and indices are raised or lowered using the Euclidean metric $\delta=dx^i\otimes dx^i$, so that $(\Lie_Z\delta)^{ij}=\delta^{ik}\delta^{jl}(\del_k Z_l+\del_l Z_k)$.
\end{proposition}

\begin{proof} As for \autoref{lem:lin-constr-coord}, the proof relies on~\eqref{ynablaT}, that is, $\znabla_i T^{j\dots}_{k\dots} = \del_i T^{j\dots}_{k\dots} + \del\zg * T$. 
For instance, the Lie derivative in~$\zdM^*$ reads
\[
(\Lie_Z \zg)_{kl}
= \znabla_k Z_l + \znabla_l Z_k = \del_k Z_l + \del_l Z_k + \del\zg * Z ,
\]
and raising the indices with $\zg^{-1}$ differs from raising them with~$\delta$ by terms of the form
\[
(\zg^{-1} - \delta) * (\Lie_Z \zg) = (\zg - \delta) * \del Z + \del\zg * Z .
\]
The second-order derivative terms in~$\zdH^{*\flat\flat}$ also need to be evaluated, namely 
\[
\znabla_i \znabla_j u
= \bigl(\del_i + \del\zg*\bigr) \bigl(\del_j + \del\zg*\bigr) u 
= \del_i \del_j u + \del\zg * \del u + \del\del\zg * u + \del\zg * \del\zg * u.
\]
The Laplacian is obtained by contracting this Hessian with $\zg^{ij}=\delta^{ij}+(\zg- \delta)^{ij}$ and the Ricci term is $R_{ij}u = \del\del\zg * u + \del\zg * \del\zg * u$. The remaining terms in $\zdG^*$ all involve~$\zh$ and are written with the $*$~notation.
\end{proof}


\section{Poincaré--Korn--Hardy inequalities on conical domains}
\label{appendix=D}

\subsection{Density of compactly supported functions}
\label{appendix=D.1}

\paragraph{The class of localization functions}

Throughout this section we work in a fixed Euclidean conical domain.  We derive density results (\autoref{appendix=D.1}) used in the main text to justify the absence of some angular boundary terms (e.g.\ in \autoref{section=5.2}), and geometric inequalities (\autoref{appendix=D.2} and \autoref{appendix=D.3}) in this flat setting that are extended to curved space in \autoref{appendix=E} and play a crucial role there to prove the existence of the seed-to-solution map.

Let us recall our notation.
We write $r=|x|$ and $\xh=x/r$ for $x\in\RR^n$ with $n\geq 3$.
We work on the open truncated cone
\be
\Omega_R = \{r\xh:\ r>R,\ \xh\in\Lambda\}\subset\RR^n,
\ee
where $R>0$ and $\Lambda\subset\Sphe^{n-1}$ is a connected domain with $C^2$ boundary.
The localization function $\lambda\in C^2(\overline\Lambda)$ is a defining function of the boundary in the sense that
\bel{equa-append33}
\lambda > 0 \text{ in } \Lambda , \qquad
\lambda = 0 \text{ on } \partial\Lambda , \qquad
|\nablaslash\lambda|>0 \text{ on } \partial\Lambda .
\ee
After multiplying $\lambda$ by a fixed positive constant, we also assume $0 < \lambda \leq 1$ on $\Lambda$. 
All constants below may depend on these fixed data~$(\Lambda,\lambda)$.

We parametrize the level sets of~$\lambda$ by coordinates $(t,y)\in[0,t_0]\times\del\Lambda$ near the boundary of~$\overline{\Lambda}$ such that $\lambda(t,y)=t$ and constant-$y$ lines are orthogonal to the level sets; these can be constructed by a normalized gradient flow of~$\lambda$ starting at $y\in\del\Lambda$.
These are coordinates on $\{0\leq\lambda\leq t_0\}$ for sufficiently small $0<t_0<m_0\leq 1$.  The measure induced on $\partial\Lambda$ is denoted by $d\sigma(y)$.

\begin{remark}
One could specify the dependence on~$\lambda$ as follows.  Consider the set
\be
\Ucal_\rho \coloneqq \bigl\{\xh\in\overline\Lambda \bigm| \operatorname{dist}_{\Sphe^{n-1}}(\xh,\partial\Lambda) < \rho \bigr\}.
\ee
For sufficiently small $\rho_0>0$, the sets $\Ucal_\rho$ for $0 < \rho\leq 2\rho_0$ are tubular collars of~$\partial\Lambda$ with nowhere vanishing $\nablaslash\lambda$.
Then for some constants $c_0,c_1,M_2,m_0>0$ one has
\be
\inf_{\overline\Lambda\setminus\mathcal U_{\rho_0/2}}\lambda\geq m_0 , \qquad
0<c_0 \leq|\nablaslash\lambda|\leq c_1 \ \text{in } \Ucal_{\rho_0} , \qquad
\|\lambda\|_{C^2(\overline\Lambda)} \leq M_2 .
\ee
All constants below can be shown to be uniform over localization functions~$\lambda$ satisfying these bounds, and may depend on $\rho_0,c_0^{-1},c_1,M_2,m_0^{-1}$ and the fixed $C^2$ geometry of~$\Lambda$.
The lower bound by $m_0$ is needed in addition to the $C^2$ bound in order to control uniformly the region away from~$\del\Lambda$.
\end{remark}


\paragraph{Density properties}

The squared weighted Sobolev norms of interest here are 
\bel{f-Hlp-a-def}
\|f\|_{H^l_{p,-a}(\Omega_R)}^2 \coloneqq \sum_{k=0}^{l}\int_{\Omega_R} |\nabla^k f|^2\lambda^{2a}r^{2p-n+2k}\,dx ,
\ee
for any given $p,a\in\RR$ and any integer $l \geq 0$. We denote $L^2_{p,-a}(\Omega_R)=H^0_{p,-a}(\Omega_R)$, and set
\be
\overline{\Omega}_R = \{ r\xh\mid r\geq R,\ \xh\in\overline{\Lambda} \}, \qquad
\Lambda_R = \{ R\xh\mid \xh\in\Lambda \}, \qquad
\Gamma_R = \{ r\xh\mid r\geq R,\ \xh\in\partial\Lambda \},
\ee
with a common corner $\del\Lambda_R = \del\Gamma_R = \{ R\xh \mid \xh\in\partial\Lambda \}$.
The appropriate test space is
\bel{DcalR-def}
\DcalR(\Omega_R, E) \coloneqq C_c^\infty(\overline{\Omega}_R \setminus \Gamma_R, E),
\ee
where $E$ is any finite-dimensional Euclidean space and we omit $E$ from the notation when $E=\RR$.
This test space consists of functions on the manifold with boundary $\overline{\Omega}_R\setminus\Gamma_R$ that extend smoothly across the radial boundary~$\Lambda_R$ and have compact support in $\overline{\Omega}_R\setminus\Gamma_R$.  Thus they may be nonzero on~$\Lambda_R$, but they vanish for large~$r$ and in a neighborhood of the angular boundary~$\Gamma_R$.
Observe that \emph{no boundary condition is imposed on $\Lambda_R$.}
Radial cutoffs are taken at large $\log(r/R)$ and never near $r=R$.  Angular
cutoffs are taken in the collar variable $\lambda$ and are justified by \autoref{lem-density}, below, which is then extended from one dimension up to a density result in~$\Omega_R$, \autoref{prop-density}.  At this stage, we restrict attention to $H^1$~functions.

\begin{lemma}[One-dimensional density lemma]
  \label{lem-density} 
  Let $A,B\in\RR$, and let $H$ be a Hilbert space. Consider\footnote{Functions in $C_c^\infty((0,1])$ have compact support contained in~$(0,1]$, namely they vanish in a neighborhood of~$0$ and smoothly extend through the regular end-point $t=1$ with no boundary condition.  The tensor product $C_c^\infty((0,1])\otimes H$ consists of finite sums of smooth scalar functions multiplied by fixed vectors of~$H$.}
  the space $V_{A,B}(H)$ of locally $H^1$ maps $f:(0,1]\to H$ with finite squared norm 
  \be
  \|f\|_{V_{A,B}(H)}^2
  \coloneqq \int_0^1 \bigl( t^{2A-1}\|f(t)\|_H^2 + t^{2B+1}\|f'(t)\|_H^2 \bigr)\,dt .
  \ee

  \bei
\item If $A\leq 0$ or $0\leq B$ then $C_c^\infty((0,1]) \otimes H$ is dense in $V_{A,B}(H)$.

\item If $B<0<A$, then the trace map $\mathrm{T}_0f=\lim_{t\to 0}f(t)$ is bounded from $V_{A,B}(H)$ to $H$, and the closure of $C_c^\infty((0,1]) \otimes H$ is its kernel.
  When $H$ is finite-dimensional, this kernel has codimension~$\dim H$.
  \eei
\end{lemma}

\begin{proof}
  To limit clutter, we write $\|\cdot\|_{A,B}$ for $\|\cdot\|_{V_{A,B}(H)}$ and $|\cdot|$ for the norm in~$H$.
  First, observe that any locally-$H^1$ function $f\in H^1_{\textnormal{loc}}((0,1],H)$ that \emph{vanishes in a neighborhood of~$0$} can be approached in $V_{A,B}(H)$ norm by elements of $C_c^\infty((0,1])\otimes H$ by a standard mollification argument.  The only difficulties thus arise near $t=0$.
  We distinguish cases $A\leq B$ and $A>B$, with the latter being further separated into $B>0$, $B<0$, and $B=0$.

  \medskip

  \bse
  \noindent{\it Case 1.} For $A\leq B$ we show density by exhibiting an explicit regularization of $f\in V_{A,B}(H)$.
  Consider the monotonically decreasing bijection $U:(0,1]\to[0,+\infty)$
  \bel{equa-density-cutoff}
  U(t)=\int_t^1s^{-1-B+A}\,ds ,
  \ee
  which has $U(t)\to+ \infty$ as $t\to 0$.  Choose $\kappa\in C^\infty(\RR,[0,1])$ with $\kappa=0$ on $(-\infty,-1]$ and $\kappa=1$ on $[1,+ \infty)$, and set $\eta_v(t)=\kappa(v-U(t))$.  Thus the product $\eta_vf$ vanishes near~$0$, $\eta_v\to1$ pointwise on $(0,1]$ as $v\to+\infty$, and
  \bel{equa-density-error}
  \int_0^1 t^{2B+1} |\eta_v'f|^2\,dt
  \lesssim \int_{\{v-1\leq U\leq v+1\}} t^{2A-1} |f|^2\,dt .
  \ee
  The right-hand side tends to zero by absolute continuity of the integral. The other terms in $\|(1-\eta_v)f\|_{A,B}$ are tails of integrable
  functions, so $\eta_vf\to f$ in~$V_{A,B}$ as $v\to+\infty$.

  \medskip

  \noindent{\it Case 2.}  For $A>B>0$, the weighted Hardy inequality with an interior
  remainder gives
  \be
  \int_0^{1/2}t^{2B-1} |f|^2\,dt
  \lesssim \int_0^1t^{2B+1} |f'|^2\,dt + \int_{1/2}^1 |f|^2\,dt
  \ee
  (with a $B$-dependent constant).
  Since $A>B$, this shows that $V_{A,B}=V_{B,B}$ with equivalent norms, so that the preceeding density result with $A=B$ concludes.

  \medskip

  \noindent{\it Case 3.} 
  Suppose next that $A>B$ and $B<0$. The Cauchy--Schwarz inequality yields
  \be
  \int_0^1|f'|\,dt
  \leq \biggl(\int_0^1 t^{2B+1} |f'|^2\,dt\biggr)^{1/2} \biggl(\int_0^1 t^{-2B-1}\,dt\biggr)^{1/2},
  \ee
  so every $f\in V_{A,B}$ admits a trace $\mathrm{T}_0f$ at $t=0$.
  Combining this bound with the ordinary $H^1$ estimate on $[1/2,1]$ gives $|\mathrm{T}_0f|\lesssim\|f\|_{A,B}$.  
  \bei
\item If $A\leq 0$, the integrability of $t^{2A-1}|f|^2$ forces $\mathrm{T}_0f=0$.
\item 
  If $A>0$, constants belong to $V_{A,B}$ so the trace map~$\mathrm{T}_0$ is onto $H$.
  Hence, when $H$ is finite-dimensional, its kernel has codimension $\dim H$.
\eei
We now show density in the space of functions $f\in\ker\mathrm{T}_0$.  The weighted Hardy inequality with $\mathrm{T}_0f=0$, namely
\be
\int_0^{1/2} t^{2B-1} |f|^2\,dt
\lesssim \int_0^{1/2} t^{2B+1} |f'|^2\,dt ,
\ee
reduces the case $A>B$ to $A=B$, which was already treated.  This proves both assertions when $B<0$.

  \medskip

  \noindent{\it Case 4.} 
  There remains to treat $B=0<A$.  With $s=-\log t$ and
  $g(s)=f(e^{-s})$, the squared norm is
  \be
  \int_0^\infty\bigl(e^{-2As}|g|^2+|g'|^2\bigr)\,ds .
  \ee
  Let $\chi_L=1$ on $[0,L]$, $\chi_L=0$ on $[2L,\infty)$, and
  $|\chi_L'|\lesssim L^{-1}$.  Since $g'\in L^2(0,\infty)$ implies
  $g(s)=o(\sqrt{s})$, we obtain 
  \be
  \int_0^\infty |\chi_L'g|^2\,ds
  \lesssim L^{-2} \int_L^{2L} |g(s)|^2\,ds \longrightarrow 0 , \qquad
  L\to+\infty .
  \ee
  Consequently $\chi_Lg\to g$ in the displayed norm.
  \ese
\end{proof}

\begin{proposition}[Density of compactly supported functions in a conical domain]
  \label{prop-density}
  Let $q,a\in\RR$.  
  \bei
  \item If $a\leq-1/2$ or $a\geq 1/2$, then $\DcalR(\Omega_R,E)$ is dense in $H^1_{q,-a}(\Omega_R,E)$.

  \item If $-1/2<a<1/2$, the angular trace operator, initially defined on smooth fields by the collar limit 
  \be
  \mathrm{T}_{\Gamma_R}f(r,y) \coloneqq \lim_{t\to 0}f(r,t,y),
  \ee
  extends uniquely to a bounded linear map
  \be
  \mathrm{T}_{\Gamma_R} : H^1_{q,-a}(\Omega_R,E)
  \longrightarrow L^2\bigl(\Gamma_R, r^{2q-1}\,dr\,d\sigma,E\bigr) .
  \ee
  The closure of $\DcalR(\Omega_R,E)$ is $\ker\mathrm{T}_{\Gamma_R}$.
  \eei
  \noindent Moreover, in all cases, the trace on $\Lambda_R$ is unrestricted and in particular the functions under consideration need not vanish at $r=R$. 
\end{proposition}

\begin{proof}
  In the collar coordinates $(r,t,y)$ the metric and inverse metric are uniformly bounded, and the volume factor in the $H^1_{q,-a}$ norm is $r^{2q-1}dr\,t^{2a}dt\,d\sigma$ up to uniformly bounded factors.  Thus for the angular part of the $H^1_{q,-a}$ norm, the one-dimensional exponents and Hilbert space in \autoref{lem-density} are
  \be
  A=a+ 1/2,\qquad B=a-1/2 , \qquad
  H = L^2(\Gamma_R, r^{2q-1}dr\,d\sigma,E)
  \ee
  and the exceptional case $B<0<A$ is $-1/2<a<1/2$.  In that interval the one-dimensional trace~$\mathrm{T}_{\Gamma_R}$ exists, and the closure (of functions vanishing near $t=0$, and under the angular part of the $H^1_{q,-a}$ norm) consists exactly of functions with zero angular trace.  Outside that interval there is no trace obstruction.
  Since the angular cutoff depends only on $t$, the $r$- and $y$-derivative terms in the norms produce no commutator and converge separately by dominated convergence.
  Away from the collar, ordinary mollification applies.
  There remains only to make the support bounded in the radial direction, which is done using the same cutoffs as \autoref{lem-density} with $(t,A,B)$ replaced by $(1/r,-q,-q)$.  This fits in Case 1 in the proof of \autoref{lem-density}, which has no trace condition.
  The continuity of the trace operator proves the converse inclusion and hence identifies the closure exactly with its kernel.
\end{proof}

\begin{proposition}[Second-order density in the variational range]
  \label{prop-density-H2}
  Let $q\in\RR$, let $a>3/2$, and let $E$ be a finite-dimensional Euclidean space. Then $\DcalR(\Omega_R,E)$ is dense in $H^2_{q,-a}(\Omega_R,E)$.
  In particular, this applies to the Hamiltonian variational space $H^2_{n-2-p,-\expoP}(\Omega_R)$ since $\expoP\geq 2$.
  No trace condition is imposed.
  Moreover, the traces $f|_{\Lambda_R}$ and $\partial_r f|_{\Lambda_R}$ at $r=R$ exist and are unconstrained.
\end{proposition}

\begin{proof}
  The proof is essentially identical to \autoref{prop-density}, with some simplifications due to the restriction to a single case $a>3/2$.
  We introduce a radial and angular cutoff and show that its effect is controlled thanks to angular Hardy inequalities.
  Choose $\kappa\in C^\infty(\RR,[0,1])$ with $\kappa=0$ on $(-\infty,-1]$ and $\kappa=1$ on $[1,+ \infty)$, and consider the function (for some $v\in\RR$)
  \be
  \eta_v = \kappa(v+\log\lambda) \kappa(v-\log r) ,
  \ee
  which belongs to $C_c^2(\overline{\Omega}_R \setminus \Gamma_R,[0,1])$ since it vanishes as $\lambda\to 0$ or $r\to+\infty$ and $\lambda$ is assumed to have $C^2$ regularity.
  We now prove that for $f\in H^2_{q,-a}(\Omega_R,E)$ one has $\eta_v f\to f$ as $v\to+\infty$.

  The function $1-\eta_v$ is supported in the region $\Delta_v=\{e^{v-1}<r<e^{v+1},e^{-1-v}<\lambda<e^{1-v}\}$, which is contained in the collar for sufficiently large~$v$.
  By expanding all derivatives in collar coordinates $(r,t,y)$ and noting that $|(r\del_r)^j\del_t^k\eta_v|\lesssim \lambda^{-k}$ for $k\leq 2$ one gets
  \be
  \|f-\eta_v f\|_{H^2_{q,-a}(\Omega_R)}^2
  \lesssim
  \|f\|_{H^2_{q,-a}(\Delta_v)}^2
  + \|f\|_{H^1_{q,-a+1}(\Delta_v)}^2
  + \|f\|_{L^2_{q,-a+2}(\Delta_v)}^2 .
  \ee
  The first term tends to zero as $v\to+\infty$ by absolute continuity of the integral.
  The other two terms are bounded by the first thanks to the Hardy inequality.

  For $b>1/2$, the Hilbert-valued Hardy inequality with an interior remainder gives
  \be
  \int_0^{t_0/2}t^{2b-2}\|h(t)\|_H^2\,dt
  \lesssim_b
  \int_0^{t_0}t^{2b}\|h'(t)\|_H^2\,dt
  +\int_{t_0/2}^{t_0}\|h(t)\|_H^2\,dt.
  \ee
  Applying it to derivatives $Df=\{r\del_r f,\del_t f,\del_y f\}$ with $b=a$, and to $f$ with $b=a-1$ gives the desired bounds.  We conclude that $\eta f\to f$.  A standard bounded $H^2$ extension across the regular face $r=R$ followed by mollification then produces approximants in $\DcalR(\Omega_R,E)$.

  Finally, the two traces along~$\Lambda_R$ are independent. Indeed, for any
  $\varphi_0,\varphi_1\in C_c^\infty(\Lambda_R,E)$, choose a radial cutoff
  $\zeta_0$ equal to one near $r=R$. The field
  \be
  w(r,\xh)
  =
  \zeta_0(r)
  \bigl(
  \varphi_0(\xh)+(r-R)\varphi_1(\xh)
  \bigr)
  \ee
  belongs to $\DcalR(\Omega_R,E)$ and satisfies
  $w|_{\Lambda_R}=\varphi_0$ and $\partial_rw|_{\Lambda_R}=\varphi_1$.
\end{proof}


\subsection{Weighted Poincar\'e--Hardy inequalities}
\label{appendix=D.2}

We now turn to geometric inequalities.  Contrarily to the main text where Sobolev norms on $\Lambda_R$ are normalized by a factor $\aire[\Lambda,\lambda]$, here, none of the norms involve this factor as it would make inconvenient the relation between $n$-dimensional norms on~$\Omega_R$ and those on its boundary.

\begin{proposition}[Weighted Poincar\'e--Hardy inequalities]
  \label{lem-Hardy}
  Let $q>0$, $a\in\RR$, and set $\delta=\Oneone_{a>1}$, equal to~$1$ if $a>1$ and $0$ if $a\leq 1$.

  \bei
\item For every $w\in H^1_{q,-a}(\Omega_R)$, one has
  \bel{equa-Poin}
  \|w\|_{L^2_{q,-a+\delta}(\Omega_R)}
  \lesssim \|\nabla w\|_{L^2_{q+1,-a}(\Omega_R)} .
  \ee

\item For every scalar distribution $w$ such that $\nabla w\in H^1_{q,-a}(\Omega_R,\RR^n)$, one has
  \bel{equa-Poin-Hess}
  \|\nabla w\|_{L^2_{q,-a+\delta}(\Omega_R,\RR^n)}
  \lesssim \mathopen\|\Hess w\|_{L^2_{q+1,-a}(\Omega_R)} .
  \ee
  \eei
  The corresponding traces along $\Lambda_R$ satisfy the stronger estimates
  \bel{equa-Poin-inner-traces}
  \aligned
    q^2\|w\|_{L^2_{q,-a}(\Omega_R)}^2
    + q R^{2q} \|w(R,\cdot)\|_{L^2(\Lambda,\lambda^{2a}d\xh)}^2
    & \leq \|\partial_r w\|_{L^2_{q+1,-a}(\Omega_R)}^2
    \leq \|\nabla w\|_{L^2_{q+1,-a}(\Omega_R)}^2 ,
    \\
    q^2\|\nabla w\|_{L^2_{q,-a}(\Omega_R)}^2
    + q R^{2q} \|\nabla w(R,\cdot)\|_{L^2(\Lambda,\lambda^{2a}d\xh)}^2
    & \leq \|\partial_r \nabla w\|_{L^2_{q+1,-a}(\Omega_R)}^2
      \leq \mathopen\|\Hess w\|_{L^2_{q+1,-a}(\Omega_R)}^2 .
  \endaligned
  \ee
  The implied constants depend only on $q$, $a$, $\Lambda$, $\lambda$, and in particular they are independent of $R$.
\end{proposition}

\begin{proof}
  \bse
  \noindent {\bf Radial estimate.}
  Regard $w(r,\cdot)$ as a map with values in $L^2(\Lambda,\lambda^{2a}d\xh)$.  With
  $s=\log(r/R)$ and $v(s,\xh)=r^qw(r,\xh)$, the weighted radial norms become the ordinary $L^2$ norms of $v$ and $\partial_sv-qv$ on the half-line.
  Thus $v$ and $\partial_sv=(\partial_sv-qv)+qv$ both belong to $L^2((0,\infty),L^2(\Lambda,\lambda^{2a}d\xh))$.
  Consequently $v\in H^1((0,\infty),L^2(\Lambda,\lambda^{2a}d\xh))$, and its continuous representative tends to zero at infinity.
  The one-dimensional trace theorem and an approximation in $s$ therefore justify the following identity for every $w\in H^1_{q,-a}(\Omega_R)$:
  \bel{equa-Poin-proof}
  \aligned
  0 & \leq \int_{\Omega_R}(r\del_r w+qw)^2 \lambda^{2a}r^{2q-n}\,dx
  \\
  & = \|\del_r w\|_{L^2_{q+1,-a}}^2 - q^2\|w\|_{L^2_{q,-a}}^2 - q R^{2q} \int_{\Lambda} w(R,\xh)^2 \lambda^{2a}\,d\xh .
  \endaligned
  \ee
  This is the first estimate in~\eqref{equa-Poin-inner-traces}.  Since the trace along $\Lambda_R$ has a favorable sign, we also obtain~\eqref{equa-Poin} with $\delta=0$.
  Applying the estimate, with the same exponent~$q$, to the components $u=\del_i w$ produces the second estimate in~\eqref{equa-Poin-inner-traces} as well as \eqref{equa-Poin-Hess} with $\delta=0$.

  \medskip

  \textbf{Angular estimate.}
  Assume $a>1$.  \autoref{prop-density} makes $\DcalR(\Omega_R)$ dense in $H^1_{q,-a}(\Omega_R)$, so it is sufficient to prove~\eqref{equa-Poin} for $w\in\DcalR(\Omega_R)$.
  Since the radial estimate controls an~$L^2$ norm whose angular weight is equivalent to the desired one away from the collar, we can focus on bounding the norm on the collar, only, and work in the collar coordinates $(r,t,y)$ introduced above, in which $t=\lambda$.
  Fix a cutoff $\kappa:[0,+\infty)\to[0,1]$ equal to $1$ on $[0,t_0/2]$ and $0$ on $[t_0,+\infty)$.
  For fixed $r,y$ and setting $f(t)=w(r,t,x)$, the one-dimensional Hardy inequality applied to $\kappa f$, which equals $f$ for $0<t<t_0/2$ gives (uniformly in $r,y$)
  \be
  \int_0^{t_0/2} t^{2a-2} f(t)^2\,dt
  \lesssim \int_0^{t_0} t^{2a} (\del_t f(t))^2\,dt + \int_{t_0/2}^{t_0} t^{2a} f(t)^2\,dt .
  \ee
  After integration in $r,y$ with measure $r^{2q-1}dr\,d\sigma$, this completes as needed the radial estimate to get the $\delta=1$ improvement for $a>1$.
  For the second inequality~\eqref{equa-Poin-Hess}, simply apply~\eqref{equa-Poin} componentwise to $\nabla w$.
  \ese
\end{proof}


\subsection{Weighted Korn--Hardy inequality}
\label{appendix=D.3}

\paragraph{Key inequality for the momentum constraint.}

We turn to a geometric inequality for vector fields that is used in our study of the momentum constraint.
The results of \autoref{appendix=D.1} imply that vector fields $Z\in H^1_{q,-a}(\Omega_R,\RR^n)$ with $a>1$ admit traces at $r=R$ that need not vanish.
The proof relies on several lemmas:
\autoref{lem-uniform-conical-Korn} provides a similar inequality on regions $\{\lambda>t\}$ uniformly as $t\to 0$, without angular weight;
\autoref{prop-unshifted-Korn} gives the inequality with no angular improvement and for smooth compactly-supported~$Z$, only;
the proof of \autoref{lem:KornHardyD}, given at the end of this section, extends the result by density and obtains an angular improvement thanks to the Poincaré--Hardy inequality in \autoref{lem-Hardy}.

\begin{proposition}[Weighted Korn--Hardy inequality]
  \label{lem:KornHardyD}
  Let $q>0$ and $a\geq 1/2$.  Let $\delta=\Oneone_{a>1}$.
  For every $Z\in H^1_{q,-a}(\Omega_R,\RR^n)$, the radial trace along~$\Lambda_R$ belongs to $L^2(\Lambda,\lambda^{2a}d\xh)$ and one has
  \bel{equa-Korn}
  \|Z\|_{L^2_{q,-a+\delta}(\Omega_R)}
  + \|\nabla Z\|_{L^2_{q+1,-a}(\Omega_R)} 
  + R^q \|Z(R,\cdot)\|_{L^2(\Lambda,\lambda^{2a}d\xh)}
  \lesssim \|\Dcal Z\|_{L^2_{q+1,-a}(\Omega_R)} .
  \ee
  The implied constant depends only on $n,q,a,\Lambda,\lambda$, and is independent of~$R$.
\end{proposition}

\begin{remark}\label{rem:smallaKorn}
  For $a\in[0,1/2)$, the same inequality holds with $\delta=0$ when restricted to vector fields~$Z$ whose trace along the angular boundary~$\Gamma_R$ vanishes.
\end{remark}

\paragraph{Inequality without angular weights.}

We introduce the domains (in~$\Lambda$ and $\Omega_R$, respectively)
\be
\Lambda_{>t} = \{ \xh\in\Lambda \mid \lambda(\xh)>t \}, \qquad
\Omega_R(\Lambda_{>t}) = \{ r\xh\mid r>R,\ \lambda(\xh) > t \} .
\ee
After decreasing the fixed~$t_0$ if necessary, these domains are connected for $0\leq t\leq 2t_0$.
Note that $\Lambda_{>0}=\Lambda$.

\begin{lemma}[Uniform Korn inequality on parallel cones]
  \label{lem-uniform-conical-Korn}
  Let $q>0$ and let $\Lambda_{>t}$ for some $0\leq t\leq t_0$.
  For every vector field $U$ that is smooth up to $\partial\Omega_R(\Lambda_{>t})$ and has bounded radial support, one has
  \bel{equa-slice-Korn}
  \int_{\Omega_R(\Lambda_{>t})} \Bigl( |U|^2 r^{2q-n} + |\nabla U|^2 r^{2q+2-n} \Bigr) dx
  \lesssim \int_{\Omega_R(\Lambda_{>t})} |\Dcal U|^2 r^{2q+2-n}\,dx .
  \ee
  The constant is uniform in $t\in[0,t_0]$ and $R$ and depends only on $n,q,\Lambda,\lambda$. No boundary condition is imposed on the radial boundary $r=R$ or angular boundary.
\end{lemma}

\begin{proof}
  \bse
  Let $\rho_j=2^jR$ for $j\geq 0$ and set
  \be
  A_j(\Lambda_{>t}) = \bigl\{ \rho_j < r < \rho_{j+2},\ \xh\in\Lambda_{>t} \bigr\}.
  \ee
  These annuli cover the cone and consecutive ones have finite overlap $\{\rho_j<r<\rho_{j+1},\ \xh\in\Lambda_{>t}\}$.
  By the local Korn inequality modulo Euclidean isometries (translations, rotations), one has
  \bel{equa-local-Korn}
  \min_{K_j\text{ Killing}} \int_{A_j(\Lambda_{>t})} \Bigl( |\nabla(U-K_j)|^2 + \rho_j^{-2}|U-K_j|^2 \Bigr)\,dx
  \lesssim \int_{A_j(\Lambda_{>t})}|\Dcal U|^2\,dx ,
  \ee
  where the minimum is taken over Killing vectors
  \be
  K_j(x) = b_j + B_j x, \qquad B_j^T = -B_j .
  \ee
  The constant is uniform in $j$ by scaling.  Let us show that it is also uniform in $t\in[0,t_0]$.
  Assume by contradiction that uniformity fails along a sequence $(t_k,U_k)_{k\geq 1}$ with $t_k\in[0,t_0)$.  After extraction, $t_k\to t_\infty$.
  Select a family of diffeomorphisms $[t_\infty,2t_0]\to[t_k,2t_0]$ with unit derivatives and vanishing second-derivatives at the end-points, and first and second derivatives converging uniformly to zero as $k\to+\infty$.  Pull-back $U_k$ to vector fields~$V_k$ on $A_j(\Lambda_{>t_\infty})$.
  After extraction, $V_k$ converges weakly in~$H^1$ and strongly in~$L^2$ to $V_\infty$.
  In particular, the effects of the pull-back on first-order derivatives tend to zero as $k\to+\infty$.
  By lower semi-continuity, $\Dcal V_\infty$ vanishes hence $V_\infty$ is a Killing vector $V_\infty = K_{j\infty}$.
  Then, strong $L^2$ convergence $V_k-K_{j\infty}\to 0$ and $\Dcal V_k-\Dcal K_{j\infty}\to 0$ together with the second Korn inequality implies that $V_k\to K_{j\infty}$ strongly in $H^1(A_j(\Lambda_{>t_\infty}))$.  Upon pushing forward to $A_j(\Lambda_{>t_k})$, the field~$V_k$ is mapped to~$U_k$ while $K_{j\infty}$ is mapped to a vector field~$W_k$.  The sequence $(W_k)$ converges in~$H^1$ to the Killing vector field itself since the diffeomorphisms converge to the identity.
  Altogether we have convergence of $V_k$ to a fixed Killing vector field, which contradicts the assumed normalization of the left-hand side of~\eqref{equa-local-Korn}.
  \ese

  \bse
  We now collect some inequality that will combine into the desired result.
  Firstly, on the overlap $A_j(\Lambda_{>t}) \cap A_{j+1}(\Lambda_{>t})$ of two successive annuli, two estimates~\eqref{equa-local-Korn} are available.
  Comparing the two and accounting for the norm equivalence on the space of rigid motions and \eqref{equa-local-Korn} gives, after rescaling by a power of $r\simeq\rho_j$,
  \bel{equa-rigid-difference}
  \rho_j^{2q}|b_{j+1}-b_j|^2 + \rho_j^{2q+2}|B_{j+1}-B_j|^2
  \lesssim \int_{A_j(\Lambda_{>t})\cup A_{j+1}(\Lambda_{>t})} |\Dcal U|^2 r^{2q+2-n}\,dx .
  \ee
  The bounded radial support of~$U$ and the definition of~$K_j$ as a minimizer implies that $K_j=0$ for all sufficiently large~$j$, namely $b_j=0$ and $B_j=0$.
  Secondly, one has the discrete Hardy inequalities for $q>0$
  \bel{equa-discrete-Hardy}
  \aligned
  \sum_{j\geq 0} \rho_j^{2q} |b_j|^2 & \lesssim_q \sum_{j\geq 0} \rho_j^{2q} |b_{j+1}-b_j|^2, \\
  \sum_{j\geq 0} \rho_j^{2q+2} |B_j|^2 & \lesssim_q \sum_{j\geq 0} \rho_j^{2q+2} |B_{j+1}-B_j|^2 ,
  \endaligned
  \ee
  proven by writing $b_j=-\sum_{k\geq j}(b_{k+1}-b_k)$ and applying the Cauchy--Schwarz inequality with a weight~$\rho_j^q$, and likewise for~$B_j$.
  Thirdly, finite-dimensional norm equivalence on a scaled annular sector, and~\eqref{equa-local-Korn} yield respectively
  \be
  \aligned
  \int_{A_j(\Lambda_{>t})} \Bigl( |K_j|^2 r^{2q-n} + |\nabla K_j|^2 r^{2q+2-n} \Bigr)\,dx
  & \lesssim \rho_j^{2q} |b_j|^2 + \rho_j^{2q+2} |B_j|^2 ,
  \\
  \int_{A_j(\Lambda_{>t})} \Bigl( |U-K_j|^2r^{2q-n} + |\nabla(U-K_j)|^2r^{2q+2-n} \Bigr)\,dx
  & \lesssim \int_{A_j(\Lambda_{>t})}|\Dcal U|^2 r_j^{2q+2-n}\,dx.
  \endaligned
  \ee
  Summing these estimates, applying~\eqref{equa-discrete-Hardy} and~\eqref{equa-rigid-difference}, and using finite overlap establishes~\eqref{equa-slice-Korn}.  The constant is independent of~$R$ by a scaling argument.
  \ese
\end{proof}

\paragraph{Introducing the angular weight.}

By applying \autoref{lem-uniform-conical-Korn}, we obtain the following result with an angular weight, which we state for smooth vector fields with compact support.  We recall from~\eqref{DcalR-def} that $\DcalR(\Omega_R,\RR^n)$ consists of smooth functions (valued in~$\RR^n$) with bounded support and that vanish in a neighborhood of the angular boundary~$\Gamma_R$.

\begin{lemma}[Weighted Korn inequality with compact support]
  \label{prop-unshifted-Korn}
  Let $q>0$ and $a\geq 0$.  For every $Z\in\DcalR(\Omega_R,\RR^n)$, one has
  \bel{equa-unshifted-Korn}
  \|Z\|_{L^2_{q,-a}(\Omega_R)}
  + \|\nabla Z\|_{L^2_{q+1,-a}(\Omega_R)}
  + R^q\|Z(R,\cdot)\|_{L^2(\Lambda,\lambda^{2a}d\xh)}
  \lesssim \|\Dcal Z\|_{L^2_{q+1,-a}(\Omega_R)} .
  \ee
  The implied constant depends only on $n,q,a,\Lambda,\lambda$, and is independent of~$R$.  In particular, the trace at $r=R$ is controlled, not required to vanish and the test fields may be nonzero there.
\end{lemma}

\begin{proof}
  \bse
  When $a=0$, the first two terms in \eqref{equa-unshifted-Korn} follow directly from \autoref{lem-uniform-conical-Korn} with $t=0$.

  Assume now that $a>0$.  To transfer the result of \autoref{lem-uniform-conical-Korn}, we rely on the layer-cake identity over the superlevel domains $\Lambda_{>t} = \{\lambda>t\}$.  Pointwise in $\xh$, we have
  \bel{equa-cake-pointwise}
  \lambda(\xh)^{2a}
  = 2a \int_0^{t_0} t^{2a-1} \Oneone_{\Lambda_{>t}}(\xh)\,dt
  + \Oneone_{\Lambda_{>t_0}}(\xh) \bigl(\lambda(\xh)^{2a}-t_0^{2a}\bigr).
  \ee
  Both terms on the right-hand side are nonnegative.  Multiplying~\eqref{equa-cake-pointwise} by a nonnegative function $F(x)$, integrating over $\Omega_R$ with respect to~$dx$,
  \bel{equa-layer-cake}
  \int_{\Omega_R}F\lambda^{2a}\,dx
  = 2a \int_0^{t_0} t^{2a-1} \int_{\Omega_R(\Lambda_{>t})} F\,dx\,dt
  + \int_{\Omega_R(\Lambda_{>t_0})} F \bigl(\lambda^{2a}-t_0^{2a}\bigr)\,dx .
  \ee

  Put now 
  \be
  F = |Z|^2 r^{2q-n} + |\nabla Z|^2 r^{2q+2-n},
  \qquad
  G = |\Dcal Z|^2 r^{2q+2-n}.
  \ee
  \autoref{lem-uniform-conical-Korn}, whose constant is uniform in $t$, gives
  \bel{cake-FleqG}
  \int_{\Omega_R(\Lambda_{>t})}F\,dx
  \lesssim \int_{\Omega_R(\Lambda_{>t})}G\,dx,
  \qquad 0\leq t\leq t_0.
  \ee
  Bounding $\lambda^{2a}-t_0^{2a}\lesssim 1$ in~\eqref{equa-layer-cake} and applying~\eqref{cake-FleqG} gives
  \be
  \aligned
  \int_{\Omega_R}F\lambda^{2a}\,dx
  & \lesssim
  2a \int_0^{t_0} t^{2a-1} \int_{\Omega_R(\Lambda_{>t})} G\,dx\,dt
  + \int_{\Omega_R(\Lambda_{>t_0})} G\,dx
  \\
  & \leq (1 + t_0^{-2a}) \int_{\Omega_R}G\lambda^{2a}\,dx
  \endaligned
  \ee
  where the second inequality comes from~\eqref{equa-layer-cake} applied to~$G$.
  This proves the first two terms in~\eqref{equa-unshifted-Korn} for $a>0$.

  In either case (that is, for $a\geq 0$), applying \eqref{equa-Poin-inner-traces} to the components of~$Z$ and then using the gradient estimate just proven yields
  \bel{equa-Korn-inner-trace}
  q R^{2q} \|Z(R,\cdot)\|_{L^2(\Lambda,\lambda^{2a}d\xh)}^2
  \leq \|\partial_r Z\|_{L^2_{q+1,-a}}^2
  \lesssim \|\Dcal Z\|_{L^2_{q+1,-a}}^2 .
  \qedhere
  \ee
  \ese
\end{proof}

\paragraph{Proof of \autoref{lem:KornHardyD}.}

\bse
We first complete \autoref{prop-unshifted-Korn} by density.
Since we assume $a\geq 1/2$, \autoref{prop-density} yields a sequence $Z_m\in\DcalR(\Omega_R,\RR^n)$ converging to $Z$ in~$H^1_{q,-a}$:
\be
Z_m\longrightarrow Z \qquad \text{in } H^1_{q,-a}(\Omega_R,\RR^n).
\ee
As stated in \autoref{rem:smallaKorn}, for $a\in[0,1/2)$ the trace along the angular boundary~$\Gamma_R$ obstructs density.
The first estimate in \eqref{equa-Poin-inner-traces}, applied componentwise to $Z_m-Z_k$, gives
\be
R^q \|Z_m(R,\cdot)-Z_k(R,\cdot)\|_{L^2(\Lambda,\lambda^{2a}d\xh)}
\lesssim \|Z_m-Z_k\|_{H^1_{q,-a}(\Omega_R)} .
\ee
Thus the inner traces form a Cauchy sequence in $L^2(\Lambda,\lambda^{2a}d\xh)$.
Its limit is the radial Sobolev trace~$Z(R,\cdot)$.
Applying \autoref{prop-unshifted-Korn} to $Z_m$ and passing to the limit proves the asserted estimate with $\delta=0$.

Then, for $a>1$, the Poincar\'e--Hardy estimate in \autoref{lem-Hardy}, applied to the Cartesian components of $Z$, gives as wanted
\be
\|Z\|_{L^2_{q,-a+1}(\Omega_R)}
\lesssim \|\nabla Z\|_{L^2_{q+1,-a}(\Omega_R)} .
\ee
Combining the two estimates proves \eqref{equa-Korn}.
\ese


\section{Existence of the seed-to-solution projection (\autoref{thm:sts-existence})}
\label{appendix=E}

\subsection{Geometric inequalities}
\label{appendix=E.1}

\paragraph{Overview of the proof.}

In a conical gluing domain with at least one asymptotic end (as defined in \autoref{section=2.3}), we consider the projection of a seed data set $(\seedg,\seedh)$ onto an exact solution $(g,h)$.  We prove convergence of the Newton iteration scheme~\eqref{Newton-it}, which we reproduce here as~\eqref{Newton-again} for convenience, and the resulting sub-harmonic estimates stated in \autoref{thm:sts-existence}.
Our proof broadly follows Carlotto and Schoen~\cite{CarlottoSchoen} as well as Corvino and Schoen~\cite{Corvino-2000,CorvinoSchoen} and~Chrusciel and Delay~\cite{ChruscielDelay-memoir,ChruscielDelay-2021,Delay}. Some differences must be taken care of within our presentation, since the operator of interest includes both a reference data set $(g_0,h_0)$ and a data set $(g_k,h_k)$ obtained after iteration, and we allow for a very low radial decay exponent $p_G>0$.

Starting from $(g_1,h_1)=(\seedg,\seedh)$, we solve an elliptic problem for $u_k,Z_k$ to obtain the next data set $(g_{k+1},h_{k+1})$, for $k\geq 1$,
\bel{Newton-again}
\aligned
\Jcal_{(g_k,h_k;g_0,h_0)}[u_k,Z_k] = - \Gcal(g_k,h_k) ,
\qquad
g_{k+1} & = g_k + \omegabf_p^2 \, d\Hcal_{(g_0,h_0)}^{*\flat\flat}[u_k,Z_k] ,
\\
h_{k+1} & = h_k + \omegabf_{p+1}^2  \, d\Mcal_{(g_0,h_0)}^{*\sharp\sharp}[u_k,Z_k] .
\endaligned
\ee
The differential operator~$\Jcal$ is the squared constraint operator
\be
\Jcal_{(g_k,h_k;g_0,h_0)}[u,Z] \coloneqq d\Gcal_{(g_k,h_k)}\Bigl[ \omegabf_p^2 \, d\Hcal_{(g_0,h_0)}^{*\flat\flat}[u,Z], \omegabf_{p+1}^2  \, d\Mcal_{(g_0,h_0)}^{*\sharp\sharp}[u,Z] \Bigr] .
\ee
Importantly, instead of a Picard iteration scheme which uses the same linearized operator at every step of the iteration, we rely on linearizations $d\Gcal_{(g_k,h_k)}$ around successive data sets $(g_k,h_k)$.  This avoids the need for Lipschitz continuity of nonlinearities: we use instead a control of nonlinearities around each data set in the iteration.

The proof proceeds in four steps in successive subsections.
Firstly, weighted Poincaré, Korn, and Hardy inequalities on~$\Omega$ hold because the decay at infinity does not allow $d\Gcal^*_{(g_0,h_0)}$ to have a kernel.
Secondly, this ensures coercivity of quadratic functionals associated to $\Jcal_{(g_0,h_0;g_0,h_0)}$ and~$\Jcal_{(g_k,h_k;g_0,h_0)}$, so that the Lax--Milgram theorem shows invertibility of the linear operator $\Jcal_{(g_k,h_k;g_0,h_0)}$, with Sobolev bounds on $(u_k,Z_k)$.
Thirdly, Douglis--Nirenberg ellipticity enhances these integral bounds to pointwise bounds on $(u_k,Z_k)$.
Lastly, nonlinearities are controlled as needed in Sobolev and H\"older norms to conclude convergence.


\paragraph{Geometric inequalities.}

In the special case $(g_k,h_k)=(g_0,h_0)$ the operator $\Jcal_{(g_0,h_0;g_0,h_0)}$ is self-adjoint and associated to a non-negative quadratic form
\be
\Jrm_{(g_0,h_0)}[u,Z]
\coloneqq \int_{\Mbf} \Bigl(
\omegabf_p^2 \, \bigl|d\Hcal_{(g_0,h_0)}^*[u,Z]\bigr|_{g_0}^2
+ \omegabf_{p+1}^2 \bigl|d\Mcal_{(g_0,h_0)}^*[u,Z]\bigr|_{g_0}^2
\Bigr) \dVol_{g_0} .
\ee
As we state in \autoref{lem:squaredquadraticform} below, it controls weighted $H^2\times H^1$ norms of $(u,Z)$, multiplied by a \emph{Poincaré--Korn--Hardy constant} $\CPKHzero>0$ introduced in \autoref{def:PoincareKornHardyConst}.
The proof relies on three geometric inequalities stated now, which generalize the ones in Euclidean space in \autoref{appendix=D}. The two Poincaré--Hardy inequalities state that derivatives of a scalar field~$w$ control lower derivatives with improved weights. The Korn--Hardy inequality states that symmetrized derivatives
$\Lie_{Z^\sharp} g_0 = 2\Sym(\nabla_{g_0} Z)$ control all derivatives, and $Z$~itself with improved weights.
Importantly, the equations $\nabla_{g_0}w=0$ or $\Hess_{g_0}(w)=0$ or $\Lie_{Z^\sharp} g_0=0$ have no solutions with $L^2_{q,-\expoP}$ decay at infinity in the asymptotic end(s).

\begin{lemma}[Poincaré--Hardy and Korn--Hardy inequalities]\label{lem:PoincareKornHardyD}
Assume that every connected component of~$\Omega$ contains at least one asymptotic end.
Given a localization manifold $(\Mbf,\Omega,\allowbreak g_0,\allowbreak\rbf,\lambdabf)$ and any exponents $q>0$ and $\expoP>1$, one has
\bel{equa-Poin-Korn}
\aligned
\| w \|_{L^2_{q,-\expoP+1}(\Omega, g_0)} 
& \lesssim \| \nabla_{g_0} w \|_{L^2_{q+1,-\expoP}(\Omega, g_0)} ,
\\
\| \nabla_{g_0} w \|_{L^2_{q,-\expoP+1}(\Omega, g_0)}
& \lesssim \bigl\| \Hess_{g_0}(w) \bigr\|_{L^2_{q+1,-\expoP}(\Omega, g_0)},
\\
\|Z\|_{L^2_{q,-\expoP+1}(\Omega, g_0)} + \|\nabla_{g_0} Z\|_{L^2_{q+1,-\expoP}(\Omega, g_0)}
& \lesssim \| \Lie_{Z^\sharp} g_0 \|_{L^2_{q+1,-\expoP}(\Omega, g_0)}
\endaligned
\ee
Here the first inequality holds for $w\in H^1_{q,-\expoP}(\Omega,g_0)$, the second for every scalar distribution satisfying $\nabla_{g_0}w\in H^1_{q,-\expoP}(\Omega,g_0)$, the third for vector fields $Z\in H^1_{q,-\expoP}(\Omega;T^*\Omega)$,
and no boundary trace is prescribed on~$\partial\Omega$.
The implied constants depend upon the exponents $q,\expoP$ and the localization data set. More generally, the same conclusion holds on a Lipschitz subdomain having the same end structure, with a constant that may additionally depend on its Lipschitz geometry.
\end{lemma}

\begin{proof}
These inequalities are proven in \refwithname{Propositions}{lem-Hardy} \refwithname{and}{lem:KornHardyD} in the special case 
$(\Omega',g_0)=(\Omega_R,\delta)$, namely a truncated cone $\Omega_R\subseteq\RR^n$ equipped with the Euclidean metric~$\delta$, without imposing any boundary condition at the radial boundary $r=R$.
In that Euclidean case, the optimal constants in~\eqref{equa-Poin-Korn} depend on the angular domain and weight $(\Lambda,\lambda)$ and the exponents $q,\expoP$ but not on~$R$.

We now explain how the Euclidean inequalities imply that the inequalities apply to each asymptotic end $(\Omega_\iota,g_0)$ \emph{outside a large enough ball}.
For brevity we omit the label~$\iota$ when possible.
We recall that $g_0-\delta$ and its derivatives decay pointwise as per~\eqref{equa-slightlybetter}
and $\wtrr=r$ and $\lambdabf=\lambda$ under the diffeomorphism from $\Omega_\iota$ to~$\Omega_R$.
The weighted $H^k$ norms involving~$g_0$ and $\delta$ are thus equivalent (for $k\leq N+3$).  However, $\dot{H}^k$ seminorms with only $k$-th derivatives are not obviously equivalent.  For the first Poincaré inequality, we note that $|\nabla_{g_0}w|_{g_0}^2=|dw|_{g_0}^2\simeq |dw|_{\delta}^2$ pointwise, so that the inequality on $(\Omega_\iota,g_0)$ is equivalent to that on $(\Omega_R,\delta)$.
For the second Poincaré--Hardy inequality in~\eqref{equa-Poin-Korn}, we denote by $C_{\Poin\Hardy 2}^{\delta}$ the optimal constant on $(\Omega_R,\delta)$.
By scale invariance the Euclidean constant $C_{\Poin\Hardy 2}^{\delta}$ does not depend on the truncation radius $R>0$.
The $o(1)$ decay of $g_0-\delta$ and $o(r^{-1})$ decay of the Christoffel symbols~$\Gamma_{g_0}$ ensures that for large enough~$R$, we have $|g_0-\delta|_\delta\leq 1/2$ and $|\Gamma_{g_0}|_\delta\leq 1/(2C_{\Poin\Hardy 2}^\delta r)$.
Hence,
\be
\aligned
4 \bigl|\Hess_{g_0}(w)\bigr|_{g_0}^2
& \geq 2 \sum_{i,j} \bigl(\del_i \del_j w - \Gamma_{g_0}{}^k{}_{ij} \del_k w\bigr)^2
\\
&
\geq |\del^2 w|_\delta^2 - 2 |\Gamma_{g_0}|_\delta^2 |\del w|_\delta^2
\geq |\del^2 w|_\delta^2 - \frac{1}{2(C_{\Poin\Hardy 2}^\delta)^2r^2} |\del w|_\delta^2 .
\endaligned
\ee
Upon integrating this last expression against $\lambda^{2\expoP}r^{2q+2-n}d^nx$, the first term controls the second term together with the desired weighted norm thanks to the Euclidean Poincaré--Hardy inequality.
Finally, for the Korn--Hardy inequality the situation is the same, with the differences $\Lie_{Z^\sharp}(g_0-\delta)$ being lower-order terms that are small for large enough~$R$.

We now return to the whole domain~$\Omega$ (or more generally a Lipschitz subdomain~$\Omega'$ of the type specified in the statement).  As we just saw, the left-hand sides of~\eqref{equa-Poin-Korn}, restricted to the outside of a large enough ball~$\Ball_R$, are controlled by the right-hand sides.  We then consider a smooth cutoff function~$\kappa$ identically equal to~$1$ in~$\Ball_R$ and supported in~$\Ball_{2R}$.  The relevant weighted $L^2$ norms of $\nabla_{g_0}(\kappa w)$, $\Hess_{g_0}(\kappa w)$, and $\Lie_{(\kappa Z)^\sharp} g_0$ are controlled by the same $L^2$ norms without~$\kappa$ because derivatives of~$\kappa$ are supported on $\Omega\setminus \Ball_R$ where lower derivatives of $w,Z$ are controlled by the already-established inequalities.
Standard weighted Poincaré--Hardy and Korn--Hardy inequalities on $\Omega\cap \Ball_{2R}$ for functions/vector fields vanishing on part of the boundary (namely the spherical caps~$\Lambda_{\iota,2R}$) then control lower-derivatives of $\kappa w,\kappa Z$ in $\Omega\cap \Ball_{2R}$, hence of $w,Z$ on $\Omega\cap \Ball_R$, which concludes.
\end{proof}


\subsection{Integral bounds for the linearized operator}
\label{appendix=E.2}

\paragraph{Coercivity of the quadratic functional for a single data set.}

We now arrive at the coercivity property that shows that $\Jrm_{(g_0,h_0)}$ is equivalent to the weighted $H^2\times H^1$ norms, namely justifies the existence of a Poincaré--Korn--Hardy constant $\CPKHzero>0$ in \autoref{def:PoincareKornHardyConst}.
The intuition is that $\Jrm_{(g_0,h_0)}[u,Z]$ vanishes only when $d\Gcal_{(g_0,h_0)}^*[u,Z]=0$, whose only solution decaying at infinity is zero.
Incidentally, if $(g_0,h_0)$ are exact solutions of the constraints, this equation defines Killing initial data sets (KIDs)~\cite{Moncrief-1975}.

\begin{lemma}[Coercivity of the squared constraints]\label{lem:squaredquadraticform}
There exists $\CPKHzero>0$ such that for $u\in H^2_{n-2-p,-\expoP}(\Omega,g_0)$ and $Z\in H^1_{n-2-p, -\expoP}(\Omega,g_0,\RR^n)$ one has
\[
\|u\|_{H^2_{n-2-p,-\expoP}(\Omega,g_0)}^2
+ \|Z\|_{H^1_{n-2-p, -\expoP}(\Omega,g_0,\RR^n)}^2
\leq (\CPKHzero)^2 \Jrm_{(g_0,h_0)}[u,Z] .
\]
Conversely, $\Jrm_{(g_0,h_0)}[u,Z]$ is controlled by the norms on the left-hand side, with a constant depending on the $C^0_2(\Omega,g_0)$ norm of $\Ric_{g_0}$ and $C^1_1(\Omega,g_0)$ norm of~$h_0$.
\end{lemma}

\begin{proof}
We begin with the straightforward control of $\Jrm_{(g_0,h_0)}$ by $H^2\times H^1$ norms.
From the expressions of the adjoint constraints in \autoref{lem:lin-constr}, we have the bounds (here, $\nabla=\nabla_{g_0}$ and $|T|=|T|_{g_0}$ for any tensor~$T$)
\bel{dHg0h0-bound}
\aligned
|d\Hcal_{(g_0,h_0)}^{*\flat\flat}[u,Z]|
& \lesssim |\nabla^2 u| + \bigl( |\Ric_{g_0}| + |h_0|^2 \bigr) |u| + |h_0| |\nabla Z| + |\nabla h_0| |Z| ,
\\
|d\Mcal_{(g_0,h_0)}^{*\sharp\sharp}[u,Z]|
& \lesssim |\nabla Z| + |h_0| |u| ,
\endaligned
\ee
with implicit constants depending on the dimension, only.  We then conclude by accounting for the weights $\omegabf_p^2$ and~$\omegabf_{p+1}^2$ and for the pointwise $\rbf^{-1}$ decay of~$h_0$ and $\rbf^{-2}$~decay of $\Ric_{g_0}$ and $\nabla h_0$.

The strategy to prove coercivity is essentially the same as for \autoref{lem:PoincareKornHardyD}.
In each asymptotic end the functional $\Jrm$ associated to the Euclidean data set $(\delta,0)$ is coercive.  Indeed,
\be
\aligned
\quad & \unquad \int_{\Omega_R} \Bigl( \omega_p^2 \, \bigl|d\Hcal_{(\delta,0)}^*[u,Z]\bigr|_{\delta}^2
+ \omega_{p+1}^2 \bigl|d\Mcal_{(\delta,0)}^*[u,Z]\bigr|_{\delta}^2 \Bigr) d^nx
\\
& = \|\del^2 u\|_{L^2_{n-p,-\expoP}(\Omega_R,\delta)}^2 + (n-2) \|\Delta u\|_{L^2_{n-p,-\expoP}(\Omega_R,\delta)}^2
+ \|\Sym(\del Z)\|_{L^2_{n-1-p,-\expoP}(\Omega_R,\delta)}^2 ,
\endaligned
\ee
which controls the norms $\|u\|_{H^2_{n-2-p,-\expoP}(\Omega_R,\delta)}^2 + \|Z\|_{H^1_{n-2-p,-\expoP}(\Omega_R,\delta)}^2$, thanks to the (Euclidean) Poincaré--Hardy and Korn--Hardy inequalities proved in \refwithname{Propositions}{lem-Hardy} \refwithname{and}{lem:KornHardyD}.
Outside a large enough ball~$\Ball_R$, the difference $d\Gcal_{(g_0,h_0)}^*[u,Z] - d\Gcal_{(\delta,0)}^*[u,Z]$ is small compared to $d\Gcal_{(\delta,0)}^*[u,Z]$ in the relevant weighted $L^2$~norms, thus the restriction of $\Jrm_{(g_0,h_0)}[u,Z]$ to~$\Omega\setminus \Ball_R$ controls the weighted $H^2\times H^1$ norm of $(u,Z)$ on that domain.

Assume by contradiction that coercivity fails. There then exists a sequence $(u_s,Z_s)_{s=1,2,\ldots}$ such that
\bel{usZs-normalization}
\|u_s\|_{H^2_{n-2-p,-\expoP}(\Omega,g_0)}^2
+ \|Z_s\|_{H^1_{n-2-p, -\expoP}(\Omega,g_0,\RR^n)}^2
= 1 , \qquad
\lim_{s\to+\infty} \Jrm_{(g_0,h_0)}[u_s,Z_s]  = 0 .
\ee
Since $\Jrm_{(g_0,h_0)}[u_s,Z_s]$ controls its restriction to $\Omega\setminus\Ball_R$, the exterior estimate just proven shows that the weighted $H^2\times H^1$ norm of $(u_s,Z_s)$ on that domain tends to zero as $s\to+\infty$.
Fix a smooth radial cutoff function $\kappa$ supported in $\Ball_{2R}\subset\Mbf$ and equal to~$1$ on $\Ball_R$.
Replacing $(u_s,Z_s)$ by $(\kappa u_s,\kappa Z_s)$ shifts the $H^2\times H^1$ norm and the functional $\Jrm_{(g_0,h_0)}$ by terms that tend to zero as $s\to+\infty$, hence we can assume (after normalization) that \eqref{usZs-normalization} holds, with all $u_s,Z_s$ supported in~$\Ball_{2R}$.

By compact embedding of $H^2\Subset H^1\Subset L^2$ in the subdomains $\Delta_{R,\varepsilon}\subset\Omega$ defined by $\lambdabf>\varepsilon$ and $R<\rbf<2R$, and by a diagonal argument, one obtains a subsequence $(u_s,Z_s)$ that converges to a limit $(u_\infty,Z_\infty)$ weakly in $H^2\times H^1$ and strongly in $H^1\times L^2$ restricted to~$\Delta_{R,\varepsilon}$ for any $\varepsilon>0$.  On the other hand, in a collar where $\lambdabf<2\varepsilon$, the angular improvement in~\eqref{equa-Poin-Korn} gives
\be
\|u_s,Z_s\|_{H^1_{n-2-p,-\expoP}\times L^2_{n-2-p,-\expoP}(\{\lambdabf<2\varepsilon\})}^2
\lesssim
\varepsilon^2
\|u_s,Z_s\|_{H^2_{n-2-p,-\expoP}\times H^1_{n-2-p,-\expoP}(\Omega)}^2
= \varepsilon^2 ,
\ee
uniformly in~$s$, hence also for $s=+\infty$.  Thus, $(u_s,Z_s)\to (u_\infty,Z_\infty)$ strongly in weighted $H^1\times L^2$ norms on~$\Omega\cap\Ball_{2R}\setminus\Ball_R$.
By \eqref{usZs-normalization}, one has $\Jrm_{(g_0,h_0)}[u_\infty,Z_\infty]=0$ hence
\be
d\Gcal_{(g_0,h_0)}^*[u_\infty,Z_\infty] = 0 .
\ee

Solutions $(u_\infty,Z_\infty)$ of this (overdetermined) system supported in the bounded region~$\Ball_{2R}$ vanish, as we now show for completeness.
Introduce $A_{\infty\,ij}\coloneqq\nabla_i Z_{\infty\,j}-\nabla_j Z_{\infty\,i}$.
The variables $\mathcal{Y}\coloneqq(u_\infty,\nabla u_\infty,Z_\infty,A_\infty)$ are such that $\nabla\mathcal{Y}$ can be expressed linearly in terms of~$\mathcal{Y}$ using the constraints: indeed, $\nabla Z_\infty = A_\infty/2 + \bigl(\frac{2}{n-1}\Tr_{g_0}(h_0)g_0 - 2 h_0^{\flat\flat}\bigr)u_\infty$ is given by the adjoint momentum constraint, $\nabla^2 u_\infty$ by the adjoint Hamiltonian constraint, and lastly
\be
\nabla_k A_{\infty\,ij}
= 2 \nabla_i \nabla_{(k} Z_{\infty\,j)}
- 2 \nabla_j \nabla_{(k} Z_{\infty\,i)}
+ R^l{}_{jki} Z_{\infty\,l}
- R^l{}_{ikj} Z_{\infty\,l}
- R^l{}_{kij} Z_{\infty\,k} ,
\ee
where the symmetrized derivatives of~$Z_\infty$ are to be replaced by their expression given by the momentum constraint.  It follows that $\mathcal{Y}$ satisfies a closed first-order linear system along every curve.  Since all connected components of~$\Omega$ include asymptotic ends in which $(u_\infty,Z_\infty)$, and thus $\mathcal{Y}$, vanish, the solution~$\mathcal{Y}$ vanishes identically in~$\Omega$.

It follows that $u_s\to0$ in weighted $H^1$ and $Z_s\to0$ in weighted $L^2$. We now use these properties in $d\Gcal^*_{(g_0,h_0)}[u_s,Z_s]\to 0$.  The adjoint momentum constraint simplifies to $\Lie_{Z_s^\sharp}g_0 \to 0$, from which the Korn--Hardy inequality implies $Z_s\to 0$ in weighted $H^1$~norm.  Then the (trace-flipped) adjoint Hamiltonian constraint simplifies to $\nabla^2 u_s\to 0$, hence $u_s\to 0$ in weighted $H^2$~norm, contradicting~\eqref{usZs-normalization}. This proves coercivity.
\end{proof}


\paragraph{Coercivity for a pair of data sets.}

The elliptic problem can be reformulated as a variational problem based on the bilinear form on
$H^2_{n-2-p,-\expoP}(\Omega,g_0)\times H^1_{n-2-p,-\expoP}(\Omega,g_0)$ defined by
\bel{Jrm-def}
\aligned
& \Jrm_{(g_k,h_k;g_0,h_0)}[u,Z;v,Y]
\coloneqq \int_{\Omega} (v,Y) \cdot \Jcal_{(g_k,h_k;g_0,h_0)}[u,Z] \, \dVol_{g_k}
\\
& \quad = \int_{\Omega} \Bigl(
\omegabf_p^2 \, g_k\bigl(d\Hcal_{(g_0,h_0)}^{*\flat\flat}[u,Z] , d\Hcal_{(g_k,h_k)}^{*\flat\flat}[v,Y]\bigr) \\
& \qquad\qquad + \omegabf_{p+1}^2 g_k\bigl(d\Mcal_{(g_0,h_0)}^{*\sharp\sharp}[u,Z] , d\Mcal_{(g_k,h_k)}^{*\sharp\sharp}[v,Y]\bigr)
\Bigr) \, \dVol_{g_k} .
\endaligned
\ee
The metric $g_k$ must be used throughout because, by definition, $(d\Hcal_{(g_k,h_k)}^{*\flat\flat},d\Mcal_{(g_k,h_k)}^{*\sharp\sharp})$ denotes the adjoint with respect to that metric.
Note that $\Jrm_{(g_0,h_0;g_0,h_0)}[u,Z;u,Z] = \Jrm_{(g_0,h_0)}[u,Z]$ is the quadratic form appearing in \autoref{lem:squaredquadraticform}.

\begin{lemma}[Invertibility of the squared constraints]\label{lem:appE-invertibility}
  There exists a constant $\etainv=\etainv(g_0,h_0)\in(0,1/2]$ depending non-increasingly on~$\CPKHzero$, $|\Ric_{g_0}|_{C^0_2(\Omega,g_0)}$, and $|h_0|_{C^1_1(\Omega,g_0)}$ such that the following properties hold if
  \be
  \|g_k-g_0\|_{C^2_0(\Omega,g_0)}
  + \|h_k-h_0\|_{C^1_1(\Omega,g_0)}
  \leq \etainv .
  \ee

  1. The quadratic form associated to~\eqref{Jrm-def} is continuous and obeys the coercivity inequality
  \[
    \|u\|_{H^2_{n-2-p,-\expoP}(\Omega,g_0)}^2
    + \|Z\|_{H^1_{n-2-p, -\expoP}(\Omega,g_0,\RR^n)}^2
    \leq 2 (\CPKHzero)^2 \Jrm_{(g_k,h_k;g_0,h_0)}[u,Z;u,Z] .
  \]

  2. For any $f\in H^{2*}_{2+p, \expoP}(\Omega,g_0)$, $V\in H^{1*}_{2+p, \expoP}(\Omega,g_0,\RR^n)$, there exists a unique solution $u\in H^2_{n-2-p, -\expoP}(\Omega,g_0)$, $Z\in H^1_{n-2-p, -\expoP}(\Omega,g_0,\RR^n)$ to the equation $\Jcal_{(g_k,h_k;g_0,h_0)}[u,Z] = (f, V)$ understood in the weak sense
  \[
    \Jrm_{(g_k,h_k;g_0,h_0)}[u,Z;v,Y] = \int_{\Omega} (v f + Y\cdot V) \, \dVol_{g_k} , \qquad
    \aligned
    v & \in H^2_{n-2-p, -\expoP}(\Omega,g_0) , \\
    Y & \in H^1_{n-2-p, -\expoP}(\Omega,g_0,\RR^n) .
    \endaligned
  \]
  Furthermore,
  \[
    \aligned
    & \|u\|_{H^2_{n-2-p, -\expoP}(\Omega,g_0)}+\|Z\|_{H^1_{n-2-p, -\expoP}(\Omega,g_0,\RR^n)}
    \\
    & \quad \leq 4 (\CPKHzero)^2 \bigl( \|f\|_{H^{2*}_{2+p, \expoP}(\Omega,g_0)}+\|V\|_{H^{1*}_{2+p, \expoP}(\Omega,g_0,\RR^n)} \bigr) .
    \endaligned
  \]
\end{lemma}

\begin{remark}
  The condition $\etainv\leq 1/2$, hence $|g_k-g_0|\leq 1/2$ everywhere, ensures the uniform equivalence of $g_k$ and~$g_0$.  As a result, norms of $g_k^{-1}$ are controlled by those of~$g_k$.
\end{remark}

\begin{remark}
The integral of $v f \dVol_{g_k}$ (and likewise $Y\cdot V \, \dVol_{g_k}$) should be understood as the duality pairing of~$f$ with the function $v\,\dVol_{g_k}/\dVol_{g_0}$, which is in the same weighted $H^2$~space as~$v$ since $g_k-g_0$ and its derivatives are bounded uniformly near the boundary.
\end{remark}

\begin{proof}
\bse
  By the Lax--Milgram theorem, the first property (continuity and coercivity of~$\Jrm$) implies the existence of $u,Z$ and the Sobolev bounds.
  Both continuity and coercivity hold for $(g_k,h_k)=(g_0,h_0)$ by \autoref{lem:squaredquadraticform}.
  Continuity of $\Jrm_{(g_k,h_k;g_0,h_0)}$ is proven in the same way, by writing the bounds~\eqref{dHg0h0-bound} for $(g_k,h_k)$ instead of $(g_0,h_0)$.
  We turn to coercivity of the quadratic form.
  We denote by $\dVol$ the volume form of~$g_0$, $\nabla=\nabla_{g_0}$, and $|T|=|T|_{g_0}$ for any tensor~$T$, for brevity, and omit the domain~$\Omega$ and metric~$g_0$ from all norms in this proof.
  We bound the difference of quadratic forms by
  \be
  \aligned
  \quad & \unquad \bigl| \Jrm_{(g_k,h_k;g_0,h_0)}[u,Z;u,Z] - \Jrm_{(g_0,h_0)}[u,Z] \bigr|
  \\
  & \lesssim \!
  \int_{\Omega} \Bigl( \aleph_1 \bigl|d\Hcal_{(g_0,h_0)}^{*\flat\flat}[u,Z]\bigr|\, \omegabf_p^2
  + \aleph_2 \bigl|d\Mcal_{(g_0,h_0)}^{*\sharp\sharp}[u,Z]\bigr| \, \omegabf_{p+1}^2 \Bigr) \dVol ,
  \\
  \aleph_1 \, & \! \coloneqq |g_k-g_0| \bigl|d\Hcal_{(g_0,h_0)}^{*\flat\flat}[u,Z]\bigr|
  + \bigl|d\Hcal_{(g_k,h_k)}^{*\flat\flat}[u,Z]-d\Hcal_{(g_0,h_0)}^{*\flat\flat}[u,Z]\bigr| ,
  \\
  \,\aleph_2 \, & \! \coloneqq |g_k-g_0| \bigl|d\Mcal_{(g_0,h_0)}^{*\sharp\sharp}[u,Z]\bigr|
  + \bigl|d\Mcal_{(g_k,h_k)}^{*\sharp\sharp}[u,Z] - d\Mcal_{(g_0,h_0)}^{*\sharp\sharp}[u,Z]\bigr| .
  \endaligned
  \ee
  The expressions of the adjoint constraints in \autoref{lem:lin-constr} then yield the bounds
  \bel{aleph12}
  \aligned
  \aleph_1
  & \lesssim |g_k-g_0| \Bigl( |\nabla^2 u| + |\Ric_{g_0}| |u| + |h_0|^2 |u| + |\nabla h_0| |Z| + |h_0| |\nabla Z| \Bigr) 
  \\
  & \quad + |\nabla g_k| \bigl( |\nabla u| + |h_0| |Z| \bigr) + |\nabla g_k|^2 |u| + |\nabla^2 g_k| |u|
  \\
  & \quad + |h_k-h_0| \bigl( |h_0| |u| + |\nabla Z| \bigr)
  + |h_k-h_0|^2|u|
  + |\nabla(h_k-h_0)| |Z| ,
  \\
  \aleph_2
  & \lesssim |g_k-g_0| \Bigl( |h_0| |u| + |\nabla Z| \Bigr)
  + |\nabla g_k| |Z| + |h_k-h_0| |u| ,
  \endaligned
  \ee
  with implicit constants that depend on the dimension and the $C^0$ norms of~$g_k$ and~$g_k^{-1}$.
  Under the assumption $|g_k-g_0|\leq 1/2$, both of these norms are bounded by functions of~$n$ so that the implicit constants in~\eqref{aleph12} are functions of the dimension, only.
  After rewriting $\nabla g_k=\nabla(g_k-g_0)$, we bound in each term one factor by the $C^2_0$ norm of $g_k-g_0$ or $C^1_1$ norm of $h_k-h_0$.
  Under the mild assumption $\|g_k-g_0\|_{C^2_0}+\|h_k-h_0\|_{C^1_1}\leq 1/2$ (so as to prevent quadratic terms from cluttering subsequent expressions), we obtain
  \bel{Jrmgkhk-closeto-g0h0}
  \aligned
  \quad & \unquad \bigl| \Jrm_{(g_k,h_k;g_0,h_0)}[u,Z;u,Z] - \Jrm_{(g_0,h_0)}[u,Z] \bigr|
  \\
  & \leq C_{\Jrm}(g_0,h_0)
  \Bigl( \|g_k-g_0\|_{C^2_0} + \|h_k-h_0\|_{C^1_1} \Bigr)
  \Bigl( \|u\|_{H^2_{n-2-p, -\expoP}}^2 + \|Z\|_{H^1_{n-2-p, -\expoP}}^2 \Bigr). 
  \endaligned
  \ee
  where the constant~$C_{\Jrm}$ depends on the $C^0_2$ norm of $\Ric_{g_0}$ and $C^1_1$ norm of~$h_0$.
  Under the condition
  \bel{gkbounde13}
  \|g_k-g_0\|_{C^2_0} + \|h_k-h_0\|_{C^1_1}
  \leq \frac{1}{2} \min\bigl(1, (\CPKHzero)^{-2} \bigl(C_{\Jrm}(g_0,h_0)\bigr)^{-1} \bigr) ,
  \ee
  the right-hand side of~\eqref{Jrmgkhk-closeto-g0h0} is bounded by $\frac{1}{2}\Jrm_{(g_0,h_0)}[u,Z]$ hence $\Jrm_{(g_k,h_k;g_0,h_0)}[u,Z;u,Z]$ has the desired coercivity.
  \ese
\end{proof}


\subsection{Pointwise bounds for the linearized operator}
\label{appendix=E.3}

\paragraph{From Sobolev to H\"older bounds.}

We now improve the control of $u,Z$ by using ellipticity of the system in the sense of Douglis and Nirenberg~\cite{DouglisNirenberg}.  As pointed out in \autoref{rem:DN}, the interior estimates in~\cite{DouglisNirenberg} are insufficient and we rely on Morrey's book~\cite{Morrey-1966}.  Besides extending the exponent range compared to~\cite[Lemma 5.7]{CarlottoSchoen} we also allow metrics that are far from Euclidean away from the asymptotic ends.

\begin{lemma}[Linearized pointwise estimates]\label{lem:appE-linearized-pointwise}
In the set-up of \autoref{lem:appE-invertibility} (namely $g_k-g_0$ and $h_k-h_0$ small in low-order norms), and with $\expoPp-\expoP=\expoP-\expoPm\geq (n+4)/2$, one has
\[
\aligned
\quad & \unquad \|u\|^{N+2,\alpha}_{\Omega, g_0, n-2-p, -\expoPp+2}
+ \|Z\|^{N+1,\alpha}_{\Omega, g_0, n-2-p, -\expoPp+1}
\\
& \lesssim \|u,Z\|_{L^2_{n-2-p, -\expoP}(\Omega,g_0)}
+ \|f\|_{\Omega,g_0, 2+p,\expoPm-2}^{N-2,\alpha} + \|V\|_{\Omega,g_0, 2+p,\expoPm-1}^{N-1,\alpha}
\endaligned
\]
provided the weighted H\"older norms of the sources $f,V$ on the right-hand side are finite.
The implicit constant depends on the localization data set and on $\|g_k\|_{C^{N,\alpha}_0(\Omega,g_0)}$ and $\|h_k\|_{C^{N,\alpha}_1(\Omega,g_0)}$.
\end{lemma}

\begin{proof}
\bse
Since we detail all the steps for the Euclidean operators in asymptotic ends $(\Omega_R,\delta)$ when proving \autoref{thm-sharp-h-localized}, we shall only outline the proof here, highlighting the complications due to the coupled nature of the curved-space operators.
First, the equation $\Jcal_{(g_k,h_k;g_0,h_0)}[u,Z] = (f, V)$ is rewritten as
\be
\Lcal[\omegabf_p u, \omegabf_{p+1} Z] = (\omegabf_p^{-1} f, \omegabf_{p+1}^{-1} V) ,
\ee
where $\Lcal$ is a simple reweighting of~$\Jcal_{(g_k,h_k;g_0,h_0)}$ by $\omegabf_p=\lambdabf^\expoP\rbf^{n/2-p}$ and~$\omegabf_{p+1}$.
Let us see how this equation fits in the framework of elliptic equations of mixed orders analyzed by Douglis and Nirenberg~\cite{DouglisNirenberg}.  This formalism involves assigning integer weights (differential orders) $s_\Hcal=0$ and $s_\Mcal=-1$ to the equations and $t_u=4$, $t_Z=3$ to the variables, such that at most $s_i+t_j$ derivatives of the $j$-th variable appear in the $i$-th equation.  The characteristic matrix $L(x,\xi)$ is obtained by keeping only these $s_i+t_j$ derivatives and changing derivatives to a mode variable~$i\xi$.  Since $\Lcal$ is a composition of operators, one finds a product of $(1+n)\times(n^2+n^2)$ and $(n^2+n^2)\times(1+n)$ block matrices associated to $d\Gcal$ and $d\Gcal^*$ (with indices properly raised/lowered),
\be
L(x,\xi) = \diag(\omegabf_p^{-1},\omegabf_{p+1}^{-1}) \charac(d\Gcal_{(g_k,h_k)}) \diag(\omegabf_p^2, \omegabf_{p+1}^2) \charac(d\Gcal^*_{(g_0,h_0)}) \diag(\omegabf_p^{-1},\omegabf_{p+1}^{-1}) ,
\ee
with
\be
\aligned
\charac(d\Gcal_{(g,h)})\begin{pmatrix}\gdiff\\\hdiff\end{pmatrix}
& = \begin{pmatrix} - \bigl( g^{ik} g^{jl} \xi_i \xi_j - g(\xi,\xi) g^{kl} \bigr) \gdiff_{kl} \\
\bigl( h^{jk} g^{\bullet l} - \frac{1}{2} g^{\bullet j} h^{kl} + \frac{1}{2} h^{\bullet j} g^{kl} \bigr) i \xi_j \gdiff_{kl} + i \xi_j \hdiff^{j\bullet}
\end{pmatrix} ,
\\
\charac(d\Gcal^*_{(g,h)})\begin{pmatrix}u\\Z\end{pmatrix}
& = \begin{pmatrix}
- (\xi_\bullet \xi_\bullet - g(\xi,\xi) g_{\bullet\bullet}) u + \frac{1}{2} (h_{\bullet\bullet} \delta^k_l - 2 h^k_{(\bullet} g^{}_{\bullet)l} - g_{\bullet\bullet} h_l^k) i \xi_k Z^l
\\
- i \xi_k g^{k(\bullet}Z^{\bullet)}
\end{pmatrix} .
\endaligned
\ee
To show that the matrix $L(x,\xi)$ is Douglis--Nirenberg elliptic, we use the same perturbative comparison as for $\Jrm_{(g_k,h_k,g_0,h_0)}$ in the proof of \autoref{lem:appE-invertibility}.
By smallness of $g_k-g_0$ and $h_k-h_0$ it is enough to consider the case $(g_k,h_k)=(g_0,h_0)$. In that case, the factors in $L(x,\xi)$ are adjoint of each other, hence $L(x,\xi)$ is positive semi-definite.
For any $(u,Z)$ we find
\bel{uZLxxiuZ}
(\omegabf_p u,\omegabf_{p+1}Z)\cdot L(x,\xi)\bigl(\begin{smallmatrix}\omegabf_p u\\\omegabf_{p+1}Z\end{smallmatrix}\bigr)
= \omegabf_p^2 \bigl|\charac(d\Hcal^*_{(g_0,h_0)})\bigl(\begin{smallmatrix}u\\Z\end{smallmatrix}\bigr)\bigr|_{g_0}^2
+ \omegabf_{p+1}^2 \bigl|\charac(d\Mcal^*_{(g_0,h_0)})\bigl(\begin{smallmatrix}u\\Z\end{smallmatrix}\bigr)\bigr|_{g_0}^2 .
\ee
The momentum constraint contribution is (raising and lowering indices with~$g_0$)
\be
\bigl|\charac(d\Mcal^*_{(g_0,h_0)})\bigl(\begin{smallmatrix}u\\Z\end{smallmatrix}\bigr)\bigr|_{g_0}^2
= \frac{1}{4} (\xi^i Z^j + \xi^j Z^i) (\xi_i Z_j + \xi_j Z_i)
= \frac{1}{2} |\xi|_{g_0}^2 |Z|_{g_0}^2 + \frac{1}{2} (\xi\cdot Z)^2 ,
\ee
which controls all products~$\xi_i Z_j$, hence \eqref{uZLxxiuZ}~controls $\omegabf_{p+1}\xi_i Z_j$.
Because $h_0$ is uniformly bounded by $\rbf^{-1}$ on~$\Omega$, its contribution to the $\charac(d\Hcal^*_{(g_0,h_0)})$ term is controlled by such products, so that \eqref{uZLxxiuZ}~controls $\omegabf_p(\xi_\bullet \xi_\bullet - g(\xi,\xi) g_{\bullet\bullet}) u$ hence all~$\omegabf_p\xi_i\xi_j u$.
Altogether, we obtain the desired ellipticity
\be
\Re\bigl\langle L(x,\xi)(U,W),(U,W)\bigr\rangle
\gtrsim |\xi|_{g_0}^4|\omegabf_p u|^2+|\xi|_{g_0}^2|\omegabf_{p+1}Z|_{g_0}^2 ,
\ee
where the real part $\Re$ is needed to accomodate the non-self-adjoint case $(g_k,h_k)\neq(g_0,h_0)$.
This is the required uniform Douglis--Nirenberg ellipticity.
Note that the work of Carlotto and Schoen avoided these complications by selecting a set-up where $(g_0,h_0)$ was close to a Euclidean metric $(\delta,0)$ on the whole gluing domain.

We can then seek interior elliptic estimates.  Recall the covering of $\Omega$ by the bounded region $\Omega_0$ and the asymptotic ends~$\Omega_\iota$.  The latter are equipped with global coordinate charts by definition, and $\Omega_0$ can be covered by finitely many coordinate charts.  The following discussion is done in one coordinate chart.  All norms below arise first as coordinate norms involving partial derivatives, and we replace them for convenience by the equivalent norms involving covariant derivatives and contraction with the metric~$g_0$.

For all $x$, consider a ball $\Ball_{d(x)}(x)$ of radius $d(x) = c_1 \lambdabf(x) \rbf(x)$ for $c_1\in (0,1)$ small enough (uniform in~$x$) that $\Ball_{2d(x)}$ is contained in the coordinate chart under consideration in~$\Omega$.
In particular, $\dbf(\cdot,{}^\complement\Omega),\rbf,\lambdabf$ on the ball $\Ball_{d(x)}(x)$ are bounded above and below by $d(x),\rbf(x),\lambdabf(x)$ up to multiplicative constants that are uniform in~$x$.
Beyond the leading terms captured by the characteristic matrix, the coefficients of subleading terms in  $\Jcal_{(g_k,h_k;g_0,h_0)}[u,Z]$ are bounded by suitable negative powers of the distance $\dbf(\cdot,{}^\complement\Omega) \simeq d(x)$, such that \cite[Theorem 6.2.6]{Morrey-1966} applies, and bounds H\"older norms of $\omegabf_p u,\omegabf_{p+1}Z$ as
\be
\aligned
K_{1/2}^{N+2}(\omegabf_p u) + d(x) K_{1/2}^{N+1}(\omegabf_{p+1} Z)
& \lesssim d(x)^{-n} \bigl\|\omegabf_p u,d(x)\omegabf_{p+1}Z\bigr\|_{L^1(\Ball_{d(x)/2}(x),g_0)} \\
& \quad + d(x)^4 K_{3/4}^{N-2}(\omegabf_p^{-1} f) + d(x)^3 K_{3/4}^{N-1}(\omegabf_{p+1}^{-1} V) ,
\\
K^M_\tau(X) \, & \! \coloneqq \sum_{i= 0}^M d(x)^i \sup_{\Ball_{\tau d(x)}(x)} |\del^i X|_{g_0} + d(x)^{M+ \alpha} \big[ \del^M X \big]_{\alpha, \Ball_{\tau d(x)}(x)} .
\endaligned
\ee
The constants depend on the localization data set, on $|g_k|_{C^{N,\alpha}_0(\Omega,g_0)}$, $|h_k|_{C^{N,\alpha}_1(\Omega,g_0)}$ and on $c_1$ (which can simply be fixed once and for all in terms of the geometry), and are uniform in~$x$.
Multiplying by $\rbf(x)^{n/2-2}\lambdabf(x)^{\expoPp-\expoP-2}$, accounting for $d(x)\sim\rbf(x)\lambdabf(x)$ and taking the sup over~$x$ gives
\be
\aligned
\quad & \unquad \|u\|^{N+2,\alpha}_{\Omega, g_0, n-2-p, -\expoPp+2}
+ \|Z\|^{N+1,\alpha}_{\Omega, g_0, n-2-p, -\expoPp+1}
\\
& \lesssim \sup_{x\in\Omega} \bigl( \rbf(x)^{-p-2}\lambdabf(x)^{\expoPp-2-n} \bigl\| u,\lambdabf Z\bigr\|_{L^1(\Ball_{d(x)/2}(x),g_0)} \bigr) \\
& \quad
+ \|f\|_{\Omega,g_0, 2+p,\expoPm-2}^{N-2,\alpha}
+ \|V\|_{\Omega,g_0, 2+p,\expoPm-1}^{N-1,\alpha} .
\endaligned
\ee

There remains to bound the $L^1$ norms on $\Ball_{d(x)/2}$ by weighted $L^2_{n-2-p,-\expoP}$ norms appearing in the Lemma's statement.
This is a consequence of the Cauchy--Schwarz inequality on $u,Z$ and the characteristic function of the ball with suitable powers of $\rbf$ and $\lambdabf$:
\be
\aligned
\|u\|_{L^1(\Ball_{d(x)/2}(x),g_0)}^2
& \leq \int_{\Ball_{d(x)/2}(x)} u^2 \rbf^{n-4-2p} \lambdabf^{2\expoP} dx'
\int_{\Ball_{d(x)/2}(x)} \rbf^{-n+4+2p} \lambdabf^{-2\expoP} dx'
\\
& \lesssim \rbf(x)^{4+2p} \lambdabf(x)^{n-2\expoP} \|u\|^2_{L^2_{n-2-p,-\expoP}(\Omega,g_0)}
\endaligned
\ee
and likewise for~$\lambdabf Z$.  This is enough to conclude, provided $\expoPp-2-n+(n-2\expoP)/2\geq 0$, namely $\expoPp\geq\expoP+(n+4)/2$.  Observe that this inequality is slightly stronger than the one needed in \autoref{section=5.4}.  This is due to the fact that in \autoref{section=5.4}, the $\unL^2$~norms on spherical shells were under control, while here one has $n$-dimensional $L^2$~norms.
\ese
\end{proof}


\subsection{Convergence of the Newton iteration.}
\label{appendix=E.4}

\paragraph{Nonlinearities and convergence of the Newton iteration.}

After bounding the nonlinear effect of perturbing a data set $(g_k,h_k)$ by $(\gdiff,\hdiff)$, we will explain why the sequence of data sets produced by the Newton method converges in Sobolev and H\"older norms.
The following pointwise bound is obtained immediately from the expression of the nonlinear terms $\Qcal\Hcal$ and $\Qcal\Mcal$ in \autoref{lem:constr-expand} and needs no proof.  It is very similar to the bound~\eqref{T-quadra} in an asymptotic end.  In contrast to \autoref{lem:step2-nonlinear}, the integral bound is seen here as a consequence of the pointwise bound rather than being derived directly from a product of pointwise and integral bounds of $(\gdiff,\hdiff)$.  This avoids the need of a negative Sobolev space on $(\Omega,g_0)$ in this appendix, at the cost of an otherwise unnecessary assumption $2\expoPm-2>\expoP$ on angular exponents.

\begin{lemma}[Control of nonlinearities]\label{lem:appE-control-nonlin}
In the set-up of \autoref{lem:appE-invertibility} (namely $g_k-g_0$ and $h_k-h_0$ small in low-order norms), for any symmetric two-tensors $\gdiff,\hdiff$ on~$\Omega$ such that
\bel{gdiffNalpha-unif}
\|\gdiff\|_{C^{N,\alpha}_0(\Omega,g_0)}
\leq 1/2,
\ee
and hence $g_k+\gdiff$ is uniformly Riemannian, one has the pointwise bounds
\be
\bigl\| \Qcal\Hcal_{(g_k,h_k)}[\gdiff,\hdiff] \bigr\|^{N-2,\alpha}_{\Omega,2p+2,2\expoPm-2}
+ \bigl\| \Qcal\Mcal_{(g_k,h_k)}[\gdiff,\hdiff] \bigr\|^{N-1,\alpha}_{\Omega,2p+2,2\expoPm-1}
\lesssim \bigl( \|\gdiff\|_{\Omega,p,\expoPm}^{N,\alpha} + \|\hdiff\|_{\Omega,p+1,\expoPm}^{N,\alpha} \bigr)^2 ,
\ee
with implicit constants depending on the norms $\|g_k\|_{C^{N,\alpha}(\Omega,g_0)}$ and $\|h_k\|_{C^{N,\alpha}_1(\Omega,g_0)}$.
Provided the angular exponents obey $2\expoPm-2>\expoP$, one deduces the integral bounds
\be
\bigl\| \Qcal\Hcal_{(g_k,h_k)}[\gdiff,\hdiff] \bigr\|_{L^2_{p+2,\expoP}(\Omega)}
+ \bigl\| \Qcal\Mcal_{(g_k,h_k)}[\gdiff,\hdiff] \bigr\|_{L^2_{p+2,\expoP}(\Omega)}
\lesssim \bigl(\|\gdiff\|^{N,\alpha}_{\Omega,p,\expoPm} + \|\hdiff\|^{N,\alpha}_{\Omega,p+1,\expoPm}\bigr)^2 .
\ee
\end{lemma}

\begin{remark}
  The smallness condition~\eqref{gdiffNalpha-unif} on~$\gdiff$ in unweighted H\"older norm allows cubic and higher-order terms in $\Qcal\Gcal_{(g_k,h_k)}[\gdiff,\hdiff]$ to be controlled by squared norms of~$(\gdiff,\hdiff)$.
  The H\"older norms appearing in \autoref{lem:appE-control-nonlin} may equally well be defined with the metric $g_0$ or~$g_k$, as they are equivalent up to constants that depend on $\|g_k\|_{C^{N-1,\alpha}(\Omega,g_0)}$.
  For $k=1$, $(g_1,h_1)=(\seedg,\seedh)$ and the norms $\|\seedg\|_{C^{N,\alpha}}$ and $\|\seedh\|_{C^{N,\alpha}_1}$ are controlled by the norms in the definition of localized seed data set in \autoref{def-aset} and conical localization data set in \autoref{def-conical}.
\end{remark}

We are ready to prove (a refined version of) \autoref{thm:sts-existence}.
We recall that
\bel{Errp-recall}
\Err_p[\seedg,\seedh] = \Norm{\Hcal(\seedg,\seedh)}^{N-2,\alpha}_{\Omega, g_0, p+2, \expoPm-2, \expoP} + \Norm{\Mcal(\seedg,\seedh)}^{N-1,\alpha}_{\Omega, g_0, p+2, \expoPm-1, \expoP}
\ee
where $\Norm{\cdot}_{\Omega,p,a,b}^{l,\alpha}$ is the sum~\eqref{equa-def-weightLC} of a weighted $L^2$ norm and a H\"older norm.

\begin{theorem}[Existence of the seed-to-solution map]\label{thm:sts-Sobolev}
Consider the setup of \autoref{thm:sts-existence}, with a conical localization data set $(\Mbf, \Omega, g_0,h_0, \wtrr, \lambdabf)$, admissible exponents $(p, p_G, p_A)$ and $(\expoPm, \expoP, \expoPp)$, and $\eps_G,\eps_A>0$ sufficiently small.
If a seed data set $(\seedg, \seedh)$ satisfies
\bel{sts-premise}
\aligned
\| \seedg - g_0 \|^{N, \alpha}_{\Mbf, p_G} + \| \seedh - h_0 \|^{N,\alpha}_{\Mbf, p_G+1} & \leq \eps_G ,
\\
\Norm{\Hcal(\seedg,\seedh)}^{N-2,\alpha}_{\Omega,  p_A+2, \expoPm-2,\expoP} + \Norm{\Mcal(\seedg, \seedh)}^{N-1, \alpha}_{\Omega, p_A+2, \expoPm-1,\expoP} & \leq \eps_A ,
\endaligned
\ee
then there exists a solution $(u,Z)$ to
\be
\aligned 
\Hcal(g,h) & = 0,
& \qquad g & = \seedg + \omegabf_p^2 \, d\Hcal^{*\flat\flat}_{(g_0,h_0)}(u,Z), 
\\
\Mcal(g,h) & = 0,
& \qquad h & = \seedh + \omegabf_{p+1}^2 \, d\Mcal^{*\sharp\sharp}_{(g_0,h_0)}(u,Z),
\endaligned
\ee
satisfying the estimates (with implicit constants depending on the localization data set and the exponents)
\be
\aligned
\|u\|^{N+2, \alpha}_{\Omega, n-2-p, -\expoPp+2}
+ \|u\|_{H^2_{n-2-p, -\expoP}(\Omega)}
+ \|Z\|^{N+1, \alpha}_{\Omega, n-2-p, -\expoPp+1}
+ \|Z\|_{H^1_{n-2-p, -\expoP}(\Omega)}
& \lesssim \Err_p[\seedg,\seedh],
\\
\Norm{g - \seedg}^{N, \alpha}_{\Omega, p,\expoPm, \expoP}
+ \Norm{h - \seedh}^{N,\alpha}_{\Omega, p+1, \expoPm, \expoP}
& \lesssim \Err_p[\seedg,\seedh] ,
\endaligned
\ee
and such that the Riemannian metrics $g,\seedg,g_0$ are uniformly equivalent on~$\Mbf$.
\end{theorem}

\begin{proof}
\bse\label{sts-all-bounds}
We consider the iteration scheme~\eqref{Newton-again}.
First, $\Jcal_{(g_k,h_k;g_0,h_0)}[u_k,Z_k] = - \Gcal(g_k,h_k)$ has a unique solution $(u_k,Z_k)$ by \autoref{lem:appE-invertibility}, provided $\|g_k-g_0\|_{C^2_0(\Omega,g_0)}$ and $\|h_k-h_0\|_{C^1_1(\Omega,g_0)}$ are small enough.
The iteration scheme then takes $g_{k+1}=g_k+\gdiff_k$ and $h_{k+1}=h_k+\hdiff_k$, where $\gdiff_k=\omegabf_p^2 \, d\Hcal_{(g_0,h_0)}^{*\flat\flat}[u_k,Z_k]$ and $\hdiff_k=\omegabf_{p+1}^2  \, d\Mcal_{(g_0,h_0)}^{*\sharp\sharp}[u_k,Z_k]$.

Fix $\teps_G,\teps_A>0$ sufficiently small as determined below.  We shall prove by induction that $(g_1,h_1)=(\seedg,\seedh)$ and all iterates $(g_k,h_k)$ satisfy a variant of~\eqref{sts-premise} with different radial exponents:
\bel{sts-iteration-bounds}
\aligned
\Err_p[\seedg,\seedh] & \leq \teps_A ,
\\
\Err_p[g_k,h_k] & \leq 2^{-2^{k-1}} \Err_p[\seedg,\seedh] ,
\\
\| g_k - g_0 \|^{N, \alpha}_0 + \| h_k - h_0 \|^{N,\alpha}_1 & \leq (1 - 2^{-k}) \teps_G .
\endaligned
\ee
Since $p_A\geq p$ and $p_G>0$, and $\wtrr$~is bounded below, \eqref{sts-premise}~for sufficiently small $\eps_G,\eps_A>0$ implies the iteration bounds~\eqref{sts-iteration-bounds} for $k=1$.

Assume that \eqref{sts-iteration-bounds} holds for some $k\geq 1$.
The explicit formulas of $(\gdiff_k,\hdiff_k)$ imply (see \autoref{lem:sec11-estimates-geom}) that they are controlled in Lebesgue--H\"older norms by $(u_k,Z_k)$, themselves controlled thanks to \refwithname{Lemmas}{lem:appE-invertibility} \refwithname{and}{lem:appE-linearized-pointwise} in terms of the source $\Gcal(g_k,h_k)$:
\bel{gdiffkbound}
\aligned
\quad & \unquad
\Norm{\gdiff_k}^{N,\alpha}_{p,\expoPm,\expoP}
+ \Norm{\hdiff_k}^{N,\alpha}_{p+1,\expoPm,\expoP}
\\
& \lesssim \|u_k\|_{H^2_{n-2-p, -\expoP}} + \|Z_k\|_{H^1_{n-2-p, -\expoP}}
+ \|u_k\|^{N+2,\alpha}_{n-2-p, -\expoPp+2}
+ \|Z_k\|^{N+1,\alpha}_{n-2-p, -\expoPp+1}
\\
& \lesssim \Err_p[g_k,h_k] \leq 2^{-2^{k-1}} \Err_p[\seedg,\seedh] .
\endaligned
\ee
\refwithname{Lemmas}{lem:appE-invertibility} \refwithname{and}{lem:appE-linearized-pointwise} require smallness of $(g_k-g_0,h_k-h_0)$ in low-order pointwise norms, which is ensured by the last condition in~\eqref{sts-iteration-bounds} for sufficiently small~$\teps_G$.

We then observe that, because $(g_{k+1},h_{k+1})$ is chosen to solve the linear problem, its failure to be an exact solution comes only from nonlinear terms (see~\eqref{Gcalgk1hk1}):
\be
\Gcal(g_{k+1},h_{k+1}) = \Qcal\Gcal_{(g_k,h_k)}[\gdiff_k, \hdiff_k] .
\ee
This, in turn, is controlled thanks to \autoref{lem:appE-control-nonlin} by the \emph{square} of~\eqref{gdiffkbound}
\bel{NewtGbound}
\Err_p[g_{k+1},h_{k+1}]
\lesssim \bigl(\|\gdiff_k\|^{N,\alpha}_{p,\expoPm} + \|\hdiff_k\|^{N,\alpha}_{p+1,\expoPm}\bigr)^2
\lesssim 2^{-2^k} \Err_p[\seedg,\seedh]^2 \leq 2^{-2^k} \Err_p[\seedg,\seedh] \teps_A .
\ee
For $\teps_A$ sufficiently small, this gives the second bound in~\eqref{sts-iteration-bounds}.
Then, \eqref{gdiffkbound} controls the weighted norms of $(\gdiff_k,\hdiff_k)$.
Since $\expoPm>N+\alpha$ and the normal derivative of $\lambdabf$ on $\del\Omega$ is bounded below, $C^{N,\alpha}$ norms with angular weights control those without, so that
$\|\gdiff_k\|^{N,\alpha}_0 + \|\hdiff_k\|^{N,\alpha}_1 \leq 2^{-2^{k-1}}C\teps_A$ for some constant $C>0$ that does not depend on~$k$ nor on the seed data.
We obtain
\be
\aligned
\| g_{k+1} - g_0 \|^{N, \alpha}_0 + \| h_{k+1} - h_0 \|^{N,\alpha}_1
& \leq \| g_k - g_0 \|^{N, \alpha}_0 + \| h_k - h_0 \|^{N,\alpha}_1
+ \| \gdiff_k \|^{N, \alpha}_0 + \| \hdiff_k \|^{N,\alpha}_1
\\
& \leq (1 - 2^{-k}) \teps_G + 2^{-2^{k-1}} C \teps_A ,
\endaligned
\ee
which implies the third bound in~\eqref{sts-iteration-bounds} for $\teps_A$ sufficiently small compared to $\teps_G$.
This concludes the proof of~\eqref{sts-iteration-bounds} for all $k\geq 1$ by induction, hence of all bounds in~\eqref{sts-all-bounds}.
From these bounds one deduces the desired control of $(u,Z,g-\seedg,h-\seedh)$ by summing the bounds on $(u_k,Z_k,\gdiff_k,\hdiff_k)$.
The condition $\expoPm>N+\alpha$ implies that $(g,h)$ are $C^{N,\alpha}$ regular across the boundary~$\del\Omega$.
Since $\Gcal(g_k,h_k)\to 0$ and $(g_k,h_k)\to(g,h)$ in $C^{N,\alpha}$ one deduces $\Gcal(g,h)=0$.
\ese
\end{proof}


\section{Observations on linear differential equations}
\label{appendix=F}

\subsection{Definition of solution operators}
\label{appendix=F.1}

The shell averages and functionals both obey radial differential equations expressed in terms of $\vartheta=r\del_r$.
Let us repeat our notation~\eqref{IJdef-first}.  For any exponent $\beta\in\RR$ and any function~$f:[R,+\infty) \to\RR$ that is suitably integrable\footnote{The operator $I_\beta$ is defined for locally-integrable functions, while $J_\beta$ is defined provided $r\mapsto f(r)r^{\beta-1}$ is integrable.}:
\bel{IJdef}
I_\beta[f](r) \coloneqq r^{-\beta} \int_R^r f(s) s^\beta \frac{ds}{s},
\qquad
J_\beta[f](r) \coloneqq r^{-\beta} \int_r^{+\infty} f(s) s^\beta \frac{ds}{s}.
\ee
Both operators map non-negative functions to non-negative functions, and they are (up to a sign) inverses of the differential operator $\vartheta+\beta$, as we state in \autoref{lem:inte-more}, below.
We recall that $\vartheta=r\del_r$.
As a preliminary step we mention two obvious inequalities that hold for $\alpha\leq\beta$ and for a \emph{non-negative} function~$f$,
\bel{IJmonotonic}
I_\beta[f](r) \leq I_\alpha[f](r) , \qquad
J_\alpha[f](r) \leq J_\beta[f](r) , \qquad r\in[R,+\infty) .
\ee

\begin{lemma}[An explicit formula for first-order ODEs]
\label{lem:inte-more}
The operators $I_\beta$ and $J_\beta$ are determined by the conditions 
\bel{IJprop}
(\vartheta+\beta)I_\beta[f](r) = - (\vartheta+\beta)J_\beta[f](r) = f(r),
\ee
together with the boundary conditions $I_\beta[f](R)=0$ and $J_\beta[f](r)=o(r^{-\beta})$ as $r\to+\infty$. In addition, one has
\bel{IJprop-2}
\aligned
I_\beta\Bigl[(\vartheta+\beta)f\Bigr](r) & = f(r) - f(R) R^\beta r^{-\beta},
\\
J_\beta\Bigl[(\vartheta+\beta)f\Bigr](r) & = - f(r)  \quad \text{provided } \lim_{r\to+\infty} f(r) r^\beta = 0 .
\endaligned
\ee
\end{lemma}

We can also establish the following elementary result. 

\begin{lemma}[An explicit formula for high-order ODEs]
\label{lem:integrate-ODE}
The solutions to the $m$-th-order differential equation $(\vartheta+\beta_1) \dots(\vartheta+\beta_m)f=g$ for pairwise distinct $\beta_1,\dots,\beta_m\in\RR$ can be expressed as
\be
f = \sum_{i=1}^m \Bigl( A_i r^{-\beta_i} + B_i I_{\beta_i}[g](r) \Bigr)
= \sum_{i=1}^m r^{-\beta_i} \Bigl( A_i + B_i \int_R^r g(s)s^{\beta_i-1}\,ds\Bigr),
\ee
 where $B_i = \prod_{j\neq i} (\beta_j - \beta_i)^{-1}$ and $A_i$ are arbitrary constants.
\end{lemma}

\begin{proof} The homogeneous solutions are obviously $r^{-\beta_i}$ for $i=1,\dots,m$, which explains the free constants~$A_i$. To check the inhomogeneous terms, we act on the given Ansatz with the differential operator, using~\eqref{IJprop}:
\be 
  (\vartheta+\beta_1) \dots(\vartheta+\beta_m) f = \polyP_\beta(\vartheta) g, 
  \qquad
  \polyP_\beta(X) = \sum_{i=1}^m \prod_{j\neq i} \frac{X+\beta_j}{\beta_j - \beta_i}.
\ee
For each $k=1,\ldots,m$ we evaluate $\polyP_\beta(-\beta_k)$: all terms vanish except the $i=k$ term, which equals~$1$.  Thus, $\polyP_\beta-1$ is a polynomial of degree at most~$(m-1)$ that vanishes at $m$ points, hence $\polyP_\beta=1$ identically.
\end{proof}

In the main text we also require pointwise decay of certain integrals, stated as follows.

\begin{lemma}\label{lem:appF-pointwise}
  For exponents $\alpha<\beta<\gamma$, if a locally integrable function $f\colon[R,+\infty)\to\RR$ obeys $\lim_{r\to+\infty} r^\beta f(r)=0$ then
  \be
  \lim_{r\to+\infty} r^\beta J_\alpha[f](r)
  = \lim_{r\to+\infty} r^\beta I_\gamma[f](r) = 0 .
  \ee
\end{lemma}

\begin{proof} For all sufficiently large $r$, denote by $N_f(r) = \sup_{s\in[r,+\infty)} |s^\beta f(s)|$ the $L^\infty_\beta$ norm of $f$ on $[r,+\infty)$.
By assumption $N_f(r)\to 0$ as $r\to+\infty$, and $|f(s)|\leq N_f(r)s^{-\beta}$ for $s\in[r,+\infty)$,
thus
\be
|r^\beta J_\alpha[f](r)|
\leq N_f(r) r^{\beta-\alpha} \int_r^{+\infty} s^{\alpha-\beta} \frac{ds}{s}
= \frac{N_f(r)}{\beta-\alpha} .
\ee
For every fixed sufficiently large $\rho>R$, split the integral defining $I_\gamma$ into intervals $[R,\rho]$ and $[\rho,r]$:
\be
|r^\beta I_\gamma[f](r)|
\leq
r^{\beta-\gamma}\int_R^\rho |f(s)|s^\gamma\frac{ds}{s}
+N_f(\rho)r^{\beta-\gamma}\int_\rho^r s^{\gamma-\beta}\frac{ds}{s} .
\ee
The contribution of the first interval only depends on~$r$ through a power-law prefactor that tends to zero, and the second contribution is bounded by $N_f(\rho)/(\gamma-\beta)$.  First letting $r\to+\infty$ for fixed~$\rho$, and then letting $\rho\to+\infty$, proves the claim.
\end{proof}


\subsection{Hardy-type inequalities}
\label{appendix=F.2}

Our aim here is to control some nested integrals that arise in the analysis of our shell identities. We build upon the standard Hardy inequality in the radial direction, which states
\bel{eqE11}
\int_R^{+\infty} |u(r)|^2 \, dr
+ 2 R |u(R)|^2
\leq 4 \int_R^{+\infty} |u'(r)|^2 r^2 \, dr
\ee
for functions $u:[R, +\infty) \to \RR$ with $u(r)=o(r^{-1/2})$ as $r\to+\infty$. After listing a set of Cauchy--Schwarz inequalities (in \autoref{lem:rad-CS}), we prove a weighted generalization (in \autoref{lem:rad-Hardy}) of~\eqref{eqE11} on intervals $[R,r]$ and $[r,+\infty)$, then on a mixed combination of such intervals (in \autoref{lem:rad-mixed}).

\begin{lemma}[Cauchy--Schwarz inequalities]
\label{lem:rad-CS}
For any exponents $\alpha,\beta\in\RR$, and any pair of functions $f,g:[R,+\infty) \to\RR$ that are locally square-integrable, one has
\bse
\bel{cs-1-a}
\aligned
  \bigl( I_{\alpha+\beta}[f g](r) \bigr)^2 & \leq I_{2\alpha}[f^2](r) \, I_{2\beta}[g^2](r) , \qquad\qquad r\in [R,+\infty) , \\
  \bigl( I_{\alpha+\beta}[f](r) \bigr)^2 & \leq \frac{1}{2\beta} I_{2\alpha}[f^2](r) , \qquad\quad \beta > 0 , \quad r\in [R,+\infty) .
\endaligned
\ee
  If the functions $r^{\alpha-1/2}f(r)$ and $r^{\beta-1/2}g(r)$ are square-integrable on $[R,+\infty)$ then one has
\bel{cs-1-b}
\aligned
  \bigl( J_{\alpha+\beta}[f g](r) \bigr)^2 & \leq J_{2\alpha}[f^2](r) \, J_{2\beta}[g^2](r), \qquad\qquad r\in [R,+\infty) , \\
  \bigl( J_{\alpha+\beta}[f](r) \bigr)^2 & \leq \frac{1}{-2\beta} J_{2\alpha}[f^2](r) , \qquad \beta < 0 , \quad r\in [R,+\infty) .
\endaligned
\ee
\ese
In particular, assuming only that $r^{\alpha-1/2}f(r)$ is square integrable on $[R,+\infty)$, the following decay properties hold as $r\to+\infty$:
  $I_{\alpha+\beta}[f](r) = \Obig(r^{-\alpha})$ for $\beta>0$, and
  $J_{\alpha+\beta}[f](r) = \osmall(r^{-\alpha})$ for $\beta<0$.
\end{lemma}

\begin{proof} The inequalities involving $f$ and~$g$ are simply restatements of the Cauchy--Schwarz inequality (on the intervals $[R,r]$ and $[r,+\infty)$, respectively) for the functions $s\mapsto s^{\alpha-1/2}f(s)$ and $s\mapsto s^{\beta-1/2}g(s)$.  The prefactors in front of integrals are $r^{-2\alpha-2\beta}$ on all sides of these inequalities.
  The inequalities apply to $g=1$ with constants arising from an explicit evaluation of $I_{2\beta}[1]$ and $J_{2\beta}[1]$: for $\gamma>0$,
\[
  I_\gamma[1](r) = \frac{1}{\gamma} \bigl(1 - (R/r)^\gamma\bigr) \leq \frac{1}{\gamma} ,
  \qquad
  J_{-\gamma}[1](r) = \frac{1}{\gamma} .
\]
The $r^{-2\alpha}$ decay is then an immediate consequence of these Cauchy--Schwarz inequalities, by noting that $r^{2\alpha} I_{2\alpha}[f^2](r)$ and $r^{2\alpha} J_{2\alpha}[f^2](r)$ are integrals on $[R,r]$ and $[r,+\infty)$ of a non-negative integrable function, hence are respectively bounded (by the integral on $[R,+\infty)$) and $o(1)$ as $r\to+\infty$.
\end{proof}

\begin{lemma}[Hardy-type inequalities]
\label{lem:rad-Hardy}
Fix a pair of exponents $\alpha,\beta\in\RR$ and a locally square-integrable function $f:[R,+\infty) \to\RR$. If $\alpha>\beta$ then the function $u=I_\alpha[f]$ obeys the Hardy-type inequality
  \bse\label{rad-Hardy}
\bel{rad-Hardy-a}
  I_{2\beta}\bigl[ u^2 \bigr](r) + \frac{1}{\alpha-\beta} u(r)^2
  \leq \frac{1}{(\alpha-\beta)^2} I_{2\beta}[f^2](r) ,
  \qquad r\in[R,+\infty) .
\ee
If $\alpha<\beta$ and $r\mapsto r^{\beta-1/2} f(r)$ is square-integrable on $[R,+\infty)$, then the function $v=J_\alpha[f]$ obeys
\bel{rad-Hardy-b}
  J_{2\beta}\bigl[ v^2 \bigr](r) + \frac{1}{\beta- \alpha} v(r)^2
  \leq \frac{1}{(\beta- \alpha)^2} J_{2\beta}[f^2](r) ,
  \qquad r\in[R,+\infty) .
\ee
  \ese
\end{lemma}

\begin{proof}
We prove the first inequality for $\alpha>\beta$. For any $\gamma\in\RR$, we expand the following square and use $f=(\vartheta+\alpha)u$ to integrate by parts the cross-term:
\[
\aligned
  0 & \leq I_{2\beta}\bigl[ (\gamma u - f)^2 \bigr]
 = \gamma^2 I_{2\beta}\bigl[ u^2 \bigr]
  - \gamma I_{2\beta}\Bigl[ 2 u (\vartheta+\alpha)u \Bigr]
  + I_{2\beta}[f^2]
\\
& =
\gamma (\gamma - 2\alpha + 2\beta) I_{2\beta}\bigl[ u^2 \bigr]
  - \gamma I_{2\beta}\Bigl[ (\vartheta+2\beta)(u^2) \Bigr]
  + I_{2\beta}[f^2]
  \\
  &
  = \gamma (\gamma - 2\alpha + 2\beta) I_{2\beta}\bigl[ u^2 \bigr]
  - \gamma u^2
  + I_{2\beta}[f^2],
\endaligned
\]
where in the last line we used the integration by parts formula~\eqref{IJprop-2}, which has no boundary term at $R$ thanks to $u(R)=I_\alpha[f](R)=0$. Taking $\gamma=2\alpha-2\beta>0$ gives a Cauchy--Schwarz inequality, specifically the second line in~\eqref{cs-1-a}. Taking $\gamma=\alpha-\beta>0$ gives the Hardy inequality~\eqref{rad-Hardy-a} we wished to prove.

Next we prove the second inequality for $\alpha<\beta$. The decay statement in \autoref{lem:rad-CS}, applied with the pair of exponents $(\beta,\alpha-\beta)$, states that $v(r)=J_\alpha[f](r)=o(r^{-\beta})$ as $r\to+\infty$. For any $\gamma\in\RR$, we can then expand the following square and use $f=- (\vartheta+\alpha)v$ to integrate by parts the cross-term:
\[
\aligned
  0 \leq J_{2\beta}\bigl[ (\gamma v + f)^2 \bigr]
& = \gamma^2 J_{2\beta}\bigl[ v^2 \bigr]
  - \gamma J_{2\beta}\Bigl[ 2 v (\vartheta+\alpha)v \Bigr]
  + J_{2\beta}[f^2]
\\
& = \gamma (\gamma - 2\alpha + 2\beta) J_{2\beta}\bigl[ v^2 \bigr]
  - \gamma J_{2\beta}\Bigl[ (\vartheta+2\beta)(v^2) \Bigr]
  + J_{2\beta}[f^2]
\\
& = \gamma (\gamma - 2\alpha + 2\beta) J_{2\beta}\bigl[ v^2 \bigr]
  + \gamma v^2
  + J_{2\beta}[f^2],
\endaligned
\]
where in the last line we used the integration by parts formula~\eqref{IJprop-2}, which has no boundary term at infinity thanks to the decay $v(r)^2=o(r^{-2\beta})$. Taking $\gamma=2\alpha-2\beta$ yields a previously-proven Cauchy--Schwarz inequality. For $\gamma=\alpha-\beta<0$ we obtain the desired Hardy-type inequality.
\end{proof}

\begin{lemma}[Hardy-type inequalities with different intervals] \label{lem:rad-mixed}
Consider three exponents $\alpha,\beta,\gamma\in\RR$ and a function $f:[R,+\infty) \to\RR$ such that $r^{\beta-1/2}f(r)$ is square-integrable on $[R,+\infty)$. If $\beta,\gamma<\alpha$ then one has 
\bse
\bel{rad-mixed-a}
J_{2\beta}\bigl[ I_\alpha[f]^2 \bigr](r)
\leq \frac{1}{2(\alpha-\beta)(\alpha-\gamma)} I_{2\gamma}[f^2](r)
+ \frac{2}{(\alpha-\beta)^2} J_{2\beta}[f^2](r) ,
\qquad r\in[R,+\infty) .
\ee
If $\alpha<\beta,\gamma$ then one has 
\be
I_{2\gamma}\bigl[ J_\alpha[f]^2 \bigr](r)
\leq \frac{2}{(\gamma-\alpha)^2} I_{2\gamma}[f^2](r) + \frac{1}{2(\gamma-\alpha)(\beta-\alpha)} J_{2\beta}[f^2](r) ,
\qquad r\in[R,+\infty) .
\ee
\ese
\end{lemma}

\begin{proof}
For $r,s\in[R,+\infty)$ with $r\leq s$, we decompose
$
s^\alpha I_\alpha[f](s) = r^\alpha I_\alpha[f](r) + \int_r^s t^\alpha f(t) \frac{dt}{t}.
$
Inserting this decomposition into the explicit expression of $J_{2\beta}$ yields
\[
\aligned
J_{2\beta}\bigl[ I_\alpha[f]^2 \bigr](r)
& \leq r^{-2\beta} \int_r^{+\infty} s^{2\beta-2\alpha} \biggl(2 \bigl( r^\alpha I_\alpha[f](r) \bigr)^2 + 2 \Bigl( \int_r^s t^\alpha f(t) \frac{dt}{t}\Bigr)^2\biggr) \frac{ds}{s}
\\
& = \frac{1}{\alpha-\beta} \bigl( I_\alpha[f](r) \bigr)^2
+ 4 r^{-2\beta} \iiint_{r\leq t_1\leq t_2\leq s} s^{2\beta-2\alpha} t_1^\alpha f(t_1) t_2^\alpha f(t_2) \frac{dt_1}{t_1} \frac{dt_2}{t_2} \frac{ds}{s}
\\
& = \frac{1}{\alpha-\beta} \bigl( I_\alpha[f](r) \bigr)^2
+ \frac{2}{\alpha-\beta} J_{2\beta}\Bigl[f \, J_{2\beta-\alpha}[f] \Bigr](r)
\\
& \leq \frac{1}{\alpha-\beta} \bigl( I_\alpha[f](r) \bigr)^2
+ \frac{2}{\alpha-\beta} \Bigl( J_{2\beta}[f^2](r) \, J_{2\beta}\Bigl[J_{2\beta-\alpha}[f]^2\Bigr](r) \Bigr)^{1/2} ,
\endaligned
\]
where we swapped the order of integrals over $t_1,t_2,s$ to obtain nested integrals that all range all the way to infinity, hence can be expressed with the $J$~notation, and we have then applied the Cauchy--Schwarz inequality~\eqref{cs-1-b}. To reach the desired bound~\eqref{rad-mixed-a}, we use
\[
\aligned
I_\alpha[f]^2 \leq \frac{1}{2(\alpha-\gamma)} I_{2\gamma}[f^2] , \qquad
J_{2\beta}\bigl[ J_{2\beta-\alpha}[f]^2 \bigr]
\leq \frac{1}{(\alpha-\beta)^2} J_{2\beta}[f^2] ,
\endaligned
\]
which are~\eqref{cs-1-a} with the exponents $(\gamma,\alpha-\gamma)$ and~\eqref{rad-Hardy-b} with the exponents $(2\beta-\alpha,\beta)$, respectively. Next, for the second inequality we begin with a decomposition valid for $s\leq r$,
\[
s^\alpha J_\alpha[f](s) = r^\alpha J_\alpha[f](r) + \int_s^r t^\alpha f(t) \frac{dt}{t} .
\]
We follow the same steps as the first inequality, but the intervals of integration are different. The main change is to replace the integral
\[
r^{2\alpha-2\gamma} \int_r^{+\infty} s^{2\gamma-2\alpha} \frac{ds}{s} = \frac{1}{2(\alpha-\gamma)} , \qquad \alpha>\gamma ,
\]
whose convergence required $\alpha>\gamma$, by an integral
\[
r^{2\alpha-2\gamma} \int_R^r s^{2\gamma-2\alpha} \frac{ds}{s} = \frac{1-(R/r)^{2\gamma-2\alpha}}{2(\gamma-\alpha)} \leq \frac{1}{2(\gamma-\alpha)} , \qquad \gamma>\alpha ,
\]
whose convenient upper bound requires $\gamma>\alpha$. We arrive at
\be
\aligned
I_{2\gamma}\bigl[ J_\alpha[f]^2 \bigr](r)
& \leq r^{-2\gamma} \int_R^r s^{2\gamma-2\alpha} \biggl(2 \bigl( r^\alpha J_\alpha[f](r) \bigr)^2 + 2 \Bigl( \int_s^r t^\alpha f(t) \frac{dt}{t}\Bigr)^2\biggr) \frac{ds}{s}
\\
& \leq \frac{1}{\gamma-\alpha} \bigl( J_\alpha[f](r) \bigr)^2
+ \frac{2}{\gamma-\alpha} I_{2\gamma}\Bigl[f \, I_{2\gamma-\alpha}[f] \Bigr](r)
\\
& \leq \frac{1}{\gamma-\alpha} \bigl( J_\alpha[f](r) \bigr)^2
+ \frac{2}{\gamma-\alpha} \Bigl( I_{2\gamma}[f^2](r) \, I_{2\gamma}\Bigl[I_{2\gamma-\alpha}[f]^2\Bigr](r) \Bigr)^{1/2} .
\endaligned
\ee
To finish up, we bound $J_\alpha[f]^2$ by $J_{2\gamma}[f^2]$ using~\eqref{cs-1-b} and we bound $I_{2\gamma}\bigl[I_{2\gamma-\alpha}[f]^2\bigr]$ using the Hardy inequality~\eqref{rad-Hardy-a}.
\end{proof} 


\subsection{Distributional definition of solution operators}
\label{appendix=F.3}

\paragraph{The space of distributions.}

We consider first the space $\Dcal'((R,+\infty))$ of distributions with test functions in the space $C_c^\infty((R,+\infty))$ of smooth functions with compact support contained in~$(R,+\infty)$.  Such compactly-supported test functions vanish in a neighborhood of~$R$.  We use the notation $\la f,\varphi\ra$ for the distributional pairing, and any locally integrable function $f\colon(R,+\infty)\to\RR$ is viewed as a distribution by setting
\bel{fvarphi-dual}
\la f, \varphi\ra \coloneqq \int_R^{+\infty} f(r) \varphi(r) \frac{dr}{r} .
\ee
The choice of measure, together with compact support in $(R,+\infty)$, ensures the absence of boundary term when integrating by parts to get (for locally integrable~$f$)
\be
\la \vartheta f, \varphi\ra = - \la f, \vartheta\varphi\ra .
\ee
We use this identity more generally as the definition of $\vartheta$ acting on distributions $f\in\Dcal'((R,+\infty))$.
The differential equations on the shell functional~$\Phi^\notreH[u]$ and average~$\la u\ra$ are derived as identities in this space of distributions.  However, the solution operators $I_\beta$ and $J_\beta$ cannot be defined here, so that the equations must be solved in a different space introduced now.


\paragraph{Dual Sobolev spaces.}

We are interested in the spaces $H^{k*}_{\alpha}([R,+\infty)) = \bigl(H^k_{-\alpha}([R,+\infty))\bigr)^*$ that are dual to weighted Sobolev spaces for $k\geq 0$ and $\alpha\in\RR$, with the norm
\bel{Hkstar-def}
\|f\|_{H^{k*}_{\alpha}([R,+\infty))}
\coloneqq \sup_{\varphi\not\equiv 0} \frac{\bigl| \la f, \varphi\ra \bigr|}{\|\varphi\|_{H^k_{-\alpha}([R,+\infty))}} .
\ee
In contrast to $\Dcal'((R,+\infty))$, the test functions $\varphi$ are not required to vanish near~$R$.
The choice of radial exponents $\alpha$ and $-\alpha$ ensures that the spaces $L^2_\alpha$ and $H^{0*}_\alpha$ are identified by the bracket $\la f,\varphi\ra$ in~\eqref{fvarphi-dual}, namely integration against the measure~$dr/r$.
One has the natural inclusions $\dots\subset H^1_{\alpha}([R,+\infty))\subset L^2_{\alpha}([R,+\infty))=H^{0*}_{\alpha}([R,+\infty))\subset H^{1*}_{\alpha}([R,+\infty))\subset\dots$ and bounds (for $k\geq 0$)
\be
\|f\|_{H^{(k+1)*}_{\alpha}([R,+\infty))} \leq \|f\|_{H^{k*}_{\alpha}([R,+\infty))} , \qquad
\|f\|_{H^{0*}_{\alpha}([R,+\infty))} = \|f\|_{L^2_{\alpha}([R,+\infty))} .
\ee
Restricting the duality brackets $\la f,\cdot\ra$ to test functions with compact support in $(R,+\infty)$ defines linear maps $H^{k*}_\alpha([R,+\infty))\to\Dcal'((R,+\infty))$ for $k\geq 0$ that are compatible with each other.  We denote them by $\dots|_{(R,+\infty)}$ in this section, and leave them implicit in the main text when no confusion arises.

We also define a radial derivative operator $\vartheta_*\colon H^{k*}_{\alpha}([R,+\infty))\to H^{(k+1)*}_{\alpha}([R,+\infty))$ for $k\geq 0$ by
\bel{vartheta-star-def}
\la \vartheta_* f, \varphi \ra \coloneqq - \la f, \vartheta\varphi \ra , \qquad
\varphi\in H^{k+1}_{-\alpha}([R,+\infty)) ,
\ee
which obeys $\|\vartheta_*f\|_{H^{(k+1)*}_{\alpha}([R,+\infty))}\leq\|f\|_{H^{k*}_{\alpha}([R,+\infty))}$.
An important consideration is that this (dual) derivative operator does not reduce to $\vartheta$ when restricted to $H^1_\alpha([R,+\infty))\subset H^{k*}_\alpha([R,+\infty))$:
indeed, for $f\in H^1_\alpha([R,+\infty))$ and $\varphi\in H^1_{-\alpha}([R,+\infty))$ one has $\la\vartheta_* f-\vartheta f,\varphi\ra = f(R)\varphi(R)$.
Observe however that
\be
(\vartheta_* f)|_{(R,+\infty)} = \vartheta\bigl(f|_{(R,+\infty)}\bigr) ,
\ee
which makes it convenient to state differential equations in~$\Dcal'((R,+\infty))$ instead of dual Sobolev spaces.


\paragraph{Solution operators in dual Sobolev spaces.}

For locally-integrable functions~$f$ one checks that
\be
\la I_\beta[f],\varphi\ra
= \int_{R<r<s<+\infty} r^\beta s^{-\beta} f(r) \varphi(s) \frac{dr}{r}\, \frac{ds}{s}
= \la f,J_{-\beta}[\varphi]\ra
\ee
by expanding out the definitions in terms of integrals and swapping the two integrals.
This suggests generalizing the operators $I_\beta$ and~$J_\beta$ to $f\in H^{k*}_\alpha([R,+\infty))$ for $k \geq 1$ by using the definitions
\bel{IJ-distrib}
\aligned
\bigl\la I_\beta[f] , \varphi \bigr\ra
& \coloneqq \bigl\la f, J_{-\beta}[\varphi] \bigr\ra , \qquad \varphi\in H^{k-1}_{-\alpha}([R,+\infty)) , \quad \text{if } \alpha < \beta ,
\\
\bigl\la J_\beta[f] , \varphi \bigr\ra
& \coloneqq \bigl\la f, I_{-\beta}[\varphi] \bigr\ra , \qquad \varphi\in H^{k-1}_{-\alpha}([R,+\infty)) , \quad \text{if } \beta < \alpha .
\endaligned
\ee
The restriction on exponents, and the space of test functions, arise as a result of the bounds~\eqref{Ibeta-norm} below.
Since decreasing the exponent $\alpha$ enlarges the space $H^{k*}_\alpha([R,+\infty))$, one can always define $I_\beta[f]$ regardless of the ordering of $\alpha$ and $\beta$, simply by treating $f$ as having $r^{-\gamma}$ decay for some $\gamma<\min(\alpha,\beta)$.  In contrast, $J_\beta[f]$ is only defined if $f$ has sufficient radial decay.
For $\beta<\alpha$ and $f\in H^{k*}_\alpha([R,+\infty))$, the following identity holds, whose terms lie in $H^{(k-1)*}_\gamma([R,+\infty))$ for any $\gamma<\beta$,
\bel{Ibeta-plus-Jbeta}
I_\beta[f] + J_\beta[f] = \la f, r^\beta\ra r^{-\beta} .
\ee

\begin{lemma}[Estimates for distributional solution operators]\label{lem:Ibeta-norm}
For any $\alpha,\beta\in\RR$, given a distribution $f\in H^{k*}_{\alpha}([R,+\infty))$ for $k\geq 1$, one has\footnote{There exist similar $H^{(k-2)*}_{\alpha}([R,+\infty))$ bounds on $J_{\beta_0}[f]-J_{\beta_1}[f]$ for $\beta_0<\beta_1<\alpha$ and on $I_{\beta_0}[f]-I_{\beta_1}[f]$ for $\alpha<\beta_0<\beta_1$ but we will not need them.}
\bel{Ibeta-norm}
\aligned
\bigl\| I_\beta[f] \bigr\|_{H^{(k-1)*}_{\alpha}([R,+\infty))}
& \lesssim \| f \|_{H^{k*}_{\alpha}([R,+\infty))} , \qquad \text{if } k\geq 1 \text{ and } \alpha < \beta ,
\\
\bigl\| J_\beta[f] \bigr\|_{H^{(k-1)*}_{\alpha}([R,+\infty))}
& \lesssim \| f \|_{H^{k*}_{\alpha}([R,+\infty))} , \qquad \text{if } k\geq 1 \text{ and } \beta < \alpha ,
\\
\bigl\| J_{\beta_0}[f] + I_{\beta_1}[f] \bigr\|_{H^{(k-2)*}_{\alpha}([R,+\infty))}
& \lesssim \| f \|_{H^{k*}_{\alpha}([R,+\infty))} , \qquad \text{if } k\geq 2 \text{ and } \beta_0 < \alpha < \beta_1 ,
\endaligned
\ee
where implicit constants depend on the exponents $\alpha,\beta,\beta_0,\beta_1$.
\end{lemma}

\begin{proof}
For $\varphi\in H^{k-1}_{-\alpha}([R,+\infty))$ we bound
\be
\Bigl| \bigl\la I_\beta[f] , \varphi \bigr\ra \Bigr|
= \Bigl| \bigl\la f, J_{-\beta}[\varphi] \bigr\ra \Bigr|
\leq \| f \|_{H^{k*}_{\alpha}([R,+\infty))} \bigl\|J_{-\beta}[\varphi]\bigr\|_{H^k_{-\alpha}([R,+\infty))} .
\ee
For $0\leq j\leq k-1$, observe that $\vartheta^j (\vartheta-\beta) J_{-\beta}[\varphi] = - \vartheta^j \varphi$ by~\eqref{IJprop}, hence $\vartheta^{j+1} J_{-\beta}[\varphi]$ is controlled by lower derivatives and by $\vartheta^j\varphi$.
On the other hand, the $L^2_{-\alpha}([R,+\infty))$ norm of $J_{-\beta}[\varphi]$ is bounded by that of $\varphi/(\beta-\alpha)$ thanks to the Hardy inequality~\eqref{rad-Hardy-b} evaluated at $r=R$, with $(\alpha,\beta)\to(-\beta,-\alpha)$.
The bound on $J_\beta[f]$ is proven identically, but using the Hardy inequality~\eqref{rad-Hardy-a} for $I_{-\beta}[\varphi]$.

For the sum $J_{\beta_0}[f]+I_{\beta_1}[f]$ we take a less regular test function $\varphi\in H^{k-2}_{-\alpha}([R,+\infty))$ and we bound
\bel{Jbeta0Ibeta1-norm-bound}
\aligned
\Bigl| \bigl\la J_{\beta_0}[f]+I_{\beta_1}[f], \varphi \bigr\ra \Bigr|
& = \Bigl| \bigl\la f, J_{-\beta_1}[\varphi] + I_{-\beta_0}[\varphi] \bigr\ra \Bigr|
\\
& \leq \| f \|_{H^{k*}_{\alpha}([R,+\infty))} \bigl\|J_{-\beta_1}[\varphi] + I_{-\beta_0}[\varphi]\bigr\|_{H^k_{-\alpha}([R,+\infty))} .
\endaligned
\ee
By the previous results (with $k$ shifted by one) we control the $H^{k-1}_{-\alpha}([R,+\infty))$ norm of $J_{-\beta_1}[\varphi] + I_{-\beta_0}[\varphi]$, so we only need to control the $k$-th derivative.  Note that
\be
\vartheta\bigl(J_{-\beta_1}[\varphi] + I_{-\beta_0}[\varphi]\bigr)
= \beta_1 J_{-\beta_1}[\varphi] + \beta_0 I_{-\beta_0}[\varphi]
\ee
because $(\vartheta-\beta_1)J_{-\beta_1}[\varphi]=-\varphi=-(\vartheta-\beta_0)I_{-\beta_0}[\varphi]$.  Thus, the $k$-th derivative is controlled by $(k-1)$-th order derivatives, and we are done bounding~\eqref{Jbeta0Ibeta1-norm-bound}.
\end{proof}

\paragraph{Solution of ordinary differential equations.}

The operators $I_\beta$ and $J_\beta$ defined in a distributional sense obey some of the properties listed in \refwithname{Lemmas}{lem:inte-more} \refwithname{and}{lem:integrate-ODE}, provided one applies the restriction map at the appropriate place.  We state \autoref{prop:integrate-ODE-distrib} below with a choice of functional spaces that is relevant in the main text.

\begin{lemma}
\label{lem:IJprop-distrib}
For any $\alpha,\beta\in\RR$, given a distribution $f\in H^{k*}_\alpha([R,+\infty))$ for $k\geq 0$, one has
\bel{IJprop-distrib}
\aligned
(\vartheta+\beta)\bigl( I_\beta[f]|_{(R,+\infty)}\bigr) & = f|_{(R,+\infty)} ,
\\
(\vartheta+\beta)\bigl( J_\beta[f]|_{(R,+\infty)}\bigr) & = - f|_{(R,+\infty)} , \qquad \text{if } \beta < \alpha .
\endaligned
\ee
\end{lemma}

\begin{proof}
For $k=0$, the operators $I_\beta,J_\beta$ are defined as standard integrals and~\eqref{IJprop-distrib} is an immediate calculation.
For $k\geq 1$, they are defined in~\eqref{IJ-distrib} (and below for $I_\beta[f]$ with $\beta\leq\alpha$).
For $I_\beta$ we consider $\varphi\in C^\infty_c((R,+\infty))$ and evaluate
\be
\aligned
\bigl\la (\vartheta+\beta)\bigl( I_\beta[f]|_{(R,+\infty)}\bigr) , \varphi \bigr\ra
& = \bigl\la I_\beta[f]|_{(R,+\infty)} , (-\vartheta+\beta) \varphi \bigr\ra
\\
& = \bigl\la I_\beta[f] , (-\vartheta+\beta) \varphi \bigr\ra
= \bigl\la f , J_{-\beta}[(-\vartheta+\beta) \varphi] \bigr\ra = \la f,\varphi\ra .
\endaligned
\ee
The calculation for $J_\beta$ is identical: the test function is compactly supported in $(R,+\infty)$, so its trace at~$R$ vanishes and no boundary contribution remains.
\end{proof}

\begin{proposition}[An explicit formula for high-order ODEs, distributional version]
\label{prop:integrate-ODE-distrib}
For integers $k\geq 1$ and $0\leq j\leq m$ and exponents $\beta_1<\dots<\beta_j<\alpha<\beta_{j+1}<\dots<\beta_m$, the function $f\in L^2_\alpha([R,+\infty))$ and distribution $g\in H^{k*}_\alpha([R,+\infty))$ obey the $m$-th-order differential equation
\be
(\vartheta+\beta_1) \dots (\vartheta+\beta_m) \bigl( f|_{(R,+\infty)} \bigr) = g|_{(R,+\infty)}
\quad \text{in } \Dcal'((R,+\infty))
\ee
if and only if
\bel{integrateODE-f-value}
f|_{(R,+\infty)} = \biggl( - \sum_{i=1}^j B_i J_{\beta_i}[g](r)
+ \sum_{i=j+1}^m \Bigl( A_i r^{-\beta_i} + B_i I_{\beta_i}[g](r) \Bigr) \biggr)\biggr|_{(R,+\infty)}
\quad \text{in } \Dcal'((R,+\infty)) ,
\ee
for some constants $A_{j+1},\dots,A_m$, where $B_i = \prod_{\ell\neq i} (\beta_{\ell} - \beta_i)^{-1}$.
Furthermore, if $k=1$ the identity~\eqref{integrateODE-f-value} holds without restriction to $(R,+\infty)$, as an identity in $L^2_\alpha([R,+\infty))$.
\end{proposition}

\begin{proof}
  The ordering of $\beta_i$ and $\alpha$ ensures that $J_{\beta_i}[g]$ is well-defined for $1\leq i\leq j$ and that $r^{-\beta_i}\in H^{k*}_\alpha([R,+\infty))$ for $j+1\leq i\leq m$ so that the right-hand side of~\eqref{integrateODE-f-value} without restriction to $(R,+\infty)$ is well-defined in $H^{k*}_\alpha([R,+\infty))$.  
  The proof of \autoref{lem:integrate-ODE} then applies, using the identities~\eqref{IJprop-distrib}, so that it is a solution of the differential equation.

  To show the converse, by subtracting the known solution (with $A_i=0$, say), the question reduces to finding solutions $h=f|_{(R,+\infty)}\in\Dcal'((R,+\infty))$ to $(\vartheta+\beta_1)\dots(\vartheta+\beta_m)h=0$.  This further reduces to the case $m=1$, and for brevity we set $\beta=\beta_1$.  Choose an arbitrary smooth function~$\varphi_0$ with compact support in $(R,+\infty)$ and with $J_{-\beta}[\varphi_0](R) = R^\beta \int_R^{+\infty} \varphi_0 s^{-\beta-1}ds$ non-zero.  Then for any smooth function $\varphi$ with compact support in $(R,+\infty)$, let $c_\varphi = J_{-\beta}[\varphi](R) / J_{-\beta}[\varphi_0](R)$.  One has
  \be
  \aligned
  \la h,\varphi\ra
  & = c_\varphi \la h,\varphi_0\ra
  + \la h,\varphi - c_\varphi \varphi_0\ra
  = c_\varphi \la h,\varphi_0\ra
  + \bigl\la h,(-\vartheta+\beta) J_{-\beta}[\varphi - c_\varphi \varphi_0]\bigr\ra
  \\
  & = c_\varphi \la h,\varphi_0\ra
  + \bigl\la (\vartheta+\beta)h,J_{-\beta}[\varphi - c_\varphi \varphi_0]\bigr\ra = \frac{\la h,\varphi_0\ra}{R^{-\beta} J_{-\beta}[\varphi_0](R)} \la r^{-\beta}, \varphi\ra ,
  \endaligned
  \ee
  so that $h$ is a multiple of~$r^{-\beta}$.  For general $m$, we learn that $h$ is a linear combination of~$r^{-\beta_i}$.
  For $1\leq i\leq j$ the homogeneous solutions $r^{-\beta_i}$ do not have bounded $H^{k*}_\alpha([R,+\infty))$ norms hence are ruled out.
  Finally, for $k=1$ both sides of~\eqref{integrateODE-f-value} without restriction to $(R,+\infty)$ are in $L^2([R,+\infty))$, and their equality when restricted to $(R,+\infty)$ implies that they are equal in $L^2([R,+\infty))$.
\end{proof}


\paragraph{Solutions of second-order differential equations.}

A particular consequence of \autoref{prop:integrate-ODE-distrib} which is useful in the main text concerns the equation
\bel{eq-ODE2}
-(\vartheta +\beta_0)(\vartheta+\beta_1) \bigl(g_2|_{(R,+\infty)}\bigr)
= g_0|_{(R,+\infty)} + (\vartheta+\beta_1)\bigl(g_1|_{(R,+\infty)}\bigr)
\quad \text{in } \Dcal'((R,+\infty)) ,
\ee
for $g_0,g_1,g_2\in H^{k*}_\alpha([R,+\infty))$ with $\beta_0<\alpha<\beta_1$.
The general solution is
\bel{equa-formuleODE}
g_2|_{(R,+\infty)} = C_1 r^{-\beta_1} + \frac{1}{\beta_1-\beta_0} \Bigl( J_{\beta_0}\bigl[(\beta_1-\beta_0)g_1 + g_0\bigr] + I_{\beta_1}[g_0] \Bigr) \Bigr|_{(R,+\infty)} .
\ee
Indeed, one checks the contribution of $g_1$ by using $-(\vartheta+\beta_0) J_{\beta_0}[g_1]|_{(R,+\infty)}=g_1|_{(R,+\infty)}$, so that~\eqref{equa-formuleODE} is indeed a solution, and it is the most general solution by \autoref{prop:integrate-ODE-distrib}.
If in addition $g_2\in L^2_\alpha([R,+\infty))$ and $g_0,g_1\in H^{1*}_\alpha([R,+\infty))$ then the identity~\eqref{equa-formuleODE} holds in $L^2_\alpha([R,+\infty))$ without restricting to $(R,+\infty)$.

\paragraph{Large radius limit of dual Sobolev norms.}
Usual Sobolev norms are integrals over the domain of interest, hence their restriction to $\Omega_{R'}=\{r>R'\}$ tends to zero in the limit $R'\to+\infty$.  Dual Sobolev norms satisfy a similar decay property.

\begin{lemma}[No concentration at infinity in dual Sobolev spaces]\label{lem:Hstar-cutoff}
For $\alpha\in\RR$, $k\geq 0$, and any $f\in H^{k*}_{\alpha,-\expoP}(\Omega_R)$, the norm of $f$ restricted to $\Omega_{R'}$ tends to zero as $R'\to+\infty$, in the sense that for any $\eps>0$ there exists a (large) $R'\geq R$ such that for all test functions $\varphi\in H^k_{\alpha,-\expoP}(\Omega_R)$ with support in~$\Omega_{R'}$,
\bel{fRpeps}
|\la f,\varphi\ra| \leq \eps \|\varphi\|_{H^k_{\alpha,-\expoP}(\Omega_R)} .
\ee
\end{lemma}

\begin{proof}
Fix $\eps>0$ and assume by contradiction that for all $R'$ there exists $\varphi_{R'}$ violating~\eqref{fRpeps}.  By density of compactly-supported functions in $H^k_{\alpha,-\expoP}(\Omega_R)$ we can assume that all $\varphi_{R'}$ are compactly supported, and normalized to have unit norm and positive $\la f,\varphi_{R'}\ra$.
One can then choose successive radii $R'_i$ by taking $R'_{i+1}$ large enough such that $\varphi_i=\varphi_{R'_i}$ vanishes on~$\Omega_{R'_{i+1}}$.
The functions $\varphi_i$ constructed along the way have disjoint support, unit norm, and $\la f,\varphi_i\ra>\eps$, hence $\varphi=\sum_{i\geq 1} \frac{1}{i}\varphi_i\in H^k_{\alpha,-\expoP}(\Omega_R)$ while $\la f,\varphi\ra>\sum_k\frac{1}{k}\eps=+\infty$.
\end{proof}



\begin{thebibliography}{99}

\bibitem{AndersonCorvinoPasqualotto}
{\sc J. Anderson, J. Corvino, and F. Pasqualotto,}
Multi-localized time-symmetric initial data for the Einstein vacuum equations,
J. Reine Angew. Math. 808 (2024), 67--110. 

\bibitem{AretakisCzimekRodnianski}
{\sc S. Aretakis, S. Czimek, and I. Rodnianski,}
The characteristic gluing problem for the Einstein equations and applications,
Duke Math. J. 174 (2025), no.~2, 355--402. 

\bibitem{Babuska}
{\sc I. Babu\v{s}ka,} 
Error-bounds for finite element method,
Numerische Math. 16 (1971), 322--333.  

\bibitem{Bartnik}
{\sc R. Bartnik,} 
The mass of an asymptotically flat manifold,
Comm. Pure Appl. Math. 49 (1986), 661--693. 

\bibitem{BeigChrusciel-1996}
{\sc R. Beig and P.T. Chru\'sciel,}
Killing vectors in asymptotically flat space-times. I. Asymptotically translational Killing vectors and the rigid positive energy theorem,
J. Math. Phys. 37 (1996), 1939--1961.

\bibitem{BeigChrusciel-2017}
{\sc R. Beig and P.T. Chru\'sciel,}
Shielding linearised gravity, 
Phys. Rev. D 95 (2017), 064063. 

\bibitem{BeigMurchadha}
{\sc R. Beig and N. O. Murchadha,}
The momentum constraints of general relativity and spatial conformal isometries,
Commun. Math. Phys. 176 1996), 723--738. 

\bibitem{Bieri-constraints}
{\sc L. Bieri, D. Garfinkle, J. Isenberg, D. Maxwell, and J. Wheeler,}
Asymptotically Euclidean solutions of the constraint equations with prescribed asymptotics,
Available at \url{https://arxiv.org/abs/2512.21274}.

\bibitem{BrendleHuisken:2017}
{\sc S. Brendle and G. Huisken,}
A fully nonlinear flow for two-convex hypersurfaces in Riemannian manifolds,
Inventiones Math. 210 (2017), 559--613.

\bibitem{Carlotto-Review}
{\sc A. Carlotto,}
The general relativistic constraint equations, 
Living Reviews in Relativity (2021), 24:2. 

\bibitem{CarlottoSchoen}
{\sc A. Carlotto and R. Schoen,} 
Localizing solutions of the Einstein constraint equations, 
Invent. Math. 205 (2016), 559--615.

\bibitem{CederbaumSakovich}
{\sc C. Cederbaum and A. Sakovich,}
On center of mass and foliations by constant spacetime mean curvature surfaces for isolated systems in general relativity, 
Calc. Var. Part. Diff. Equa. 60 (2021), 214.

\bibitem{ChenK}
{\sc X. Chen and S. Klainerman,}
Solving the constraint equation for general free data,
Available at \url{https://arxiv.org/abs/2512.22704}.

\bibitem{CDHS}
{\sc O. Chodosh, J.M. Daniels-Holgate, and F. Schulze,}
Mean curvature flow from conical singularities, 
Inventiones Math. 238 (2024), 1041--1066 

\bibitem{Choquet-book} 
{\sc Y. Choquet-Bruhat,}
{\sl General relativity and the Einstein equations,}
Oxford Math. Monographs, Oxford University Press, 2009.

\bibitem{ChoquetC}
{\sc Y. Choquet-Bruhat and D. Christodoulou,}
Elliptic systems in $H^{s, \delta}$ spaces on manifolds which are Euclidean at infinity,
Acta Math. 146 (1981), 129--150. 

\bibitem{Chrusciel-bourbaki}
{\sc P.T. Chru\'sciel,}
Anti-gravit\'e \`a la Carlotto-Schoen [after Carlotto and Schoen],
S\'eminaire Bourbaki, Vol. 2016/2017, Exp. 1120,  
Ast\'erisque 407 (2019), 1--25. 

\bibitem{ChrusCorvinoIsenberg} 
{\sc P.T. Chru\'sciel, J. Corvino, and J. Isenberg,}
 Construction of N-body initial data sets in general relativity,
Comm. Math. Phys. 304 (2011), 637--647. 

\bibitem{ChruscielDelay-memoir}
{\sc P.T. Chru\'sciel and E. Delay,}
On mapping properties of the general relativistic constraints operator in weighted function spaces with applications, 
M\'em. Soc. Math. France, Vol 94. French Math. Society, 2003.  

\bibitem{ChruscielDelay-2018}
{\sc P.T. Chru\'sciel and E. Delay,}
Exotic hyperbolic gluings,
J. Differential Geom. 108 (2018), 243--293. 

\bibitem{ChruscielDelay-2021}
{\sc P.T. Chru\'sciel and E. Delay,}
On Carlotto-Schoen-type scalar-curvature gluings, 
Bull. Soc. Math. France 149 (2021), 641--662. 
 
\bibitem{ColdingMinicozzi}
{\sc T.H. Colding and W.P. Minicozzi II,}
Generic mean curvature flow I: generic singularities,
Annals of Math. 175 (2012), 755--833. 

\bibitem{Corvino-2000} 
{\sc J. Corvino,}
Scalar curvature deformation and a gluing construction for the Einstein constraint equations, 
Comm. Math. Phys. 214 (2000), 137--189. 

\bibitem{CorvinoHuang} 
{\sc J. Corvino and L.-H. Huang,} 
Localized deformation for initial data sets with the dominant energy condition, 
Calc. Variations PDEs 59 (2020), 42.

\bibitem{CorvinoSchoen} 
{\sc J. Corvino and R. Schoen,}
On the asymptotics for the vacuum Einstein constraint equations, 
J. Diff. Geom. 73 (2006), 185--217.

 \bibitem{CzimekRodnianski}
{\sc S. Czimek and I. Rodnianski,} 
Obstruction-free gluing for the Einstein equations. 
Available at \url{https://arxiv.org/abs/2210.09663}.

\bibitem{Delay}
{\sc E. Delay,}
Localized gluing of Riemannian metrics in interpolating their scalar curvature, 
Differential Geom. Appl. 29 (2011), 433--439.  

\bibitem{DouglisNirenberg}
{\sc A. Douglis and L. Nirenberg,}
Interior estimates for elliptic systems of partial differential equations,
Comm. Pure Appl. Math. 8 (1955), 503--538.
 
\bibitem{DruetHebey}
{\sc O. Druet and E. Hebey,}
Stability and instability for Einstein-scalar field Lichnerowicz equations on compact Riemannian manifolds,
Math. Z. 263 (2009), 33--67. 

\bibitem{DruetPremoselli}
{\sc O. Druet and B. Premoselli,}
Stability of the Einstein-Lichnerowicz constraint system,
Math. Ann. 362 (2015), 839--886. 

\bibitem{Fang-1}
{\sc A.J. Fang, J. Szeftel, and A. Touati,} 
Initial data for Minkowski stability with arbitrary decay,
Adv. Theor. Math. Phys. 29 (2025), 933--1043. 
Also available at \url{https://arxiv.org/abs/2401.14353}.

\bibitem{Fang-2}
{\sc A.J. Fang, J. Szeftel, and A. Touati,} 
Spacelike initial data for black hole stability,
Commun. Math. Phys. 406 (2025), 235.
Also available at \url{https://arxiv.org/abs/2405.02071}.

\bibitem{FischerMarsden-1973} 
{\sc A.E. Fischer and J.E. Marsden,}
Linearization stability of the Einstein equations,
Bull. AMS. 79 (1973), 997--1003.

\bibitem{FischerMarsden-1975} 
{\sc A.E. Fischer and J.E. Marsden,}
Deformations of the scalar curvature, 
Duke Math. J. 42 (1975), 519--547. 

\bibitem{GallowayMiaoSchoen}
{\sc G. Galloway, P. Miao, and R. Schoen,}
Initial data and the Einstein constraint equations, in ``General relativity and gravitation'', 
Cambridge Univ. Press, Cambridge, 2015, pp.~412--448. 

\bibitem{Gicquaud}
{\sc R. Gicquaud,}
The conformal method is not conformal, in preparation.  

\bibitem{HanLin-book}
{\sc Q. Han and F.-H. Lin,}
{\sl Elliptic partial differential equations,}
Courant Lecture Notes in Mathematics, 
Vol. 1, Amer. Math. Soc., Providence, RI, 2nd edition, 2011.
 
\bibitem{Henneaux}
{\sc M. Henneaux,}
 Corvino-Schoen theorem and supertranslations at spatial infinity,
 Int. J. Mod. Phys. A 39 (2024), 2447007.

\bibitem{Hintz1}
{\sc P. Hintz,}
Gluing small black holes along timelike geodesics I: formal solution. 
Available at \url{https://arxiv.org/abs/2306.07409}.

\bibitem{Hintz2}
{\sc P. Hintz,}
Gluing small black holes along timelike geodesics II: uniform analysis on glued spacetimes.
Available at \url{https://arxiv.org/abs/2408.06712}.

\bibitem{Holst}
{\sc M. Holst, D. Maxwell, and R. Mazzeo,}
Conformal fields and the structure of the space of solutions of the Einstein constraint equations. 
Available at \url{https://arxiv.org/abs/1711.01042}.

\bibitem{Isenberg-1995}
{\sc J. Isenberg,}
Constant mean curvature solutions of the Einstein constraint equations on closed manifolds,
Class. Quant. Grav. 12 (1995), 2249--2274.

\bibitem{IsenbergMaxwellPollack}
{\sc J. Isenberg, D. Maxwell, and D. Pollack,}
A gluing construction for non-vacuum solutions of the Einstein-constraint equations,
Adv. Theor. Math. Phys. 9 (2005), 129--172.

\bibitem{IsenbergMoncrief}
{\sc J. Isenberg and V. Moncrief,}
A set of non-constant mean curvature solutions of the Einstein constraint equations on closed manifolds, 
Class. Quantum Grav. 13 (1996), 1819--1847.

\bibitem{LeeLesourdUnger}
{\sc D.A. Lee, M. Lesourd, and R. Unger,}
Density and positive mass theorems for initial data sets with boundary, 
Commun. Math. Phys. 395 (2022), 643--677. 

\bibitem{LeeParker}
{\sc J. Lee and T. Parker,}
The Yamabe problem,
Bull. Amer. Math. Soc. 17 (1987), 37--81.

\bibitem{LL-first-version}
{\sc B. Le Floch and P.G. LeFloch,}
Optimal localization for the Einstein constraints. 
First version distributed in Dec. 2023 and available at \url{https://arxiv.org/abs/2312.17706v1}.

\bibitem{LL-Letter}
{\sc B. Le Floch and P.G. LeFloch,} 
Optimal shielding for Einstein gravity, 
Class. Quantum Grav. 41 (2024) 13LT02.  

\bibitem{LL-PoincareKornHardy}
{\sc B. Le Floch and P.G. LeFloch,}
Harmonic, radial, and shell stability of the weighted Einstein constraints on the sphere at infinity.
Available at \url{https://arxiv.org/abs/2607.28599}.

\bibitem{LL-next}
{\sc B. Le Floch and P.G. LeFloch,} 
The seeds and silhouettes of optimal gravitational shielding, in preparation.

\bibitem{LeFlochNguyen-preprint} 
{\sc P.G. LeFloch and T.-C. Nguyen,}
The seed-to-solution method for the Einstein constraints and the asymptotic localization problem. 
Distributed in March 2019 and available at \url{https://arxiv.org/abs/1903.00243}.

\bibitem{LeFlochNguyen}
{\sc P.G. LeFloch and T.-C. Nguyen,}
The seed-to-solution method for the Einstein constraints and the asymptotic localization problem, 
J. Funct. Analysis 285 (2023), 110106.

\bibitem{Lichne}
{\sc A. Lichnerowicz,}
L'intégration des équations de la gravitation relativiste et le problème des n corps, 
Jour. Math. Pures Appl. 23 (1944), 37--63.

\bibitem{MaoOhTao}
{\sc Y.-C. Mao, S.-J. Oh and Z.-K. Tao,}
Initial data gluing in the asymptotically flat regime via solution operators with prescribed support properties. 
Available at \url{https://arxiv.org/abs/2308.13031}.

\bibitem{MaoTao}
{\sc Y.-C. Mao and Z.-K. Tao,}
Localized initial data for Einstein equations. 
\url{https://arxiv.org/abs/2210.09437}.

\bibitem{Maxwell-2005}
{\sc D. Maxwell,}
Rough solutions of the Einstein constraint equations on compact manifolds,
J. Hyper. Differ. Equ. 2 (2005), 521--546. 

\bibitem{Maxwell-2021}
{\sc D. Maxwell,}
Initial data in general relativity described by expansion, conformal deformation and drift, 
Comm. Anal. Geom. 29 (2021), 207--281.

\bibitem{Moncrief-1975}
{\sc V. Moncrief,}
Spacetime symmetries and linearization stability of the Einstein equations. I,
J. Math. Phys. 16 (1975), 493--498. 

\bibitem{Morrey-1966}
{\sc C.B. Morrey Jr,}
{\sl Multiple integrals in the calculus of variations,}
Springer Berlin Heidelberg, Berlin, Heidelberg, 1966.

\bibitem{Sansom-2024}
{\sc R. Sansom,} 
The $C^3$-null gluing problem: linear and nonlinear analysis, 
Ann. Henri Poincar\'e (2026).
Available at \url{https://arxiv.org/abs/2408.10859}.

\bibitem{SchoenYau-79} 
{\sc R. Schoen and S.T. Yau,}
On the proof of the positive mass conjecture in general relativity, 
Comm. Math. Phys. 65 (1979), 45--76. 

\bibitem{ShenWan}
{\sc D. Shen and J. Wan,}
Vacuum initial data with minimal decay and borderline decay,
Available at \url{https://arxiv.org/abs/2602.01557}.

\bibitem{Witten-81}
{\sc E. Witten,}
A new proof of the positive energy theorem,
Comm. Math. Phys. 80 (1981), 381--402.

\end{thebibliography}
\end{document}